\documentclass[10pt,english,pra,aps,superscriptaddress,showpacs,longbibliography,showkey,twocolumn]{revtex4-2}
\usepackage{microtype}
\usepackage{graphicx}

\usepackage{lmodern}
\usepackage[colorlinks,citecolor=blue,urlcolor=blue,linkcolor = blue]{hyperref}
\usepackage[multiple]{footmisc}
\usepackage{todonotes}

\usepackage{enumitem}
\usepackage{amssymb,amsfonts,amsthm,mathtools,mathrsfs}
\usepackage{xcolor}

\usepackage{braket}
\usepackage{makecell}
\usepackage{comment}

\usepackage[table]{xcolor}
\usepackage{tabularray}

\theoremstyle{plain}
\newtheorem{theorem}{Theorem}[section]

\newtheorem{lemma}[theorem]{Lemma}

\newtheorem{corollary}[theorem]{Corollary}
\theoremstyle{remark}
\newtheorem{remark}[theorem]{Remark}
\theoremstyle{definition}
\newtheorem{definition}[theorem]{Definition}
\theoremstyle{remark}
\newtheorem{example}[theorem]{Example}
\newcommand{\C}{\mathbb{C}}

\newcommand{\BB}{\mathrm{GF}(2)}

\newcommand{\abs}[1]{\left| #1 \right|}

\renewcommand{\set}[1]{\left\lbrace #1\right\rbrace} %pleasue leave it as it is. The pre-defined command adds extra spaces after the brackets
\newcommand{\symdif}{\mathop{\triangle}}

\newcommand{\SG}{\mathscr{G}}

\newcommand{\lab}{\ell}

\newcommand{\good}{\mathrm{good}}
\newcommand{\pred}{\mathrm{pred}}
\newcommand{\alltoall}{\mathrm{all-to-all}}
\newcommand{\grid}{\mathrm{grid}}

\newcommand{\CNOT}{\mathrm{CX}}

\newcommand{\CXC}{\mathrm{CXC}}
\newcommand{\CXCM}{\mathrm{CXC'}}

\newcommand{\depth}{D_\CNOT}

\newcommand{\G}{\mathcal{G}}

\newcommand{\NCS}{\mathrm{NCS}}

\newcommand{\XS}{\mathrm{DX}}
\newcommand{\XSC}{\mathrm{PTC}}
\newcommand{\XSCM}{\mathrm{PTC'}}
\newcommand{\XSN}{\mathrm{PTN}}
\newcommand{\XSNM}{\mathrm{PTN'}}

\newcommand{\SW}{\mathrm{SW}}
\newcommand{\SWC}{\mathrm{SWC}}

\newcommand{\CL}{\mathrm{CL}}
\newcommand{\concat}[3][]{#2\odot_{#1} #3}
\newcommand{\cat}{\odot}

\newcommand{\CXNAME}{CNOT}
\newcommand{\DXNAME}{DCNOT}
\newcommand{\SWNAME}{SWAP}
\newcommand{\TWINE}{Parity Twine}

\newcommand{\labaction}[2]{#1 #2}
\newcommand{\labseq}[1]{(#1)}

\renewcommand{\ket}[1]{\left\vert #1\right\rangle}
\renewcommand{\SG}{\mathcal{SG}}
\renewcommand{\depth}[1]{\mathrm{depth}\left(#1\right)}
\newcommand{\size}[1]{\mathrm{size}\left(#1\right)}

\renewcommand\mod{\,\mathrm{mod}\,}

\begin{document}

\thanks{These two authors contributed equally}

\title{Connectivity-aware Synthesis of Quantum Algorithms}

\author{Florian Dreier}
\thanks{These two authors contributed equally; \\
c.fleckenstein@parityqc.com \\
f.dreier@parityqc.com}
\affiliation{Institute for Theoretical Physics, University of Innsbruck, A-6020 Innsbruck, Austria}
\affiliation{Parity Quantum Computing GmbH, A-6020 Innsbruck, Austria}
\author{Christoph Fleckenstein}
\thanks{These two authors contributed equally; \\
c.fleckenstein@parityqc.com \\
f.dreier@parityqc.com}
\affiliation{Parity Quantum Computing GmbH, A-6020 Innsbruck, Austria}
\author{Gregor Aigner}
\affiliation{Institute for Theoretical Physics, University of Innsbruck, A-6020 Innsbruck, Austria}
\author{Michael~Fellner}
\affiliation{Institute for Theoretical Physics, University of Innsbruck, A-6020 Innsbruck, Austria}
\affiliation{Parity Quantum Computing GmbH, A-6020 Innsbruck, Austria}
\author{Philipp Aumann}
\affiliation{Parity Quantum Computing GmbH, A-6020 Innsbruck, Austria}
\author{Reinhard Stahn}
\affiliation{Parity Quantum Computing Germany GmbH, 20095 Hamburg, Germany}
\author{Martin Lanthaler}
\affiliation{Institute for Theoretical Physics, University of Innsbruck, A-6020 Innsbruck, Austria}
\affiliation{Parity Quantum Computing GmbH, A-6020 Innsbruck, Austria}
\author{Wolfgang Lechner}
\affiliation{Institute for Theoretical Physics, University of Innsbruck, A-6020 Innsbruck, Austria}
\affiliation{Parity Quantum Computing GmbH, A-6020 Innsbruck, Austria}
\affiliation{Parity Quantum Computing Germany GmbH, 20095 Hamburg, Germany}

\date{\today}
\begin{abstract}
    We present a general method for the implementation of quantum algorithms that optimizes both gate count and circuit depth. Our approach introduces connectivity-adapted CNOT-based building blocks called Parity Twine chains. It outperforms all known state-of-the art methods for implementing prominent quantum algorithms such as the quantum Fourier transform or the Quantum Approximate Optimization Algorithm across a wide range of quantum hardware, including linear, square-grid, hexagonal, ladder and all-to-all connected devices. {\color{black} We show that even moderate increments in connectivity can yield significant efficiency improvements and reach the proven optimum for specific cases. Furthermore, we demonstrate a practical performance advantage of this approach for a wide range of compilation problems and quantum hardware.}
\end{abstract}
\maketitle

\section{Introduction}

Quantum computers \cite{Arute2019, Moses2023, Kim2023, bluvstein_logical_2024, Acharya2024} promise a computational speedup for classically intractable problems ranging from efficient and accurate simulation of quantum systems \cite{Brown2010, Bernien2017} over quantum chemistry \cite{Cao2019} to financial modeling \cite{Orus2019} and solving real world large scale optimization problems \cite{Farhi2000, Farhi2014, Cerezo2021}. Recently, both quantum hardware and software experienced enormous leaps of improvements putting modern quantum hardware in the range of state-of-the art classical computing systems for specialized tasks \cite{Kim2023, Acharya2024}. However, unlike for classical computation, present day quantum resources of near-term quantum hardware systems~\cite{Preskill2018} are sparse which is why efficient quantum algorithms are crucial to leverage quantum advantage.

\begin{figure}[t!]
    \centering
    \includegraphics[width=\columnwidth]{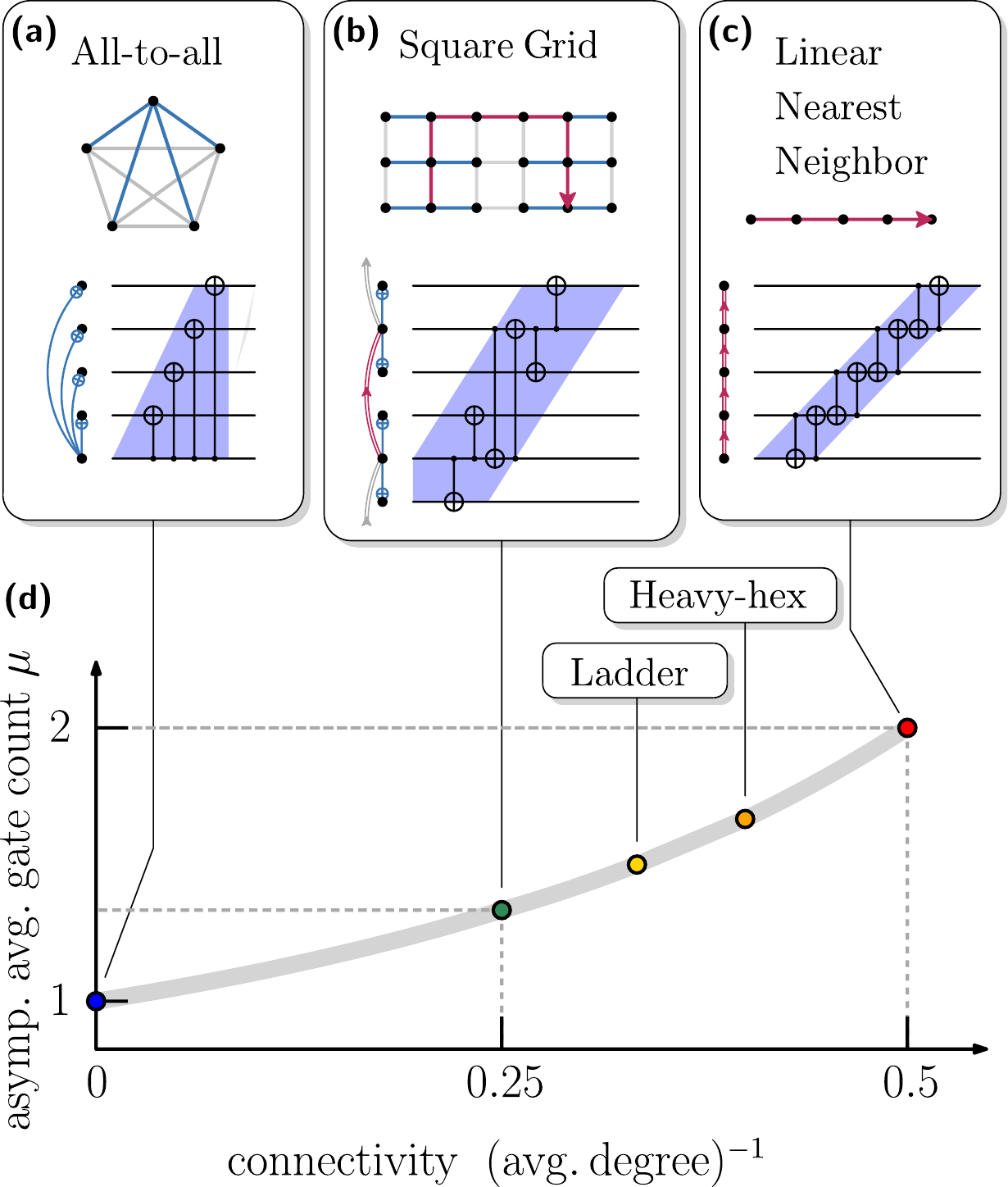}
\caption{The average asymptotic gate count as a function of connectivity ranging from linear nearest neighbor, square grid, heavy hexagon, ladder to all-to-all for algorithms considered in this work. The depicted \TWINE{} chain constitute basic building blocks used throughout the article. The resulting average gate counts per interaction are schematically indicated as a function of the qubit layout connectivity. Connectivity gains lead to significant reductions in gate count approaching the theoretical optimum for all-to-all connected devices.}
    \label{fig:overview}
\end{figure}

The efficiency of quantum algorithms can be measured by various means such as the total execution time \cite{Gyongyosi2020, Zhang2021}, the number of involved quantum operations \cite{Holmes2020}, or the total count of specific costly operations \cite{Gidney2018}. Optimizations, measured according to these metrics, can be distinguished by dividing them into two groups: (i) hardware agnostic optimizations aiming to find more efficient algorithms, and (ii) optimizations improving the compilation of algorithms to fit the specific connectivity constraints of corresponding quantum devices. In fact, current quantum hardware rarely fits the algorithmic connectivity requirements so that ab initio implementations of (multi-body) quantum operations are often infeasible. Although there exist ion-based quantum computers that provide intrinsic all-to-all connectivity~\cite{Schindler2013, Piltz2016}, these are usually limited to a few tens of qubits. In many cases, quantum devices, in particular those where the position of the information carriers is physically fixed, provide two-dimensional grids of qubits with nearest-neighbor connectivity~\cite{Rempfer2024}, like for example square~\cite{Arute2019, Acharya2024}, hexagonal~\cite{Chamberland2020, Kim2023} or octagonal~\cite{Dupont2023} lattices. On such platforms, arbitrary connectivity is typically emulated by costly \SWNAME{} networks~\cite{Crooks2018, Kivlichan2018, Ogorman2019, Hashim2022, Weidenfeller2022, Yuan2024} or shuttling operations~\cite{Kaushal2020, Bluvstein2022, Moses2023, Zwerver2023, Künne2024}, maneuvering quantum states or physical qubits, respectively. Recently, parity label tracking was introduced as a compelling alternative to SWAP networks enabling the redesign of established algorithms using heuristic routing algorithms~\cite{Schmitz2021, van-de-Griend2022} and exact solutions~\cite{Klaver2024}. In particular, the approach introduced in Ref.~\cite{Klaver2024} utilizes aspects of the ZX-calculus~\cite{Cowtan2019} and combines them with the Lechner-Hauke-Zoller (LHZ) architecture for universal quantum computing~\cite{Lechner2015, Fellner2022, Fellner2022app}: Every physical qubit carries a logical parity label which can be altered by Clifford operations. Tracking of these labels throughout a circuit provides a means to understand the corresponding information flow and enables the design of efficient quantum algorithms.

Based on the parity label tracking approach, in this work we focus on optimizations of type (ii) for designing efficient, connectivity-aware quantum algorithms that implement logical many-body operators. {\color{black} We propose a constructive approach for the design of gate-count and -depth optimized algorithms leveraging the device connectivity.
More specifically, we present a generic construction recipe how the connectivity can be used to significantly improve algorithm implementations. While our results are generic,} we showcase the efficiency of our approach in five different exemplary platforms: linear nearest-neighbor (LNN) systems, all-to-all connected systems, square grids, heavy hexagon and ladder architectures. As a metric for the efficiency, we introduce the asymptotic measures $\mu$ and $\nu$, where $\mu$ determines the average number of \CXNAME{} gates necessary to generate a single logical many-body operator and $\nu$ represents a normalized depth for an arbitrarily large number of qubits. To synthesize algorithm implementations, we use connectivity-adapted variations of so-called \emph{\TWINE} chains depicted in Fig.~\ref{fig:overview}(c) for LNN devices and schematically indicated for all-to-all and square grids in (a) and (b). As connectivity of the underlying graph increases, we obtain reductions in the corresponding asymptotic average gate counts $\mu$ [see Fig.~\ref{fig:overview}(d)], where eventually on all-to-all connected devices, our approach yields the provable optimal value ${\mu=1}$. 

{\color{black}
Our formalism has broad applications across quantum computing. 
We demonstrate more efficient implementations of cornerstone quantum algorithms, such as the quantum Fourier transform (QFT). In quantum optimization routines, such as quantum approximate optimization algorithms (QAOA), it significantly improves the encoding of optimization problems. Additionally, in non-equilibrium dynamics and Hamiltonian simulation we use the developed methods to optimize the encoding of long-range many-body Hamiltonians.
} 
Tabs.~\ref{tab:qaoa} and \ref{tab:qft} compare the results for implementations of QAOA for quadratic unconstrained binary optimization (QUBO) problems (Tab.~\ref{tab:qaoa}) and the QFT (Tab.~\ref{tab:qft}) obtained within our formalism to the best known implementations outside our framework. Our approach outperforms alternative implementations on all the investigated platforms in gate count. Simultaneously, circuit depths are, at most, equivalent to the best known alternatives. Interestingly, recursive extensions of \TWINE{} chains allow us to synthesize logical $k$-body operators (with ${k>2}$). Thereby the corresponding circuits for ${k>2}$ inherit the average count from lower order leading to the results summarized in Tab.~\ref{tab:average_count_depth_kbody}. 

{\color{black}
To facilitate our theoretical findings, we investigate the performance of our methods on noisy quantum hardware. To this end we compare our implementations to alternative implementations in the presence of noise. Our analysis suggests a performance advantage of our implementations for the entire parameter regime spanned by finite qubit lifetime and two-qubit gate errors. This result is supported by series of numerical simulations as well as experimental runs on contemporary quantum hardware. Leveraging connectivity with Parity Twine we find performance advantages of up to several order of magnitude in the obtained fidelities for experiments performed on superconducting square grid devices.
}

The article is organized as follows. In Sec.~\ref{sec:label_tracking} we introduce and recapitulate the label-tracking formalism used throughout the article. In Sec.~\ref{sec:fundamental_bounds} we introduce the quantities $\mu$ and $\nu$ and investigate fundamental properties regarding gate count and depth regarding the algorithms of interest. This is followed by Sec.~\ref{sec:LNN_kbody} where we discuss in depth the different building blocks of our approach exemplified on LNN devices. In Secs.~\ref{sec:all-to-all} and Sec.~\ref{sec:2d_architectures} we generalize the concepts of Sec.~\ref{sec:LNN_kbody} to a wide range of qubit-connectivity graphs and explicitly discuss complete graphs (Sec.~\ref{sec:all-to-all}) as well as a number of planar graphs such as square grids and heavy hexagons (Sec.~\ref{sec:2d_architectures}).
{\color{black} This is followed by Sec.~\ref{app:dept_optimization}, where we discuss techniques to reduce the circuit depth. In Sec.~\ref{sec:QAOA} we apply our approach to optimization problems within the QAOA framework and we compare our exact solutions to state-of-the-art compilation tools. This is followed by Sec.~\ref{sec:QFT} where we outline the implementation of the QFT and Sec.~\ref{sec:hamiltonian_simulation} where we discuss applications for Hamiltonian simulation. Finally, in Sec.~\ref{sec:performance_on_noisy_hardware} we investigate the performance of our circuits on noisy quantum hardware. To this end we employ theoretical models, numerical simulations and we conduct experiments on quantum hardware}. We conclude with a summary and an outlook in Sec.~\ref{sec:conclusion}.

\begin{table}[t!]
\centering
\resizebox{\columnwidth}{!}{
\begin{tblr}{
colspec={|l |cccc | cccc|}, rowsep=1.5pt, colsep=2.5pt, hline{1,2,3,8},
cells={valign=m,halign=c},
cell{1}{1} = {r=2}{c},
cell{1}{2,6} = {c=4}{c},
}

Connectivity graph    & QAOA ours & & & & QAOA best known~\cite{Weidenfeller2022}  & & & \\ 
        & count & depth & $\mu$ & $\nu$ & count & depth & $\mu$ & $\nu$\\  
    LNN  & $n^2$ & $2n\left(1\! +\!\frac{1}{p}\right)$ & 2 & $2\left(1 \!+\!\frac{1}{p}\right)$ & $\frac{3}{2}n^2$ & $3n$ & 3 &3\\ 
    heavy hexagon & $\frac{5}{6}n^2$ & $\frac{5}{2}n\left(1 \!+\!\frac{1}{p}\right)$ & $\frac{5}{3}$ & $\frac{5}{2}\left(1 \!+\!\frac{1}{p}\right)$ & $\frac{3}{2}n^2$ & $3n$ & 3 &3\\ 
    ladder & $\frac{3}{4}n^2$ & $\frac{3}{2}n \left(1\! +\!\frac{1}{p}\right)$ & $ \frac{3}{2}$ & $\frac{3}{2}\left(1 \!+\!\frac{1}{p}\right)$ & $\frac{3}{2}n^2$ & $3n$ & 3 &3\\
    square grid & $\frac{2}{3} n^2$ & $3n\left(1\! +\!\frac{1}{p}\right)$ & $\frac{4}{3}$ & $3\left(1 \!+\!\frac{1}{p}\right)$ & $\frac{3}{2}n^2$ & $3n$ & 3 & 3\\
    all-to-all & $\frac{1}{2} n^2$ & $n\left(1 \!+\!\frac{1}{p}\right)$ & 1 & $1 \!+\!\frac{1}{p}$ & $n^2$ & $2n$ & 2 & 2  
\end{tblr}
}
    \caption{Summary of the leading order terms of gate count and depth as well as the asymptotic average gate count $\mu$, defined in Eq.~\eqref{eq:asymptotic_average_count}, and the asymptotic normalized depth $\nu$, defined in Eq.~\eqref{eq:asymptotic_avaerage_depth}, of different implementations of QAOA (with $p$ QAOA cycles) for various connectivity graphs. All provided values refer to count and depth per QAOA cycle. Our results refer to an odd number of QAOA cycles $p$. For even $p$ we obtain slightly lower depth (cf.~Sec.~\ref{sec:QAOA}). We compare our results to the best known algorithms. As noted in Ref.~\cite{Weidenfeller2022} a linear chain represents the best known approach also for two-dimensional connectivity graphs. For heavy hexagon architectures, where no Hamiltonian path exists, this implies that not all qubits are used. Instead, the provided numbers are only valid for a sub-graph of the heavy hexagon layout. This is also the case in our approach, although with a different sub-graph (compare App.~\ref{app:heavy-hex} for details).}
    \label{tab:qaoa}
\end{table}

\begin{table*}[t!]
\centering
\begin{tblr}{
colspec={|l |cccc | cccc| cccc|}, rowsep=1.5pt, colsep=2.5pt, hline{1,2,3,8},
cells={valign=m,halign=c},
cell{1}{1} = {r=2}{c},
cell{1}{2} = {c=4}{c},
cell{1}{6} = {c=4}{c},
cell{1}{10} = {c=4}{c},
hspan=even,
}

Connectivity graph    & QFT ours & & & & QFT best known  (count-optimized)~\cite{Holmes2020} & & & & QFT best known (depth-optimized)~\cite{Park2023} & & & \\ 
        & count & depth & $\mu$ & $\nu$ & count & depth & $\mu$ & $\nu$ & count & depth & $\mu$ & $\nu$ \\  
    LNN  & $n^2$ & $4n$ & 2 & 4 & $n^2$ & $n^2$ & 2 & $n$ &  $\frac{3}{2}n^2$ & $6n$ & 3 & 6 \\ 
    heavy hexagon & $\frac{5}{6}n^2$ & $5n$ & $1 + \frac{2}{3}$ & 5 & $n^2$ & $n^2$ & 2 & $n$ &  $\frac{3}{2}n^2$ & $6n$ & 3 & 6 \\ 
    ladder & $\frac{3}{4}n^2$ & $3n$ & $1 + \frac{1}{2}$ & 3 & $n^2$ & $n^2$ & 2 & $n$ &  $\frac{3}{2}n^2$ & $6n$ & 3 & 6 \\ 
    square grid & $\frac{2}{3} n^2$ & $6n$ & $1+\frac{1}{3}$ & $6$ & $n^2$ & $n^2$ & 2 & $n$ &  $\frac{3}{2}n^2$ & $6n$ & 3 & 6 \\
    all-to-all & $\frac{1}{2} n^2$ & $2n$ & 1 & 2 & $n^2$ & $4n$ & 2 & 4 & $n^2$ & $4n$ & 2 & 4 \\ 
\end{tblr}
    \caption{Summary of the leading order terms of gate count and depth as well as the asymptotic average gate count $\mu$, defined in Eq.~\eqref{eq:asymptotic_average_count}, and the asymptotic normalized depth $\nu$, defined in Eq.~\eqref{eq:asymptotic_avaerage_depth}, of different implementations of the QFT for various connectivity graphs. We compare our results to the best known algorithms which amounts to the use of a linear chain approach for all the investigated connectivity graphs \cite{Weidenfeller2022} (compare caption of Tab.~\ref{tab:qaoa}).}
    \label{tab:qft}
\end{table*}

\begin{table}[t!]
\centering

\begin{tblr}{
colspec={|l |cc|}, rowsep=1.5pt, colsep=2.5pt, hline{1,2,3,8},
cells={valign=m,halign=c},
cell{1}{1} = {r=2}{c},
cell{1}{2} = {c=2}{c},
cell{5}{1} = {r=2}{c},
cell{5}{2} = {r=2}{c},
cell{5}{3} = {r=2}{c},
hspan=even
}
Connectivity graph   & ours for $k>2$ &\\ 
         & $\mu$& $\nu$\\  
    LNN  & 2 & $k$ \\
    heavy hexagon & $1 + \frac{2}{3}$ & $\frac{5}{3}k$ \\
    square grid & $1+\frac{1}{3}$ & {$\frac{25}{12}k$ for $k=3$ \\ $\frac{7}{3}k$ for $k>3$} \\
    & & \\
    all-to-all & 1 & $k/2$\\
\end{tblr}

    \caption{Summary of the average \CXNAME{} count $\mu$, defined in Eq.~\eqref{eq:asymptotic_average_count}, and the normalized depth $\nu$, defined in Eq.~\eqref{eq:asymptotic_avaerage_depth}, for generator circuits $\mathcal{G}_k$ (see Sec.~\ref{sec:special_4body_generator}) with $k>2$ for the different connectivity graphs investigated.}
    \label{tab:average_count_depth_kbody}
\end{table}

\section{Parity label tracking and $k$-body operators \label{sec:label_tracking}}
Quantum computers inherently rely on entangling gates to systematically distribute information and perform computations. A typical (Clifford) gate utilized for this purpose is the \CXNAME{} gate. The logical action of a \CXNAME{} can be understood as encoding the parity information of the two qubits involved, the target and the control qubit, on one qubit. In the $z$-basis, a \CXNAME{} encodes the $z$ parity information on the target qubit, while in the $x$ basis a \CXNAME{} gate encodes the $x$ parity information on the control qubit. To capture this abstract notion of the action of a \CXNAME{} on a set of physical qubits ${Q=\{1, ... ,n\}}$, we attribute to each qubit ${j \in Q}$ a logical parity label $\lab_j$. If not explicitly stated differently, we interpret the labels $\lab_j$ as the logical $z$ parity encoded on the qubit $j$. In this language the action of the \CXNAME{} gate $\CNOT_{c,t}$ on the sequence of labels  ${\lab = \labseq{\lab_1, ..., \lab_n}}$ is given by 
\begin{equation}
    \label{eq:CNOT_action}
    \labaction{\labseq{\lab_1, ..., \lab_n}}{\CNOT_{c,t}} = \labseq{\lab_1,..., \lab_{t-1},\lab_t \symdif \lab_c, ...,\lab_n},
\end{equation}
where $\symdif$ denotes the symmetric difference operator. We adopt the notation that operators act from the right onto sequences of labels. 

In the following, we denote the symmetric difference of two labels with shorthand notation $\lab_{c}\lab_{t}\coloneqq\lab_{c} \symdif \lab_{t}$ as well as we waive the set notation. For example, if the logical parity label $\lab_i$ contains the labels $p$ and $q$ we denote $\lab_i= pq $. Moreover, if not explicitly stated differently, we assume an initial state with logical parity labels ${\lab_j = j}$ for all ${j \in Q}$.

{\color{black}
A description of quantum circuits via gates applied to sets of parity labels provides a formal means to express the action of quantum gates on arbitrary states in the computational basis. For example, in terms of computational basis states, Eq.~\eqref{eq:CNOT_action} corresponds to
\begin{align*}
    \CNOT_{c,t}\vert a(\lab_1)\!\ldots &a(\lab_n)\rangle
    \\
    &=\ket{a(\lab_1)\ldots a(\lab_{t-1}) \, a(\lab_c)\oplus a(\lab_t)\ldots a(\lab_n)}\\
    &=\ket{a(\lab_1)\ldots a(\lab_{t-1}) \, a(\lab_c\symdif\lab_t)\ldots a(\lab_n)}
\end{align*}
where we use $a(\lab) = \oplus_{i \in \lab}a_i$
 for any computational basis state with entries $a_1,\ldots,a_n\in\set{0,1}$ and parity label $\lab$.}

The logical parity labels together with the symmetric difference operation as an addition form a vector space $V$ over $\BB$ where the empty set is the zero vector. Interpreting the $\CNOT_{c,t}$ gate as an operator on the $n$-fold direct product $V^n$, the $\CNOT_{c,t}$ operator is a bijective operator which maps label sequences that form a basis in $V$ to another label sequences that are again a basis in $V$. Hence, starting with a label sequence which is a basis in $V$, the logical parity labels form a basis at any moment within a \CXNAME{} circuit. In particular, none of the labels will be empty. Moreover, transforming the $\CNOT_{c,t}$ operator to an operator on $\BB^{n\times n}$, $\CNOT_{c,t}$ can be viewed as a Gauss operation performing an addition of rows on matrices.

We define a \emph{circuit} as a sequence of time steps - called \emph{moments} - enumerated from $1$ to $d$. Each of these moments contains a number of non-overlapping gates where we require for non-empty circuits at least one gate in the first and in the last moment. The \emph{depth} of the circuit is $d$ and the \emph{gate count} or \emph{size} is the number of gates contained in the circuit. To emphasize that a circuit $C$ acts on $n$ qubits we write $C^{(n)}$ and if the circuits only act on a subset of qubits $\{p,...,q\}$ we denote it with $C^{(p,q)}$. Moreover, $\concat{A}{B}$ denotes (a possibly shifted) concatenation of two circuits $A$ and $B$, and $C^\dagger$ the \emph{adjoint circuit} of some unitary circuit $C$. Circuits may be used to map between two label sets $\lab$ and $\lab'$: if $\lab$ and $\lab'$ both form a basis in $V$, then there exists a Clifford circuit $C$ so that $\labaction{\lab}{C} = \lab'$ and $\labaction{\lab'}{C^{\dagger}} = \lab$. 

In general, circuits can be described by specifying the unitaries to be implemented without detailing the actual implementation in terms of gates available on a quantum computer. However, this abstract description omits details that could be crucial for optimizations. Therefore, it is important to consider circuits composed solely of gates native to the quantum computer. Moreover,  the fact that different circuits with various gate count and depths can produce the same unitary operator highlights another key reason why distinguishing between circuits and their corresponding unitary operators is essential. Consider for example a logical many-body rotation ${\exp( -i \alpha \prod_{j \in \lab} Z_j)}$: A concrete implementation on quantum hardware with  access to \CXNAME{} and single-body rotation gates is given by

\begin{equation}
    \label{eq:logical_physical}
    \exp \bigg( -i \alpha \prod_{j \in \lab} Z_j  \bigg)= \concat{C}{\concat{\mathrm{RZ}_i(\alpha)}{C^\dagger}},
\end{equation}
where we first use a \CXNAME{} circuit $C$ to encode the label $\lab$ on, say, qubit $i$. This is followed by a single-body $z$ rotation applied to qubit $i$, $\mathrm{RZ}_i(\alpha)=\exp(-i\alpha Z_i)$, which in effect implements a logical many-body rotation associated with all qubits ${j\in \lab}$. Eventually, to regain the initial label sequence we use a decoding circuit $C^{\dagger}$. Note that in Eq.~\eqref{eq:logical_physical}, the right-hand side is interpreted as the induced operator of the circuit composed of $C$, $\mathrm{RZ}_i(\alpha)$ and $C^{\dagger}$.
We emphasize that RZ gates act as the identity on the $z$ parity labels which follows from the fact that RZ gates commute with each other. This in turn implies that the logical effect of a single RZ gate does not depend on other RZ gates in a circuit composed of CNOT and RZ gates.

Similarly, a series of logical operators can be encoded with
\begin{multline*}
   \prod_{ \lab \in L } \exp\bigg(-i \alpha_\lab \prod_{j \in \lab} Z_j \bigg)\\
   = \bigodot_{ (\lab ,k)\in S }\left(C_{\lab, k} \cat \mathrm{RZ}_k(\alpha_\lab) \cat C_{\lab,k}^\dagger \right),
\end{multline*}
where $L$ is a set of desired labels and $S$ is set of tuples containing the labels of $L$ paired with physical qubits $k$. The circuit $C_{\lab,k}$ produces the label $\lab$ at qubit $k$: The output label set just after $C_{\lab,k}$ contains label $\lab$ on physical qubit k. Note that the left hand side has no dependence on the physical qubits (specified only on the right-hand side by $k$). This reflects the fact that a logical operator is not aware of the device on which eventually it is implemented. 

The inverse Clifford circuits $C_{\lab,k}^\dagger$ are used to decode the produced logical parity label to return to the initial label sequence after each application of a physical non-Clifford gate. This is where parity label tracking can yield an advantage: instead of decoding after every logical non-Clifford gate, the parity labels can be tracked and the next parity label can be encoded using potentially shorter \CXNAME{} circuits. Eventually, the decoding is applied once after the last physical non-Clifford gate
\begin{multline}
    \label{eq:sequence_of_many_body_operators}
   \prod_{ \lab \in L } \exp\bigg(-i \alpha_\lab \prod_{j \in \lab} Z_j \bigg)\\
   =\bigg(\bigodot_{p = 1}^{\abs{L} } C_p \cat \mathrm{RZ}_{k(p)}(\alpha_{\lab(p)})\bigg) \cat C_{\mathrm{clean}}
\end{multline}
where $\bigodot_{i=1}^{p}C_i$ generates label $\lab(p)$ at qubit $k(p)$. The $\lab(p)$ enumerate the elements of $L$. Moreover,

\begin{equation*}
    C_{\mathrm{clean}}=  \bigodot_{p=1}^{\abs{L}} C_{\abs{L}-p+1}^\dagger.
\end{equation*}

Due to potential cancellations, the actual size and depth of the circuit $C_{\mathrm{clean}}$ can be significantly smaller than the above formula seems to imply.
Furthermore, as long as the sequence of logical operators to encode [on the left-hand side of Eq.~\eqref{eq:sequence_of_many_body_operators}] forms a pairwise commuting set, their order is irrelevant. However, depending on the order of physical operators [right-hand side of \eqref{eq:sequence_of_many_body_operators}], the corresponding \CXNAME{} circuits $C_p$ can drastically differ in depth and count. Optimizing the order of physical operators with respect to the total gate count or depth of the corresponding \CXNAME{} circuits $C_p$ is in general a NP-hard problem akin to the traveling salesman problem.

{\color{black}
The notion of utilizing labels to track the action of quantum gates can be extended to the entire Clifford group. Then, bookkeeping requires two labels and one $\mathcal{Z}_4$ phase per physical qubit \cite{Klaver2025}. For the standard generators of the Clifford group (Hadamard, phase gate and \CXNAME{} gate) there exists a simple set of label update rules. Since for the most part of this paper we are interested manipulations of the labels in the computational $z$-basis, we refer the reader to Ref.~\cite{Klaver2025} for a detailed discussion on parity label tracking via Clifford tableaus. A formal mathematical approach to parity labels, their properties and manipulation via circuits tailored to this work is provided in App.~\ref{sec:notation_auxiliary_results}.
}
\section{Average gate count and normalized depth\label{sec:fundamental_bounds}}
In this article, we are mainly interested in $z$-diagonal non-Clifford gates, which leave the logical parity labels in the $z$ basis unchanged. As long as we maintain the label tracking, relevant non-Clifford gates can always be inserted at corresponding moments of the circuit (i.e., just after a circuit $C_p$ produced a logical label of interest). Thus, in the following we focus solely on the Clifford parts that \textit{generate} the desired labels. Furthermore, throughout this article, we assume that Clifford circuits constitute \CXNAME{} circuits.

Consider $n$ qubits and a set $L^{(n)}$ of non-trivial labels over these qubits. We say that a \CXNAME{} circuit $C^{(n)}$ \textit{generates} the labels $L^{(n)}$ from a starting configuration $\lab  = (\lab_1,\ldots,\lab_{n})$ if each label in $L^{(n)}$ appears at least in one moment of the circuit $C^{(n)}$. Note that if $C^{(n)}$ generates the labels $L^{(n)}$ from $\lab $ it also generates any subset of $L^{(n)}$. For such a generator circuit $C^{(n)}$ of the label set $L^{(n)}$ we define the \textit{average \CXNAME{} count} as 
$$\mu_n(C^{(n)}, L^{(n)})= \size{C^{(n)}} / \abs{L^{(n)}}$$ 
with the total \CXNAME{} count $\size{C^{(n)}}$ of $C^{(n)}$. Furthermore, we define the \textit{asymptotic average \CXNAME{} count} of some generator circuit family $C$ and some label family $L$, where $C$ contains circuits $C^{(n)}$ and $L$ some label sets $L^{(n)}$ on $n$ qubits, as
\begin{equation}
\label{eq:asymptotic_average_count}
    \mu(C, L)=\liminf_{n\to\infty} \mu_n(C^{(n)}, L^{(n)}).
\end{equation}
A similar measure can be defined for the \CXNAME{} depth of the generator circuit family $C$. We recognize that every moment of $C^{(n)}$ allows at most $\lfloor n/2 \rfloor$ \CXNAME{} gates applied in parallel. Thus, we define 
\begin{equation}
\label{eq:asymptotic_avaerage_depth}
\nu(C,L)=\liminf_{n\to\infty}\frac{\depth{C^{(n)}} n}{2\vert L^{(n)} \vert },
\end{equation}
where $\depth{C^{(n)}}$ denotes the depth of the \CXNAME{} circuit $C^{(n)}$ (see App.~\ref{def:average_cnot_count_norm_depth}). We leave out single-body gates in the gate count and depth analysis since they typically require significantly shorter execution times than two-body gates~\cite{Cheng2023} and are not subject of the applied optimizations.

From Eq.~\eqref{eq:sequence_of_many_body_operators} it is evident that every logical many-body operator can be associated with one pair of Clifford \CXNAME{} circuits combined with a physical single-body (non-Clifford) gate. The encoding of many-body operators thus requires at least one \CXNAME{} gate per many-body operator since we demand all many-body operators and the associated labels in $L^{(n)}$ to be different and generating logical parity labels can only happen via application of \CXNAME{} gates (or related two-body entangling gates). Thus, we conclude that for any generator circuit $C^{(n)}$ of $L^{(n)}$ we expect ${\mu_n(C^{(n)}, L^{(n)})\geq 1}$ (see also App.~\ref{lem:trivial-size-bound}). 

For efficient generator circuits with low average gate count of large enough label sets, expressed in the property ${\lim_{n\rightarrow \infty}\vert L^{(n)} \vert / n = \infty}$, we can draw some fundamental conclusions. For that, let us assume we find a generator circuit $C^{(n)}$ of $L^{(n)}$ with ${\mu_n(C^{(n)}, L^{(n)}) = 1 + \gamma}$, i.e., using only ${(1 + \gamma) \vert L^{(n)}\vert}$ \CXNAME{} gates. Then, using generic counting arguments, one can demonstrate that at least ${(1-\gamma)\vert L^{(n)}\vert}$ of all ${(1 + \gamma) \vert L^{(n)}\vert}$ \CXNAME{} gates of $C^{(n)}$ transform labels ${\ell\in L^{(n)}}$ into other labels ${\ell' \in L^{(n)}}$, where $\ell'$ has not been generated in prior moments of $C^{(n)}$. For small ${\gamma>0}$ almost all gates of $C^{(n)}$ have to act in this way. By this property, we call \CXNAME{} gates implementing the above mapping \emph{typical} \CXNAME{} gates. With this in mind we can employ further counting arguments for LNN devices to find lower bounds for possible values of $\gamma$ in the limit $n\to\infty$. In particular, in case the label set $L^{(n)}$ consists of labels which cannot be pairwise combined to form other labels of $L^{(n)}$, i.e., if ${\lab_i\lab_j \notin L^{(n)}}$ with ${\lab_i, \lab_j \in L^{(n)}}$, then ${\gamma \geq \frac{1}{9}}$ for ${n\to\infty}$ holds. This is notable since label sets with this property appear frequently in the subsequent sections. As an example, consider the set of all three-body labels for $n$ qubits on a LNN device: combinations of two three-body labels cannot yield another three-body label. Similar considerations hold for all odd-bodied label sets or label sets where all labels share (at least) one constituent. In all these cases, an increased theoretical lower bound of ${\mu \geq 1 + \frac{1}{9}}$ holds. A detailed proof of the concepts outlined above can be found in App.~\ref{sec:proof-impossibility-theorem}.

\section{$k$-body generators for LNN connectivity}
\label{sec:LNN_kbody}
{\color{black}
In this section, we present the theoretical framework to construct multi-body generator circuits. Sec.~\ref{sec:two_body} outlines the basic concepts that are generalized subsequently in more technical sections.
}
\subsection{Generator circuits for two-body operators \label{sec:two_body}}
In this section we recapitulate the circuit construction introduced in Ref.~\cite{Klaver2024}. There, the authors present a \emph{scalable} quantum circuit for LNN chains that generates \emph{all} two-body labels with significantly reduced gate count and depth compared to established algorithms. When combined with corresponding physical non-Clifford single-qubit operators, this approach can be used to implement all relevant logical two-body operators diagonal in the $z$ basis. As the authors demonstrate, this can be useful for algorithms like QAOA~\cite{Farhi2014} or QFT. 

The most prominent building block of the circuit outlined in Ref.~\cite{Klaver2024} is a \emph{double-controlled NOT} (\emph{\DXNAME}) gate~\cite{ColLinNoa01}, ${\XS_{c,t}=\concat{\CNOT_{c, t}}{\SW_{c,t}}}$, which combines a \CXNAME{} gate with a corresponding \SWNAME{} gate acting on the same qubits. Using trivial operator identities, we observe for the induced operators
\begin{align*}
    \XS_{c,t}&=\left(\CNOT_{c, t}\cat (\CNOT_{c,t}\cat \CNOT_{t,c}\cat \CNOT_{c,t})\right)\\
    &=\left(\CNOT_{t,c}\cat \CNOT_{c,t}\right).
\end{align*}
The $\XS_{c,t}$ gate is depicted in Fig.~\ref{fig:2To3BodyGenerator}(a). Note that the \DXNAME~gate shares aspects with an i\SWNAME{} gate, native to many state-of-the-art hardware platforms~\cite{Sung2021, Wei2024}.

\begin{figure}
    \centering
    \includegraphics[width=\columnwidth]{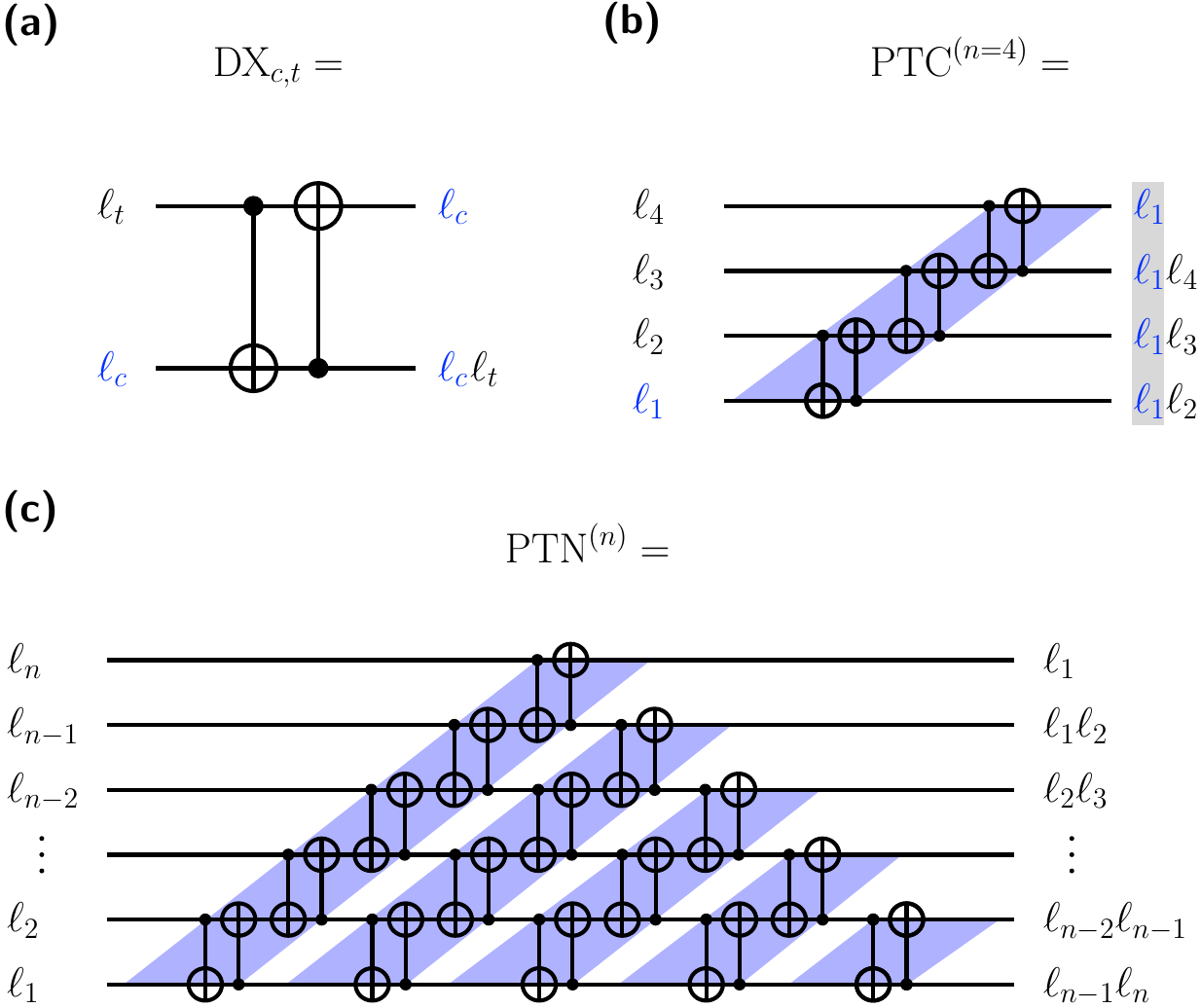}
    \caption{(a) Schematic of a \DXNAME{} gate. (b) Concatenations of \DXNAME{} gates forming a \TWINE{} chain. (c) Concatenated \TWINE{} chains acting on a decreasing number of qubits and forming the \TWINE{} network $\XSN^{(n)}$, generating all possible two-body terms.}
    \label{fig:2To3BodyGenerator}
\end{figure}

Logically, the action of a $\XS_{c,t}$ gate can be understood as swapping the parity information of the two qubits involved and a subsequent encoding of the combined parity on qubit $c$. These two features are essential for an efficient encoding of many-body operators on the LNN chain, where the location of logical parity information on the devices plays a pivotal role. Concatenations of \DXNAME{} gates, which we call a \emph{\TWINE{} chain}, allow to systematically distribute logical parity information from one physical qubit to all other qubits
\begin{equation}
\label{eq:cswap_chain}
    \left( \lab_1, \ldots ,\lab_n \right) \XSC^{(n)} = \left( (\lab_2, \ldots, \lab_n)\lab_1, \lab_1 \right) 
\end{equation}
where the \TWINE{} chain is given by
\begin{equation}
    \XSC^{(n)} = \concat{\XS_{1, 2}}{\concat{\XS_{2, 3}}{\concat{\ldots}{ \XS_{n-1, n}}}}.
\end{equation}
Fig.~\ref{fig:2To3BodyGenerator}(b) schematically depicts a \TWINE{} chain for $n=4$ qubits. Here and in the following, we use the shorthand notation ${\labseq{(\lab_2,\ldots,\lab_n)\lab_1, \lab_1} =\labseq{\lab_2\lab_1, \lab_3\lab_1, \ldots, \lab_n\lab_1, \lab_1}}$. Note that a simple \CXNAME{} chain is not able to accomplish the same encoding. \CXNAME{} chains either encode an increasing number of logical labels or just pair the logical labels of neighboring qubits when applying their adjoint operation. In contrast, $\XSC^{(1,n)}$ encodes the logical parity information originally located on the first physical qubit with the logical parity information of \emph{all} other physical qubits. Consequently, starting from the label sequence $\labseq{\lab_1, \dots, \lab_n}$, $\XSC^{(n)}$ generates all logical two-body parity labels that include the label $\lab_1$. Interestingly, we obtain the generator circuit for all two-body labels containing $\lab_1$ and all two-body labels containing $\lab_2$ from a concatenation of $\XSC$ circuits, since
\begin{equation*}
    \labaction{\labseq{(\lab_2, \ldots, \lab_n)\lab_1, \lab_1}}{\XSC^{(1, n-1)}} = \labseq{ (\lab_3, \ldots, \lab_n)\lab_2, \lab_2\lab_1 ,\lab_1} .
\end{equation*}
Similarly, we can construct the generator circuit of all two-body labels just from repeated application of \TWINE{} chains on a shrinking set of qubits
\begin{multline}
\label{eq:cswap_network}
    \XSN^{(n)} = \concat{\XSC^{(n)}}{\concat{\XSC^{(1,n-1)}}{\concat{\dots}{\XSC^{(1,2)}}}}.
\end{multline}
We call the circuit $\XSN^{(n)}$ a \emph{\TWINE{} network of type I}. Application of $\XSN^{(n)}$ on a label set $(\lab_1, \dots, \lab_n)$ yields 
\begin{equation}
\label{eq:2BodyGenerator}
 (\lab_1, \dots, \lab_n) \XSN^{(n)} = \left(\lab_n\lab_{n-1}, \lab_{n-1}\lab_{n-2},\dots,\lab_2\lab_1, \lab_1\right). \nonumber
 \end{equation}
Fig.~\ref{fig:2To3BodyGenerator}(c) illustrates the building block $\XSN^{(n)}$. As outlined in Ref.~\cite{Klaver2024}, the circuit $\XSN^{(n)}$ can be interpreted as a spatio-temporal version of the LHZ encoding where certain spanning tree lines of an extended LHZ code can be mapped to a moment of the circuit $\XSN^{(n)}$.  

Following $\XSN^{(n)}$, a \CXNAME{} chain $\overline{\CXC}^{(n)} = \CNOT_{n,n-1} \cat \CNOT_{n-1, n-2} \cat \dots \cat \CNOT_{2,1}$ acts as a decoding circuit and may be used to restore the original labels (in reversed order)
\begin{equation*}
\label{eq:reverted_cnot_chain}
    \labaction{\labseq{\lab_n\lab_{n-1}, \lab_{n-1}\lab_{n-2},\dots,\lab_2\lab_1, \lab_1}}{\overline{\CXC}^{(n)}} = \labseq{\lab_n, \dots, \lab_1}.
\end{equation*}
With $\overline{O}$ we denote the \emph{reversed circuit} of a circuit $O$, which simply maps each $\CNOT_{i , j}$ to $\CNOT_{n-i+1 , n-j+1}$ (see App.~\ref{def:reversed_circuit}). 
The concatenation 
\begin{equation}
\label{eq:clean_two_body generator}
    \mathcal{G}_2^{(n)} = \XSN^{(n)} \cat \overline{\CXC}^{(n)}
\end{equation}
constitutes a generator circuit for all two-body labels which maps an input sequence of single-body labels to an output of single-body labels. In the subsequent sections we call generator circuits with this property \emph{clean} generator circuits.
For a more detailed analysis of the above circuits and their properties, we refer to App.~\ref{subsec:two-body_generators}. 

\subsection{Generator circuits for three-body operators \label{sec:three_body}}
Our strategy for a circuit that generates all three-body logical parity labels $\mathcal{L}^{(n)}_3$ builds on the two-body generator circuit. Here, the central idea is to encode an additional label $\lab_s$ repeatedly on physical qubit $2$ just before every application of a \TWINE{} chain (starting at physical qubit 2). The corresponding modified building block is thus a \emph{modified \TWINE{} chain} 
\begin{eqnarray}
    \XSCM^{(n)}=\CNOT_{1,2}\cat\XSC^{(2,n)}. \nonumber
\end{eqnarray}
Concatenations of modified \TWINE{} chains yield a \emph{modified  \TWINE{} network of type I}
\begin{eqnarray}
\label{eq:controlled_cswap_network}
    \XSNM^{(n)} &=& \XSCM^{(n)} \cat \XSCM^{(1,n-1)} \cat \dots \\ &\dots & \cat \nonumber ~\XSCM^{(1,3)}\cat\CNOT_{1,2}.
\end{eqnarray}
The application of $\XSNM^{(n)}$ on the label sequence $\labseq{\lab_s, \lab_1, \dots, \lab_{n-1}}$ yields
\begin{multline}
    \label{eq:special_three_body_output}
    \labaction{\labseq{\lab_s, \lab_1 \dots, \lab_{n-1}}}{\XSNM^{(n)}}\\
    =\labseq{\lab_s, \lab_{n-1}\lab_{n-2}, \lab_{n-2}\lab_{n-3}, \ldots, \lab_{2}\lab_1, \lab_1\lab_s},
\end{multline}
where $\XSNM^{(n)}$ is generating the label set $$\set{\lab_s\lab_i\lab_j\mid 1\leq i<j\leq n-1}$$ from the label sequence $\labseq{\lab_s, \lab_1 \ldots, \lab_{n-1}}$, i.e.~all logical three-body labels that include $\lab_s$. By this property the label $\lab_s$ is a \emph{special label} wherefore we also call the modified \TWINE{} network of type I a \emph{special three-body generator} (in App.~\ref{def:special_k-body_generator}, we give a detailed definition of arbitrary special $k$-body generators).

While $\XSNM^{(n)}$ in Eq.~\eqref{eq:special_three_body_output} generates all labels containing the special label $\lab_s$, the output label set has an inconvenient form as the special label $\lab_s$ is encoded on physical qubit $1$ and $n$. For a LNN chain with periodic boundary conditions, $\lab_s$ can be annihilated on qubit $n$ using a single \CXNAME{} gate. However, given open boundary conditions, extra SWAP operations would be required to achieve the same. This issue is circumvented recognizing that the circuit $\XSNM^{(n)}$ constitutes a special three-body generator not only from the label sequence $\labseq{\lab_s, \lab_1, \ldots,\lab_{n-1}}$ but also from the label sequence $\labseq{\lab_s, (\lab_1 \ldots \lab_{n-1})\lab_s}$. Then, application of $\XSNM^{(n)}$ yields
\begin{multline}
\label{eq:special_three_body_from_encoded}
     \labaction{\labseq{\lab_s, (\lab_1,\ldots,  \lab_{n-1})\lab_s}}{\XSNM^{(n)}}\\
     =\labseq{\lab_s, \lab_{n-1}\lab_{n-2},\lab_{n-2}\lab_{n-3},\ldots,\lab_2\lab_1, \lab_1}.
\end{multline}

\begin{figure}
    \centering
    \includegraphics[width=\columnwidth]{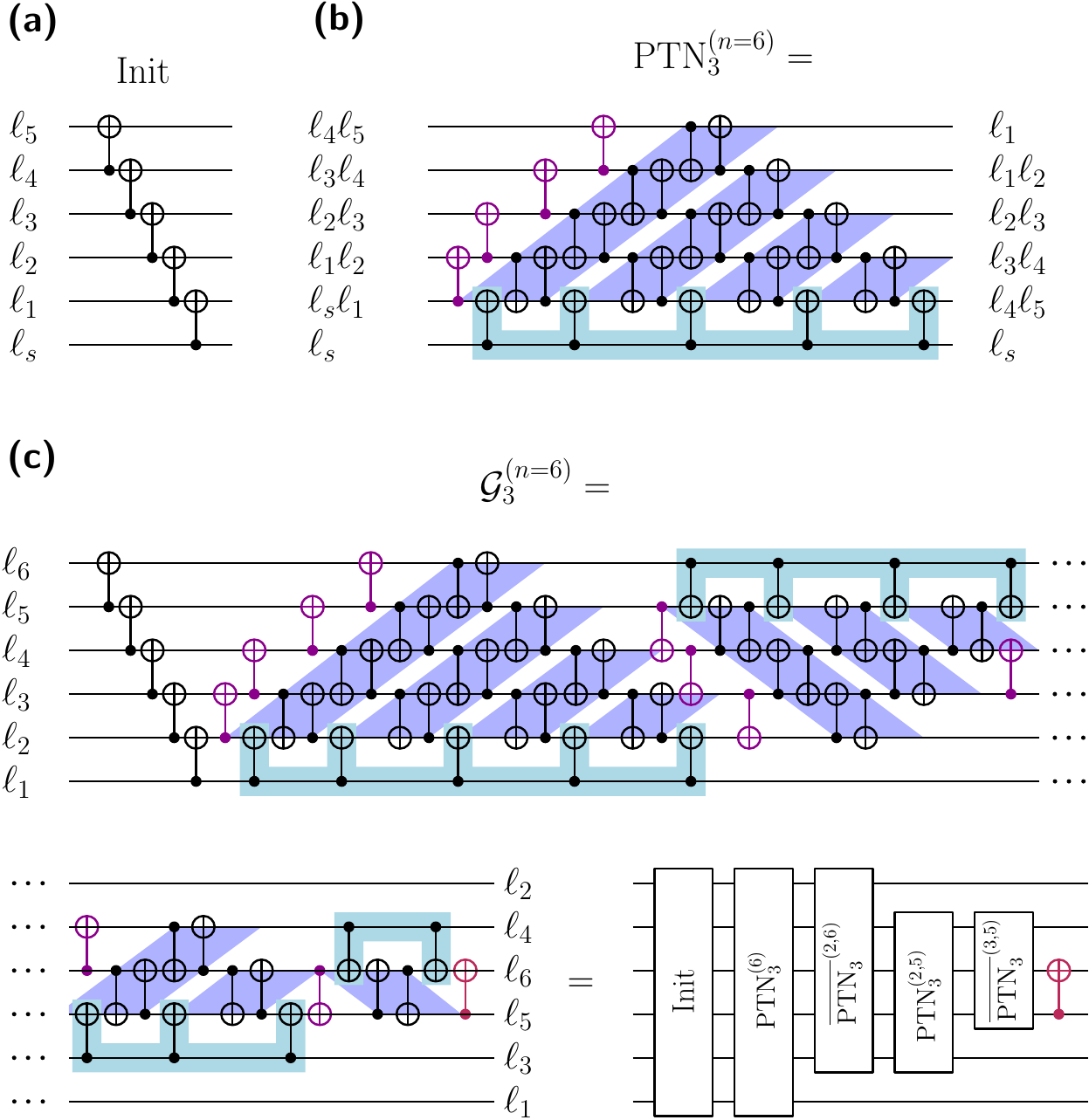}
    \caption{(a) An initialization circuit preparing the input label set for special three-body generators. (b) Schematic of the special three-body generator circuit $\XSN_3^{(n)}$ for the special label $\lab_s$. (c) Schematic of a three-body generator circuit $\G_3^{(n)}$ composed of special three-body generators $\XSN_3^{(n)}$ and $\overline{\XSN}_3^{(n)}$, respectively.}
    \label{fig:special_three_body_gen}
\end{figure}

The output label sequence of Eq.~\eqref{eq:special_three_body_from_encoded} coincides with
that of Eq.~\eqref{eq:special_three_body_output} except on qubit $n$ which does not anymore contain the special label. The starting label sequence $( \lab_s, (\lab_1 \ldots  \lab_{n-1})\lab_s )$ [left-hand side of Eq.~\eqref{eq:special_three_body_from_encoded}] can be obtained via two simple \CXNAME{} chains. The first \CXNAME{} chain prepares a label set similar to the right-hand side of Eq.~\eqref{eq:special_three_body_from_encoded} [see Fig.~\ref{fig:special_three_body_gen}(a)]
\begin{multline}
\label{eq:CNOTchain_for_three_body_generator}
    \labaction{\labseq{\lab_s, \lab_1,\ldots,\lab_{n-1}}}{(\CXC^{(n)})^\dagger}\\
    =\labseq{\lab_s, \lab_s\lab_1, \lab_1\lab_2 \dots,\lab_{n-2}\lab_{n-1}}
\end{multline}
and via a second \CXNAME{} chain we can transform the result into the required input of Eq.~\eqref{eq:special_three_body_from_encoded}
\begin{equation*}
    \labaction{\labseq{\lab_s, \lab_s\lab_1, \lab_1\lab_2 \dots,\lab_{n-2}\lab_{n-1}}}{\CXC^{(2,n)}} = \labseq{\lab_s, (\lab_1, ...,  \lab_{n-1})\lab_s}.
\end{equation*}
The combined building block 
\begin{equation}
\label{eq:SG_3}
    \XSN_3^{(n)}=\CXC^{(2,n)} \cat \XSNM^{(n)}
\end{equation}
is depicted in Fig.~\ref{fig:special_three_body_gen}(b) and is called \emph{\TWINE{} network of type II}. Note that $\XSN_3^{(n)}$ is also a special three-body generator since it generates all three-body logical labels that include the special label $\lab_s$. Importantly, however, $\XSN_3^{(n)}$ maps the input label set (which corresponds to the output label set of Eq.~\eqref{eq:CNOTchain_for_three_body_generator}) on an output label set reusable for forthcoming special three-body generators
\begin{multline}
\label{eq:special_three_body}
    \labaction{\labseq{\lab_s, \lab_s\lab_1, \lab_1\lab_2 \dots,\lab_{n-2}\lab_{n-1}}}{\XSN_3^{(n)}}\\
    =\labseq{\lab_s, \lab_{n-1}\lab_{n-2},\lab_{n-2}\lab_{n-3}, ... ,\lab_2\lab_1, \lab_1}
\end{multline}
The label set on the right-hand side of Eq.~\eqref{eq:special_three_body} has the same structure of logical labels encoded on qubits $\{2, \dots, n\}$ (in reversed order) as the input label set has on qubits $\{1,\dots, n\}$. Furthermore, the special label $\lab_s$ is isolated on physical qubit $1$. Thus, in the next step, we can now use $\lab_1$ as the special label and apply a reversed $\overline{\XSN}_3^{(2,n)}$ circuit, which is a special three-body generator for the label set $\{\lab_1,\dots,\lab_{n-1}\}$ acting on the physical qubits $\{2, \dots, n\}$. The parity label $\lab_s$, which is isolated on qubit $1$ remains unaffected by this. Pursuing this strategy on a shrinking set of qubits, we form a clean generator circuit $\mathcal{G}_3^{(n)}$ for all possible logical three-body labels. The corresponding circuit is schematically depicted in Fig.~\ref{fig:special_three_body_gen}(c). After the application of the last \TWINE{} network of type II (which only acts on three qubits) and a final \CXNAME{} gate that acts as a clean-up step, the output label sequence is given by
\begin{equation}
\label{eq:3body_output_label_set}
    \labseq{\lab_1, \lab_3, \ldots, \lab_4, \lab_2}.
\end{equation}
If we demand the input and output label set to coincide with $\labseq{\lab_1, \dots, \lab_n}$, this can be achieved with a sorting network, which has $\mathcal{O}(n)$ depth and $\mathcal{O}(n^2)$ gate count, which is (relatively) cheap in comparison to the $\mathcal{O}(n^2)$ depth and $\mathcal{O}(n^3)$ gate count of $\G_3^{(n)}$. For a comprehensive definition and rigorous analysis of $\G_3^{(n)}$, we refer the reader to App.~\ref{def:three-body_generator} and \ref{thm:three-body-generator}

Our circuit constructions have the neat advantage that building blocks can be partially executed in parallel. Using shifted concatenations as described in App.~\ref{def:shifted_concatenation} in Eq.~\eqref{eq:cswap_network} allows to start each $\XSC^{(1,p)}$ circuit with only a constant delay of four moments [compare Fig.~\ref{fig:2To3BodyGenerator} (c)]. In contrast, without shifted concatenation, each $\XSC^{(1,p)}$ circuit would acquire a delay of $2p$ moments. Eventually, shifted concatenations yield a reduction in total depth of $\mathcal{O}(n)$. For $\XSN^{(n)}$ this yields a total depth of $4n+ \mathcal{O}(1)$. Similar concatenation shifts can be applied for $\XSN_3^{(n)}$ and $\mathcal{G}_3^{(n)}$. The respective depth-reduced circuits are depicted in Fig.~\ref{fig:2To3BodyGenerator} and \ref{fig:special_three_body_gen}. 

With the additional concatenation shifts for consecutive special three-body generators, we obtain for $\mathcal{G}_3^{(n)}$ a reduced total depth of $n^2+\mathcal{O}(n)$, whereas the total gate count amounts to $\frac{1}{3}n^3 + \mathcal{O}(n^2)$. This result can readily be deduced from the involved building blocks: every special three-body generator is mainly based on \DXNAME{} gates (up to $\mathcal{O}(n^2)$ additional \CXNAME{} gates), whereby every \DXNAME{} gate (consisting of two \CXNAME{} gates) generates one required three-body label. In total, there are $\vert \mathcal{L}^{(n)}_3 \vert = \binom{n}{3}=\frac{1}{6} n^3 + \mathcal{O}(n^2)$ three-body labels. Thus, in leading order, we expect a total count of $\frac{2}{6} n^3$ gates.
According to Eqs.~\eqref{eq:asymptotic_average_count} and \eqref{eq:asymptotic_avaerage_depth}, this yields an asymptotic average gate count of $\mu(\mathcal{G}_3, \mathcal{L}_3)=2$ as well as an asymptotic depth of $\nu(\mathcal{G}_3,\mathcal{L}_3)=3$. The label set $\mathcal{L}^{(n)}_3$ falls into the class of label sets investigated in Sec.~\ref{sec:fundamental_bounds} and App.~\ref{thm:impossibility_theorem} for which we find a hypothetical lower bound of $\mu \geq 1+\frac{1}{9}$. We note that for large enough $n$ also most of the special three-body generator circuits must obey the same lower bounds since two labels of the corresponding special label set cannot be combined to form another label from the set. This holds equivalently for every (large enough) special label set in which all labels share at least one constituent. To obtain an average asymptotic gate count smaller than two for LNN chains, this requires a finite fraction of labels from $\mathcal{L}^{(n)}_3$ to be created using (on average) only a single \CXNAME{} gate. Constructing such a circuit becomes exceedingly difficult, if not outright impossible, even for moderately large values of $n$ ($\sim 10$).

\subsection{Generator circuits for arbitrary many-body operators \label{sec:special_4body_generator}}
The concepts outlined in Sec.~\ref{sec:two_body} and \ref{sec:three_body} can be used to construct special four-body generators. To do so, we recycle the optimized three-body generator circuit with inserting an extra label on top of the existing ones. This can be achieved using an additional \CXNAME{} to encode an additional label $\lab_{s'}$ on the physical qubit that holds the special label $\lab_s$. Then, the circuit $\CNOT_{1,2} \cat \XSN_3^{(2,n)}$ generates all four-body labels that include $\lab_s\lab_{s'}$. 
However, if we want to utilize the full depth-optimized three-body generator circuit $\mathcal{G}_3$ via mirroring every second special three-body generator $\XSN_3^{(n)}$, we need to encode the label $\lab_{s'}$ also on the other end of the LNN chain. Let us assume a similar initial label set as used for the special three-body generator $\XSN_3^{(n)}$, $\labseq{\lab_{s'},\lab_s,\lab_s\lab_2,\lab_2\lab_3,\ldots,\lab_{n-2}\lab_{n-1}}$ encoded using the tools of the latter section. Next, we apply a \CXNAME{} chain to obtain
\begin{multline}
\label{eq:special_4body_in}
\labaction{\labseq{\lab_{s'},\lab_s,\lab_s\lab_2,\lab_2\lab_3,\ldots,\lab_{n-2}\lab_{n-1}}}{\CXC^{(3,n)}}\\
    =\labseq{\lab_{s'},\lab_s,\lab_s(\lab_2\ldots\lab_{n-1})}
\end{multline}
Now we can use a \TWINE{} chain to encode the label $\lab_{s'}$ on physical qubit $n$
\begin{multline*}
\labaction{\labseq{\lab_{s'},\lab_s,\lab_s(\lab_2\ldots\lab_{n-1})}}{\XSC^{(n)}}\\
    =\labseq{\lab_{s'}\lab_s,\lab_{s'}\lab_s(\lab_2\ldots\lab_{n-1}),\lab_{s'}}.
\end{multline*}
Up the label encoded on qubit $n$, the resulting label sequence has the same structure as the input label sequence in Eq.~\eqref{eq:special_three_body_from_encoded}. Therefore, all four-body labels containing $\lab_s\lab_{s'}$ can now be generated using 
$\XSNM^{(1,n-1)}$ (note that we leave out the very last qubit $n$ where $\lab_{s'}$ is encoded)
\begin{multline}
\label{eq:special_4body_out}
\labaction{\labseq{\lab_{s'}\lab_s,\lab_{s'}\lab_s(\lab_2\ldots\lab_{n-1}),\lab_{s'}}}{\XSNM^{(1,n-1)}}\\
   =\labseq{\lab_{s'}\lab_s,\lab_{n-2}\lab_{n-1},\lab_{n-3}\lab_{n-2},\ldots,\lab_3\lab_2,\lab_2,\lab_{s'}}.
\end{multline}
Importantly, the output label set of Eq.~\eqref{eq:special_4body_out} encodes the same label structure (in reversed order) as the input state in Eq.~\eqref{eq:special_4body_in}. This pattern can be repeated and eventually yields a circuit which generates all four-body labels that containing $\lab_{s'}$ (initially encoded on qubit $1$). 

As a post-processing step, we add an additional clean-up circuit with $\mathcal{O}(n)$ \CXNAME{} count and depth to regain single-bodied labels. The resulting (combined) circuit constitutes a clean special four-body generator $\SG_4^{(n)}$. Fig.~\ref{fig:4body_gen}(a) illustrates the circuit $\SG_4^{(n)}$. A precise definition of the circuit $\SG_4^{(n)}$ and its components as well as their properties can be found in App.~\ref{def:label_generator} and \ref{def:special_k-body_generator}-\ref{thm:clean_special_four-body_generator}.

Similar to the design of the three-body generator $\mathcal{G}_3^{(n)}$, we can build a four-body generator $\G_4^{(n)}$ by concatenating $n-3$ special four-body generators each with different special labels. Thereby, later clean special four-body generators ignore all previous special labels, until the last clean special four-body generator only acts on four qubits. The circuit design is schematically shown in Fig.~\ref{fig:4body_gen}(b).
\begin{figure}
    \centering
    \includegraphics[width=1.0\linewidth]{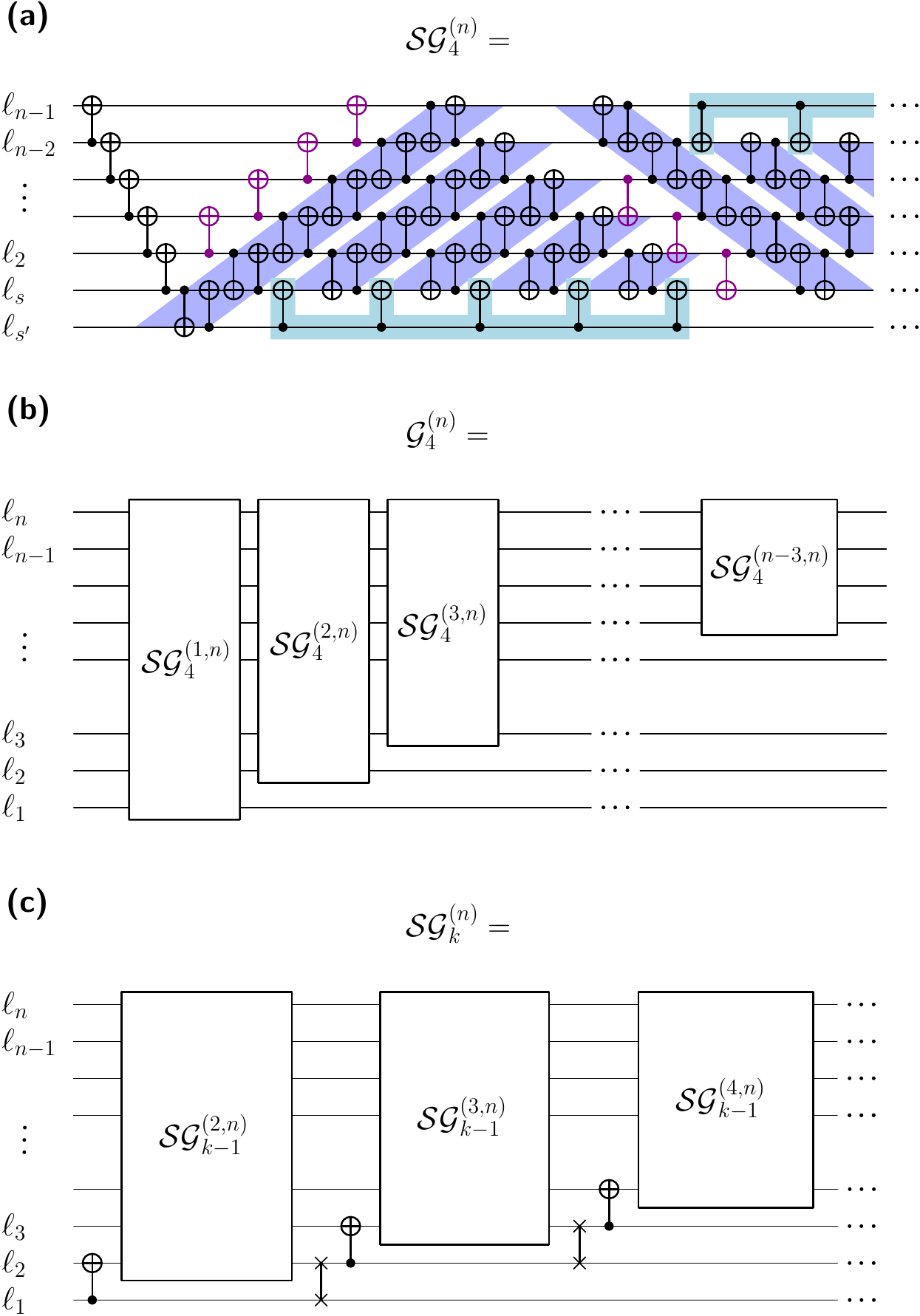}
    \caption{(a) Schematic of the clean special four-body generator $\SG_4^{(n)}$. (b) Schematic of the four-body generator circuit $\G_4^{(n)}$ from the clean special four-body generator $\SG_4^{(n)}$. (c) Design of a clean special $k$-body generator $\SG_k^{(n)}$ from the clean special $(k\!-\!1)$-body generator $\SG_{k-1}^{(n)}$. We indicate a \SWNAME{} gate by crossed-out qubits connected by a line.}
    \label{fig:4body_gen}
\end{figure}

The circuit construction for the clean special four-body generator $\SG_4^{(n)}$ can now be utilized to constructively design clean special $k$-body generators for arbitrary $k >4$. The strategy here is to recycle the clean special $(k\!-\!1)$-body generator $\SG_{k-1}^{(n)}$ for the construction of a clean special $k$-body generator $\SG_{k}^{(n)}$. To do so, we encode the an extra special label $\lab_{s'}$ into each of the special label $\lab_s$ of the corresponding clean special $(k\!-\!1)$-body generator. The corresponding circuit then becomes a clean special $(k\!-\!1)$-body generator which generates all $k$-body labels with the special labels $\lab_{s'}\lab_s$. To proceed with the next clean special $(k\!-\!1)$-body generator, we apply a \SWNAME{} gate to exchange the special labels $\lab_s$ and $\lab_{s'}$ enabling the repetition of the same procedure. Chaining $n-k+1$ clean special $(k\!-\!1)$-body generators (with the extra special label $\lab_{s'}$) and a subsequent clean-up step then yields a clean special $k$-body generator (i.e. a generator for all labels that contain $\lab_{s'}$). Concatenations of clean special $k$-body generators can then be composed into a $k$-body generator $\mathcal{G}_k^{(n)}$. The recursive construction of our $k$-body generator $\G_k^{(n)}$ implies that $\G_k^{(n)}$ is also a generator for all $(k-p)$-body labels with $p<k$. For more details, we refer to App.~\ref{def:clean_special_k-body_generator} - \ref{thm:k-body_generator_from_special_k-body_generator}. 

An interesting feature of this construction of $k$-body generators $\mathcal{G}_k^{(n)}$ is that each $\mathcal{G}_k^{(n)}$ inherits the asymptotic average gate count and asymptotic normalized depth from the corresponding clean special $k$-body generator since for each $k$ we just introduce an asymptotically negligible amount of extra \CXNAME{} gates. Recursively, this yields the asymptotic average gate count for the family of $k$-body generators $\mathcal{G}_k$ with the corresponding family of label sets $\mathcal{L}_k$ to be $\mu(\G_k,\mathcal{L}_k)= 2$. Likewise, for the asymptotic normalized depth we obtain $\nu(\G_k, \mathcal{L}_k)=k$ (see App.~\ref{thm:kbodies_lnn} for a detailed gate count and depth analysis). In Fig.~\ref{fig:count_and_depth_analysis} we investigate the convergence properties of the average gate count $\mu_n$  and the normalized depth $\nu_n$ for the different $k$-body generator circuits as a function of the number of qubits $n$. The average count shows fast convergence to the expected asymptotic value ($\mu=2$). In contrast, convergence of the normalized depth slows down with increasing $k$. This is expected as we do not employ further depth optimizations beyond $k=4$. Instead higher order terms are recursively constructed using lower order building blocks which amounts to increasing pre-factors in the sub-leading order contributions to the normalized depth.

\begin{figure}
    \centering
    \includegraphics[width=1.0\linewidth]{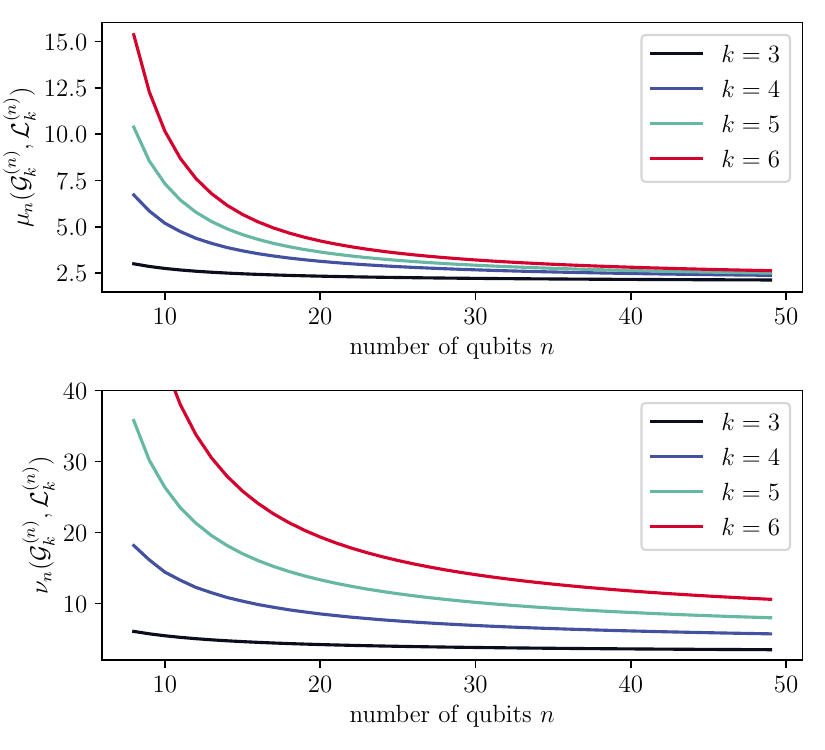}
    \caption{Average count $\mu_n(\mathcal{G}_k^{(n)}, \mathcal{L}_k^{(n)})$ and normalized depth $\nu_n(\mathcal{G}_k^{(n)}, \mathcal{L}_k^{(n)})$ for different $k$-body generator circuits on LNN devices as a function of the number of qubits $n$.}
    \label{fig:count_and_depth_analysis}
\end{figure}

\section{Generic connectivity graphs} \label{sec:generic_connectivity_graphs}
{\color{black}
In this section, we generalize the results of the last sections to a  wide range of different connectivity graphs and demonstrate how the developed building blocks can be customized to fit connectivity constraints and yield performance gains.

\subsection{All-to-all connectivity \label{sec:all-to-all}}

\begin{figure}
    \centering
    \includegraphics[width=1.0\linewidth]{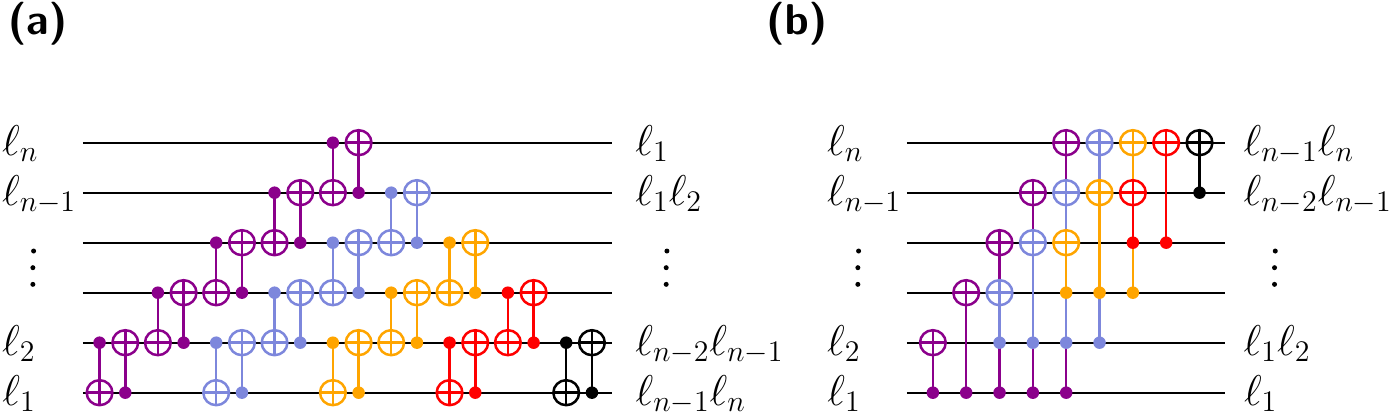}
    \caption{(a) Circuit diagram of the two-body generator circuit $\XSN^{(n)}$ on LNN devices. (b) Equivalent circuit $\XSN_{\alltoall}^{(n)}$ for a hardware with all-to-all connectivity. Note that this version uses half the amount of \CXNAME s and has half the depth.}
    \label{fig:all-to-all}
\end{figure}

Before directly discussing generic connectivity graphs, it is worthwhile to first investigate all-to-all connected devices as they offer the simplest generalization.
}
On all-to-all connected devices, logical information of any qubit is accessible from any other qubit. This is dramatically different to LNN devices where each qubit is connected to (at most) \emph{two} other qubits. On first glance one might assume corresponding generator circuits to reflect these dramatic architectural differences since all-to-all connected devices render the positioning of labels inconsequential and label transportation via \TWINE{} chains is unnecessary. However, in this section we present a simple scheme to map the circuits developed in Sec.~\ref{sec:LNN_kbody} to gate count and depth optimized circuits on all-to-all connected devices.

On LNN devices, ${\XS_{i,j} = \concat{\CNOT_{i,j}}{\SW_{i,j}}}$ constitutes the basis of our circuit constructions, providing two essential features: (i) it encodes the combined parity information on the control qubit $i$ and (ii) it transports the initial logical parity information of the control qubit $i$ to the target qubit $j$. Since the positioning of logical labels is irrelevant on all-to-all connected devices, this implies that the action of a \DXNAME{} gate can be emulated on an all-to-all connected device with just a single \CXNAME{} gate. For that, we replace each \SWNAME{} gate with a \emph{virtual} \SWNAME{} gate, which virtually swaps the positions of the qubits involved. Virtual \SWNAME{} gates are merely a trick to map the output label sequence of a \DXNAME{} gate on the LNN chain to its corresponding counterpart for all-to-all connected devices. One can use classical software to keep track of all virtual \SWNAME{} gates applied and, in the end of the circuit, we obtain the actual qubit layout by reverting all virtual \SWNAME{} gates. 

With this approach, we can now construct $\XSC^{(n)}$ and $\XSN^{(n)}$ circuits for all-to-all systems. For example, Fig.~\ref{fig:all-to-all}(b) shows the circuit $\XSN_{\alltoall}^{(n)}$ obtained by applying this reverting process of virtual \SWNAME{} swaps on the circuit $\XSN^{(n)}$. For comparison, Fig.~\ref{fig:all-to-all}(a) displays the corresponding LNN circuit. For each moment $m$ of the circuit in Fig.~\ref{fig:all-to-all}(b), there exists a bijection which maps the label sequence at $m$ to the corresponding label sequence at moment $2m$ in the circuit of Fig.~\ref{fig:all-to-all}(a). In terms of generated labels, the action of $\XSN_{\alltoall}^{(n)}$ is identical with that of $\XSN^{(n)}$ (given in Eq.~\eqref{eq:cswap_network}). However, $\XSN_{\alltoall}^{(n)}$ only requires \emph{half} as many gates and \emph{half} of the depth as compared to $\XSN^{(n)}$.

We can readily extend these ideas to $\XSNM_{\alltoall}^{(n)}$ and $\XSN_{3,{\alltoall}}^{(n)}$ circuits [see Eq.~\eqref{eq:controlled_cswap_network} and \eqref{eq:SG_3} and Fig.~\ref{fig:special_three_body_gen}(b)]. Their dominant contribution consists of $\XSC_{\alltoall}^{(n)}$ circuits and, consequently, we achieve (almost) half of the gate count and depth compared to the LNN construction. Likewise, all the generator circuits profit from this inheritance. For the family of $k$-body generators ${\mathcal{G}_{k,\alltoall}}$ with ${k>2}$, this yields an asymptotic average gate count of ${\mu(\mathcal{G}_{k,\alltoall}, \mathcal{L}_k) = 1}$ and a normalized depth of ${\nu(\mathcal{G}_{k,\alltoall}, \mathcal{L}_k) = k/2}$ (see App.~\ref{sec:k-body_generators_on_other_devices}). Let us emphasize again that ${\mu =1}$ corresponds to the theoretical lower bound. Thus, all the (special) $k$-body generators for any ${k\geq 2}$ are asymptotically optimal circuits (in gate count) where on average each \CXNAME{} in the circuit $\G_k^{(n)}$ creates one label of the label set $\mathcal{L}_k^{(n)}$ for ${n\to\infty}$.

\subsection{Locally connected graphs}\label{sec:2d_architectures}
From the prospective of connectivity, LNN and all-to-all connected devices constitute two extreme cases. LNN devices possess the minimal required connectivity so that every qubit can be reached from any other qubit via residual connections. Except for the first and last qubit, each qubit has two directly connected neighbor qubits, resulting in an average neighbor count of ${\eta_n \sim 2}$. In all-to-all deviecs, the qubit connectivity graph is described by a complete graph where all qubits are directly connected resulting in an average neighbor count of ${\eta_n= n-1}$. In between these two extremes, we can allocate qubit layouts with connectivity graphs with an average neighbor count ${2<\eta_n < n-1}$. As a prototypical example, here we discuss square grid devices \cite{harrigan_quantum_2021, Acharya2024}, however, our results can be directly transferred to other popular architectures such as hexagonal \cite{Chamberland2020, Weidenfeller2022} or octagonal \cite{Dupont2023} layouts.

If we aim to implement generator circuits on square grid devices, from the results of Secs.~\ref{sec:LNN_kbody} and ~\ref{sec:all-to-all} we can already deduce an average gate count $1\leq \mu \leq 2$. This follows directly from the connectivity graph described by square grid devices, where all qubits can easily be connected in a Hamiltonian path [Fig.~\ref{fig:hamiltonian_grid_path}(a)]. A chosen Hamiltonian path can then be used as an effective LNN model to implement the generator circuits outlined in Sec.~\ref{sec:LNN_kbody} with an average gate count of ${\mu(\mathcal{G}_k, \mathcal{L}_k)=2}$. Interpreting square grid devices as effective LNN devices might be convenient, however, this neglects a considerable amount of additional connections. To utilize this untapped potential, we pursue a different strategy: instead of laying out a Hamiltonian path, we define a path so that all qubits of the square grid are either within the path or directly neighbor this path. Fig.~\ref{fig:hamiltonian_grid_path}(b) depicts such a path on a ${3\times 3}$ square grid where the black solid line indicates the path itself and the blue dashed lines represent the connections to path neighbors. Subsequently, we call paths of this kind \emph{Hamiltonian grid paths} (HGP). Qubits are either directly part of a HGP or they are direct neighbors to a HGP in which case we call them HGP neighbors. We call a HGP \emph{minimal}, if the number of qubits directly contained in it is minimal (and, as consequence, the number of HGP neighbors is maximal). For example, the path depicted in Fig.~\ref{fig:hamiltonian_grid_path}(a) represents a HGP but it is not minimal. In contrast, Fig.~\ref{fig:hamiltonian_grid_path}(b) shows a minimal HGP on a ${3\times 3}$ grid. In the following, we demonstrate how minimal HGPs can be used to efficiently implement the circuit building blocks of Sec.~\ref{sec:LNN_kbody}. 

\begin{figure}
    \centering
    \includegraphics[width=\columnwidth]{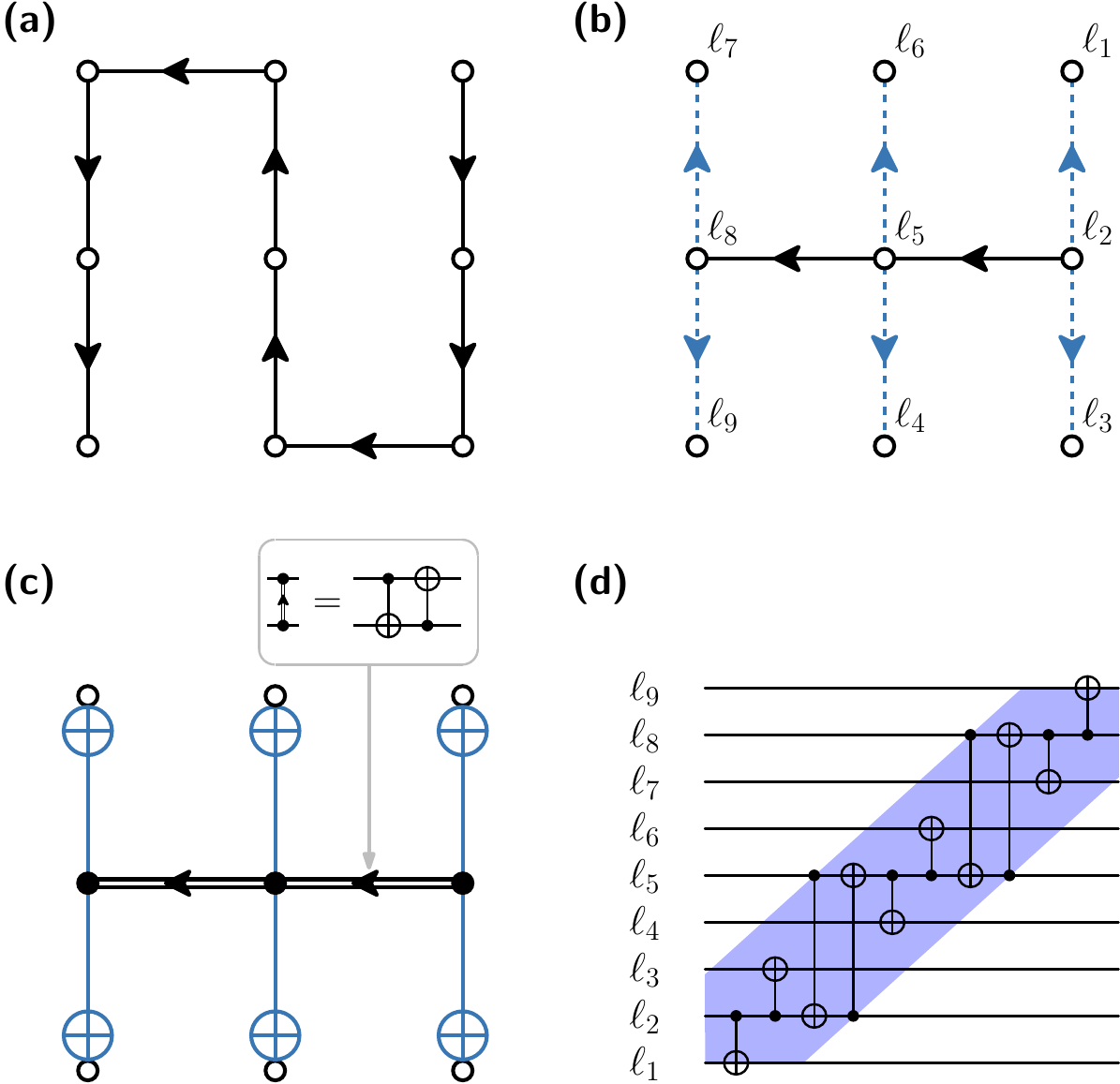}
    \caption{(a) A Hamiltonian path on a square grid of nine qubits. (b) A HGP on a square grid of nine qubits. (c) An implementation of $\XSC_\grid^{(9)}$ on a square grid of nine qubits based on the HGP of (b). (d) Same as (c) as a circuit diagram.}
    \label{fig:hamiltonian_grid_path}
\end{figure}

Let us recapitulate the action of $\XS_{i,j}$, which encodes the combined parity information, $\lab_i\lab_j$ of qubit $i$ and $j$ onto qubit $i$ while transporting the (initial) parity information of qubit $i$ onto qubit $j$. The latter property allows to use concatenations of \DXNAME{} gates, i.e., \TWINE{} chains, to encode the label $\lab_i$ on all qubits of the concatenation. On square grids, we can use the additional connections to implement this even more efficiently: we use concatenations of \DXNAME{} gates, akin to \TWINE{} chains on LNN devices, which is applied along a defined minimal HGP [black arrowed lines in Fig.~\ref{fig:hamiltonian_grid_path}(b)]. However, we intertwine it with \CXNAME{} gates that target the direct HGP neighbors [blue dashed lines in Fig.~\ref{fig:hamiltonian_grid_path}(b)]. In this way, the corresponding circuit, subsequently denoted $\XSC_{\mathrm{grid}}^{(n)}$, encodes the label $\lab_2$ of Fig.~\ref{fig:hamiltonian_grid_path}(b) onto all qubits of the square grid. 
Fig.~\ref{fig:hamiltonian_grid_path}(c) and (d) show a $\XSC_\grid^{(n)}$ circuit following a minimal HGP for a $3\times 3$ square grid. Compared to a corresponding $\XSC^{(n)}$ circuit (along a conventional Hamiltonian path), $\XSC_{\mathrm{grid}}^{(n)}$ saves $\sim 1/3$ of the involved \CXNAME{} gates.

Following the construction principles outlined in Sec.~\ref{sec:LNN_kbody}, we repeatedly apply $\XSC_{\mathrm{grid}}^{(1,p)}$ circuits on a shrinking set of qubits $p$ for the construction of $\XSN_{\mathrm{grid}}^{(n)}$ circuits.
To allow this, after each $\XSC_{\mathrm{grid}}^{(1,p)}$ a final circuit is required to pigeonhole labels and allow subsequent circuits to address all relevant labels. In Fig.~\ref{fig:2body_generator_square_grid_1}(a) we schematically display the first four consecutive $\XSC_{\mathrm{grid}}^{(1,p)}$ circuits together with the corresponding pigeonhole circuit for a ${6\times 4}$ square grid. 
As a general rule for the pigeonhole circuit, we first transport a label to one of the HGP neighbors, the next label is transported to the corresponding HGP node while the third label then occupies the second HGP neighbor.

\begin{figure}
    \centering
    \includegraphics[width=\columnwidth]{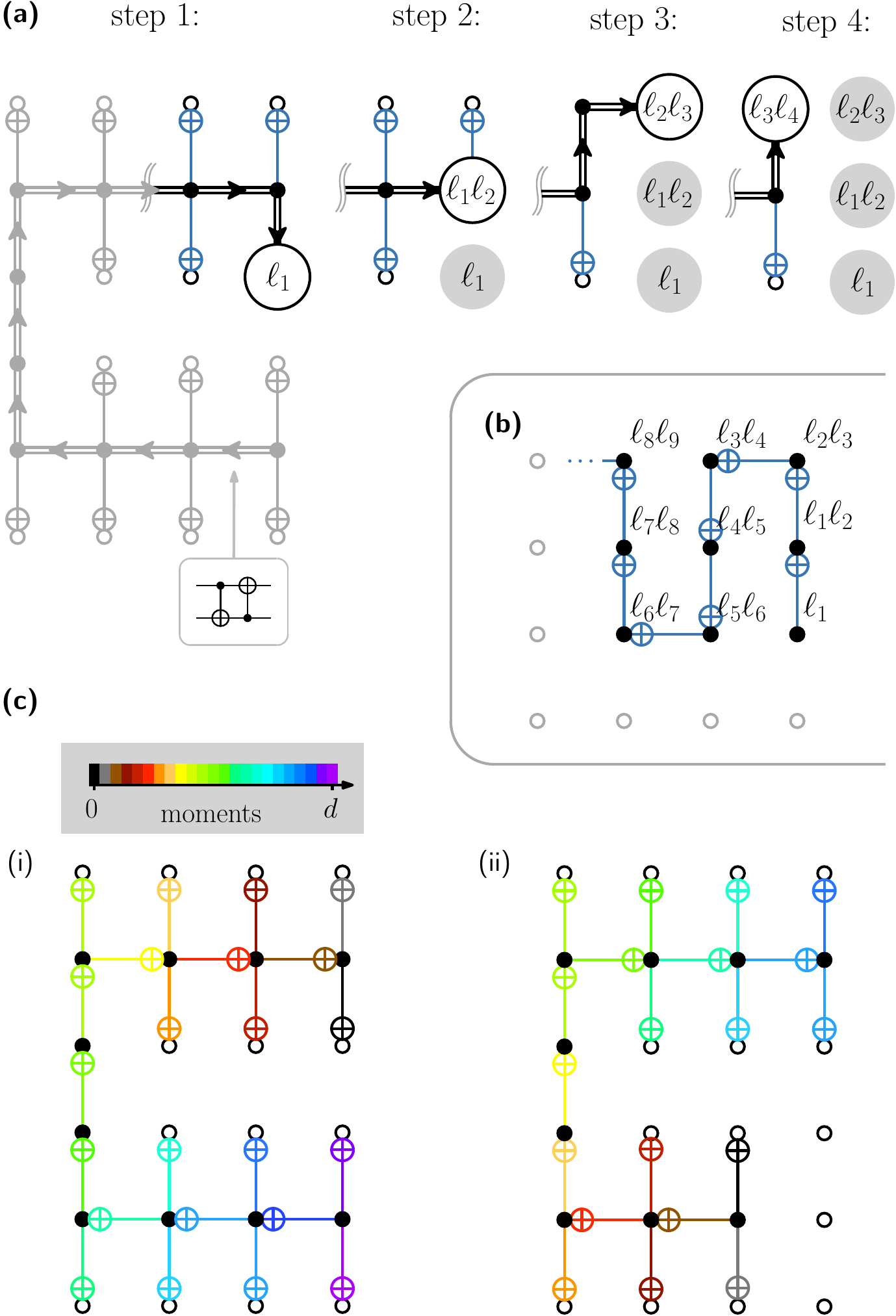}
    \caption{(a) Schematic of the first four consecutive $\XSC_\grid^{(1,p)}$ circuits combined with a respective circuit to pigeonhole labels to corresponding locations to simplify the decoding. (b) Schematic of the decoding step given by a \CXNAME{} chain along a conventional Hamiltonian path. Depending on the device dimensions, extra \SWNAME{} gate are required along the boundaries. (c) The initializer consisting of two $\CNOT$ chains on a minimal HGP.}
    \label{fig:2body_generator_square_grid_1}
\end{figure}

Concatenation of ${n-1}$ $\XSC_{\mathrm{grid}}^{(1,p)}$ circuits (with the final pigeonhole circuits) leaves us with a label pattern depicted in Fig.~\ref{fig:2body_generator_square_grid_1}(b). Similar to the LNN circuit construction, we can decode the labels using a chain of \CXNAME{} gates [Fig.~\ref{fig:2body_generator_square_grid_1}(b)]. Depending on the geometry of the device, extra \SWNAME{} gates are necessary to breach the boundaries. Note that, as a consequence of the employed HGP, after decoding we do not regain the initial label order, however, labels can be tracked and appear in deterministic positions dictated by the algorithm.

As a whole, the outlined steps yield a clean generator circuit for all two-body labels $\mathcal{G}_{2,\mathrm{grid}}^{(n)}$. In comparison to the corresponding LNN circuit derived in Sec.~\ref{sec:LNN_kbody}, here we save ${\sim 1/3}$ of the involved \CXNAME{} gates which leaves us with an asymptotic average \CXNAME{} count of ${\mu(\mathcal{G}_{2,\mathrm{grid}}, \mathcal{L}_2) = 1 + \frac{1}{3}}$. Interestingly, in close analogy to LNN, the $\XSC_\grid^{(1,p)}$ and the final pigeonhole circuits can be stacked using shifted concatenations to form a depth-optimized $\XSN_{\mathrm{grid}}^{(n)}$ circuit. Fig.~\ref{fig:2body_generator_square_grid_2} displays the $\XSN_{\mathrm{grid}}^{(12)}$ circuit together with a decoding \CXNAME{} chain for a ${3\times 4}$ square grid. Each of the ${n-1}$ $\XSC_{\mathrm{grid}}^{(1,p)}$ circuits starts with a constant delay of six moments which eventually yields a total depth of ${6n+\mathcal{O}(1)}$ for the combined $\XSN_{\grid}^{(n)}$ circuit.

\begin{figure*}
    \centering
    \includegraphics[width=1.0\linewidth]{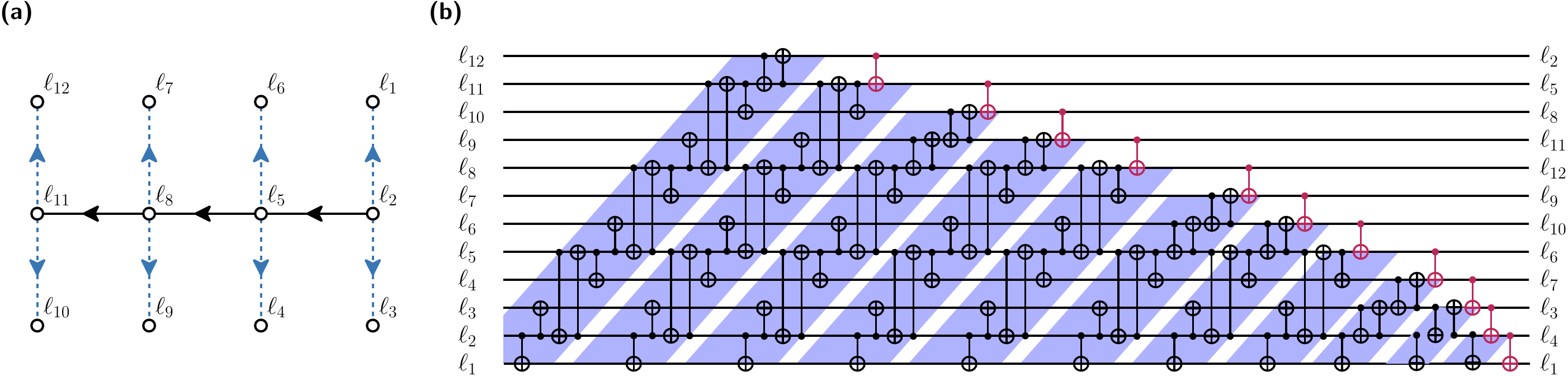}
    \caption{A clean two-body generator designed for a square grid device. (a) The HGP used for the design of the clean two-body generator depicted in (b) on a square grid device with 12 qubits together with their initial labels enumerated with $\ell_1, \dots, \ell_{12}$. (b) Two-body generator implementation for the square grid device schematically shown in (a).}
    \label{fig:2body_generator_square_grid_2}
\end{figure*}

Beyond two-body generator circuits, we can repeat the steps derived for LNN devices using our adapted building blocks $\XSC_{\mathrm{grid}}^{(n)}$. The only additional building block required is an \emph{initializer} that prepares a label sequence akin to that of Eq.~\eqref{eq:special_three_body_from_encoded}, i.e.~encoding a special label $\lab_s$ in all other labels while keeping the label $\lab_s$ fixed. On a minimal HGP, this can be done with two \CXNAME{} sequences, reaching all HGP nodes as well as all HGP neighbors [see Fig.~\ref{fig:2body_generator_square_grid_1}(c), note that for the square grid, this can equivalently accomplished with a \CXNAME{} chain along a conventional Hamiltonian path with the same count and depth]. With this, special three-body generators are constructed using the initializer of Fig.~\ref{fig:2body_generator_square_grid_1}(c)(ii) together with a modified \TWINE{} network of type I $\XSN'^{(n)}_\grid$. $\XSN'^{(n)}_\grid$ can be derived from the definition of $\XSN'^{(n)}$ [Eq.~\eqref{eq:controlled_cswap_network}] by replacing each $\XSC^{(n)}$ circuit with the respective square grid version $\XSC^{(n)}_\grid$ (followed by a suitable pigeonhole circuit). Concatenating an initializer according to Fig.~\ref{fig:2body_generator_square_grid_1}(c)(ii) with $\XSN'^{(n)}_\grid$, this yields the special three-body generator $\XSN^{(n)}_{3,\grid}$ adapted to square grids. Similarly, the full three-body generator $\mathcal{G}_3^{(n)}$ can be assembled using concatenations of the initaliser akin to Fig.~\ref{fig:2body_generator_square_grid_1}(c)(i) as well as $\XSN^{(n)}_{3,\grid}$ and $\overline{\XSN}^{(n)}_{3,\grid}$ (replacing the building blocks in Fig.~\ref{fig:special_three_body_gen}(c) with the respective square grid adapted version). 

In general, to adapt the LNN circuit to the square grid, we need to replace the respective building blocks (see Fig.~\ref{fig:2To3BodyGenerator}) with the ones suitable for square grids. However, following this recipe, we obtain a depth overhead originating from the misalignment of the special three-body generators $\XSN^{(n)}_{3,\grid}$ and $\overline{\XSN}^{(n)}_{3,\grid}$ [in contrast to the respective case for LNN devices; compare Fig.~\ref{fig:special_three_body_gen}(c)]. The concatenation of the special three-body generator for square grids $\XSN^{(n)}_{3,\grid}$ and its reversed counterpart $\overline{\XSN}^{(n)}_{3,\grid}$ can be shifted at most by $4/3n + \mathcal{O}(\sqrt{n})$. This originates from the form of the $\XSN_\grid$ circuits (compare Fig.~\ref{fig:2body_generator_square_grid_2}) where each constituent $\XSC_\grid$ circuit is displaced by six moments with a corresponding length of $4n/3 + \mathcal{O}(\sqrt{n})$. The resulting misalignment of $\XSN^{(n)}_{3,\mathrm{grid}}$ and $\overline{\XSN}^{(n)}_{3,\mathrm{grid}}$ then leads to an increased depth of $7/3 n^2 +\mathcal{O}(n\sqrt{n})$ for the corresponding three-body generator $\mathcal{G}^{(n)}_3$. To improve on this result, in App.~\ref{sec:k-body_generators_on_other_devices} we investigate a construction based on clean special three-body generators which allows to use Lemma \ref{lem:adjoint_circuit}. This has the advantage that we can concatenate not only $\XSN^{(n)}_{3,\grid}$ and $\overline{\XSN}^{(n)}_{3,\grid}$, i.e.~ a special three body generator and its reversed circuit, but also the adjoint reversed. With this improvement, we obtain a depth of $25/12 n^2 + \mathcal{O}(n)$ for the full three body generator. However, starting from special four-body generators, this approach does not yield any depth gains and, eventually, for four-body generators and beyond, we find an asymptotic average \CXNAME{} count of $\mu(\mathcal{G}_k, \mathcal{L}_k) = 1 + \frac{1}{3}$ and an normalized asymptotic depth of $\nu(\mathcal{G}_k, \mathcal{L}_k) = 7/3 k$ (see App.~\ref{sec:k-body_generators_on_other_devices}).

We emphasize that the procedure outlined above is generic and can be generalized to all connectivity graphs where the qubits can be connected in a HGP such as ladders, hexagonal layouts \cite{Chamberland2020} but also graphs with non-local and/or non-planar connectivity \cite{bluvstein_quantum_2022, bluvstein_logical_2024, Moses2023}. For a connectivity graph of interest, we define a (minimal) HGP. Based on this, we construct the respective building blocks: a $\XSC$ circuit which maps the respective input label sets as derived in Eq.~\eqref{eq:cswap_chain}, a decoding circuit and an initializer circuit. Up to implementation details, we can then readily design any $k$-body generator of interest. As an additional example of utility, in App.~\ref{app:heavy-hex} we demonstrate our approach for heavy hexagon layouts. All results regarding asymptotic normalized depth and asymptotic average \CXNAME{} count are collected in Tab.~\ref{tab:average_count_depth_kbody}. 

The approach of using HGPs can also be applied for LNN and all-to-all connected devices, though the respective paths are not recognized as such or do not yield performance gains. In the LNN case, a minimal Hamiltonian grid (almost) coincides with a conventional Hamiltonian path: the HGP includes all but the first and the last qubits which are the only HGP neighbors. For all-to-all connected devices, a minimal HGP is trivial as it contains just a single qubit with all other qubits being HGP neighbors. 

Using HGPs allows to convert increments in connectivity into reductions of the total \CXNAME{} count. In Fig.~\ref{fig:overview} we depict the asymptotic average \CXNAME{} count $\mu$ as a function of the inverse asymptotic average neighbor count $1/\eta$ for the investigated connectivity graphs. With the approach outlined above, for $k$-body generator circuits this yields a reduction in gate count on the order of $\mathcal{O}(n^k)$. However, this comes at the expense of an increased depth on the order of $\mathcal{O}(n^{k-1})$ as compared to LNN (except for all-to-all, compare Tab.~\ref{tab:average_count_depth_kbody}). Motivated by this observation, in App.~\ref{app:dept_optimization} we investigate an alternative approach for a ladder layout, which allows to reduce both, total count and depth (as compared to the LNN approach). Here, the main idea is to use multiple HPGs in parallel. For the two-body generator circuits, this allows to obtain a \CXNAME{} count of $3/4n^2 + \mathcal{O}(n)$ with a total depth of $3n+\mathcal{O}(1)$, which beats the corresponding LNN approach in both metrics. 

{\color{black}
\section{Optimizing depth\label{app:dept_optimization}}

In Sec.~\ref{sec:2d_architectures} we have introduced a robust strategy to reduce the gate count of $k$-body generator algorithms by efficiently implementing them on a given connectivity graph. Simultaneously, this comes at the expense of a depth increase compared to the corresponding construction for nearest neighbor connectivity graphs. If we aim to improve both, count and depth, several additional considerations are necessary. 

Specifically, achieving lower depth (as compared to implementations for nearest neighbor connectivity graphs) requires stricter connectivity requirements on the hardware. The overall size of our algorithm implementations is largely determined by the size of the \DXNAME{} gates which constitutes the most prominent component. Meanwhile, the depth depends on the required shift between the constituent circuit components, in particular the shifts between subsequent $\XSC$ chains. As the number of HGP neighbors connected to a single HGP node grows, so does the required shift between subsequent $\XSC$ chains and, consequently we obtain an increased depth. For example, the $\XSN_{\mathrm{LNN}}^{(n)}$ circuit (with zero HGP neighbors) has a depth of $4n + \mathcal{O}(1)$ while $\XSN_{\mathrm{grid}}^{(n)}$ (with two HGP neighbors per HGP node) results in $6n + \mathcal{O}(1)$.

To reduce these shifts we need to leverage parallelization. One way to achieve this is by using multiple HGPs. To ensure functionality, the corresponding HGP nodes must not intersect, collectively cover all qubits and an edge must exist that connects labels that have traveled $m$ moments on one path with those that have traveled $m$ moments on another.

\begin{figure}
    \centering
    \includegraphics[width=1.0\linewidth]{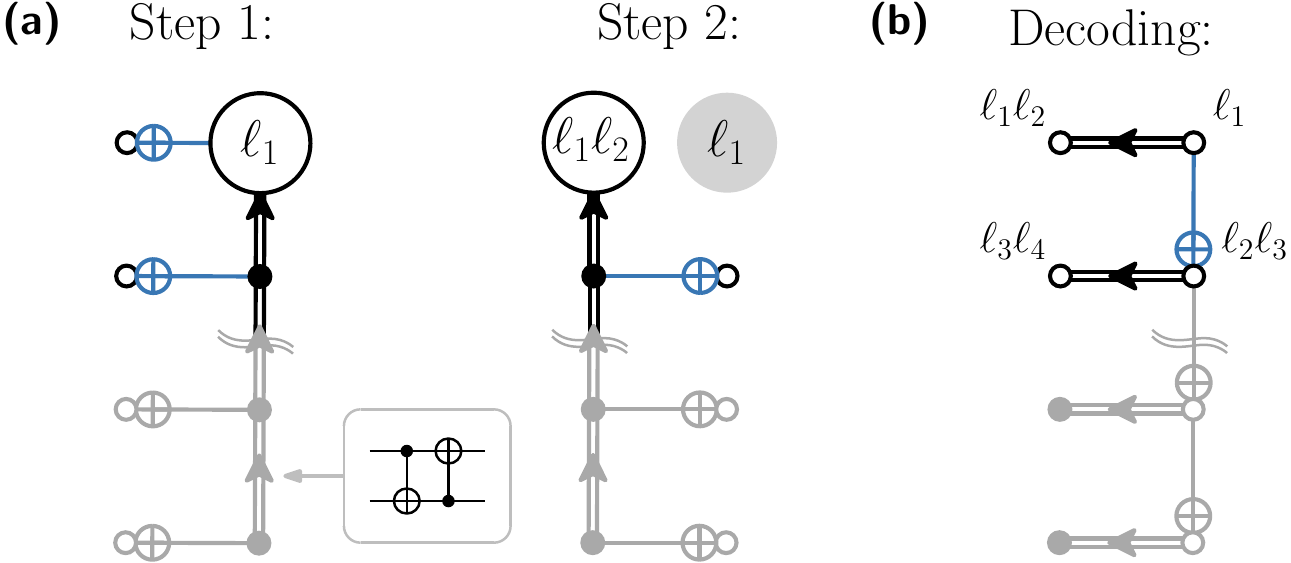}
\caption{
(a) Schematic of the first two consecutive \TWINE{} chains adapted to a ladder connectivity graph, illustrating how labels are transported along two distinct HGPs.  
(b) Schematic of the decoding procedure, which begins with a single \DXNAME{} gate and is followed by one \CXNAME{} gate. This sequence is then applied repeatedly from top to bottom. The resulting decoding step contributes only in sub-leading order to both the circuit size and depth.}
    \label{fig:ladder}
\end{figure}

Assuming a planar qubit layout, the above requirements are for example met by a ladder connectivity graph (see Fig.~\ref{fig:ladder}). On a ladder we can naturally define two HGPs where the nodes of one are the HPG neighbors of the other. In that way, the required shift of subsequent $\XSC$ chains reduces to $3$ \CXNAME{} gates (as compared to $4$) when using the (sequential) LNN approach. Notably, this parallelization scheme only affects the depth but leaves the \CXNAME{} count unchanged. That is, using two connected HGPs (where each node is associated with one HGP neighbor) results in the exact same \CXNAME{} as the sequential approach where just a single HGP of the same kind is used. More concretely, we find
\[
    \size{\XSN_{\mathrm{ladder}}^{(n)}} = \frac{3}{4}n^2 + \mathcal{O}(n)
\]
\[
    \depth{\XSN_{\mathrm{ladder}}^{(n)}} = 3n + \mathcal{O}(1).
\]
Note that this improves both depth and count compared to the LNN implementation (see also Tabs. \ref{tab:qaoa} and \ref{tab:qft}). Generalizing this approach for $k$-body generators while preserving the improved depth scaling, one would need to introduce ``checkerboard connectivity'' - that is, diagonal couplings between qubits that are already connected in the vertical and horizontal directions.

}

\section{Applications}

\subsection{Quantum Approximate Optimization
Algorithms \label{sec:QAOA}}
\begin{figure*}
    \centering
    \includegraphics[width=1.0\linewidth]{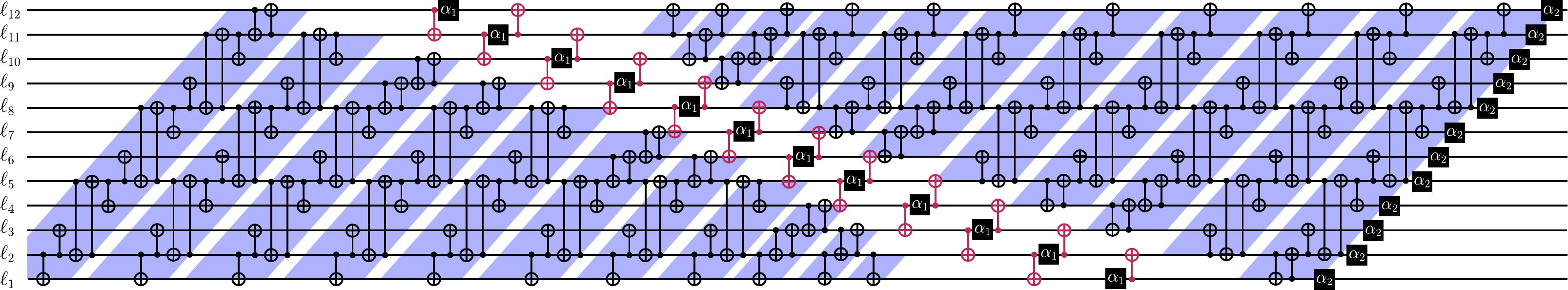}
    \caption{Two subsequent QAOA cycles for the $3 \times 4$ square grid of Fig.~\ref{fig:2body_generator_square_grid_2}(a) using shifted concatenation of two-body generators $\mathcal{G}_{2, \mathrm{grid}}^{(12)}$ and $(\overline{\mathcal{G}}_{2,\mathrm{grid}}^{(12)})^\dagger$ with single-qubit $x$-rotations of the driver unitary in between, indicated by black rectangular boxes. For illustration purposes we neglect the (many) single-qubit $z$ rotations.}
    \label{fig:qaoa_square_grid}
\end{figure*}

QAOA \cite{Farhi2014, Blekos2024} is considered a promising candidate algorithm to solve combinatorial optimization problems using quantum computers. The algorithm follows a quantum-classical protocol to reach the ground state of an (typically) Ising-like problem Hamiltonian. This requires the repeated application of a parameterized problem unitary $U_P(\beta)$, which encodes the problem Hamiltonian, followed by a parameterized driver unitary $U_X(\alpha)$. The concatenation of a driver together with a problem unitary defines a QAOA cycle. The prepared candidate state after $p$ QAOA cycles
\begin{eqnarray}
    \vert \psi (\alpha,\beta) \rangle = \prod_{j=1}^p U_X(\alpha_j) U_P(\beta_j)\vert + \rangle ^{\otimes n} \nonumber
\end{eqnarray}
yields a trial energy expectation value which is then used in a quantum-classical feedback loop to optimize the $2p$ parameters $\alpha_j$, $\beta_j$ with $j \in [1,p]$. In Ref.~\cite{Klaver2024}, the authors discuss the implementation of QAOA for all-to-all connected QUBO problems on LNN devices using parity label tracking. There, the corresponding problem Hamiltonian is described by an all-to-all connected Ising Hamiltonian of the form
\begin{eqnarray}
    H_{\mathrm{QUBO}} = \sum_{k=1}^{n} \sum_{j<k}J_{jk} Z_{j}Z_{k} + \sum_{j=1}^nh_j Z_j. \nonumber
\end{eqnarray}
The implementation of the problem unitary $U_{\mathrm{QUBO}} = \exp(-i \beta H_{\mathrm{QUBO}})$ thus requires the encoding of all logical two-body rotation operators $\exp(-i\beta J_{ij}Z_iZ_j)$, a task that can be accomplished using the clean two body-generator $\G_2^{(n)}$ given in Eq.~\eqref{eq:clean_two_body generator} (intertwined with single-body rotation gates wherever corresponding parity labels are available).

Utilizing the circuit building blocks developed in Secs.~\ref{sec:LNN_kbody} and \ref{sec:generic_connectivity_graphs}, we can generalize this to arbitrary higher order binary optimization (HUBO) problems which we can efficiently implement on a multitude of different connectivity graphs. For example, to encode a problem unitary $\exp(-i \beta H_{\mathrm{HUBO}})$ with a problem Hamiltonian of the form
\begin{eqnarray}
    H_{\mathrm{HUBO}} &=& \sum_{k=1}^{n} \sum_{j<k} \sum_{l<j} M_{jkl} Z_{j}Z_{k}Z_{l} \nonumber\\
    &+&  \sum_{k=1}^{n} \sum_{j<k}J_{jk} Z_{j}Z_{k} +\sum_{j=1}^n  h_j Z_j \nonumber 
\end{eqnarray}
we can use the clean three-body generator circuit $\mathcal{G}_3^{(n)}$ (which is automatically also a generator of all two-body terms). Adapting the required building block of $\mathcal{G}_3^{(n)}$ readily yields the corresponding implementation for the connectivity graph of interest.

In Ref.~\cite{Klaver2024} the authors demonstrate a QAOA encoding with a depth of $2n + \mathcal{O}(1)$ per QAOA cycle (up to an initialization encoding circuit). In contrast, a QAOA algorithm for QUBO problems on LNN devices implemented based on $\mathcal{G}_2^{(n)}$ would yield a depth of ${4n +\mathcal{O}(1)}$ per QAOA cycle. To improve on this, we can employ a similar strategy as used for the design of the three-body generators: Instead of concatenating bare $\mathcal{G}_2^{(n)}$ generators in subsequent QAOA cycles, we use a shifted concatenation of circuits that alternate between the two-body generators $\mathcal{G}_2^{(n)}$ and $\overline{\mathcal{G}_2}^{(n)}$. Moreover, we apply shifted concatenations of the $x$ rotation gates of the driver unitary after each two-body generator. In this way, shifted concatenation of all building blocks achieves a reduction of the total QAOA algorithm depth by a factor $\sim 2$. More concretely, with this trick for $p$ QAOA cycles we obtain a depth of ${2n (1+ 1/p)+\mathcal{O}(1)}$ per cycle on LNN devices. A similar approach can be utilized for all-to-all connected devices, reducing the normalized depth per QAOA cycle from ${2n + \mathcal{O}(1)}$ to ${n (1+1/p) +\mathcal{O}(1)}$ per QAOA cycle.

For qubit layouts with other connectivity graphs, a depth efficient implementation of QAOA \cite{harrigan_quantum_2021, Sachdeva2024, Pelofske2023} requires more care.  As discussed in Sec.~\ref{sec:2d_architectures}, the two-body generator tend to become asymmetric so that shifted concatenations of, for instance, $\mathcal{G}_{2,\mathrm{grid}}^{(n)}$ and $\overline{\mathcal{G}}_{2, \mathrm{grid}}^{(n)}$ do not align properly [compare Fig.~\ref{fig:2body_generator_square_grid_2}(b)]. Each $\mathcal{G}_{2,\mathrm{grid}}^{(n)}$,  $\overline{\mathcal{G}}_{2, \mathrm{grid}}^{(n)}$, respectively, has a depth of $6n + \mathcal{O}(1)$, while their shifted concatenation allows only for an overlap of $4n/3 + \mathcal{O}(1)$. However, according to Lemma~\ref{lem:adjoint_circuit}, since $\overline{\mathcal{G}}_k^{(n)}$ is a clean generator circuit for all $k$-body terms from a sequence of single-body labels, then also the adjoint circuit $(\overline{\mathcal{G}}_k^{(n)})^\dagger$ is a clean $k$-body generator from a set of single-body labels. Thus, instead of concatenating $\mathcal{G}_{2,\mathrm{grid}}^{(n)}$ and $\overline{\mathcal{G}}_{2, \mathrm{grid}}^{(n)}$ for subsequent QAOA cycles, we can also concatenate $\mathcal{G}_{2,\mathrm{grid}}^{(n)}$ and $(\overline{\mathcal{G}}_{2, \mathrm{grid}}^{(n)})^\dagger$. This enables a better alignment and a larger overlap of subsequent QAOA blocks using shifted concatenation. In Fig.~\ref{fig:qaoa_square_grid}, we schematically illustrate this approach for two subsequent QAOA cycles on the $3\times 4$ square grid shown in Fig.~\ref{fig:2body_generator_square_grid_2}(a). In between the two Hamiltonian encoding blocks a permutation of the single-body labels is available so that we can apply the $x$ rotation of the driver unitary. Shifted concatenation of QAOA blocks built from $\mathcal{G}_{2,\mathrm{grid}}^{(n)}$ and $(\overline{\mathcal{G}}_{2,\mathrm{grid}}^{(n)})^\dagger$ reduce the total depth by a factor of $\sim 2$. More precisely, for an even number of QAOA cycles $p$ we obtain a depth of ${3n(p+\frac{4}{9})+\mathcal{O}(1)}$ while for an odd number of QAOA we find ${3n(p+1)+\mathcal{O}(1)}$. Thus, eventually we obtain a depth of ${3n(1+1/p)+\mathcal{O}(1)}$ (${3n(1+\frac{4}{9p})+\mathcal{O}(1)}$) per QAOA cycle given an odd (even) number of cycles. For ${p \gg 1}$, this is on par with the best known depth, however saving a factor of $9/4$ in gate count compared to established \SWNAME{} based approaches \cite{harrigan_quantum_2021, Weidenfeller2022}. The same approach can be applied to other layouts such as heavy-hexagon and ladder. For all the connectivity graphs we investigate, we obtain a circuit depth that is at most equal to the best known result with a simultaneously significant improvement in gate count. In fact in most cases both metrics outperform existing algorithms. Our results for implementations of QAOA on the different connectivity graphs are collected in Tab.~\ref{tab:qaoa}.

{\color{black}
Oftentimes, realistic QUBO or HUBO problems are represented by graphs with an edge density less than all-to-all connectivity. In such cases, one may still use the full generator circuits to encode such problems, yet the implementation becomes less efficient since not all generated labels are actually required for the encoding of the problem graph. Interestingly, even without any further optimizations the exact generator solutions still significantly outperforms state-of-the-art transpilation tools for large regime of edge densities. Fig.~\ref{fig:transpilers_comparison} compares different transpilation tools (at their respective highest optimization level) against our exact solution for different qubit connectivities. For sufficiently small edge densities \texttt{Qiskit} \cite{javadiabhari2024} and \texttt{TKET} \cite{Sivarajah2020} transpilation performs well in comparison to our methods. However, interestingly, already at moderate graph densities around $\sim 0.25-0.55$ our connectivity adapted solutions outperform all investigated compilation tools in gate-count as well as in circuit depth. Fig.~\ref{fig:crossing} depicts the threshold edge density (with respect to the best performing transpiler) as a function of the connectivity. Above this threshold edge density our exact solution should be preferred over explicit transpilation. As connectivity increases we observe a slight drift of the threshold edge density. While our implementations gain efficiency with increased connectivity, so do transpilation tools as more connections can be utilized. Nonetheless, this result is remarkable especially as we employ no additional optimizations on top of our exact solution. Note also that our exact solution scales favorably in terms of transpilation time: since our circuits are already tailored to given connectivity graphs no further adaptions are required other than a direct transpilation into native gate sets (which is trivial and computationally cheap). In contrast, the compute time to obtain optimized transpiled circuits via \texttt{Qiskit} or \texttt{TKET} diverges with increasing problem size. 
Our results are in accordance with recent works conducted on LNN topology \cite{MontanezBarrera2025} where the authors arrive at a similar conclusion. Notably, in contrast to Ref.~\cite{MontanezBarrera2025}, here we apply no additional optimizations on top of our exact solutions. Adding such optimizations, we expect to gain even more efficiency. 

\begin{figure*}
    \centering
    \includegraphics[width=\textwidth]{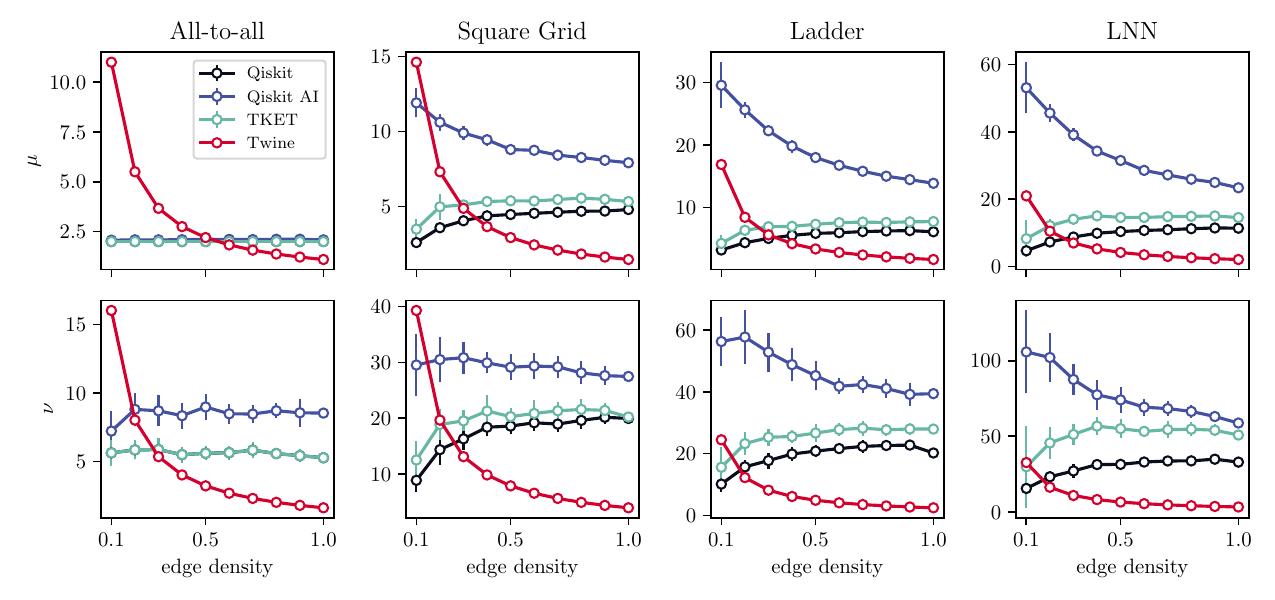}
\caption{Average two-qubit gate count $\mu$ (upper row) and normalized depth $\nu$ (lower row) of transpiled random graphs as a function of the graph edge density. Each data point shows the average of the transpilation of $20$ random graphs each with $20$ nodes. We compare our results (red) with three different transpilation tools using the highest optimization level: \texttt{Qiskit} (black), \texttt{Qiskit AI} (blue) and \texttt{TKET} (green). Error bars denote the standard deviation. The exact connectivity-adapted \TWINE{} solution is shown in red. Note that we pre-transpiled circuits with the standard transpiler of \texttt{Qiskit} with lowest optimization level before transpiling them with \texttt{Qiskit AI} to enforce routing and basis gate decomposition.}
\label{fig:transpilers_comparison}
\end{figure*}
\begin{figure}
    \centering
    \includegraphics[width=0.5\textwidth]{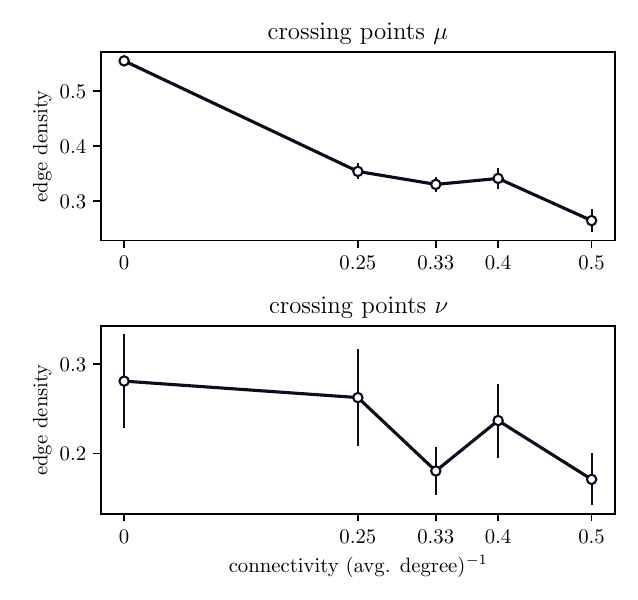}
\caption{Threshold edge density for which Parity Twine outperforms the best transpiler solution (Fig. \ref{fig:transpilers_comparison}) as a function of the device connectivity. We additionally display heavy-hexagonal connectivity graphs with a connectivity of $0.4$. Data for transpilation on heavy-hexagonal layouts is obtained in the same manner as described in Fig. \ref{fig:transpilers_comparison}.}
\label{fig:crossing}
\end{figure}

}
\subsection{Quantum Fourier Transform \label{sec:QFT}}

\begin{figure*}
    \centering
    \includegraphics[width=\textwidth]{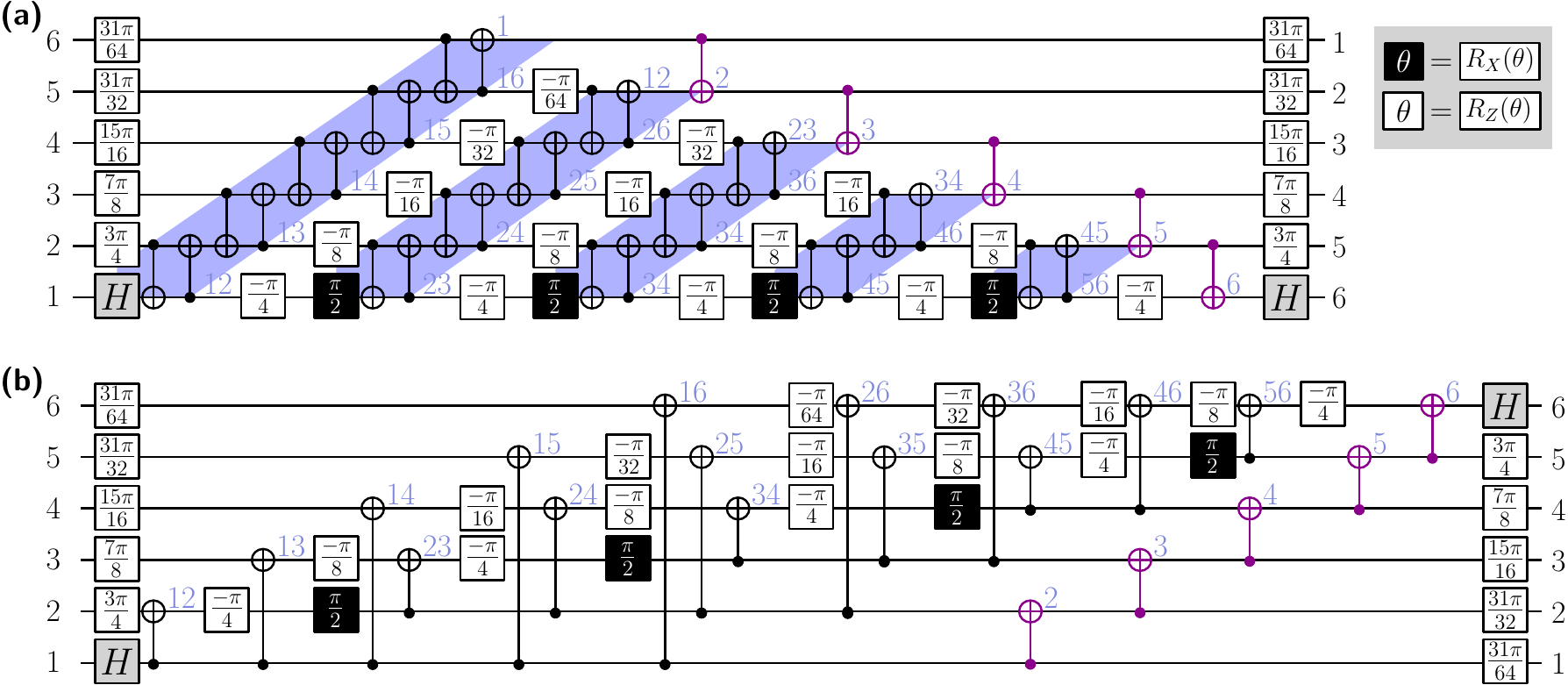}
\caption{QFT circuit implementation for LNN devices (a) and all-to-all connected devices (b). The blue numbers indicate the parity labels at the given circuit moments. Colored \CXNAME{} gates belong to a final cleanup step. Rectangular black boxes indicate $x$-rotation gates with the denoted angle.}
    \label{fig:QFT}
\end{figure*}

The QFT constitutes a major cornerstone algorithm in quantum computing and lays the foundation to a multitude of important algorithms and use cases such as modular exponentiation and phase estimation \cite{Shor1994, Kitaev1999} used in Shor's factoring algorithm, quantum arithmetic \cite{Draper2000, Ruiz-Perez2017} or solving linear systems of equations \cite{Harrow2009}. In recent years, a multitude of algorithms was developed to implement QFT on different connectivity graphs optimizing either depth \cite{Maslov2007, Zhang2021, Jin2023, Gao2024} or gate count \cite{Holmes2020, Park2023, Baumer2024}. Thereby, state-of art approaches typically use architecture optimized \SWNAME{} networks. Here we demonstrate how to apply the conclusion drawn in Secs.~\ref{sec:all-to-all} and \ref{sec:2d_architectures} for the design of gate count and depth optimized implementations of the QFT adapted to different connectivity graphs. 

Recently, some of the authors demonstrated in Ref.~\cite{Klaver2024} how parity label tracking can yield reductions in gate count and depth when implementing the QFT on LNN devices. Up to single-qubit $z$ rotations, the QFT circuit of Ref.~\cite{Klaver2024} can (mainly) be understood as a combination of $\XSC$ circuits neatly intertwined with single-qubit $x$ rotation gates. Fig.~\ref{fig:QFT}(a) depicts the corresponding QFT circuit. In between subsequent $\XSC$ circuits we add the necessary single-body $z$ rotations and Hadamard gates, which we decompose into $z$ and $x$ rotation gates. 
{\color{black}
Unlike for circuits which solely decompose into \CXNAME{} and $z$ rotation gates, the QFT requires a strict order in which labels are created and corresponding rotation gates are applied. In particular, the textbook QFT circuit requires Hadamard gates applied to individual qubits \textit{before} application of corresponding controlled-phase gates. The QFT implementation of Fig.~\ref{fig:QFT}(a) works because of another neat property of $\XSN$ chains which becomes evident if one additionally tracks the $x$-labels: while $\XSC$ chains encode the initial bottom $z$-label into all other labels, the $x$-labels remain unpaired for the most part. Only the top qubit accumulates a string of the traversed $x$-labels (see Fig.~\ref{fig:z_and_x_labels}). Therefore, after each $\XSC$ chain, we can still apply $x$-rotation gates that also logically correspond to single-qubit rotations. This property holds for all connectivity-adapted variations of $\XSC$ chains.

\begin{figure}
    \centering
    \includegraphics[width=0.65\linewidth]{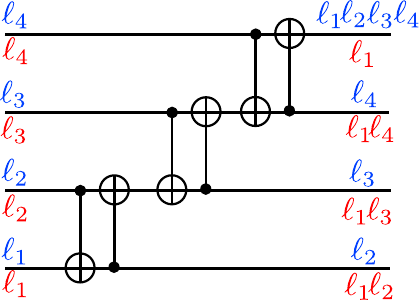}
    \caption{Transformation of $z$ labels (red) and $x$ labels (blue) under a $\XSC$ circuit.}
    \label{fig:z_and_x_labels}
\end{figure}
}
By this reason, we can replace the respective building blocks and readily adapt the approach of Fig.~\ref{fig:QFT} to qubit layouts with other connectivity graphs. In essence, this amounts to replacing each $\XSC$ circuit with its corresponding counterpart on the connectivity graph of interest and inserting corresponding single-body $z$ rotation gates. For instance, using the mapping established in Sec.~\ref{sec:all-to-all}, we can map the circuit of Fig.~\ref{fig:QFT}(a) to a corresponding circuit optimized for all-to-all connected devices shown in Fig.~\ref{fig:QFT}(b). By virtue of mapping $\XSC$ to $\XSC_{\alltoall}$, we can directly conclude that the \CXNAME{} count and depth of the algorithm halves. Interestingly, the structure of the QFT implementation on all-to-all connected devices readily allows to construct an approximate QFT \cite{Takahashi2007, Maslov2007} just from omitting CNOT gates (and corresponding rotations with angles below the approximation threshold). 

Analogously to the procedure outlined above for all-to-all connected devices, once a $\XSC$ circuit together with a decoding circuit is available for a given connectivity graph, implementing a corresponding QFT circuit is straightforward. In Sec.~\ref{sec:2d_architectures} we discuss in detail how to find $\XSC$ circuits using HGPs. The obtained gate count and depth obtained for $\XSN$ circuits can then directly transferred also to the corresponding QFT implementation (up to single-qubit gates). 
The results for our implementations of QFT on the different connectivity graphs are summarized in Tab.~\ref{tab:qft}. Notably, our implementations significantly improve the gate count while the depth is at most the same as that of the best known state-of-the-art implementations.

{\color{black}
\subsection{Hamiltonian simulation \label{sec:hamiltonian_simulation}}
Hamiltonian simulation of complex quantum many-body systems constitutes yet another promising use-case of present day quantum computers. To date, there exist various different methods to simulate the time-evolution under a given Hamiltonian on a quantum computer. Techniques like linear combination of unitaries (LCU) \cite{Childs2012} or block encoding \cite{Low2017} often require highly non-trivial state preparation circuits but have the advantage of error-free simulation. In contrast, methods based on Trotter-Suzuki decomposition \cite{Suzuki1976} do not require complex state preparation circuits, however Trotterization introduces errors and to keep total simulation errors below given tolerance an increased number of Trotter steps is required. Very recently combinations of LCU and Trotter decomposition were proposed to efficiently reduce these errors \cite{Zeng2025}.

Our method of implementing $k$-body generators naturally fits approaches using Trotter decomposition. To outline its effectiveness, here we discuss the implementation of a long-range mixed field Ising model
\begin{eqnarray}
\label{eq:mfim}
    H_{\mathrm{MFIM}} = \sum_{j<k} J_{jk} Z_j Z_k + \sum_j h_j Z_j + g_j X_j.
\end{eqnarray}
Given an operator $O$ of the form $O = t \sum_i A_i$, the second order Suzuki formula \cite{suzukiGeneralTheoryFractal1991a} is given by
\begin{equation*}
    e^{O} =
    e^{\frac{t}{2}A_1}\, \cdots e^{\frac{t}{2}A_{q-1}}\,
    e^{t A_q}\,
    e^{\frac{t}{2}A_{q-1}} \cdots  e^{\frac{t}{2}A_1} + \mathcal{O}(t^3).
\end{equation*}
This allows us to approximate the unitary of Eq.~\eqref{eq:mfim} with
\begin{multline}
\label{eq:mfim_trotter}
    \exp(-i\! \tau\! H_{\mathrm{MFIM}}) = \exp(\!-i\!\frac{\tau}{2}\!\sum_j \!g_j X_j\!)\exp(-i \!\tau\! \sum_j \!h_j Z_j\!)
      \\
 \times \exp(-i\tau\sum_{j<k} J_{jk}Z_jZ_k) \exp(\!-i\!\frac{\tau}{2}\!\sum_j \!g_j X_j\!)
 + \mathcal{O}(\tau^3).
\end{multline}

Tighter bounds on the error of this approximation can be found in \cite{childsTheoryTrotterError2021}.
To encode the unitary of Eq.~\eqref{eq:mfim_trotter} on a quantum device with limited connectivity, we utilize the two-body generators developed in Sec.~\ref{sec:LNN_kbody}. The single-body operators can readily be implemented upfront, at the end, respectively, while the two-qubit interaction terms can be implemented using the two-body generators of Sec.~\ref{sec:LNN_kbody}. Thus, the total two-qubit gate count and depth to implement the unitary of Eq.~\eqref{eq:mfim_trotter} coincides with the count and depth of the corresponding $\XSN$ circuit. Consequently, compared to implementations based on \SWNAME{} networks we save the same number of two-qubit gates as reported in Sec.~\ref{sec:LNN_kbody}.

\section{Performance on noisy quantum hardware \label{sec:performance_on_noisy_hardware}}
\subsection{Comparison of circuit fidelity \label{sec:circuit_fidelities}}
On present-day noisy quantum hardware circuit depth and gate count are similarly detrimental for the execution quality of a given quantum algorithm. Every application of a quantum gate introduces a potential source of errors, while the finite lifetime of qubits prohibits large depth circuits. Assuming only Markovian and uncorrelated noise, these two sources of error can be modeled using an idling fidelity, $F_{\mathrm{idle}}$, and a fidelity describing the expected rate of success when applying a two-qubit gate, $F_{2q}$. Note that in principle this also applies to single-qubit gates. However, since their contribution to the overall error-budget is small and cancels out for the subsequent comparison, we don't include them in this analysis.

Then, to every qubit and for every moment of the circuit we either assign $\sqrt{F_{2q}}$ in case a two-qubit gate is applied, or $F_{\mathrm{idle}}$ if the qubit is idle. With this approach we expect $F_{2q}< F_{\mathrm{idle}}^2$ since $F_{2q}$ also comprises errors emerging from the decay and dephasing of qubits.

The total expected fidelity, i.e.~the performance of a circuit $C$ is thus given by
\begin{eqnarray}
\label{eq:performance}
    \mathcal{F}_C = F_{2q}^{\size{C}}F_{\mathrm{idle}}^{n\cdot\depth{C}-2\cdot \size{C}} .
\end{eqnarray}

Equation~\eqref{eq:performance} allows to directly compare the expected performance of our circuit implementations with reference state-of-the-art implementations. Since our methods are better in count and/or at most equivalent in depth for all the discussed applications, this readily implies an improved performance. 

Many device designs embed sub-graphs that are isomorphic to other relevant layouts. For example a square grid embeds ladder or LNN designs. On such devices, which naturally fit different Parity Twine and other reference implementations, it is interesting to investigate and compare their relative performance. To this end, in App.~\ref{app:circuit_fidelities} we compute the respective performance according to Eq.~\eqref{eq:performance} for individual building blocks. Note that for the subsequent discussion, we explicitly exclude all-to-all connected devices and focus on locally connected qubit layouts. In Fig.~\ref{fig:best_performing} we display the best performing design as a function of idling and two-qubit gate errors. Throughout the physical regime, we find that implementations developed in this work outperform all existing implementations, in particular, also the implementation tailored to LNN qubit layout. Close to the unphysical error regime (gray hashed out part), circuit depth plays a detrimental role for the overall fidelity. In this regime the ladder design is favorable as it provides the most dense implementation. With decreasing ratio $F_{2q}/F_{\mathrm{idle}}$ two-qubit gate count becomes more important, and, hence, eventually the square grid design emerges on top. 

Treating decay and decoherence as quantum channels, the idling fidelity can be estimated by \cite{Economou2022}
\begin{eqnarray}
\label{eq:idle_fidelity}
    F_{\mathrm{idle}} = \frac{1}2 + \frac{1}{6} \left(2 \exp(-T_g/T_2)+\exp(-T_g/T_1)\right),
\end{eqnarray}
where $T_1$ represents the excited state decay time, $T_2$ indicates dephasing and $T_g$ denotes the gate execution time. From publicly available data of $T_1$, $T_2$ and $T_g$ we can estimate the idling fidelity of current quantum hardware. For reference, in Fig.~\ref{fig:best_performing} we include values of five recent IBM machines. 

\begin{figure}
    \centering
    \includegraphics[width=\linewidth]{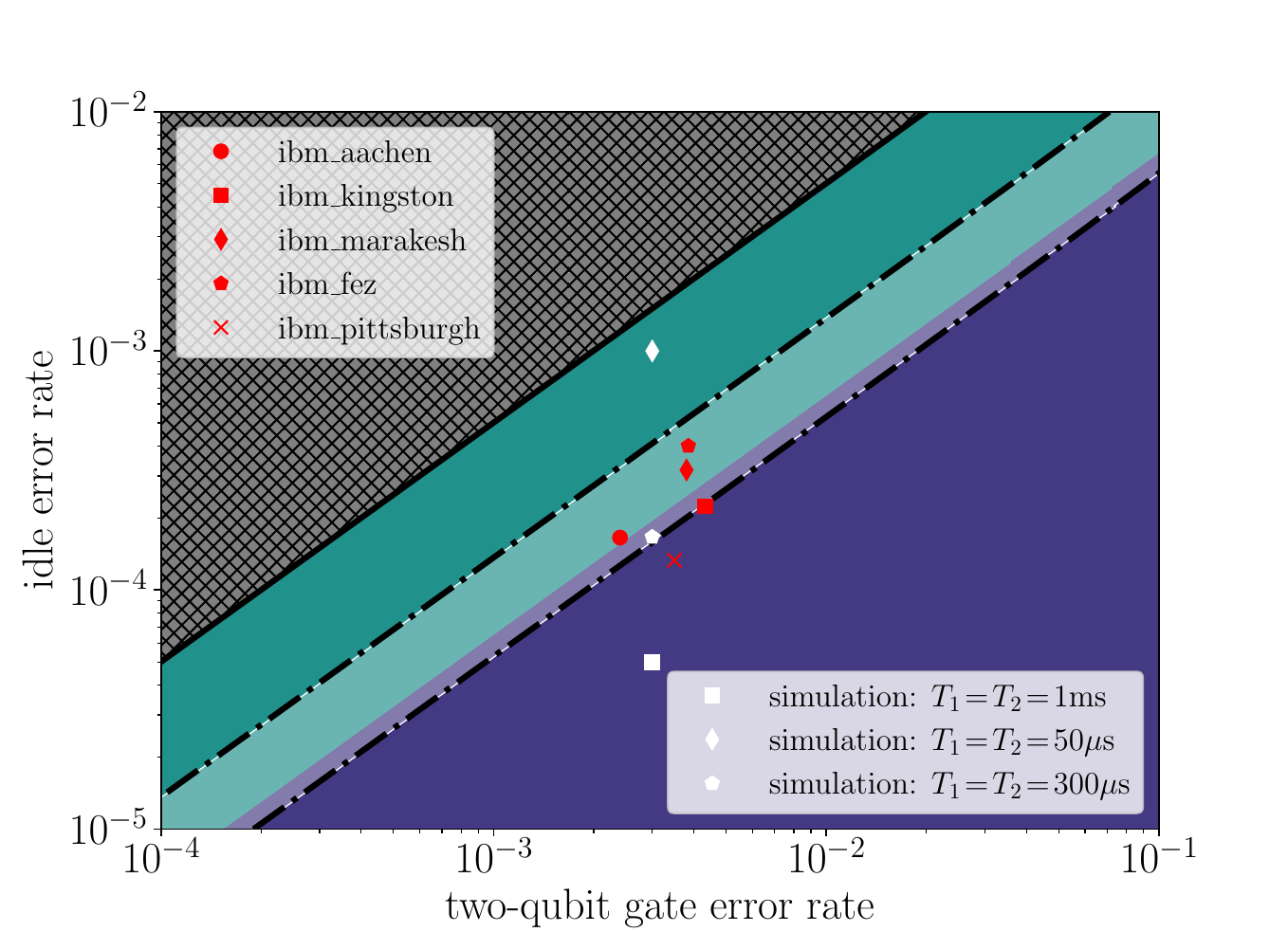}
    \caption{Best performing device design as a function of two-qubit error rate (CNOT error rate) and idling error rate. Green indicates ladder while blue indicates square grid layouts. The hashed out area corresponds to the unphysical region where $F_{2q} > F_{\mathrm{idle}}^2$ (the thick black solid line indicates $F_{2q} = F_{\mathrm{idle}}^2$). The plot displays the transition for $n=18$. In general, the transition depends on $n$. We indicate the transition region by lighter coloring bounded by the black dashed-dotted lines. As $n\rightarrow \infty$ square grid layouts outperform ladder layouts for $F_{2q}>F_{\mathrm{idle}}^{19}$ (lower black dashed-dotted line). Conversely, for the smallest meaningful value, $n=2$, we find the transition at $F_{2q}>F_{\mathrm{idle}}^{\frac{22}{3}}$ (upper black dashed-dotted line). For reference, we include error values of five IBM machines (red) as well as the error values used for some of the simulations (white).}
    \label{fig:best_performing}
\end{figure}

\subsection{Simulations}
To complement the theoretical analysis outlined in the latter section, we now demonstrate the performance by means of concrete examples. To this end, we simulate the execution of our QFT implementation on noisy quantum hardware with different qubit connectivity graphs and directly compare the obtained fidelity to state-of-the-art implementations.

We follow a similar procedure as outlined in Ref.~\cite{Baumer2024} and generate initial states representing $\sigma_k^* = \mathrm{QFT}^{\dagger}\vert k \rangle\langle k \vert \mathrm{QFT}$ with computational basis states $\vert k \rangle $ where $k$ represents a bitstring of length $n$. This can be done by using only local rotations (of which the $z$ rotations can be implemented noiseless). Subsequently, we apply the noisy QFT circuit implementations, $\tilde{\mathcal{QFT}}$, combined with $z$-basis measurements. By extracting the counts of the correct bitstrings which refer to the initial computational basis state, we can deduce the performance of the QFT implementation. The results are combined in the process fidelity~\cite{Baumer2024}

\begin{align}\label{eq:process_fidelity}
    \mathcal{F}_{\mathrm{proc}} &= \frac{m}{m-1} \left[\frac{1}{m}\sum_{l=1}^m\sqrt{ \text{Pr}\left(k_l|  \tilde{\mathcal{QFT}}(\sigma_{k_l}^\ast)\right)}\right]^2\\ \notag 
    &\quad - \frac{1}{m(m-1)}\sum_{l=1}^m \text{Pr}\left(k_l|  \tilde{\mathcal{QFT}}(\sigma_{k_l}^\ast)\right),
\end{align}
where $m$ is the number of initial basis states and $ \text{Pr}\left(k_l|  \tilde{\mathcal{QFT}}(\sigma_{k_l}^\ast)\right)$ the probability to measure the correct bitstring $k_l$ given $\tilde{\mathcal{QFT}}$ of the state $\sigma_{k_l}^\ast$ after the state preparation.

We conduct simulations using the \texttt{Qiskit} \texttt{AerSimulator} \cite{javadiabhari2024}. Thereby, noise is modeled stochastically using a Monte Carlo sampling of different (noisy) channels such as thermal relaxation and dephasing (controlled by $T_1$ and $T_2$) as well as depolarization errors when applying gates. 

First, to validate our theoretical analysis employed in the latter section, we perform a series of simulations with varying values of $T_1$, $T_2$ and two-qubit gate error rates and compare the obtained process fidelities. Similar to Fig.~\ref{fig:best_performing}, in Fig.~\ref{fig:2d_scan} we depict the process fidelity of the best performing implementation as a function of idle and two-qubit gate errors for simulations with $n=12$ qubits. In accordance to the theoretical analysis, the simulations suggest a transition between the ladder implementation, which dominates for small $T_1$ and $T_2$, and square grid, which outperforms ladder as two-qubit gate count becomes more detrimental for the overall circuit fidelity. Notably, the transition occurs approximately where the theoretical analysis (indicated by the solid white line) predicts it.

\begin{figure}
    \centering
    \includegraphics[width=\linewidth]{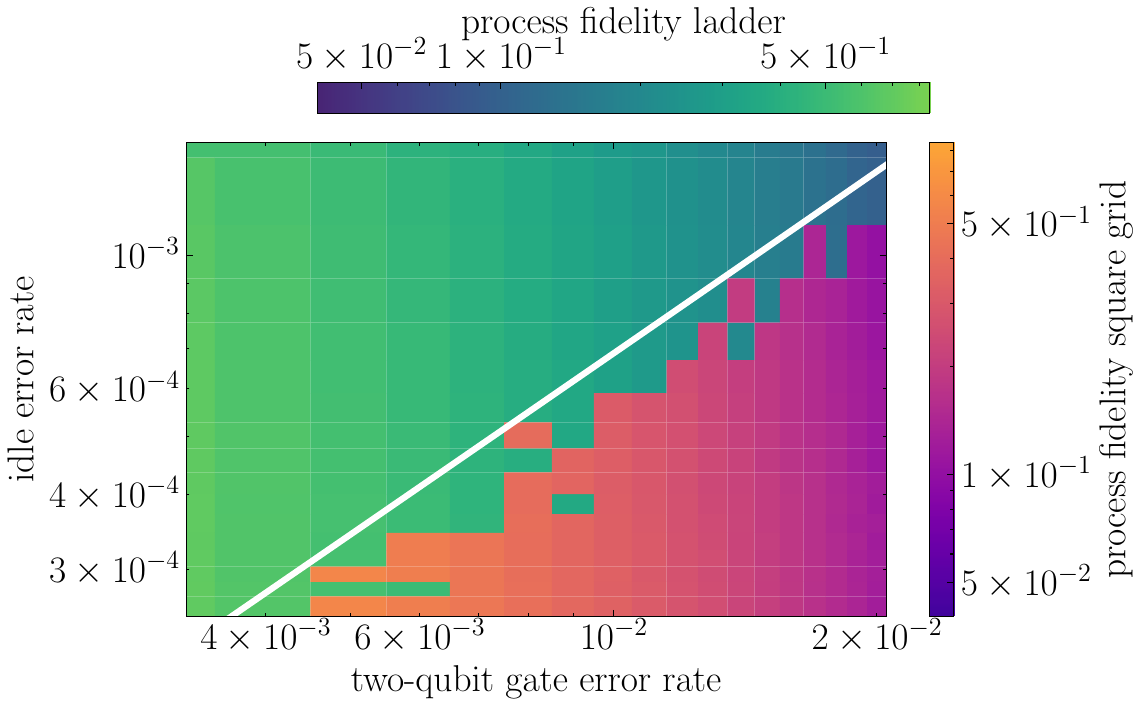}
    \caption{Simulations of the square grid QFT and the ladder QFT implementation. We depict the process fidelity of the best performing implementation as a function of idle and two-qubit gate errors for $n=12$ qubits. The white line indicates the expected transition between ladder and square grid given by $F_{2q} = F_{\mathrm{idle}}^{\frac{117}{8}}$ (see also App.~\ref{app:circuit_fidelities}). We assume a ratio of 10 for the gate execution times of two-qubit gates and single-qubit gates. In particular, we choose $T_{2q} = 100\mathrm{ns}$ for two-qubit gates and $T_{1q} = 10\mathrm{ns}$ as the execution time of single-qubit gates. For the sake of simplicity, we consider only thermal relaxation channels as the source of noise for single-qubit gates. The idle error rates are computed according to $\epsilon_{\mathrm{idle}} = 1 - F_{\mathrm{idle}}$ using Eq.~\ref{eq:idle_fidelity} with varying $T_1$ (we set $T_1=T_2$). Readout errors are neglected. Each data point is obtained using $m=20$ random initial computational basis states with $2000$ shots each.}
    \label{fig:2d_scan}
\end{figure}

Next, in Fig.~\ref{fig:twine_simulation}, we depict simulation results for the process fidelity for the QFT circuit implementations on various connectivity graphs as a function of the number of qubits for three different regimes: in Fig.~\ref{fig:twine_simulation}(a) we assume short qubit lifetimes. In this regime, the depth is detrimental which explains why the ladder implementation shows the best performance. Increasing the qubit lifetime to realistic values of contemporary quantum hardware (Fig.~\ref{fig:twine_simulation} (b), compare also Fig.~\ref{fig:best_performing}), we find that ladder and square grid perform similarly, yet significantly better in comparison to other implementations including heavy hexagon or LNN. Finally, for large qubit lifetimes (Fig.~\ref{fig:twine_simulation} (c)) two-qubit gate count becomes the determining factor of performance. In this regime, an increase in device connectivity translates directly into an increase in process fidelity: With increasing device connectivity, a (quadratically) growing amount of two-qubit gates can be saved which in turn induces an exponential boost in overall fidelity. To demonstrate this with the relatively small number of qubits available in simulations, in Fig.~\ref{fig:fidelity_vs_connectivity} we perform simulations with increased two-qubit gate error rates and plot the result as function of the connectivity.

\begin{figure}
    \centering
    \includegraphics[width=1.0\linewidth]{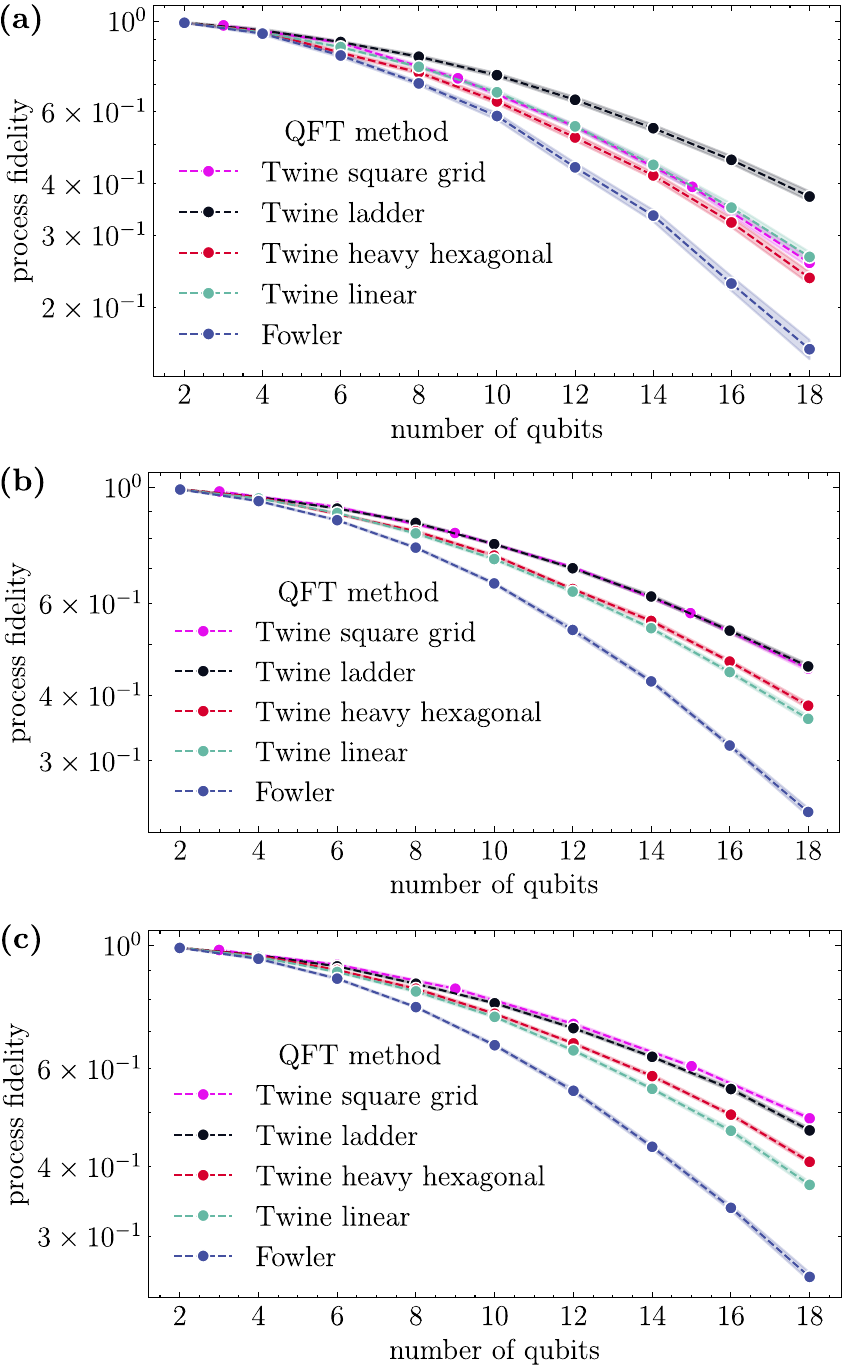}
    \caption{Simulations of different QFT implementations for various qubit connectivity using the \texttt{Qiskit} \texttt{AerSimulator}. We assume the same gate execution times and single-qubit gate errors as given in Fig.~\ref{fig:2d_scan}. The Fowler implementation follows Ref.~\cite{Fowler2004}. 
    In all three plots we assume a two-qubit gate error of $\epsilon_{2q}=3\times 10^{-3}$. Readout errors are neglected.
    (a) Short qubit lifetimes: $T_1 = T_2 = 50\mathrm{\mu s}$. (b) Realistic qubit lifetimes similar to those reported for the latest IBM devices:  $T_1 = T_2 = 300\mathrm{\mu s}$. (c) Long qubit lifetimes: $T_1 = T_2 = 1000\mathrm{\mu s}$.
    The respective shaded error bars indicate the 95\% confidence interval estimated via bootstrapping the data for $m=20$ random initial computational basis states with 2000 shots each. The error values of (a)-(c) are also indicated in Fig.~\ref{fig:best_performing}.}
    \label{fig:twine_simulation}
\end{figure}

\begin{figure}
    \centering
    \includegraphics[width=1.0\linewidth]{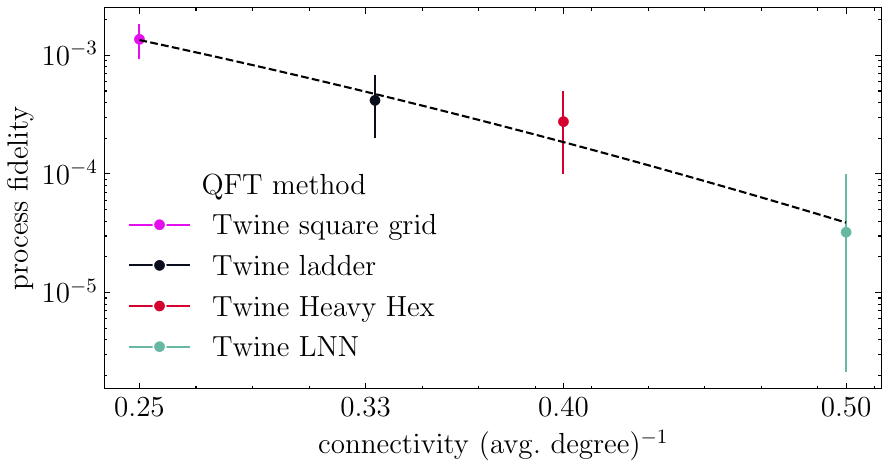}
    \caption{Simulations of different QFT implementations for various qubit connectivity using the \texttt{Qiskit} \texttt{AerSimulator} for $n=18$ qubits. We use the same parameters as in Fig.~\ref{fig:twine_simulation} (c) except for the two-qubit gate error rate for which we use $\epsilon_{2q}=3\times 10^{-2}$ and $T_1=T_2=2\mathrm{ms}$. The dashed black line indicates a fit with $a \times \exp(bx + c x^2)$ with $a =0.014$, $b=-6.92$ and $c =-9.68$. Error bars indicate the 95\% confidence interval.}
    \label{fig:fidelity_vs_connectivity}
\end{figure}

\subsection{Experiments}

\begin{figure*}[t]
  \centering
  \includegraphics[width=\textwidth]{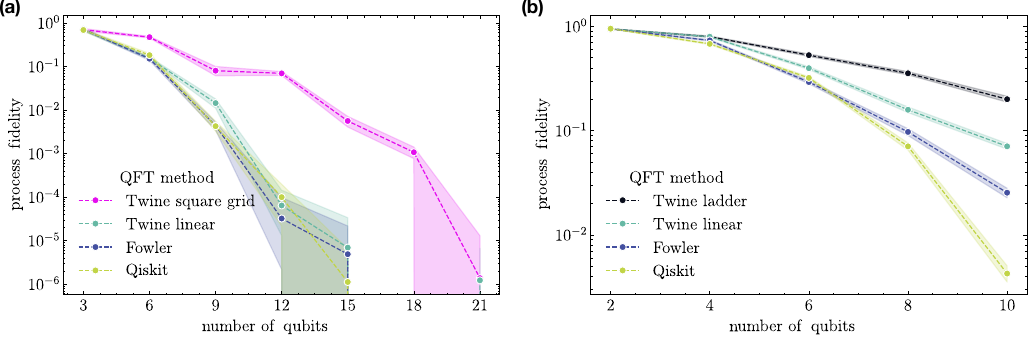}
  \caption{QFT experiments on superconducting square grid QPUs, comparing the process fidelity, as defined in Eq.~\eqref{eq:process_fidelity}, of various Twine methods against the reference implementations Fowler~\cite{Fowler2004} and a connectivity unaware QFT formulation, transpiled to the backend via \texttt{Qiskit}. We use $m=20$ random initial states where for each we measure 2000 shots. The shaded error bars are the percentile interval of the statistics gained from bootstrapping the $m=20$ data points, indicating the area where 95 \% of the data falls. (a) Experiments on the IQM Emerald including square grid and LNN QFT implementations. (b) Experiments on the IQM Garnet, including ladder and LNN QFT implementations as well as alternative QFT implementations. Different implementations in (a) and (b) all use the same set of qubits.}
  \label{fig:experiments}
\end{figure*}

While simulations allow to investigate different parameter regimes and gain some intuition, they are limited in scope and based on assumptions that might potentially influence the results. This especially concerns the employed noise models which rarely capture a comprehensive picture of all noise sources relevant for modern day quantum devices. Thus, it is important to validate theoretical and numerical results with actual experimental runs on real quantum hardware. 

In contrast to the theoretical analysis and the numerical simulations of the latter section, on real quantum hardware, qubit lifetimes and gate fidelities vary among different qubits even on a single chip. A valid comparison between different implementations of the same algorithm can thus be conducted in two ways: (i) performing a statistical analysis of fidelity distributions obtained from multiple runs on varying sets of qubits or (ii) eliminating the effects of fidelity fluctuations by using the exact same set of qubits when comparing two implementations. Clearly option (i) consumes lots of resources, while option (ii) is not possible for some device topologies. For example, on heavy hexagon devices a direct comparison of an LNN implementation with a heavy-hex-tailored implementation is not possible via (ii) simply because we cannot use the same set of qubits.

Instead, on square grids we can use the exact same set of qubits to directly compare different variations of our circuit implementations since there exist sub-graphs isomorphic to ladder and LNN layouts. To this end, we perform experimental runs of various QFT implementations on different IQM square grid devices. Fig.~\ref{fig:experiments} displays the obtained process fidelities. Note again that a direct comparison between the \TWINE{} square grid implementation and the \TWINE{} ladder implementation is not possible via (ii) as they necessarily use different qubits. However, square grid and ladder implementations can separately be compared against linear and alternative reference implementations: Fowler~\cite{Fowler2004}, based on \SWNAME{} gates, and a \texttt{Qiskit} compiled QFT implementation.

While the process fidelity of most QFT methods drops to indistinguishably small values in the data shown in Fig.~\ref{fig:experiments}(a) at a system size of around 12 qubits, the \TWINE{} square grid QFT still shows performance of several orders of magnitudes higher. Concretely, at a system size of $n=18$ our implementation still provides a process fidelity of $\sim 10^{-3}$, where all other methods didn't sample the correct bitstring at least once. Surprisingly, the LNN QFT implementation does not show a significant performance advantage over the reference implementations (in contrast to the simulations and theoretical expectations). This might be explained by single low-fidelity qubits that are addressed more often in the LNN implementation. Instead, on the IQM Garnet (Fig.~\ref{fig:experiments}(b)), we find a more versatile picture. Here the results resemble more the insights and expectations gained from our theoretical and numerical analysis: while the LNN implementation improves over the reference implementations, the ladder based implementation significantly outperforms all other implementations.

}

\section{Conclusion \label{sec:conclusion}}
Efficient quantum algorithms constitute a cornerstone for the development of quantum computing. In this work, we have contributed to this development by establishing generic algorithmic tools for the efficient implementation of logical many-body operators on quantum devices with typical contemporary qubit connectivity graphs. Resulting algorithms are highly optimized both in gate count and circuit depth and significantly outperform competing approaches. 

Using the developed tools we have constructed generator circuits that are surprisingly close to the provable lower bound for the gate count or -- in the case of all-to-all connected devices -- even align with it. We have derived a generic framework for the construction of generator \CXNAME{} circuits of parity label sets with $k$-body labels adaptable to a wide range of connectivity graphs and we have analyzed the role of the device connectivity. Thereby, we have found that even moderate increments in connectivity can yield significant efficiency improvements.
This indicates that full qubit connectivity is not necessarily required to obtain reasonably efficient implementations of quantum algorithms.

We have investigated five different qubit layouts in depth: (1) LNN devices with nearest neighbor connectivity, planar connectivity graphs among which are (2) heavy hexagon devices (3) ladder devices (4) square grid devices and (5) all-to-all connected devices with complete connectivity graphs. The average gate count of the derived generator circuits interpolates from LNN (${\mu =2}$) over heavy hexagon  (${\mu=5/3}$), ladder (${\mu=3/2}$) and square grids ($\mu=4/3$) to all-to-all with an optimal asymptotic average gate count of ${\mu=1}$. The developed formalism can be understood as a means to map the circuits between graphs of different connectivity. This lets us conjecture that the class of optimal circuits derived within our formalism extends beyond complete graphs. As a first step in this direction we have proven that on LNN devices ${\mu>1}$ for generator circuits of label sets with odd-body labels. More precisely, we have shown ${\mu \geq 1 + \frac{1}{9}}$ in this case but we have justified belief that this bound is not optimal. 
{\color{black}
We have investigated different applications of our formalism: (i) quantum optimization algorithms, specifically the QAOA applied to QUBO and HUBO problems, (ii) the QFT and its approximation on all-to-all connected devices and (iii) Hamiltonian simulation based on Trotter-Suzuki decomposition. On all investigated connectivity graphs, our approaches significantly outperform existing algorithms in gate count. Simultaneously, the depth of our algorithms beats the best known approaches on the majority of platforms, where at most we are on par with existing algorithms. Notably, our exact solutions, which are tailored to encode all-to-all connected problems, are so efficient that they even outperform the best state-of-the art compilation tools on sparse problems.

Finally, to further showcase the practical performance advantage of our approach, we have conducted a theoretical analysis for noisy devices and performed a series of numerical simulations and experiments on quantum hardware for the QFT. The results align with our theoretical analysis, demonstrating a clear performance advantage in the parameter regimes relevant to both current and future quantum computers.
}
\section*{Acknowledgments}
This study was supported by the Austrian Research Promotion Agency (FFG Project
No. FO999924030, FFG Basisprogramm) as well as funded in part by the Austrian Research Promotion Agency (FFG Project No. 884444, QFTE 2020), NextGenerationEU via FFG and Quantum Austria (FFG
Project No. FO999896208) and the Horizon Europe programme HORIZON-
CL4-2022-QUANTUM-02-SGA via the project 101113690 (PASQuanS2.1). Moreover, this publication has received funding by the Austrian Science Fund (FWF) SFB BeyondC Project No. F7108-N38, as well as funding within the QuantERA II Programme that is supported by the European Union’s Horizon 2020 research and innovation programme under Grant Agreement No. 101017733. For the purpose of open access, the authors have applied a CC BY public copyright license to any Author Accepted Manuscript version arising from this submission. Furthermore, funding is acknowledged by the German Federal Ministry of Research, Technology and Space within ATIQ (Project No. 13N16127),  MuniQC-Atoms (Project No. 13N16080) an MuniQC-SC (Project No. 13N1618). {\color{black}We acknowledge the use of IQM services for this work. The views expressed are those of the authors and do not reflect the official policy or position of IQM.}

\appendix
\onecolumngrid

\section{Notation and auxiliary results}
\label{sec:notation_auxiliary_results}
This section presents the notations and auxiliary results that are used throughout the paper.
\begin{definition}[\CXNAME{} gate]
    \label{def:cnot_gate}
    Suppose that ${n\geq 2}$ and ${i,j\in\set{1,\ldots,n}}$ with ${i\neq j}$. We denote by ${\CNOT_{i,j}\colon (\C^2)^{\otimes n}\to(\C^2)^{\otimes n}}$ the unitary operator which is defined by ${\CNOT_{i,j}(\ket{a_1\cdots a_n})\coloneqq \ket{b_1(a)\cdots b_n(a)}}$
    where
    \begin{equation*}
        b_k(a)\coloneqq\begin{cases}
            a_k,\quad & k\neq j,\\
            a_i\oplus a_j,\quad & k=j. 
        \end{cases}
    \end{equation*}
    for ${a_1,\ldots,a_n\in\set{0,1}}$. We call the operator \emph{\CXNAME{} gate with control qubit $i$ and target qubit $j$}.
\end{definition}
\begin{definition}[Circuit]
    \label{def:quantum_circuit}
    A (quantum) circuit $C$ is a finite sequence which maps every time step - \emph{moments} - enumerated from $1$ to $d$ to a set of non-overlapping gates. For a non-empty sequence, the number of elements of each set can be zero except for the set in the first and last moment. The \emph{depth} of the circuit is $d$ and denoted by $\depth{C}$, and the \emph{gate count} is the number of gates contained in the circuit which we denote by $\size{C}$. For a set of non-overlapping gates $C_1$, we also write for brevity $C_1$ for the circuit $(C_1)$ which contains only the set $C_1$. Note that we also call a circuit $C$ a circuit for \emph{complete graphs}. In general, if every gate in a circuit $C$ can be implemented on a device with a certain connectivity graph, then we call the circuit $C$ a circuit for that connectivity graph.
\end{definition}

In this article, we propose \emph{\CXNAME{} circuits} which are circuits where all gates are \CXNAME{} gates.
To emphasize that a circuit $C$ acts on $n$ qubits we write $C^{(n)}$. If the circuits only acts on the qubits $p,\ldots,q$, we write $C^{(p,q)}$. Furthermore, we also use $\CNOT_{m,l}$ to denote the circuit $(\set{\CNOT_{m,l}})$ which only contains the \CXNAME{} gate $\CNOT_{m,l}$ for $m\neq l$.

Next, we introduce the term \emph{reversed and adjoint circuit} as well as the \emph{concatenation} of two circuits. These operations play a crucial role in reducing the depth of $k$-body generators, introduced below.
\begin{definition}[Reversed \CXNAME{} circuit]
    \label{def:reversed_circuit}
    Let $C$ be a \CXNAME{} circuit with depth $d$. Then, the \emph{reversed circuit $\overline{C}$} of $C$ is defined as the circuit with depth $d$ where each moment $m$ contains the reversed \CXNAME{} gates of the \CXNAME{} gates of the circuit $C$ at the moment $m$. Here, the \emph{reversed \CXNAME{} gate} $\overline{\CNOT_{i,j}}$ of $\CNOT_{i,j}$ with $1\leq i,j\leq n$ and $i\neq j$ is defined as $\CNOT_{n+1-i,n+1-j}$. We also write $\overline{C}^{(n)}$ to emphasize that the \CXNAME{} gates are reversed between all qubits $1,\ldots,n$. Moreover, we also denote by $\overline{C}^{(p,q)}$ the circuit which reverses all \CXNAME{} gates of the circuit $C$ between qubits $p$ and $q$. In this case, every \CXNAME{} gate $\CNOT_{i,j}$ for some $p\leq i,j\leq q$ with $i\neq j$ is replaced by $\CNOT_{p+q-i,p+q-j}$.
\end{definition}
\begin{example}
    \label{examp:rev_operator}
    Let $n=7$ and $C$ be the circuit on the left-hand side in Fig.~\ref{fig:examp_perm_operator}. The corresponding reversed circuit $\overline{C}^{(2,6)}$ of $C$ between qubits two and six is shown on the right-hand side in Fig.~\ref{fig:examp_perm_operator}.
    \begin{figure}[h!]
    \centering
    \includegraphics[width=.5\textwidth]{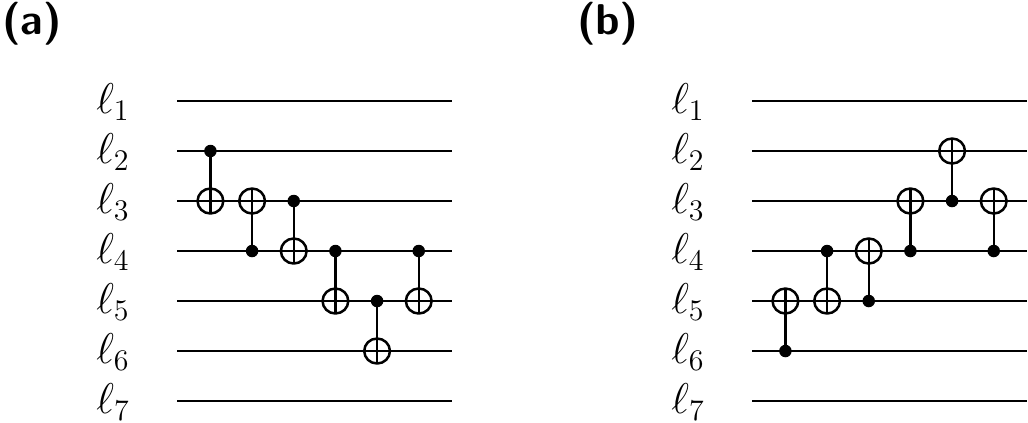}
    \caption{Circuit $C$ (a) and the reversed circuit $\overline{C}^{(2,6)}$ (b) in Example \ref{examp:rev_operator}.}
    \label{fig:examp_perm_operator}
    \end{figure}
\end{example}
\begin{definition}[Adjoint circuit]
\label{def:adjoint_circuit}
    Let ${C=(C_1,\ldots,C_{d})}$ be a circuit where ${C_1,\ldots,C_{d}}$ denote the sets of non-overlapping gates. Then, the \emph{adjoint circuit} of $C$ is denoted by $C^\dagger$ and defined as the sequence ${(C_{d}^\dagger,\ldots,C_1^\dagger)}$ where $C_{l}^\dagger$ denotes the set which contains the adjoint operators of all quantum gates in $C_{l}$ for ${l=1,\ldots,d}$.
\end{definition}
\begin{definition}[Shifted concatenation of circuits]
    \label{def:shifted_concatenation}
    For two circuits $C_1$ and $C_2$ we define the \emph{concatenation}
    $$ C_1 \cat C_2 $$
    as the circuit which first contains all moments from $C_1$ and then all moments from $C_2$. The \emph{shifted concatenation by $s\geq 0$ moments}
    $$ C_1 \cat_s C_2 $$
    is the circuit which contains all moments of $C_1$ as-is and \emph{in addition} all moments of $C_2$ but shifted by $s$ moments to the right. Of course this is only well-defined if the resulting circuit does not contain overlapping gates. Similarly, we set
    $$ C_1 \cat_{-s} C_2 $$
    as the circuit which corresponds to the circuit $C_1 \cat C_2$ but with $C_2$ shifted to the left by $s$ moments afterwards.
\end{definition}
We emphasize that the subtle difference between positive and negative shifts is that positive shifts start at the left of the first circuit whereas negative shifts start at the right to determine the time moment where the first circuit is concatenated with the second one. Moreover, we note that the concatenation operator is not an associative operation. Whenever we write $$\concat[s_1]{C_1}{\concat[s_2]{C_2}{\concat[s_{m-1}]{\cdots}{C_m}}}$$ for some circuits $C_1,\ldots,C_m$ with shifts $s_1,\ldots,s_{m-1}\geq 0$, we mean the circuit $$\concat[s_1]{C_1}{\left(\concat[s_2]{C_2}{\left(\concat[s_{m-2}]{\cdots}{\left(\concat[s_{m-1}]{C_{m-1}}{C_m}\right)}\right)}\right)}.$$

Figs.~\ref{fig:example_merge_all}(a)-(c) show examples of concatenated circuits with six qubits. We observe that the depth of $\concat{C_1}{\concat{C_2}{C_3}}$ [see Fig.~\ref{fig:example_merge_all}(a)] is nine whereas the depth of $\concat{C_1}{(\concat[-2]{C_2}{C_3})}$ [see Fig.~\ref{fig:example_merge_all}(b)] and $\concat[-2]{C_1}{(\concat[-2]{C_2}{C_3})}$ [see Fig.~\ref{fig:example_merge_all}(c)] is seven and five, respectively, and their corresponding unitary operators are equal. This shows that finite sequences of valid sets of controlled-target qubit pairs with different depths can produce the same resulting unitary operators.

\begin{figure}[h!]
    \centering
    \includegraphics[width=.8\textwidth]{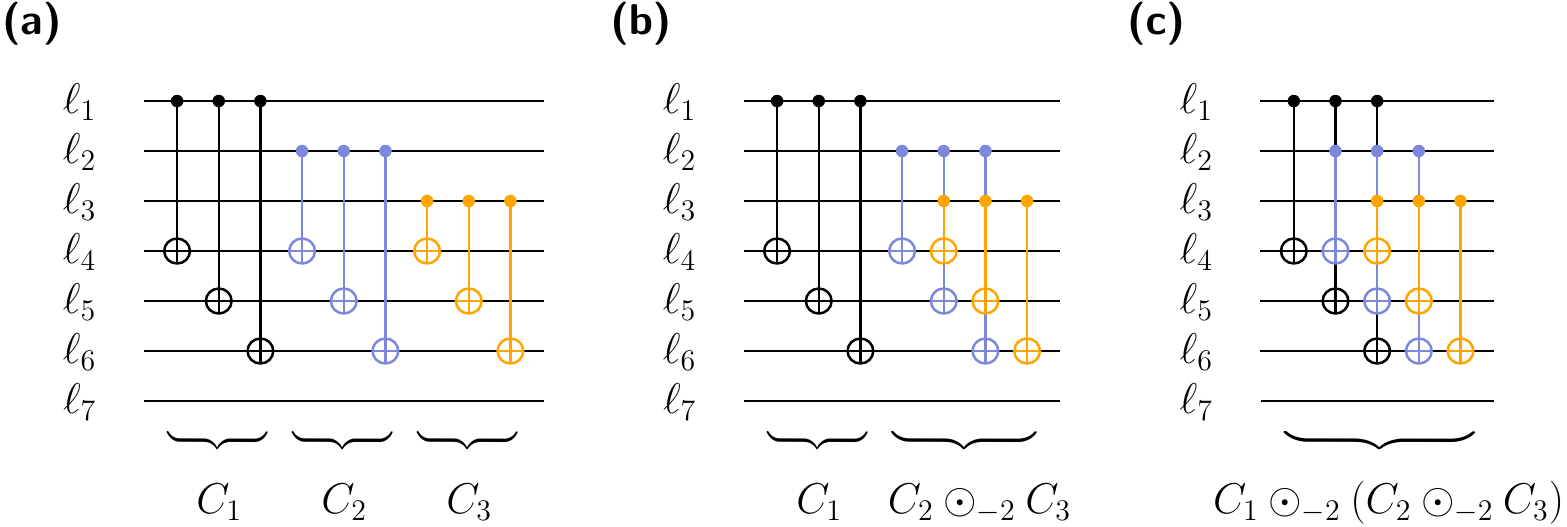}
    \caption{(a) $\concat{C_1}{\concat{C_2}{C_3}}$. (b)  $\concat{C_1}{(\concat[-2]{C_2}{C_3})}$. (c) $\concat[-2]{C_1}{(\concat[-2]{C_2}{C_3})}$}.
    \label{fig:example_merge_all}
\end{figure}
    
To define $k$-body generators, the following definitions are necessary.
\begin{definition}[Parity label]
    \label{def:parity_label}
    Given a set of qubits ${Q=\{1,\ldots,n\}}$ a \emph{(parity) label} $\lab$ is a subset of $Q$. For two labels $\lab_1$ and $\lab_2$ we denote by ${\lab_1\lab_2\coloneqq \lab_1\symdif\lab_2}$ the \emph{symmetric difference} of the two labels (all qubits which are contained in exactly one of the two labels), which is again a label. Moreover, we call the empty label the \emph{trivial} (or \emph{unphysical}) label. Note that the set of all labels over $Q$ equipped with symmetric difference as an addition is a vector space over $\BB$. The zero vector is the trivial label.
\end{definition}

\begin{definition}[Label action]
    \label{def:label_action}
    Let $C$ be a \CXNAME{} circuit on ${n\geq 2}$ qubits with depth ${d}$ and ${\lab_1,\ldots,\lab_n}$ some labels.
    \begin{enumerate}
        \item For $1\leq i,j\leq n$ with $i\neq j$ the \emph{right action of $\CNOT_{i,j}$ on $\labseq{\lab_1,\ldots,\lab_n}$} is defined as
        \begin{equation*}
            \labaction{\labseq{\lab_1,\ldots,\lab_n}}{\CNOT_{i,j}}\coloneqq\labseq{\lab_1,\ldots,\lab_{j-1},\lab_i\lab_j,\lab_{j+1},\ldots,\lab_{n}}.
        \end{equation*}
        \item Let $C_m\coloneqq\set{\CNOT_{i_1,j_1},\ldots,\CNOT_{i_{k_m},j_{k_m}}}$ be the set of all non-overlapping \CXNAME{} gates at the moment $m$ of the circuit $C$. Then, the \emph{right action of $C_m$ on $\labseq{\lab_1,\ldots,\lab_n}$} is defined as
        \begin{equation*}
            \labaction{\labseq{\lab_1,\ldots,\lab_n}}{C_m}\coloneqq\labaction{\left(\labaction{\left(\labaction{\labseq{\lab_1,\ldots,\lab_n}}{\CNOT_{i_1,j_1}}\right)}{\cdots \CNOT_{i_{k_m-1},j_{k_m-1}}}\right)}{\CNOT_{i_{k_m},j_{k_m}}}.
        \end{equation*}
        \item Let $C_1,\ldots,C_{d}$ be the sets of all non-overlapping \CXNAME{} gates at every moment of the circuit $C$. Then, the \emph{right action of the circuit $C$ on $\labseq{\lab_1,\ldots,\lab_n}$} is defined as
        \begin{equation*}
            \labaction{\labseq{\lab_1,\ldots,\lab_n}}{C}\coloneqq\labaction{\left(\labaction{\left(\labaction{\labseq{\lab_1,\ldots,\lab_n}}{C_1}\right)}{\cdots C_{d-1}}\right)}{C_{d}}.
        \end{equation*}
    \end{enumerate}
\end{definition}
\begin{comment}
    \todo[inline]{Reinhard: Here again I would not know what $\ket{a_{\lab_1}\ldots a_{\lab_n}}$ is.}
\todo[inline]{Florian: See above comment.}
\end{comment}
\begin{remark}
    Defining $$G_n\coloneqq\set{C\mid C\text{ is a \CXNAME{} circuit with $n$ qubits}}\quad \text{and}\quad M_n\coloneqq\set{\labseq{\lab_1,\ldots,\lab_n}\mid \lab_1,\ldots\lab_n \text{ are labels}},$$ it can be shown that the mapping
        \begin{equation*}
            \cdot\colon M_n\times G_n\to M_n\colon (\lab,C)\mapsto \labaction{\lab}{C}
        \end{equation*}
            fulfills the properties of a right group action of $G_n$ on $M_n$. However, note that assuming the empty sequence is the identity element, the set $G_n$ together with the concatenation $\cat$ does not form a group since there is no inverse element for any non-empty sequence.
\end{remark}

\begin{definition}[Label generator]\label{def:label-generator}
    \label{def:label_generator}
    Consider $n$ qubits and a set $L^{(n)}$ of non-trivial labels over these qubits. We say that a circuit $C^{(n)}$ generates the labels $L^{(n)}$ from a sequence of labels ${\lab=\labseq{\lab_1,\ldots,\lab_{n}}}$ if the following holds: Let $C^{(n)}_{1,m}$ be the subcircuit which only contains the first $m$ moments.
    Collect all labels occurring in each of the $$\lab C^{(n)}_{1,m} \quad  \text{ for } m = 1, \ldots, \depth{C^{(n)}}$$ into a set $L'$. If $L^{(n)}$ is a subset of $L'$, we say that the circuit generates $L^{(n)}$ from $\lab$. Moreover, we say that $C^{(n)}$ generates some label $\hat{\lab}$ from $\lab$ if $\hat{\lab}\in L'$. In addition, we call $C^{(n)}$ \emph{clean} if and only if $C^{(n)}$ generates the labels $L^{(n)}$ from $\lab$ and there exists a permutation $\pi\colon\set{1,\ldots,n}\to \set{1,\ldots,n}$ such that $\labaction{\lab}{C^{(n)}}=\labseq{\lab_{\pi(1)},\ldots,\lab_{\pi(n)}}$.
\end{definition}
This definition immediately leads us to the notion of $k$-body generators.
\begin{definition}[$k$-body generator]
    \label{def:k-body_generator}
    Let $C^{(n)}$ be a circuit with $n$ qubits, $k\geq 2$ and $\lab=\labseq{\lab_1,\ldots,\lab_{n}}$ with some labels $\lab_1,\ldots,\lab_n$. We say that $C^{(n)}$ is a \emph{$k$-body generator} if $C^{(n)}$ generates all $k$-body labels $$\set{\lab_{i_1}\lab_{i_2}\ldots\lab_{i_k}\mid i_1<\ldots<i_{k}\leq n}$$ from $\lab$.
\end{definition}
Next, we introduce a metric for \CXNAME{} circuits measuring the average \CXNAME{} count as well as the normalized depth for generating a certain set of labels.
\begin{definition}[Average \CXNAME{} count, normalized \CXNAME{} depth]
    \label{def:average_cnot_count_norm_depth}
    Let $C^{(n)}$ be a \CXNAME{} circuit and suppose that $C^{(n)}$ generates the label set $L^{(n)}$ on $n$ qubits. Then, we define
    \begin{equation*}
        \mu_n(C^{(n)}, L^{(n)})\coloneqq \frac{\size{C^{(n)}}}{\abs{L^{(n)}}}\quad\text{and}\quad\nu_n(C^{(n)},L^{(n)})\coloneqq\frac{\depth{C^{(n)}} n}{2\abs{ L^{(n)}}}
    \end{equation*}
    as the \emph{average \CXNAME{} count} and \emph{normalized \CXNAME{} depth} of the circuit $C^{(n)}$ generating the label set $L^{(n)}$, respectively. Furthermore, we set
    \begin{equation*}
        \mu(C, L)\coloneqq\liminf_{n\to\infty} \mu_n(C^{(n)}, L^{(n)})\quad\text{and}\quad\nu(C, L)\coloneqq\liminf_{n\to\infty}\nu_n(C^{(n)},L^{(n)})
    \end{equation*}
    as the \emph{asymptotic average \CXNAME{} count} and \emph{asymptotic normalized \CXNAME{} depth} of a family of circuits $C= (C^{(n)})_{n\geq m}$ and a family of label sets $L= (L^{(n)})_{n\geq m}$ starting from some natural number $m\geq 2$, respectively, where $C^{(n)}$ generates the label set $L^{(n)}$ for $n\geq m$.
\end{definition}
\begin{remark}
    \label{rem:k-body_generator} Let $C_k=(C_k^{(n)})_{n\geq k}$ be a family of $k$-body generators and $\mathcal{L}_k=(\mathcal{L}_k^{(n)})_{n\geq k}$ be the family of label sets $\mathcal{L}_k^{(n)}$ which contain all $k$-body labels on $n\geq k$ qubits.
    \begin{enumerate}
        \item \label{rem:mucx-properties-ineq} $\mu(C_k,\mathcal{L}_k)\geq 1$ and $\nu(C_k,\mathcal{L}_k)\geq 1$.
        \item If $\size{C_k^{(n)}}=cn^k+g(n)$ with $\lim_{n\to\infty}\frac{g(n)}{n^k}=0$ and $c\geq 0$, then $c\geq \frac{1}{k!}$.
    \end{enumerate}
\end{remark}
The statements from the above remark follow from the following lemma.
\begin{lemma}[Trivial size bound]\label{lem:trivial-size-bound}
    Let $C^{(n)}$ be a \CXNAME{} circuit which generates a set $L^{(n)}$ of labels not containing the labels of the start sequence. Then the size of $C^{(n)}$ is at least $\abs{L^{(n)}}$ and the depth of $C^{(n)}$ at least $\abs{L^{(n)}}/(n/2)$.
\end{lemma}
\begin{proof}
    Follows from the fact that each \CXNAME{} gate can only produce at most one new label and every moment of a circuit on $n$ qubits can have at most $n/2$  \CXNAME{} gates.\qedhere
\end{proof}
As we observe from Remark \ref{rem:k-body_generator} \ref{rem:mucx-properties-ineq}, for every family of $k$-body generators we have that $\mu_{\CNOT}(V_k)\geq 1$. The goal of the next sections is to construct families of $k$-body generators with an average \CXNAME{} count and a normalized depth as small as possible.

We conclude this section with the following two lemmas, which will be repeatedly applied to reduce the depth of the circuits presented in the next sections. The first lemma is utilized for the construction of our depth-optimized $k$-body generators [see Figs.~\ref{fig:special_three_body_gen}(c), \ref{fig:4body_gen}(a), for instance], while the application of the second lemma further reduces the depth of our QAOA circuits (see Fig.~\ref{fig:qaoa_square_grid}, for example).

\begin{lemma}
    \label{lem:rev_circuit}
    Let $C^{(n)}$ be a \CXNAME{} circuit acting on $n$ qubits with depth $d$ and $\lab=
    \labseq{\lab_1,\ldots,\lab_n}$ a sequence of labels. Then, it holds
    \begin{equation}
        \label{eq:rev_circuit}
        \labaction{\lab}{\overline{C}^{(n)}}=\overline{\overline{\lab}C^{(n)}}.
    \end{equation}
    Here, for any sequence of labels $\lab=
    \labseq{\lab_1,\ldots,\lab_n}$ we set $\overline{\lab}\coloneqq(\lab_n,\lab_{n-1},\ldots,\lab_1)$. As a consequence, if $C^{(p,q)}$ is acting on qubits $p,\ldots,q$, it holds
    \begin{equation}
        \label{eq:rev_circuit2}
        \labaction{\lab}{\overline{C}^{(p,q)}}=\overline{\overline{\lab}C^{(p,q)}},
    \end{equation}
    where in this case $\overline{\lab}\coloneqq(\lab_1,\ldots,\lab_{p-1},\lab_q,\lab_{q-1},\ldots,\lab_p,\lab_{q+1},\ldots,\lab_{n})$. Moreover, if $C^{(p,q)}$ generates some label $\lab_{i_1}\lab_{i_2}\ldots\lab_{i_m}$ with $1\leq i_1<\ldots< i_m\leq n$ from $\lab$, then $\overline{C}^{(p,q)}$ generates the label $\lab_{p+q-i_1}\lab_{p+q-i_2}\ldots\lab_{i_{p+q-i_m}}$ from $\lab$.
\end{lemma}

\begin{proof}
    First, suppose that $C$ contains only one \CXNAME{} gate, and let $i$ the control and $j$ the target qubit. Since $\overline{\lab}_m=\lab_{n+1-m}$ for $m=1,\ldots,n$, we observe from the definition of the \CXNAME{} gate
    \begin{equation*}
        \hat{\lab}_m=\begin{cases}
            \overline{\lab}_{i}\overline{\lab}_{j},\quad & \exists k\in\set{1,\ldots,n}\colon m=j,\\
            \overline{\lab}_{m},\quad &\text{else,}
        \end{cases}
    \end{equation*}
    where $\labseq{\hat{\lab}_1,\ldots,\hat{\lab}_n}\coloneqq \labaction{\overline{\lab}}{\CNOT_{i,j}}$. Therefore, we have
    \begin{equation*}
        \left(\overline{\overline{\lab}\CNOT_{i,j}}\right)_m =\hat{\lab}_{n+1-m}=\begin{cases}
            \overline{\lab}_{i}\overline{\lab}_{j},\quad & \exists k\in\set{1,\ldots,n}\colon n+1-m=j,\\
            \overline{\lab}_{n+1-m},\quad &\text{else,}
        \end{cases}
    \end{equation*}
    which equals
    \begin{equation*}
       \begin{cases}
            \lab_{n+1-i}\lab_{n+1-j},\quad & \exists k\in\set{1,\ldots,n}\colon m=n+1-j,\\
            \lab_{m},\quad &\text{else,}
        \end{cases}
        =\left(\labaction{\lab}{\overline{C}^{(n)}}\right)_m.
    \end{equation*}
    Next, assume that the statement holds when the circuit has size $m\geq 1$ and let $C$ be a circuit with size $m+1$. Then, we can find two circuits $C_1$ and $C_2$ of size at most $m$ such that $\labaction{\lab}{C}=\labaction{\lab}{\left(C_1\cat C_2\right)}$ for any sequence of labels $\lab$. Thus, we have from our induction hypothesis $$\labaction{\lab}{\overline{C}}=\labaction{\lab}{\left(\overline{C_1}\cat\overline{C_2}\right)}=\labaction{\left(\labaction{\lab}{\overline{C_1}}\right)}{C_2}=\labaction{\left(\overline{\labaction{\overline{\lab}}{C_1}}\right)}{C_2}.$$ Hence, applying our induction hyptothesis again on the right-hand side yields 
    \begin{equation*}
    \overline{\labaction{\left(\labaction{\overline{\lab}}{C_1}\right)}{C_2}}=\overline{\labaction{\overline{\lab}}{\left(C_1\cat C_2\right)}}=\overline{\labaction{\overline{\lab}}{C}},
    \end{equation*}
    which shows desired relation.\qedhere
\end{proof}
Subsequently, for a label sequence $\lab=\labseq{\lab_1,\ldots,\lab_n}$ and a set $I\coloneqq\set{i_1,\ldots,i_m}$ with $1\leq i_1<\ldots<i_m\leq n$ we set $\lab_I\coloneqq \lab_{i_1}\ldots\lab_{i_m}$. Furthermore, for a set $\mathcal{I}$ containing subsets of $\set{1,\ldots,n}$ we set $L_{\mathcal{I}}^{(n)}\coloneqq\set{\lab_I\mid I\in\mathcal{I}}$.
\begin{lemma}
    \label{lem:adjoint_circuit}
    Let $C^{(n)}$ be a \CXNAME{} circuit acting on $n$ qubits with depth $d$ and $\lab=
    \labseq{\lab_1,\ldots,\lab_n}$ a sequence of labels.
    \begin{enumerate}
        \item \label{item:adjoint_circuit_label_set_prop} The adjoint circuit $(C^{(n)})^\dagger$ generates the same labels as $C^{(n)}$ from the label sequence $\labaction{\lab}{C^{(n)}}$.
        \item Let $\mathcal{I}$ be a set of subsets of $\set{1,\ldots,n}$ and suppose that $C^{(n)}$ generates $L_{\mathcal{I}}^{(n)}$ from $\lab$ as well as there exists a permutation $\pi\colon\set{1,\ldots,n}\to\set{1,\ldots,n}$ such that $\labaction{\lab}{C^{(n)}}=\labseq{\lab_{\pi(1)},\ldots,\lab_{\pi(n)}}$. Then, the adjoint circuit $(C^{(n)})^\dagger$ generates $L_{\pi^{-1}(\mathcal{I})}^{(n)}\coloneqq\set{\lab_{\pi^{-1}(I)}\mid I\in\mathcal{I}}$ from the label sequence $\lab$.
    \end{enumerate}
\end{lemma}
\begin{proof}
    \begin{enumerate}[wide=\parindent,leftmargin=0pt,align=left]
        \item Let $C_1,\ldots,C_{d}$ be the sets of non-overlapping gates of the circuit $C^{(n)}$. Then, we observe that
        \begin{align*}
            \labaction{\lab}{\left(C_1\cat\ldots\cat C_{m}\right)}&=\labaction{\lab}{\left(C_1\cat\ldots\cat C_{m}\cat C_{m+1}\cat\ldots C_{d}\cat C_{d}^\dagger\cat C_{d-1}^\dagger\cat\ldots\cat C_{m+1}^\dagger\right)}\\
            &=\labaction{\lab}{\left(C^{(n)}\cat C_{d}^\dagger\cat C_{d-1}^\dagger\cat\ldots\cat C_{m+1}^\dagger\right)}\\
            &=\labaction{\labaction{\lab}{C^{(n)}}}{\left(C_{d}^\dagger\cat C_{d-1}^\dagger\cat\ldots\cat C_{m+1}^\dagger\right)}
        \end{align*}
        for each $m=1,\ldots,d-1$. Since $C_{d}^\dagger,\ldots,C_{m+1}^\dagger$ are the first $d-m$ moments of the circuit $(C^{(n)})^\dagger$, the above identity implies that $(C^{(n)})^\dagger$ generates the same labels as $C^{(n)}$ from the label sequence $\labaction{\lab}{C^{(n)}}$.
        \item Define the label sequence $\lab'\coloneqq\labseq{\lab_{\pi^{-1}(1)},\ldots,\lab_{\pi^{-1}(n)}}$. Then, by applying \ref{item:adjoint_circuit_label_set_prop} on the label sequence $\lab'$ we obtain that $(C^{(n)})^\dagger$ generates the label set
        \begin{equation*}
            \set{\lab'_I\mid I\in\mathcal{I}}=\set{\lab_{\pi^{-1}(I)}\mid I\in\mathcal{I}}=L_{\pi^{-1}(\mathcal{I})}^{(n)}
        \end{equation*}
        from $\labaction{\lab'}{C^{(n)}}=\labseq{\lab'_{\pi(1)},\ldots,\lab'_{\pi(n)}}=\lab$, showing the desired result.\qedhere
    \end{enumerate}
\end{proof}
Note that in the following, whenever a circuit $C^{(n)}$ is defined for each number of qubits $n\geq 2$, we denote by $C^{(p,q)}$ the circuit which replaces the control qubit $i$ and the target qubit $j$ of each \CXNAME{} gate in the circuit $C^{(q-p+1)}$ by the control qubit $i+p-1$ and the target qubit $j+p-1$, respectively.

In the following, let $n\geq 2$ denote the number of qubits and $\lab=\labseq{\lab_1,\ldots,\lab_n}$ be a finite sequence of some labels.
\section{Main building blocks}
\label{sec:main_building_blocks}
In this section, we outline the main building blocks and its key properties that are being utilized for the construction of our $k$-body generators in the subsequent sections. 
\begin{definition}[\DXNAME{} gate, \TWINE{} chain, \CXNAME{} chain]
    \label{def:dx_cx_chains}
    Let $1\leq c,t\leq n$ with $c\neq t$. 
    \begin{enumerate}
        \item The \emph{\DXNAME{} gate} is defined as $\XS_{c,t}\coloneqq \CNOT_{t,c}\cat\CNOT_{c,t}$
        \item The \emph{\TWINE{} chain} is defined as the circuit $\XSC^{(n)}\coloneqq \XS_{1,2} \cat \cdots \cat \XS_{n-1,n}$. By $\XSC^{(1)}$ we denote the empty sequence.
        \item The \emph{\CXNAME{} chain} is the circuit defined as $\CXC^{(n)}\coloneqq \CNOT_{1,2}\cat \CNOT_{2,3}\cat\ldots\cat\CNOT_{n-1,n}$.
        \item The \emph{modified \CXNAME{} chain} is the circuit defined as $\CXCM^{(n)}\coloneqq \CNOT_{1,2}\cat_3 \CNOT_{2,3}\cat_5\ldots\cat_{
        2(n-2)+1}\CNOT_{n-1,n}$.
    \end{enumerate}
\end{definition}
\begin{remark}
    \label{rem:xswap_prop}
    \begin{enumerate}
        \item \label{item:xswap_gate_prop} It holds $\labaction{\lab}{\XS_{c,t}}=\hat{\lab}$ where $\hat{\lab}_i\coloneqq\lab_i$ for $i\notin\set{c,t}$ as well as $\hat{\lab}_c\coloneqq\lab_c\lab_t$ and $\hat{\lab}_t\coloneqq\lab_c$.
        \item \label{item:xswap_chain_prop} The \TWINE{} chain satisfies $\labaction{\lab}{\XSC^{(n)}}=\labseq{\lab_{1}(\lab_{2}\ldots\lab_{n}),\lab_1}$.
        \item \label{item:cx_chain_prop} We have for the \CXNAME{} chain ${\labaction{\lab}{\CXC^{(n)}}=\labseq{\lab_1,\lab_{1}\lab_{2},\lab_{1}\lab_{2}\lab_{3},\ldots,\lab_{1}\lab_{2}\ldots\lab_{n}}}$ and ${\labaction{\lab}{(\CXC^{(n)})^\dagger}=\labseq{\lab_1,\lab_{1}\lab_{2},\lab_{2}\lab_{3},\ldots,\lab_{n-1}\lab_{n}}}$. Obviously, the statement also holds for the modified \CXNAME{} chain.
    \end{enumerate}
\end{remark}
\begin{proof}
    Property \ref{item:xswap_gate_prop} can be easily verified. For \ref{item:xswap_chain_prop} we use \ref{item:xswap_gate_prop} to obtain
    \begin{align*}
        \labaction{\lab}{\XSC^{(n)}}&=\labaction{\left(\labaction{\lab}{\XS_{1,2}}\right)}{\left(\concat{\XS_{2,3}}{\concat{\XS_{3,4}}{\concat{\cdots}{\XS_{n-1,n}}}}\right)}\\
        &=\labaction{\labseq{\lab_1\lab_{2},\lab_{1},\lab_{3}\ldots\lab_{n}}}{\left(\concat{\XS_{2,3}}{\concat{\XS_{3,4}}{\concat{\cdots}{\XS_{n-1,n}}}}\right)}.
    \end{align*}
    Consequently, using induction, we see that
    \begin{equation*}
        \labaction{\lab}{\XSC^{(n)}}=\labseq{\lab_{1}\lab_{2},\lab_{1}\lab_{3},\ldots,\lab_{1}\lab_{n},\lab_{1}}=\labseq{\lab_{1}(\lab_{2}\ldots\lab_{n}),\lab_1}.
    \end{equation*}
    Property \ref{item:cx_chain_prop} can analogously be shown.\qedhere
\end{proof}
As we will later see, the following circuits are needed for our construction of $k$-body generators for arbitrary $k\geq 4$ in App.~\ref{sec:k-body_generators} and are used in a post-cleanup step in our circuits.
\begin{definition}[\SWNAME{} gate, \SWNAME{} chain]
   Let $1\leq c,t\leq n$ with $c\neq t$. 
    \begin{enumerate}
        \item The \emph{\SWNAME{} gate} is defined as $\SW_{c,t}\coloneqq \CNOT_{c,t}\cat\CNOT_{t,c}\cat\CNOT_{c,t}$
        \item The \emph{\SWNAME{} chain} is defined as the circuit $\SWC^{(n)}\coloneqq \SW_{1,2} \cat \SW_{2,3} \cat \cdots \cat \SW_{n-1,n}$.
    \end{enumerate}
\end{definition}
\begin{remark}
    \label{rem:swap_prop}
    \begin{enumerate}
        \item \label{item:swap_gate_prop} It holds $\labaction{\lab}{\SW}=\hat{\lab}$ where $\hat{\lab}_i\coloneqq\lab_i$ for $i\notin\set{c,t}$ as well as $\hat{\lab}_c\coloneqq\lab_t$ and $\hat{\lab}_t\coloneqq\lab_c$.
        \item \label{item:swap_chain_prop} The \SWNAME{} chain satisfies $\labaction{\lab}{\SWC^{(n)}}=\labseq{\lab_{2},\lab_{3},\ldots\lab_{n},\lab_1}$.
    \end{enumerate}
\end{remark}
Next, we present our main results on $k$-body generators. We emphasize that some of the circuits in the subsequent sections are a concatenation of subcircuits with non-zero shifts. However, since these circuits produce the same labels as their counterparts where the shifts of the concatenated subcircuits are set to zero, we may assume that the shifts are zero when proving that the circuits generate a specific label set. 

\section{Two-body generators}
\label{sec:two_body_gen}
In this section we quickly recap the circuit in Ref.~\cite{Klaver2024} which generates all two-body labels for nearest neighbor connectivity graphs and introduce it with the help of a so-called \emph{\TWINE{} network of type I} which consists of a concatenation of \TWINE{} chains. 
\label{subsec:two-body_generators}
\begin{definition}[\TWINE{} network of type I]
    Define
    \begin{equation*}
        \XSN^{(n)}_l\coloneqq\begin{cases}
            \concat[4]{\XSC^{(l)}}{\XSN^{(n)}_{l-1}},\quad&\text{$3\leq l\leq n$}\\
            \XSC^{(2)},\quad&\text{$l=2$}.
        \end{cases}
    \end{equation*}
    for $2\leq l\leq n$. Then, we denote by $\XSN^{(n)}\coloneqq \XSN^{(n)}_{n-1}$ the \emph{\TWINE{} network of type I}.
\end{definition}
\begin{theorem}
    \label{thm:xswap_network_I}
    For the \TWINE{} network of type I we have that
    \begin{equation}
        \label{eq:ncs_output}
        \labaction{\lab}{\XSN^{(n)} }=\labseq{\lab_{n-1}\lab_{n},\ldots,\lab_1\lab_2,\lab_1}.
    \end{equation}
    and $\XSN^{(n)}$ generates all two-body labels $\set{\lab_i\lab_j\mid 1\leq i<j\leq n}$. Furthermore, ${\size{\XSN^{(n)}}=n^2-n}$ and ${\depth{\XSN^{(n)}}=4n-6}$.
\end{theorem}
\begin{proof}
    The case ${n=2}$ is trivial. Assume that the statement holds for some ${n-1\geq 2}$. From Remark \ref{rem:xswap_prop} \ref{item:xswap_chain_prop} we see that $\XSN^{(n)}=\XSC^{(n)}\cat\XSN^{(1,n-1)}$ is generating the labels $\set{\lab_1\lab_i\mid 1<i\leq n}$ and $\hat{\lab}\coloneqq\labaction{\lab}{\XSC^{(n)}}=\labseq{\lab_{1}(\lab_{2}\ldots\lab_{n}),\lab_1}$. Finally, using the induction hypothesis, we have that
    \begin{equation*}
        \labaction{\hat{\lab}}{\XSN^{(1,n-1)}}=\labseq{\hat{\lab}_{n-2}\hat{\lab}_{n-1},\ldots,\hat{\lab}_1\hat{\lab}_2,\hat{\lab}_1,\lab_1}=\labseq{\lab_{n-1}\lab_{n},\ldots,\lab_1\lab_2,\lab_1}
    \end{equation*}
    and $\XSN^{(n)}$ is generating the label set $\{\hat{\lab}_i\hat{\lab}_j\mid 1\leq i<j\leq n-1\}=\set{\lab_i\lab_j\mid 2\leq i<j\leq n}$ from $\hat{\lab}$, which implies the required properties. The expression for the size of $\XSN^{(n)}$ can be derived from the recursive definition of $\XSN^{(n)}$
    \begin{equation*}
        \size{\XSN^{(n)}}=\sum_{l=0}^{n-2} \size{\XSC^{(n-l)}}=\sum_{l=1}^{n} 2(l-1)=n^2-n
    \end{equation*}
    Similarly, we get ${\depth{\XSN^{(n)}}=4(n-2)+2}$.\qedhere
\end{proof}
\begin{corollary}
    \label{cor:two-body_generator}
    For $n\geq 3$ the circuit $\G_2^{(n)}\coloneqq \concat[-(n-3)]{\XSN^{(n)}}{\overline{\CXC}^{(n)}}$ is a two-body generator satisfying
    \begin{equation*}
        \labaction{\lab}{\G_2^{(n)}}=\overline{\lab}=\labseq{\lab_n,\lab_{n-1},\ldots,\lab_1}
    \end{equation*}
    as well as
    \begin{equation*}
        \size{\G_2^{(n)}}=n^2-1\quad\text{and}\quad\depth{\G_2^{(n)}}=4n-4.
    \end{equation*}
\end{corollary}
\begin{proof}
    The first identity follows from Theorem \ref{thm:xswap_network_I}, Remark \ref{rem:xswap_prop} \ref{item:cx_chain_prop} and Lemma \ref{lem:rev_circuit}. Moreover, from Theorem \ref{thm:xswap_network_I} we deduce the relations ${\size{\G_2^{(n)}}=\size{\XSN^{(n)}}+\size{\CXC^{(n)}}=n^2-1}$ and ${\depth{\G_2^{(n)}}=\depth{\XSN^{(n)}}+\depth{\CXC^{(n)}}-n+3=4n-4}$.
\end{proof}
\section{Three-body generators}
\label{sec:three-body_generators}
As we have already seen in App. \ref{subsec:two-body_generators}, the repeated execution of \TWINE{} chains replaces the encoded special label on each qubits with another one. The idea of the subsequent three-body generator is based on the two-body generator and involves linking the special labels to a fixed special label by using a \CXNAME{} gate at the beginning of each \TWINE{} chain, thereby generating three-body labels. Thus, the \TWINE{} chain is replaced by the so-called \emph{modified \TWINE{} chain}, which includes this additional \CXNAME{} gate prior to the \TWINE{} chain. In the modified \TWINE{} chain, the initial \CXNAME{} gate ensures (even with repeated usage of a modified \TWINE{} chain) that the fixed special label is contained on its initial qubit throughout the circuit. Applying this modification, the \TWINE{} network of type I (introduced in the previous section) becomes a \emph{\TWINE{} network of type II}. In addition, this new network generates an output of labels that can be reused by another \TWINE{} network of type II where the fixed special label is replaced by a new fixed label that has not previously been used. Thus, the concatenation of multiple \TWINE{} network of type II ensures that all three-body labels are generated.

Throughout the whole section, we suppose that $n\geq 3$.
\begin{definition}[Modified \TWINE{} chain, modified \TWINE{} network of type I and \TWINE{} network of type II]
    \leavevmode
    \begin{enumerate}
        \item The \emph{modified \TWINE{} chain} is defined as $\XSCM^{(n)}\coloneqq \CNOT_{1,2}\cat \XSC^{(2,n)}$.
        \item We set $\XSN_{3,1}^{(n)}\coloneqq \concat{\XSCM^{(3)}}{\CNOT_{1,2}}$ and $\XSN_{3,l}^{(n)}\coloneqq \concat[4]{\XSCM^{(l+2)}}{\XSN^{(n)}_{3,l-1}}$ for $2\leq l\leq n-2$, and call $\XSNM^{(n)}\coloneqq\XSN_{3,n-2}^{(n)}$ the \emph{modified \TWINE{} network of type I}.
        \item The \emph{\TWINE{} network of type II} is defined as $\XSN_3^{(n)}\coloneqq \concat[1]{(\CNOT_{2,3}\cat\CXCM^{(3,n)})}{\XSNM^{(n)}}$.
    \end{enumerate}
\end{definition}
\begin{theorem}
    \label{thm:mod_xswap_chain_xswap_network_II}
    \begin{enumerate}
        \item \label{item:mod_xswap_chain_prop} The modified \TWINE{} chain satisfies
    \begin{equation}
        \label{eq:csc_modified_output}
        \labaction{\lab}{\XSCM^{(n)}}=\labseq{\lab_{1},\lab_{1}\lab_{2}(\lab_{3}\ldots\lab_{n}),\lab_{1}\lab_{2}}.
    \end{equation}
        \item \label{item:xswap_network_IIa_properties} It holds
        \begin{equation}
        \label{eq:ncs3a_output}
        \labaction{\lab}{\XSNM^{(n)}}= \labseq{\lab_1,\lab_{n-1}\lab_n,\ldots,\lab_3\lab_4,\lab_2\lab_3,\lab_1\lab_2}
    \end{equation}
        and $\XSNM^{(n)}$ generates the labels $\set{\lab_1\lab_i\lab_j\mid 1< i<j\leq n}$ from $\lab$. Moreover, we have that
        \begin{equation}
            \label{eq:size_depth_xswap_network_IIa}
            \size{\XSNM^{(n)}}=n^2-2n+1\quad\text{and}\quad\depth{\XSNM^{(n)}}=4n-8
        \end{equation}
        \item \label{item:xswap_network_II_properties} For the \TWINE{} network of type II we have that
    \begin{equation}
        \label{eq:ncs3_output}
        \labaction{ \labseq{\lab_1,\lab_1\lab_2,\lab_2\lab_3,\ldots,\lab_{n-1}\lab_n}}{\XSN_3^{(n)} }= \labseq{\lab_1,\lab_{n-1}\lab_n,\ldots,\lab_3\lab_4,\lab_2\lab_3,\lab_2}.
    \end{equation}
    and $\XSN_3^{(n)}$ generates the labels $\set{\lab_1\lab_i\lab_j\mid 1< i<j\leq n}$ from $\labseq{\lab_1,\lab_1\lab_2,\lab_2\lab_3,\ldots,\lab_{n-1}\lab_n}$. Furthermore, it holds
    \begin{equation}
        \label{eq:size_depth_xswap_network}
        \size{\XSN_3^{(n)}}=n^2-n-1\quad\text{and}\quad\depth{\XSN_3^{(n)}}=4n-7.
    \end{equation}
    \end{enumerate}
\end{theorem}
\begin{proof}
    \begin{enumerate}[wide=\parindent,leftmargin=0pt,align=left]
        \item Applying $\CNOT_{1,2}$ on $\lab$ and substituting $\lab$ with $\labseq{\lab_1,\lab_1,\lab_2,\lab_3,\ldots,\lab_n}$ in Remark \ref{rem:xswap_prop} \ref{item:xswap_chain_prop} shows the first identity.
        \item The case $n=3$ can be easily verified. Suppose the statement holds for some $n-1\geq 3$. From the definition of $\XSNM^{(n)}$ and \ref{item:mod_xswap_chain_prop} we have
        that $\XSNM^{(n)}$ is producing the label set $\set{\lab_1\lab_2\lab_j\mid 2<j\leq n}$ from $\lab$. Then, applying our induction hypothesis on $\XSNM^{(n-1)}$ and the sequence $\hat{\lab}\coloneqq\labaction{\lab}{\XSCM^{(n)}}=\labseq{\lab_{1},\lab_{1}\lab_{2}(\lab_{3}\ldots\lab_{n}),\lab_{1}\lab_{2}}$ yields
        \begin{align*}
            \labaction{\hat{\lab}}{\XSNM^{(n-1)}}&=\labseq{\hat{\lab}_1,\hat{\lab}_{n-2}\hat{\lab}_{n-1},\ldots,\hat{\lab}_{3}\hat{\lab}_{4},\hat{\lab}_{2}\hat{\lab}_{3},\hat{\lab}_{2},\lab_1\lab_2}\\
            &=\labseq{\lab_1,\lab_{n-1}\lab_n,\ldots,\lab_3\lab_4,\lab_2\lab_3,\lab_1\lab_2}
        \end{align*}
        and $\XSNM^{(n)}$ is generating the label set $\{\hat{\lab}_1\hat{\lab}_i\hat{\lab}_j\mid 1<i<j\leq n-1\}=\set{\lab_1\lab_i\lab_j\mid 2<i<j\leq n}$. This shows that the above statement holds true. The expressions for the depth and size can be derived as follows: From $\size{\XSNM^{(n)}}=\size{\XSCM^{(n)}}+\size{\XSNM^{(n-1)}}$ for $n\geq 4$ and $\size{\XSNM^{(3)}}=4$ we deduce
        \begin{align*}
            \size{\XSNM^{(n)}}=\sum_{l=0}^{n-4}\size{\XSCM^{(n-l)}}+4=4+\sum_{l=0}^{n-4}2(n-l-2)+1=n^2-2n+1.
        \end{align*}
        Similarly, we have from $\depth{\XSNM^{(n)}}=4+\depth{\XSNM^{(n-1)}}$ for $n\geq 4$ and $\depth{\XSNM^{(3)}}=4$ the relation $\depth{\XSNM^{(n)}}=4n-8$.
        \item Eq.~\eqref{eq:ncs3_output} and the property that $\XSN^{(n)}_{3}$ is generating the label set $\set{\lab_1\lab_i\lab_j\mid 1< i<j\leq n}$ are derived from the relation $\labaction{\labseq{\lab_1,\lab_1\lab_2,\lab_2\lab_3,\ldots,\lab_{n-1}\lab_n}}{(\CNOT_{2,3}\cat\CXCM^{(3,n)})}=\labseq{\lab_1,\lab_1(\lab_2\ldots\lab_n)}$ and applying \ref{item:xswap_network_IIa_properties} on $\hat{\lab}\coloneqq \labseq{\lab_1,\lab_1(\lab_2\ldots\lab_n)}$.\\
        The expressions for the depth and size of $\XSN^{(n)}_3$ follow from \ref{item:xswap_network_IIa_properties}.\qedhere
    \end{enumerate}
\end{proof}
Following the previous results, we now introduce our three-body generator and its properties.
\begin{definition}[Three-body generator]
    \label{def:three-body_generator}
    Let $W_{3}^{(n)}\coloneqq\concat[s_{n-1}]{\XSN_3^{(n)}}{\overline{W_3}^{(2,n)}}$ with $W_3^{(3)}\coloneqq \XSN_3^{(3)}$ and
    \begin{equation*}
        s_n\coloneqq \begin{cases}
            -2(n-4),\quad &\text{if $n>4$}\\
            -1,\quad &\text{if $n=4$},\\
            0,\quad &\text{if $n=3$}.
        \end{cases}
    \end{equation*}
    Then, we define the circuit $$\G_3^{(n)}\coloneqq (\CXC^{(n)})^\dagger\cat W_3^{(n)}\cat \CNOT_{c_n,t_n},$$ where $(c_n,t_n)$ is $(\frac{n+1}{2}+1,\frac{n+1}{2})$ if $n$ is odd and $(\frac{n}{2},\frac{n}{2}+1)$ if $n$ is even.
\end{definition}
\begin{theorem}
    \label{thm:three-body-generator}
    The circuit $\G_3^{(n)}$ in Definition \ref{def:three-body_generator} is a three-body generator satisfying
    \begin{equation*}
        \label{eq:three-body_generator_LNN_analysis}
        \size{\G_{3}^{(n)}}=\frac{n^{3}}{3} - \frac{n}{3} \quad\text{and}\quad \depth{\G_{3}^{(n)}}=\begin{cases}
            n^{2} + 5 n - 19,\quad & n\geq 5,\\
            18, \quad & n=4,\\
            8,\quad & n=3.\\
        \end{cases}
    \end{equation*}
    Moreover, it holds
    \begin{equation*}
        \labaction{\lab}{\G_3^{(n)}}=\labseq{\lab_1,\lab_3,\ldots,\lab_{o_n}, \lab_{e_n},\ldots,\lab_4,\lab_2},
    \end{equation*}
    where $o_n$ and $e_n$ denote the largest odd number and largest even number less or equal than $n$, respectively.

\end{theorem}
\begin{proof}
    First, we show that $W_3^{(n)}$ generates all three-body labels from the sequence $\labseq{\lab_{n-1}\lab_{n},\ldots,\lab_1\lab_2,\lab_1}$ and
    \begin{equation*}
        \labaction{\labseq{\lab_1,\lab_1\lab_2,\lab_2\lab_3,\ldots,\lab_{n-1}\lab_n}}{W_3^{(n)}}=\labaction{\labseq{\lab_1,\lab_3,\ldots,\lab_{\mathrm{odd}}^{(n)}, \lab_{\mathrm{even}}^{(n)},\ldots,\lab_2}}{\CNOT_{c_n,t_n}}.
    \end{equation*}
    The case $n=3$ follows from Theorem \ref{thm:mod_xswap_chain_xswap_network_II} \ref{item:xswap_network_II_properties}. Now, suppose that the statement holds for some $n-1\geq 3$. Then, we deduce from Theorem \ref{thm:mod_xswap_chain_xswap_network_II} \ref{item:xswap_network_II_properties} again and Eq.~\eqref{eq:rev_circuit2}
    \begin{align*}
        \labaction{\labseq{\lab_1,\lab_1\lab_2,\lab_2\lab_3,\ldots,\lab_{n-1}\lab_n}}{W_3^{(n)}}&=\labaction{\labseq{\lab_1,\lab_{n-1}\lab_n,\ldots,\lab_3\lab_4,\lab_2\lab_3,\lab_2}}{\overline{W_3}^{(2,n)}}\\
        &=\overline{\labaction{\labseq{\lab_1,\lab_2,\lab_2\lab_3,\lab_3\lab_4,\ldots,\lab_{n-1}\lab_n}}{W_3^{(2,n)}}}.
    \end{align*}
    as well as $W_3^{(n)}$ is generating the label set $\set{\lab_1\lab_i\lab_j\mid 1< i<j\leq n}$. Then, applying our induction hypothesis and Lemma \ref{lem:rev_circuit} shows that $\overline{W_3}^{(2,n)}$ is generating the label set $\set{\lab_i\lab_j\lab_r \mid 2\leq i<j<r\leq n}$ from $\labseq{\lab_1,\lab_{n-1}\lab_n,\ldots,\lab_3\lab_4,\lab_2\lab_3,\lab_2}$ and therefore, $W_3^{(n)}$ is a three-body generator. Moreover, we also see that applying Eq.~\eqref{eq:rev_circuit2} and the induction hypothesis yields that right-hand side equals
    \begin{equation*}
        \overline{\labaction{\labseq{\lab_1,\lab_2,\lab_4,\ldots,\lab_{\mathrm{even}}^{(n)}, \lab_{\mathrm{odd}}^{(n)},\ldots,\lab_5,\lab_3}}{\CNOT_{c_{n-1}+1,t_{n-1}+1}^{(2,n)}}}=\labaction{\labseq{\lab_1,\lab_3,\ldots,\lab_{\mathrm{odd}}^{(n)}, \lab_{\mathrm{even}}^{(n)},\ldots,\lab_2}}{\overline{\CNOT_{c_{n-1}+1,t_{n-1}+1}}^{(2,n)}}.
    \end{equation*}
    and $W_3^{(n)}$ is generating the label set $\set{\lab_2\lab_i\lab_j\mid 2< i<j\leq n}$. Using that $\overline{\CNOT_{c_{n-1}+1,t_{n-1}+1}}^{(2,n)}=\CNOT_{c_n,t_n}$ and Remark \ref{rem:xswap_prop} \ref{item:cx_chain_prop} shows desired statement.

    Finally, we prove the statements on the depth and size of $\G_3^{(n)}$. Since $\size{W_3^{(n)}}=\size{\NCS_3^{(n)}}+\size{W_3^{(n-1)}}$, we have from Theorem \ref{thm:mod_xswap_chain_xswap_network_II} \ref{item:xswap_network_II_properties}
    \begin{equation*}
        \size{W_3^{(n)}}=\sum_{l=0}^{n-3}\size{\XSN_3^{(n-l)}}=\sum_{l=1}^{n} l^2-l-1=\frac{1}{3}n^3-\frac{4}{3}n,
    \end{equation*}
    which shows the expression for the size of $\G_3^{(n)}$. To derive the expression for the depth of $\G_3^{(n)}$, we use the relations $\depth{W_3^{(n)}}=\depth{\XSN_3^{(n)}}+s_{n-1}+\depth{W_3^{(n-1)}}$ for $n\geq 4$ and $\depth{W_3^{(5)}}=26$ as well as Theorem \ref{thm:mod_xswap_chain_xswap_network_II} \ref{item:xswap_network_II_properties} to obtain
    \begin{align*}
        \depth{W_3^{(n)}}&=\depth{W_3^{(5)}}+\sum_{l=0}^{n-6}\depth{\XSN_3^{(n-l)}}+s_{n-l-1}\\
        &=-19+\sum_{l=1}^{n}\depth{\XSN_3^{(l)}}+s_{l-1}\\
        &=-19+\sum_{l=1}^{n} 2l+3=n^2+4n-19
    \end{align*}
    for $n\geq 6$. For $n=3,4$ one can easily verify $\depth{W_3^{(3)}}=5$ and $\depth{W_3^{(4)}}=14$, implying the expression for the depth of $\G_3^{(n)}$.\qedhere
\end{proof}

\section{Generators for arbitrary $k\geq 4$}
\label{sec:k-body_generators}
Based on the previous sections, we now show that there exists a family of $k$-body generators $\G_k=(\G_k^{(n)})_{n\geq k}$ with $\mu(\G_k,\mathcal{L}_k)=2$ and $\nu(\G_k,\mathcal{L}_k)=k$ for nearest neighbor connectivity graphs for all $k\geq 2$.
For the construction of our $k$-body generator for $k\geq 4$, we first introduce a so-called \emph{clean special four-body generator}, which generates all four-body labels with one fixed special label. The definition of the clean special four-body generator is motivated by the three-body generator in App. \ref{sec:three-body_generators} and replaces \TWINE{} networks of type II by \emph{\TWINE{} networks of type III}. The main difference between these two networks is that the latter applies an additional \TWINE{} chain at beginning of the circuit which encodes the additional special label on the qubits. Additionally, the latter network incorporates a final pre-cleanup circuit that reverts the output labels to single labels. Based on the clean four-body generator, we present a constructive algorithm for establishing clean special $k$-body generators for arbitrary $k\geq 4$ which consequently allows the design of $k$-body generators for all $k\geq 4$.
Throughout the whole section, we suppose that $n\geq 4$.
The main theorem reads as follows:
\begin{theorem}
    \label{thm:kbodies_lnn}
    For all $k\geq 2$, there exists a family of $k$-body generators for nearest neighbor connectivity graphs $\G_k=(\G_k^{(n)})_{n\geq k}$ with $\size{\G_k^{(n)}}=\frac{2n^k}{k!}+\mathcal{O}(n^{k-1})$ and $\depth{\G_k^{(n)}}=\frac{2n^{k-1}}{(k-1)!}+\mathcal{O}(n^{k-2})$. As a consequence, $\mu(\G_k,\mathcal{L}_k)=2$ and $\nu(\G_k,\mathcal{L}_k)=k$, where $\mathcal{L}_k=(\mathcal{L}_k^{(n)})_{n\geq k}$ is the family of labels sets $\mathcal{L}_k^{(n)}$ containing all $k$-body labels on $n\geq k$ qubits.
\end{theorem}
As already mentioned, for the proof of Theorem \ref{thm:kbodies_lnn} we make us of so-called \emph{special $k$-body generators} which are defined as follows.
\begin{definition}[Special $k$-body generator]
    \label{def:special_k-body_generator}
    Let $C^{(n)}$ be a circuit with $n$ qubits, $k\geq 2$ and $\lab=\labseq{\lab_1,\ldots,\lab_{n}}$ with some labels $\lab_1,\ldots,\lab_n$. We say that $C^{(n)}$ is a \emph{special $k$-body generator} if $C^{(n)}$ generates all $k$-body labels $$\set{\lab_1\lab_{i_1}\lab_{i_2}\ldots\lab_{i_{k-1}}\lab_n\mid 1<i_1<\ldots<i_{k-1}\leq n}$$ from $\lab$ and $\hat{\lab_1}=\lab_1$ where $\hat{\lab}\coloneqq\labaction{\lab}{C^{(n)}}$. Furthermore, we call $\lab_1$ the \emph{special label}.
\end{definition}
Next, we introduce the \emph{\TWINE{} network of type III} as well as the \emph{post-cleanup circuit} which are one of the key-ingredients for the proof of Theorem \ref{thm:kbodies_lnn}.
\begin{definition}[\TWINE{} network of type III, post-cleanup circuit]
    \label{def:xswap_network_III_cleanup}
    \begin{enumerate}
        \item The \emph{\TWINE{} network of type III} is defined as $\XSN_4^{(n)}\coloneqq \concat[-2(n-3)]{\CXCM^{(3,n)}}{\concat[5]{\XSC^{(n)}}{\XSNM^{(1,n-1)}}}$.
        \item The post-cleanup circuit $\CL^{(n)}$ is defined as
        \begin{equation*}
            \CL^{(n)}\coloneqq\CNOT_{c_n,t_n}\cat\mathrm{OSC}^{(m_n,m_n+2)}\cat \XSC^{(m_n+2,n)}\cat\overline{\SWC}^{(m_n,n)}\cat\overline{\XSC}^{(1,m_n)},
        \end{equation*}
        where $m_n\coloneqq\lceil\frac{n-1}{2}\rceil$ and $(c_n,t_n)$ is $(m_n+1,m_n+2)$ if $n$ is odd and $(m_n+1,m_n)$ if $n$ is even. Here, $\mathrm{OSC}^{(m_n,m_n+2)}$ is defined as the circuit which is set to $\SWC^{(m_n,m_n+2)}$ if $n$ is odd and the empty sequence if $n$ is even.
    \end{enumerate}
\end{definition}
\begin{theorem}
    \label{thm:xswap_network_III_cleanup}
    \begin{enumerate}
        \item \label{item:xswap_network_III_properties} For the \TWINE{} network of type III we have that
    \begin{equation}
        \label{eq:xswap_network_III_output}
        \labaction{ \labseq{\lab_1,\lab_2,\lab_2\lab_3,\lab_3\lab_4,\ldots,\lab_{n-1}\lab_n}}{\XSN_4^{(n)} }= \labseq{\lab_1\lab_2,\lab_{n-1}\lab_n,\lab_{n-2}\lab_{n-1},\ldots,\lab_3\lab_4,\lab_3,\lab_1}.
    \end{equation}
    and $\XSN_4^{(n)}$ generates the labels $\set{\lab_1\lab_2\lab_i\lab_j\mid 2< i<j\leq n}$ from $\labseq{\lab_1,\lab_2,\lab_2\lab_3,\lab_3\lab_4,\ldots,\lab_{n-1}\lab_n}$. Furthermore, it holds
    \begin{equation}
        \label{eq:size_depth_xswap_network_III}
        \size{\XSN_4^{(n)}}=n^2-n-1\quad\text{and}\quad\depth{\XSN_4^{(n)}}=4n-7
    \end{equation}
    \item \label{item:post-clean_up_properties} Define the sequence of labels $\lab_{\mathrm{out}}\coloneqq\labseq{\lab_1\lab_2,\lab_3\lab_4,\lab_3,\lab_1}$ for $n=4$ and
    \begin{equation*}
        \lab_{\mathrm{out}}\coloneqq\labseq{\lab_1(\lab_2,\lab_4,\ldots\lab_{e_n-2}),\hat{\lab}_1,\hat{\lab}_2,\hat{\lab}_3,(\lab_{o_n-2},\lab_{o_n-4},\ldots,\lab_3)\lab_1}
    \end{equation*}
    where $(\hat{\lab}_1,\hat{\lab}_2,\hat{\lab}_3)\coloneqq(\lab_{n-1}\lab_{n},\lab_{n-1},\lab_1)$ for even $n$ and $(\hat{\lab}_1,\hat{\lab}_2,\hat{\lab}_3)\coloneqq(\lab_{1},\lab_{n-1},\lab_{n-1}\lab_n)$ for odd $n$. Here, $o_n$ and $e_n$ denote again the largest odd number and largest even number less or equal than $n$, respectively. Then, the post-cleanup circuit $\CL^{(n)}$ satisfies
    \begin{equation*}
        \labaction{\lab_{\mathrm{out}}}{\CL^{(n)}}=\labseq{\lab_1,\lab_2,\lab_4,\ldots,\lab_{e_n}, \lab_{o_n},\ldots,\lab_5,\lab_3}.
    \end{equation*}
    Furthermore, it holds
    \begin{equation}
        \label{eq:size_depth_post-cleanup}
        \size{\CL^{(n)}}=\depth{\CL^{(n)}}=5n-3m_n-5+6(n\mod 2)=\frac{7}{2}n-5\left(-\frac{1}{2}\right)^{n \mod 2}.
    \end{equation}
    \end{enumerate}
\end{theorem}
\begin{proof}
    \begin{enumerate}[wide=\parindent,leftmargin=0pt,align=left]
        \item From Remark \ref{rem:xswap_prop} \ref{item:xswap_chain_prop} and \ref{item:cx_chain_prop} we deduce that
        \begin{align*}
            \labaction{ \labseq{\lab_1,\lab_2,\lab_2\lab_3,\lab_3\lab_4,\ldots,\lab_{n-1}\lab_n}}{\XSN_4^{(n)} }&=\labaction{\labseq{\lab_1,\lab_2,\lab_2(\lab_3\ldots\lab_n)}}{\left(\concat{\XSC^{(n)}}{\XSNM^{(1,n-1)}}\right)}\\
            &=\labaction{\labseq{\lab_1\lab_2,\lab_1\lab_2(\lab_3\ldots\lab_n),\lab_1}}{\XSNM^{(1,n-1)}}.
        \end{align*}
        Therefore, applying Theorem \ref{thm:mod_xswap_chain_xswap_network_II} \ref{item:xswap_network_IIa_properties} on the label sequence $\labseq{\lab_1\lab_2,\lab_1\lab_2(\lab_3\ldots\lab_n),\lab_1}$ proves the first statement in \ref{item:xswap_network_III_properties}. Moreover, using the expressions for the size and depth in Theorem \ref{thm:mod_xswap_chain_xswap_network_II} \ref{item:xswap_network_IIa_properties} we deduce $\size{\XSN_4^{(n)}}=3n-4+\size{\XSNM^{(1,n-1)}}=n^2-n$ as well as $\depth{\XSN_4^{(n)}}=5+\depth{\XSNM^{(1,n-1)}}=4n-7$.
        \item Using Remark \ref{rem:xswap_prop} \ref{item:xswap_chain_prop},  Remark \ref{rem:swap_prop} \ref{item:swap_chain_prop} and Lemma \ref{lem:rev_circuit}, the statements in \ref{item:post-clean_up_properties} can be easily verified.\qedhere
    \end{enumerate}
\end{proof}
The definition of the \TWINE{} network of type III and the post-cleanup circuit leads us immediately to the definition of the clean special four-body generator. 
\begin{definition}[Clean special four-body generator]
    \label{def:clean_special_four-body_generator}
    Let $W_{4}^{(n)}\coloneqq\concat[s_{n-1}]{\XSN_4^{(n)}}{\overline{W}_4^{(2,n)}}$ with $W_4^{(4)}\coloneqq \XSN_4^{(4)}$ where and $s_n$ is defined as in Definition \ref{def:three-body_generator}.
    Then, we define the circuit [see Fig.~\ref{fig:4body_gen}(a)] $$\SG_4^{(n)}\coloneqq (\CXC^{(2,n)})^\dagger\cat W_4^{(n)}\cat \CL^{(n)}.$$
\end{definition}
\begin{theorem}
    \label{thm:clean_special_four-body_generator}
    The circuit $\SG_4^{(n)}$ in Definition \ref{def:clean_special_four-body_generator} is a clean special four-body generator satisfying
    \begin{align*}
        \size{\SG_{4}^{(n)}}=\frac{1}{3}n^{3} + \frac{19 }{6}n - 7-5\left(-\frac{1}{2}\right)^{n \mod 2}\quad\text{and}\\ \depth{\SG_{4}^{(n)}}=\begin{cases}
            n^2 + \frac{17}{2}n - 26-5\left(-\frac{1}{2}\right)^{n \mod 2},\quad & n\geq 5,\\
            20, \quad & n=4.\\
        \end{cases}
    \end{align*}
    Moreover, it holds
    \begin{equation}
        \label{eq:clean_special_four-body_output}
        \labaction{\lab}{\SG_4}=\labseq{\lab_1,\lab_2,\lab_4,\ldots,\lab_{e_n}, \lab_{o_n},\ldots,\lab_5,\lab_3}.
    \end{equation}
\end{theorem}
\begin{proof}
    Using the same steps as in the proof of Theorem \ref{thm:three-body-generator} and applying Theorem \ref{thm:xswap_network_III_cleanup} \ref{item:xswap_network_III_properties}, it can be verified that $W_4^{(n)}$ generates the labels $\set{\lab_1\lab_{i_1}\lab_{i_2}\lab_{i_3}\mid 1<i_1<i_2<i_3\leq n}$ from the sequence $\labseq{\lab_1,\lab_2,\lab_2\lab_3,\lab_3\lab_4,\ldots,\lab_{n-1}\lab_n}$ and
    \begin{equation*}
        \labaction{\labseq{\lab_1,\lab_2,\lab_2\lab_3,\lab_3\lab_4,\ldots,\lab_{n-1}\lab_n}}{W_4^{(n)}}=\lab_{\mathrm{out}},
    \end{equation*}
    where $\lab_{\mathrm{out}}$ is defined as in Theorem \ref{thm:xswap_network_III_cleanup} \ref{item:post-clean_up_properties}. Therefore, using that $\labaction{\lab}{(\CXC^{(2,n)})^\dagger}=\labseq{\lab_1,\lab_2,\lab_2\lab_3,\lab_3\lab_4,\ldots,\lab_{n-1}\lab_n}$ and applying Theorem \ref{thm:xswap_network_III_cleanup} \ref{item:post-clean_up_properties} shows that $\SG_4^{(n)}$ is a clean special four-body generator satisfying Eq.~\eqref{eq:clean_special_four-body_output}. Moreover, since $\size{\XSN_3^{(n)}}=\size{\XSN_4^{(n)}}$ as well as $\depth{\XSN_3^{(n)}}=\depth{\XSN_4^{(n)}}$ we have $$\size{W_4^{(n)}}=\size{W_3^{(n)}}-\size{\XSN_3^{(3)}}=\size{W_3^{(n)}}-5$$ and similarly $$\depth{W_4^{(n)}}=\depth{W_3^{(n)}}-5.$$ Using the expression for $\depth{W_3^{(n)}}$ in the proof of Theorem \ref{thm:three-body-generator} as well as Eq.~\eqref{eq:size_depth_post-cleanup} imply the remaining identities.\qedhere
\end{proof}
Subsequently to the clean special four-body generator in Definition \ref{def:clean_special_four-body_generator}, we now show how arbitrary clean special $k$-body generator with $k\geq 4$ can be constructed.
\begin{definition}[Clean special $k$-body generator]
    \label{def:clean_special_k-body_generator}
    Let $k_0\geq 2$ and suppose that $\SG_{k_0}^{(n)}$ is a clean special $k_0$-body generator. Then, we define for $k>k_0$
    \begin{equation}
        \label{eq:clean_special_k-body_generator}
        \SG_k^{(n)}\coloneqq W_k^{(n)}\cat \CNOT_{n-k+1,n-k+2}\cat\overline{\XSC}^{(1,n-k+1)}
    \end{equation}
    where $W_k^{(n)}\coloneqq \CNOT_{1,2}\cat \SG_{k-1}^{(2,n)}\cat\SW_{1,2}\cat W_k^{(2,n)}$ with $W_k^{(k)}\coloneqq \CNOT_{1,2}\cat\SG_{k-1}^{(2,k)}$. Fig.~\ref{fig:4body_gen}(c) illustrates the circuit $\SG_k^{(n)}$.
\end{definition}
\begin{theorem}
    \label{thm:clean_special_k-body_generator}
    Under the assumptions of Definition \ref{def:clean_special_k-body_generator}, the circuit $\SG_k^{(n)}$ is a clean special $k$-body generator for each $k\geq k_0$. Moreover, if the clean special $k_0$-body generator satisfies 
    \begin{equation*}
        \label{eq:assump_clean_special_k0-body_generator}
        \size{\SG_{k_0}^{(n)}}=c_1 n^{k_0-1}+\mathcal{O}(n^{k_0-2})\quad\text{and}\quad \depth{\SG_{k_0}^{(n)}}=c_2n^{k_0-2}+\mathcal{O}(n^{k_0-3})
    \end{equation*}
    where $c_1,c_2>0$ are some constants, then for every $k\geq k_0$ it holds
    \begin{equation*}
        \size{\SG_{k}^{(n)}}=c_1\frac{(k_0-1)!}{(k-1)!}n^{k-1}+\mathcal{O}(n^{k-2})
    \quad and \quad
        \depth{\SG_{k}^{(n)}}=c_2\frac{(k_0-2)!}{(k-2)!}n^{k-2}+\mathcal{O}(n^{k-3}).
    \end{equation*}
\end{theorem}
\begin{proof}
    \begin{enumerate}[wide=\parindent,leftmargin=0pt,align=left]
        \item First, we show that $\SG_k^{(n)}$ is a clean special $k$-body generator for each $k\geq k_0$. To this end, suppose that $\SG_{k-1}^{(n)}$ is a clean-special $k-1$-body generator for some $k-1\geq k_0$. We prove that $W_k^{(n)}$ is a special $k$-body generator from $\lab$ and that there exists a permutation $\pi\colon\set{1,\ldots,n}\to\set{1,\ldots,n}$ with $\pi(1)=1$ satisfying
        \begin{equation}
            \label{eq:output_wk}
            \labaction{\lab}{W_k^{(n)}}=\labseq{\hat{\lab}_1(\hat{\lab}_2\ldots\hat{\lab}_{n-k+1}),\hat{\lab}_1,\hat{\lab}_1\hat{\lab}_{n-k+2},\hat{\lab}_{n-k+3},\ldots,\hat{\lab}_n)}
        \end{equation}
        for $n>k$ and $\labaction{\lab}{W_k^{(n)}}=\labseq{\hat{\lab}_1,\hat{\lab}_1\hat{\lab}_{2},\hat{\lab}_{3},\ldots,\hat{\lab}_k}$ where $\labseq{\hat{\lab}_1,\ldots,\hat{\lab}_n}\coloneqq\labseq{\lab_{\pi(1)},\ldots,\lab_{\pi(n)}}$. The remaining statement follows from the application of the label sequence on the right-hand side in Eq.~\eqref{eq:output_wk} on the circuits $\CNOT_{k+1,n-k+2}$ and $\overline{\XSC}^{(1,n-k+1)}$ given in Eq.~\eqref{eq:clean_special_k-body_generator}. The case $n=k$ can be easily verified. For $n>k$ we first observe that there exists a permutation $\pi_1\colon\set{1,\ldots,n-1}\to\set{1,\dots,n-1}$ with $\pi_1(1)=1$ such that
        \begin{equation*}
            \labaction{(\labaction{\lab}{\CNOT_{1,2}})}{\SG_{k-1}^{(2,n)}}=\labaction{\labseq{\lab_1,\lab_1\lab_2,\lab_3,\ldots,\lab_n}}{\SG_{k-1}^{(2,n)}}=\labseq{\lab_1,\lab'_{\pi_1(1)},\ldots,\lab'_{\pi_1(n-1)}},
        \end{equation*}
        where $\lab'_1\coloneqq\lab_{1}\lab_{2}$ and $\lab'_i\coloneqq\lab_{i+1}$ for $2\leq i\leq n-1$, and $\SG_{k-1}^{(2,n)}$ is generating the labels
        \begin{equation}
            \label{eq:sg_k-1_label_set}
            \set{\lab'_1\lab'_{i_1}\ldots\lab'_{i_{k-2}}\mid 1<i_1<\ldots<i_{k-2}\leq n}=\set{\lab_1\lab_2\lab_{i_1}\ldots\lab_{i_{k-2}}\mid 2<i_1<\ldots<i_{k-2}\leq n}
        \end{equation}
        from $\labseq{\lab_1,\lab_1\lab_2,\lab_3,\ldots,\lab_n}$. Furthermore, from our induction hypothesis we can imply that there exists a further permutation $\pi_2\colon\set{1,\ldots,n-1}\to \set{1,\ldots,n-1}$ with $\pi_2(1)=1$ such that $W_k^{(2,n)}$ is generating the labels
        \begin{equation*}
            \set{\lab_1\lab'_{\pi_1(i_1+1)}\ldots\lab'_{\pi_1(i_{k-1}+1)}\mid 1<i_1<\ldots<i_{k-1}\leq n-1}=\set{\lab_1\lab_{i_1}\ldots\lab_{i_{k-1}}\mid 2<i_1<\ldots<i_{k-1}\leq n}
        \end{equation*}
        from the label sequence $\tilde{\lab}\coloneqq\labaction{(\labaction{\lab}{\CNOT_{1,2}})}{(\SG_{k-1}^{(2,n)}\cat\SW_{1,2})}=\labseq{\lab_1\lab_2,\lab_1,\lab'_{\pi_1(2)},\ldots,\lab'_{\pi_1(n-1)}}$ as well as
        \begin{align*}
            \labaction{\tilde{\lab}}{W_k^{(2,n)}}&=\labseq{\lab_1\lab_2,\hat{\lab}_1(\hat{\lab}_2\ldots\hat{\lab}_{n-k}),\hat{\lab}_1,\hat{\lab}_1\hat{\lab}_{n-k+1},\hat{\lab}_{n-k+2},\ldots,\hat{\lab}_{n-1})}\\
            &=\labseq{\lab_1\lab_2,\lab_1(\lab_{\pi(2)+1}\ldots\lab_{\pi(n-k)+1)}),\lab_1,\lab_1\lab_{\pi(n-k+1)+1},\lab_{\pi(n-k+2)+1},\ldots,\lab_{\pi(n-1)+1}}\\
            &=\labseq{\lab_{\sigma(1)}(\lab_{\sigma(2)}\ldots\lab_{\sigma(n-k+1)}),\lab_1,\lab_1\lab_{\sigma(n-k+2)},\lab_{\sigma(n-k+3)},\ldots,\lab_{\sigma(n)}}
        \end{align*}
        where $\labseq{\hat{\lab}_1,\ldots,\hat{\lab}_n}\coloneqq\labseq{\lab_{\sigma_{1}},\lab'_{\sigma(2)},\ldots,\lab'_{\sigma(n-1)}}$ and $\pi\coloneqq \pi_1\circ\pi_2$. Here, $\sigma\colon\set{1,\ldots,n}\to\set{1,\ldots,n}$ denotes the permutation defined by $\sigma(1)\coloneqq 1$, $\sigma(2)\coloneqq 2$ and $\sigma(i)\coloneqq \pi(i-1)+1$. This shows together with Eq.~\eqref{eq:sg_k-1_label_set} that $W_k^{(n)}$ is a special $k$-body generator satisfying Eq.~\eqref{eq:output_wk}.
        \item We only prove the expression for the size of the circuit $\SG_k^{(n)}$. The expression for the depth of the circuit $\SG_k^{(n)}$ can be proven analogously. Assume that the statement holds for some $k-1\geq k_0$. Then, our induction hypothesis implies
        \begin{equation*}
            \size{\SG_k^{(n)}}=\mathcal{O}(n)+\sum_{i=1}^{n-k}\size{\SG_{k-1}^{(i+1,n)}}=\mathcal{O}(n)+c_1\frac{(k_0-1)!}{(k-2)!}\sum_{i=1}^{n-2}i^{k-2}+\mathcal{O}(i^{k-3}).
        \end{equation*}
        Hence, applying Faulhaber's formula yields the above identity for the size of $\SG_k^{(n)}$.\qedhere
    \end{enumerate}
\end{proof}
Finally, based on clean special $k$-body generator we now establish $k$-body generators and derive their properties subsequently, completing the proof of Theorem \ref{thm:kbodies_lnn}.
\begin{definition}[$k$-body generator]
    \label{def:k-body_generator_from_special_k-body_generator}
    Let $k\geq 2$ and suppose that $\SG_{k}^{(n)}$ is a clean special $k$-body generator. Then, we set $\G_k^{(n)}\coloneqq \SG_{k}^{(n)}\cat \SG_{k}^{(2,n)}\cat\ldots\cat\SG_{k}^{(n-k+1,n)}$. Fig.~\ref{fig:4body_gen}(b) demonstrates an illustration of the circuit for the case $k=4$.
\end{definition}
\begin{theorem}
    \label{thm:k-body_generator_from_special_k-body_generator}
    \begin{enumerate}
        \item \label{item:k-body_generator_from_special_k-body_generator} Under the assumptions of Definition \ref{def:k-body_generator_from_special_k-body_generator}, the circuit $\G_k^{(n)}$ is a $k$-body generator.
        \item \label{item:k-body_generator_size_depth_analysis} Let $k_0\geq 2$ and $\SG_{k_0}^{(n)}$ be a clean special $k_0$-body generator. Then, for every $k>k_0$ the circuit $\SG_{k}^{(n)}$ as defined in Definition \ref{def:clean_special_k-body_generator} is a clean special $k$-body generator and the circuit $\G_k^{(n)}$ as defined in Definition \ref{def:k-body_generator_from_special_k-body_generator} is a $k$-body generator. Moreover, if $\SG_{k_0}^{(n)}$ satisfies
        \begin{equation*}
        \size{\SG_{k_0}^{(n)}}=c_1 n^{k_0-1}+\mathcal{O}(n^{k_0-2})\quad\text{and}\quad \depth{\SG_{k_0}^{(n)}}=c_2n^{k_0-2}+\mathcal{O}(n^{k_0-3}),
        \end{equation*}
        then it holds
        \begin{equation*}
        \size{\G_k^{(n)}}=c_1\frac{(k_0-1)!}{k!}n^{k}+\mathcal{O}(n^{k-1})\quad\text{and}\quad
        \depth{\G_k^{(n)}}=c_2\frac{(k_0-2)!}{(k-1)!}n^{k-1}+\mathcal{O}(n^{k-2}).
        \end{equation*}
    \end{enumerate}
\end{theorem}
\begin{proof}
    The statement in \ref{item:k-body_generator_from_special_k-body_generator} can easily be proven by induction. The first statement in \ref{item:k-body_generator_size_depth_analysis} follows from Theorem \ref{thm:clean_special_k-body_generator} and \ref{item:k-body_generator_from_special_k-body_generator}. The expressions for the depth and size of $\G_k^{(n)}$ can be similarly derived as the expressions for the depth and size of $\SG_k^{(n)}$ in the proof of Theorem \ref{thm:clean_special_k-body_generator}.\qedhere
\end{proof}
\begin{proof}[Proof of Theorem \ref{thm:kbodies_lnn}]
    For $k=2$, the statement follows from the results presented \cite{Klaver2024}. The case $k=3$ follows from Theorem \ref{thm:three-body-generator}. The case $k\geq 4$ is obtained by Theorem \ref{thm:clean_special_four-body_generator} and Theorem \ref{thm:k-body_generator_from_special_k-body_generator} \ref{item:k-body_generator_size_depth_analysis}.
\end{proof}

\section{Two-body generators for heavy-hexagon connectivity graphs \label{app:heavy-hex}}
In this section, we investigate the heavy hexagon qubit layouts \cite{Chamberland2020} and demonstrate how to utilize the strategy developed in Sec.~\ref{sec:2d_architectures} to build two-body generators. In contrast to the square grid connectivity graphs outlined in Sec.~\ref{sec:2d_architectures}, heavy hexagonal qubit layouts do not possess a Hamiltonian path. However, all qubits can be connected in a HGP, where (roughly) every sixth qubit is a HGP neighbor. In (the bulk of) a heavy hexagonal layout, every other qubit has three neighbors. Thus, HGPs, based on sub-graphs of heavy hexagons, can have as many as $1/3$ of the involved qubits as Hamiltonian grid neighbors. A corresponding HGP is schematically shown in Fig.~\ref{fig:app:heavy_hex}(a). Applying the principles of Sec.~\ref{sec:2d_architectures} along the HGP of Fig.~\ref{fig:app:heavy_hex}(a), we can immediately conclude the \CXNAME{} count of a two-body generator: Since one out of three qubits represents a HGP neighbor, for generating three new two-body labels we require five \CXNAME{} gates (in contrast to the LNN case where we would require six). Thus, we expect to find an average asymptotic \CXNAME{} count of $\mu(\mathcal{G}_2,\mathcal{L}_2) = \frac{5}{3}$.

Fig.~\ref{fig:app:heavy_hex}(b) illustrates the construction of $\mathcal{G}_{2,\mathrm{hex}}^{(n)}$: shifted concatenations of adapted \TWINE{} chains ($\XSC_{\mathrm{hex}}$, marked in blue) form 
a $\XSN_{\mathrm{hex}}^{{(n)}}$ circuit. This is followed by a decoding circuit (marked in red) to restore single-labels. For the count and depth we obtain
\begin{equation*}
     \size{\XSN_{\mathrm{hex}}^{{(n)}}}=\frac{5}{6}n^2+\mathcal{O}(n)
\end{equation*}
\begin{equation*}
    \depth{\XSN_{\mathrm{hex}}^{{(n)}}}=5n+\mathcal{O}(1).
\end{equation*}

Note that, unlike for square grid connectivity graphs, here the lack of a Hamiltonian path complicates the decoding: Instead of a simple CNOT chain, here decoding requires $2n$ CNOT gates [compare Fig.~\ref{fig:app:heavy_hex}(b)]. 
\begin{figure}
    \centering
    \includegraphics[width=0.8\linewidth]{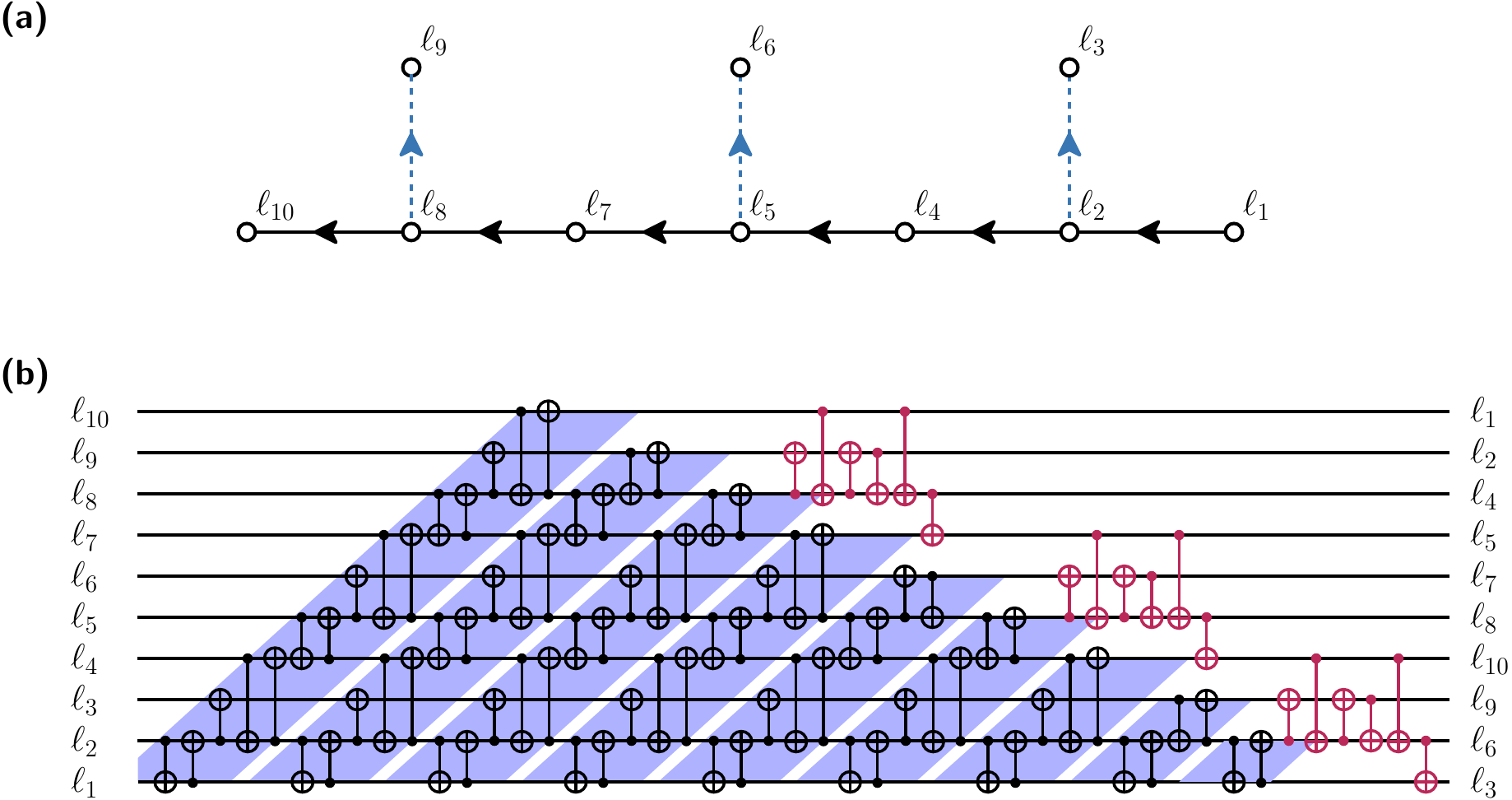}
    \caption{(a) A HGP for a sub-graph within a heavy-hexagon connectivity graph: Every other qubit has three immediate neighbors. (b) Implementation of a two-body generator $\mathcal{G}_{2,\mathrm{hex}}^{(10)}$ based on the HGP shown in (a) including a decoding circuit. \CXNAME{} gates depicted in red are attributed to the final decoding step.}
    \label{fig:app:heavy_hex}
\end{figure}

The two-body generators $\mathcal{G}_{2, \mathrm{hex}}^{(n)}$ can be used for algorithms such as QFT and QAOA as discussed in Sec.~\ref{sec:QFT} and \ref{sec:QAOA}. For QAOA, shifted concatenation of $\mathcal{G}_{2,\mathrm{hex}}^{(n)}$ and $\big(\overline{\mathcal{G}}_{2,\mathrm{hex}}^{(n)}\big)^{\dagger}$ induces a depth reduction per QAOA layer by a factor of $\sim 2$, resulting in a depth of $\frac{5}{2}n(1+\frac{1}{p}) +\mathcal{O}(1)$ per QAOA layer given an odd number of QAOA cycles $p$. In case of an even number of cycles, the depth is slightly lower due to the asymmetry of the corresponding generators. In this case we find a depth of $\frac{5}{2}n(1+\frac{2}{3p}) +\mathcal{O}(1)$.

\section{$k$-body generators on complete and planar connectivity graphs}
\label{sec:k-body_generators_on_other_devices}

\begin{figure}
    \centering
    \includegraphics[width=1.0\linewidth]{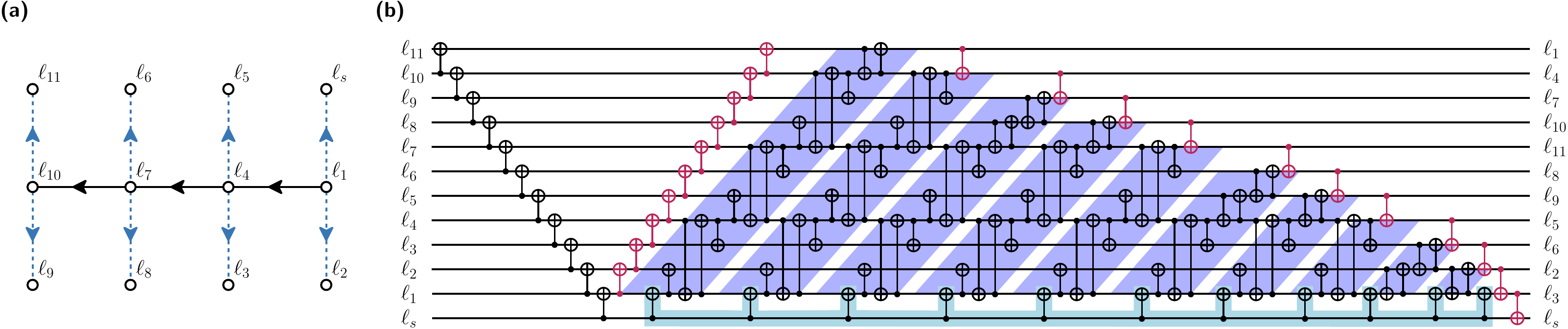}
    \caption{(a) The HGP used for the design of the adapted $\XSN_{3,\grid}$ circuit illustrated in (b) with 12 qubits. (b) The adapted special three-body generator $\XSN_{3,\grid}$ for the square grid connectivity graph pictured in (a).}
    \label{fig:special_four_body_grid}
\end{figure}

In this chapter, we discuss the size and depth scaling for the proposed $k$-body generators for complete and square grid connectivity graphs. However, rather than repeating extensive derivations, our focus is on the scaling up to leading order. We begin with the two-body generators.
\subsection{\label{app:count_and_depth_two_body_gen}Count and depth analysis of two-body generators}
As discussed in the main text, for the construction of two-body generators on complete graphs, \DXNAME{} gates can be replaced with \CXNAME{} gates resulting in half the number of gates as compared to the corresponding construction for nearest neighbor connectivity graphs. Since a two-body generator relies solely on \DXNAME{} gates, we readily obtain
\begin{equation*}
    \size{\XSN_{\alltoall}^{(n)}} = \frac{1}{2}n^2 + \mathcal{O}(n).
\end{equation*}
A similar argument can be used to derive the depth: Since the count of each $\XSC_{\alltoall}^{(n)}$ halves, so does the depth. Consequently, shifted concatenations of $\XSC_{\alltoall}^{(n)}$ forming $\XSN_{\alltoall}^{(n)}$ equivalently halve in depth resulting in
\begin{equation*}
    \depth{\XSN_{\alltoall}^{(n)}} = 2n + \mathcal{O}(1).
\end{equation*}
Fig.~\ref{fig:all-to-all} illustrates the circuit for the two-body generator for complete graphs.

For square grid connectivity graphs, we observe that the HGP consists of $\frac{1}{3}n+\mathcal{O}(\sqrt{n})$ qubits, resulting in $\frac{1}{3}n+\mathcal{O}(\sqrt{n})$ \DXNAME{} gates, and $\frac{2}{3}n+\mathcal{O}(\sqrt{n})$ HGP neighbors, corresponding to $\frac{2}{3}n+\mathcal{O}(\sqrt{n})$ \CXNAME{} gates. Hence, we deduce that $\size{\XSC_{\grid}^{(n)}}=\frac{4}{3}n+\mathcal{O}(\sqrt{n})$. This implies that
\begin{equation*}
    \size{\XSN_{\grid}^{(n)}} =\sum_{l=0}^{n-2}  \size{\XSC_{\grid}^{(n-l)}}=\sum_{l=2}^{n}\bigg(\frac{4}{3}l+\mathcal{O}\big(\sqrt{l}\big)\bigg)=\frac{2}{3}n^2+\mathcal{O}(n^{3/2}).
\end{equation*}
Note that we also neglected the final pigeonhole circuit after each $\XSC_{\grid}$ since they are of order $\mathcal{O}(1)$.
The depth depends on the shifts introduced between consecutive $\XSC_\grid$ circuits. From the circuit in Fig.~\ref{fig:2body_generator_square_grid_2}, we observe that the square grid strategy  requires a shift of six moments for a grid with three rows. Therefore, the depth is
\begin{equation*}
    \depth{\XSN_\grid^{(n)}} = 6n + \mathcal{O}(\sqrt{n}).
\end{equation*}
\subsection{Count and depth analysis of $k$-body generators}
Next, we discuss the generalization to $k$-body generators by first considering complete graphs. The case $k=3$ is not treated explicitly, but follows the same reasoning.

From Fig.~\ref{fig:4body_gen}(a) we observe that the only component of the clean special four-body generator circuit that determines the size to leading-order is the $\XSN_4^{(n)}$ circuit. $\XSN_{4,\alltoall}^{(n)}$ can be obtained similarly to $\XSN_{\alltoall}^{(n)}$ and has half the size of $\XSN_4^{(n)}$. This readily yields
\begin{equation*}
    \size{\SG_{4,\alltoall}^{(n)}}
    =\mathcal{O}(n)+\sum_{l=0}^{n-4}\size{\XSN_{4,\alltoall}^{(n-l)}}
    = \frac{1}{6}\,n^3 + \mathcal{O}(n^2).
\end{equation*}

Together with Theorem \ref{thm:k-body_generator_from_special_k-body_generator}, we find $\mu(\G_{k,\alltoall},\mathcal{L}_k) = 1$.

For the depth scaling of the circuit construction for complete graphs, we again refer to Fig.~\ref{fig:4body_gen}(a), now focusing on the shifts between circuit components. As outlined in App.~\ref{app:count_and_depth_two_body_gen} similar to the count, also the depth of respective two-body generators for complete graphs halves with respect to nearest neighbor connectivity graphs. Now there are two extra components to be considered for special four-body generators: the initial \CXNAME{} chains of the $\XSN_4$ circuits and the initial \CXNAME{} gates of the modified \TWINE{} chains. These components consist of \CXNAME{} gates with shifts of two and four moments, respectively, which can also be halved. Therefore, we have
\begin{equation*}
    \depth{\SG_{4,\alltoall}^{(n)}}
    = \tfrac{1}{2}\,n^2 + \mathcal{O}(n).
\end{equation*}

Finally, combining this with Theorem \ref{thm:k-body_generator_from_special_k-body_generator} we find $\nu(\G_{k,\alltoall},\mathcal{L}_k) = k/2$.

Similar arguments can be applied to the size scaling of respective circuits adapted to square grids.
The main size contribution for building the clean special four-body generator still arises from the concatenated $\XSN_{4,\grid}$ circuits, which can be obtained just as the $\XSN_{\grid}$ circuits. Therefore, using $\size{\XSC_{\grid}^{(n)}}=\frac{4}{3}n+\mathcal{O}(\sqrt{n})$ and applying Theorem \ref{thm:k-body_generator_from_special_k-body_generator}, we conclude
\begin{equation*}
    \mu(\G_{k,\grid},\mathcal{L}_k) = \frac{4}{3}.
\end{equation*}

For the depth analysis, we observe that the $\XSN_{\grid}$ circuit forms a ``triangle'' (compare to the circuit diagram of Fig.~\ref{fig:2body_generator_square_grid_2}) with a relatively long right side (measured in circuit moments), unlike the symmetric ``triangle'' used for the circuit construction for nearest neighbor connectivity graphs. In the main text, we discuss a strategy useful for QAOA using shifted concatenations of $\XSN_{\grid}$ and $\left(\overline{\XSN}_{\grid}\right)^\dagger$ to optimize depth. However, since $\XSN_{4,\grid}$ is not clean, utilizing the adjoint is not possible because the special label creates an asymmetry between the input and output label sequences. Consequently, when building a clean special four-body generator, we can only use reversed circuits for subsequent special $\XSN_{4,\grid}$ circuits. This allows shifted concatenations with an overlap of $\tfrac{4}{3}n + \mathcal{O}(\sqrt{n})$ leading to
\begin{equation*}
    \depth{\SG_{4,\grid}^{(n)}}
    = \mathcal{O}(n) 
      + \sum_{l=4}^{n-1}
        \Bigl(6l - \tfrac{4}{3}l
      + \mathcal{O}\big(\sqrt{l}\big)\Bigr)
    = \frac{7}{3}n^2 + \mathcal{O}(n\sqrt{n}).
\end{equation*}

Combining this result with Theorem \ref{thm:k-body_generator_from_special_k-body_generator} yields $\nu(\G_{k,\grid},\mathcal{L}_k) = \tfrac{7}{3}k$. For the analogous construction for heavy-hexagon connectivity graphs we obtain $\mu(\G_{k,\mathrm{hex}},\mathcal{L}_k) = \tfrac{5}{3}$ and $\nu(\G_{k,\mathrm{hex}},\mathcal{L}_k) = \tfrac{5}{3}k$ for $k\geq 3$.

The above results, in particular the increased depth scaling, emerges from the alignment mismatch of the corresponding generators. This is especially pronounced in the case of square grid connectivity graph and, in fact, is expected to worsen with increased connectivity. However, for the case of $k=3$, these issues can be partially circumvented using adapted generator building blocks to yield further depth reductions. The main idea of this approach is as follows: Instead of applying the initial \CXNAME{} chain solely before the first $\XSN_{3,\grid}$ circuit, we apply it before every $\XSN_{3,\grid}$ circuit. Moreover, we add an additional \CXNAME{} chain after each $\XSN_{3,\grid}$ circuit (see Fig.~\ref{fig:special_four_body_grid}). The resulting modified $\XSN_{3,\grid}$ circuit maps a sequence of single labels to a sequence of single labels while generating all three-body labels with the special label $\lab_{s}$. As a consequence, this enables us to use Lemma \ref{lem:adjoint_circuit} and concatenate $\XSN_{3,\grid}$ and $(\overline{\XSN}_{3,\grid})^{\dagger}$ (instead of $\XSN_{3,\grid}$ and $\overline{\XSN}_{3,\grid}$). Then, consecutive $\XSN_{3,\grid}$ [$(\overline{\XSN}_{3,\grid})^{\dagger}$] can be shifted with overlaps alternating between $6n-4/3n+\mathcal{O}(\sqrt{n})=14/3n+\mathcal{O}(\sqrt{n})$ and $n+\mathcal{O}(1)$ which reduces the depth as we circumvent the alignment mismatch. However, individual modified $\XSN_{3,\grid}$ blocks possess an increased depth of $7n$, so that in total the depth of the three-body generator designed in this way sums up to $\frac{25}{12}n^2+\mathcal{O}(n^{3/2})$. In principle, this approach can also be applied to heavy-hexagon connectivity graphs. However, there, the corresponding alignment mismatch is less dramatic [see Fig.~\ref{fig:app:heavy_hex}(b)]. As a consequence the modification of the $\XSN_{3,\grid}$ as outlined for square grid connectivity graphs does not lead to a depth reduction. In general, we expect this method to achieve depth reductions for layouts with sufficient connectivity with the effect getting more pronounced for increasing connectivity.

{\color{black}
\section{Circuit fidelities and performance comparison \label{app:circuit_fidelities}}
\subsection{Effective depth}

In this section, we compute the performance of the different two-body generators for various qubit connectivities as defined in Eq.~\eqref{eq:performance} of the main text. 

While the count of our circuit building blocks is equivalent for all applications, the circuit depths is more subtle. The structure of the different two-body generators $\mathcal{G}_{2}^{(n)}$ implies qubits idling for a substantial part of the circuit (see for example Fig.~\ref{fig:2body_generator_square_grid_2}). However, for the discussed applications such as QAOA or Hamiltonian simulation, repeated application eventually allows to reduce the idling time and cut the depth by a factor $\sim 2$. However, there are also applications where this is not possible (such as QFT). Thus, in principle, a detailed analysis requires a case-by-case study. 

Here, we focus on comparing individual building blocks (as for example relevant for QFT). In this case we can adopt an as-soon-as-possible (ASAP) scheduling for which results are retrieved (i.e. qubits are measured) ASAP. Thus, for this scheduling, it makes sense to define an \textit{effective} depth, $\mathrm{depth}_{\mathrm{eff}}(C)$, reflecting the early termination once a qubit goes idle. Concretely, we define $\mathrm{depth}_{\mathrm{eff}}(C)$ of a circuit $C$ by

\begin{eqnarray}
   \mathrm{depth}_{\mathrm{eff}}(C) = \frac{1}{n} \sum_{k=1}^n D(q_k, C)
\label{eq:eff_depth}
\end{eqnarray}
where $D(q_k, C)$ is defined by last the moment of the circuit $C$ in which a gate is applied to qubit $q_k$. Note that the definition of Eq.~\eqref{eq:eff_depth} is different from the depth of the circuit, which is defined as the maximum moment at which a gate is applied to any of the qubits involved in the circuit $C$. 

To give an example, the effective depth associated with the clean two-body generator on a square grid (Fig.~\ref{fig:2body_generator_square_grid_2}) computes as
\begin{eqnarray}
   \mathrm{depth}_{\mathrm{eff}}(\mathcal{G}_{2,\mathrm{grid}}^{(n)}) = \frac{1}{n} \sum_{k=1}^n \left(\frac{14}{3}k + \frac{4}{3}n -1\right) = \frac{11}{3}n + \frac{4}{3}.~~
\label{eq:d_eff_square}
\end{eqnarray}
The first $\XSC_\grid^{(1,n)}$ adds a depth of $4/3 n$ on qubit $q_n$. Every consecutive chain adds $14/3$ additional gates (on average depending on the concrete pigeonhole sequence). Note that this pattern breaks down for the final few $\XSC_\grid^{(1,p)}$ of the $\XSN_\grid^{(n)}$. Equation \eqref{eq:d_eff_square} thus represents an upper bound. In contrast the depth of $\mathcal{G}_{2,\mathrm{grid}}^{(n)}$ is given by $\depth{\mathcal{G}_{2,\mathrm{grid}}^{(n)}}=6n +\mathcal{O}(1)$. In Tab.~\ref{tab:eff_depth} we collect the values obtained for the effective depth for each qubit connectivity investigated in this work. Note that these results do not include the single-qubit rotation gates. They typically possess execution times which are around one order of magnitude smaller compared to the execution time of two-qubit gates. Moreover, the expected depth increase due to single-qubit gates is similar for all discussed implementations and should therefore not matter for the employed comparison.

\begin{table}[]
\centering
\begin{tblr}{
colspec={|l |c|c|}, rowsep=1.5pt, colsep=2.5pt, hline{1,2,6},
cells={valign=m,halign=c},
}

Connectivity graph   & count & effective depth\\ 

    LNN  & $n^2$ & $3n+2$\\
    heavy hexagon & $\frac{5}{6}n^2+\frac{17}{6}n +\mathcal{O}(1)$ & $\frac{10}{3}n + \frac{51}{3}$\\ 
    ladder & $\frac{3}{4}n^2 +n + \mathcal{O}(1)$ & $\frac{9}{4}n+\frac{3}{2}$ \\
    square grid & $\frac{2}{3}n^2+ \frac{2}{3}n + \mathcal{O}(1)$ & $\frac{11}{3}n + \frac{4}{3}$\\
\end{tblr}
    \caption{Summary of two qubit gate count (up to $\mathcal{O}(1)$ corrections) and effective depth (defined in Eq.~\eqref{eq:eff_depth}) of the clean two-body generators $\mathcal{G}_{2,\mathrm{grid}}^{(n)}$ for the locally-connected qubit layouts discussed in this work }
    \label{tab:eff_depth}
\end{table}

\subsection{Comparing circuit fidelities}
Using Eq.~\eqref{eq:eff_depth} in Eq.~\eqref{eq:performance}, we can now compute the expected performance of the clean two-body generators for each of the discussed (locally-connected) qubit layout and compare them against each other. 

As an example, here we present the comparison of the ladder and the square grid implementation. To compare the performance of ladder and square grid and obtain the respective regions of performance advantage, we search for solutions to the equation
\begin{eqnarray}
    \label{eq:square_vs_ladd}
    \frac{\mathcal{F}_{\mathcal{G}_{2,\mathrm{grid}}^{(n)}}}{\mathcal{F}_{\mathcal{G}_{2,\mathrm{ladder}}^{(n)}}}=1.
\end{eqnarray}
We can readily reformulate Eq.~\eqref{eq:square_vs_ladd} as
\begin{eqnarray}
\label{eq:comparing_fidelities}
\bigg\lbrace\size{\mathcal{G}_{2,\mathrm{grid}}^{(n)}} -\size{\mathcal{G}_{2,\mathrm{ladder}}^{(n)}}\bigg\rbrace\log\left(\frac{F_{2q}}{F_\mathrm{idle}^2}\right) + n\bigg\lbrace\mathrm{depth}_{\mathrm{eff}}\left(\mathcal{G}_{2,\mathrm{grid}}^{(n)}\right)- \mathrm{depth}_{\mathrm{eff}}\left(\mathcal{G}_{2,\mathrm{ladder}}^{(n)}\right)\bigg\rbrace \log\left(F_\mathrm{idle}\right) = 0.
\end{eqnarray}
Inserting the values given in Tab.~\ref{tab:eff_depth}, we can further reduce Eq.~\eqref{eq:comparing_fidelities} to yield
\begin{eqnarray}
    \label{eq:transition}
    -\frac{1}{12}\bigg\lbrace \log\left(\frac{F_{2q}}{F_\mathrm{idle}^{19}}\right)\bigg\rbrace n -\frac{1}{6}\log\left(\frac{F_{2q}^2}{F_{\mathrm{idle}}^3}\right) = 0.
\end{eqnarray}
From Eq.~\ref{eq:transition} we see that the transition where square grid becomes advantageous over the ladder implementation depends on $n$. For the smallest meaningful value $n=2$, we find the transition at $F_{2q}= F_{\mathrm{idle}}^{\frac{22}{3}}$. In the limit $n\rightarrow \infty$, the second term in Eq.~\eqref{eq:transition} becomes irrelevant so that the transition happens for $F_{2q}= F_{\mathrm{idle}}^{19}$.
}

\section{Proof of a non-trivial lower bound for $\mu$}\label{sec:proof-impossibility-theorem}

In this section, we derive a non-trivial lower bound for the asymptotically average \CXNAME{} count for nearest neighbor connectivity graphs.

\begin{theorem}
    \label{thm:impossibility_theorem}
    Consider a family $L=(L^{(n)})_n$ of label sets ($n$ being the number of qubits) with the following two properties:

    \begin{enumerate}
        \item \label{item:symdiff_property}For all $\lab_1, \lab_2\in L^{(n)}$ we have $\lab_1\lab_2\notin L^{(n)}$.
        \item \label{item:cardinality_property}We have $\lim_{n\to\infty} \frac{\abs{L^{(n)}}}{n} = \infty$.
    \end{enumerate}

    Let $C=(C^{(n)})_n$ be a family of \CXNAME{}-circuits for nearest neighbor connectivity graphs so that $C^{(n)}$ generates $L^{(n)}$. Then

    \begin{equation*}
        \mu(C,L)\geq 1 + \gamma ,
    \end{equation*}

    for some absolute constant $\gamma>0$.
\end{theorem}

\begin{remark}
    The largest $\gamma$ for which the theorem holds is not known to us. In the proof below we give a simple argument to see that $\gamma=\frac{1}{43}$ works. A slightly more involved argument pushes this to $\gamma=\frac{1}{9}$. But we have no reason to assume that this is optimal.
\end{remark}

\begin{proof}
    Consider an arbitrary family of labels and a corresponding family of generator circuits as described \emph{but} with the property that

    \begin{equation}
        \label{eq:assump_proof_impossibility}
        \frac{\size{C^{(n)}}}{\abs{L^{(n)}}} \leq 1 + \delta
    \end{equation}

    for an infinite number of $n$. In order to prove the theorem we have to show that necessarily $\delta\geq\gamma$. Below we always assume that $n$ is set to one of the values for which the above inequality holds. Moreover we consider the limit for large $n$ whenever we utilize Landau-notation like $o(1)$.

    Without loss of generality we may assume that every moment of the circuits contains exactly one \CXNAME{} gate. We may also assume that none of the single-body labels is an element of $L^{(n)}$ (due to property \ref{item:cardinality_property}).

    Let us define a few things. For a moment $m$ and qubit $i$ we call $(i, m)$ a \emph{spacetime point}. Next we define the notion of a \emph{good} spacetime point which basically tracks where new labels occur. We do this recursively by looking through the following list of rules and apply the first one which matches:

    \begin{itemize}
        \item Initialization: spacetime point $(i, 1)$ is \emph{bad} (meaning \emph{not good}) for all $i$,
        \item Spawning good: $(i, m)$ is good if qubit $i$ is a target (of a \CXNAME{} gate) at moment $m$ and the resulting label is in $L^{(n)}$ and occurs for the first time (did not occur in an earlier moment at any of the qubits),
        \item Spawning bad: $(i, m)$ is bad if qubit $i$ is a target at moment $m$ and the previous case did not apply.
        \item Propagation: $(i, m)$ is good if qubit $i$ is not a target in moment $m$ and $(i, m-1)$ is good. The same rule applies to bad spacetime points.
    \end{itemize}

    We say that a \CXNAME{} gate (of the circuit) at moment $m$ with target at $t$ is \emph{successful} if spacetime point $(t, m)$ is good. Finally, let us say that a \CXNAME{} gate at moment $m$ with target $t$ is \emph{typical} if it is successful, $m>1$, and spacetime point $(t, m-1)$ was already good. That is, the qubit which is targeted by a typical gate already supported a new label before the \CXNAME{} was applied and again got a new one after applying the gate. Fig.~\ref{fig:def-typical-cnot-gate} shows a summary of these definitions.

    \begin{figure}[h]
       \centering
        \includegraphics[width=.9\textwidth]{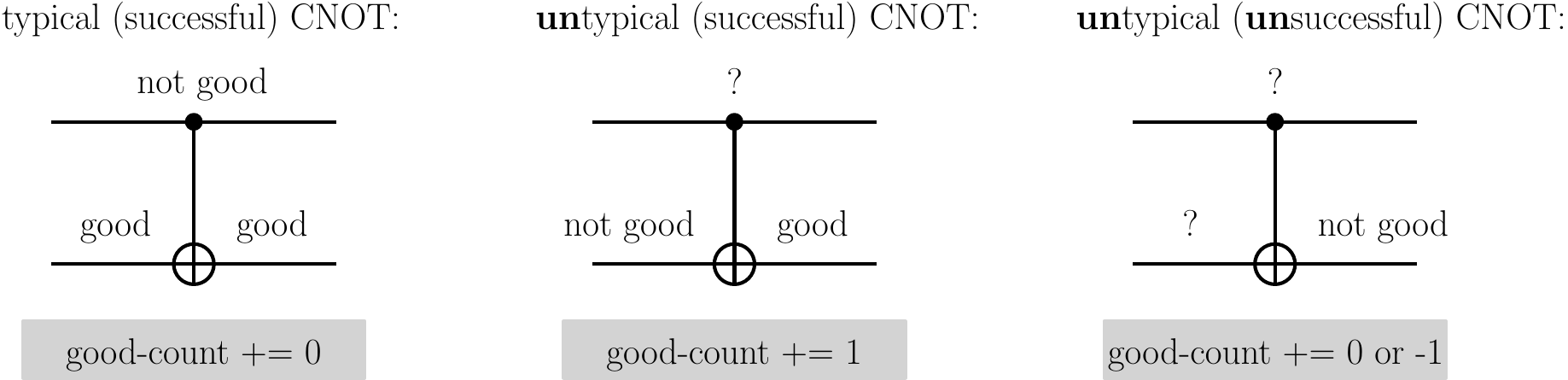}
        \caption{A \emph{successful} gate is a gate which produces a good spacetime point on its target (it contains one of the desired labels for the first time). A \emph{typical} gate is successful and in addition its target was already succeded by a good spacetime point. Interestingly, the control of a typical gate never occupies a good spacetime point. This is a consequence of the property \ref{item:symdiff_property} of the label sets we consider and is crucial for the proof.}
        \label{fig:def-typical-cnot-gate}
    \end{figure}

    In what follows we try to count the number of untypical gates in two different ways. We will obtain a contradiction if $\delta$ is too small.

    \paragraph{Counting untypical gates - part 1.}
    
    For a moment $m$ let $\good(m)$ be the number of good spacetime points in that moment. Clearly $0\leq \good(m)\leq n$. Moreover, $\good(m) - \good(m-1) \in\{-1,0,+1\}$ for all $m>1$ (recall our assumption that each moment contains exactly one CNOT gate). Observe that (c.f. Fig.~\ref{fig:def-typical-cnot-gate}):

    \begin{itemize}
        \item If $g$ is typical we have $\good(m)-\good(m-1)=0$.
        \item If $g$ is untypical and successful we have $\good(m)-\good(m-1)=1$.
        \item If $g$ is unsuccessful (hence untypical) we have $\good(m)-\good(m-1)\in \{-1,0\}$.
    \end{itemize}

    Clearly we have exactly $\abs{L^{(n)}}$ successful gates. Hence, we have at most $\delta\abs{L^{(n)}}$ unsuccessful gates (see Eq.~\eqref{eq:assump_proof_impossibility}). Furthermore, from the above listing and the fact that the $\good$-function is bounded by $n$ we deduce that
    \begin{align*}
        &\abs{\set{g\left\vert\,  \text{$g$ is an untypical and a successful gate in $C^{(n)}$}\right.}}=\sum_{\substack{m=2\\ \good(m)-\good(m-1)=1}}^{\depth{C^{(n)}}}1\\
        &\hspace{2cm}=\sum_{m=2}^{\depth{C^{(n)}}}\good(m)-\good(m-1)-\sum_{\substack{m=2\\ \good(m)-\good(m-1)\neq 1}}^{\depth{C^{(n)}}}\good(m)-\good(m-1)\\
        &\hspace{2cm}\leq n+\abs{\set{g\left\vert\, \text{$g$ is an unsuccessful gate in $C^{(n)}$}\right.}}.
    \end{align*}
    This and \ref{item:cardinality_property} implies that there are at most $(\delta+o(1))\abs{L^{(n)}}$ untypical successful gates.
    Hence,
    \begin{equation}\label{proof:untypical-gates-estimate-1}
        \text{there are at most } (2\delta + o(1)) \abs{L^{(n)}} \text{ untypical gates}
    \end{equation}
    or more precisely,
    \begin{equation}\label{proof:typical-gates-estimate-1}
        \text{there are at least } (1 - \delta - o(1)) \abs{L^{(n)}} \text{ typical gates.}
    \end{equation}

    \paragraph{Counting untypical gates - part 2.}

    For subsets $G$ of gates from $C^{(n)}$ let us define the partial function $\pred'$ by

    \begin{equation*}
        \pred'(G)\coloneqq\text{ the latest $g$ in $C^{(n)}$ before $G$ which targets one of the qubits of $G$} . 
    \end{equation*}

    It is a partial function because such a $g$ must not necessarily exist. With \emph{qubits of} $G$ we just mean the union of all controls and targets of the gates in $G$. That $g$ comes \emph{before} $G$ means that every gate in $G$ lives in a later moment than $g$. Let us define another partial function $\pred_3$, this time defined on the gates $g$ of $C^{(n)}$ by

    \begin{equation*}
        \pred_3(g) = (g_1, g_2, g_3)
    \end{equation*}

    where $g_1=\pred'(\{g\})$ and $g_{i+1}=\pred'(\{g, g_1,\ldots,g_i\})$ for $i=1,2$. The intuition behind this is that $\pred_3(g)$ gives us the first three gates \emph{immediately} before $g$ which have some influence on the label produced by $g$ (if this many predecessor gates exist).

    We call a gate $g$ \emph{super-typical} if it is typical and $\pred_3(g)$ exists.

    It is not hard to see that the number of gates for which $\pred_3$ is \emph{not} defined is at most $3(n-1)=o(1)\abs{L^{(n)}}$ (using property \ref{item:cardinality_property}). Therefore, we deduce from \eqref{proof:typical-gates-estimate-1} that

    \begin{equation}\label{proof:estimate-super-typical-gates}
        \text{there are at least } (1-\delta-o(1)) \abs{L^{(n)}} \text{ super-typical gates.}
    \end{equation}

    Finally let us define yet another partial function $U$ by

    \begin{equation*}
        U(g)\coloneqq \text{the first gate in } \pred_3(g) \text{ which is untypical.}
    \end{equation*}

    The intuition behind it is that for "some" typical gates it finds us an untypical gate. This in turn will help us to count untypical gates in a second way. Let us state two important claims

    \begin{itemize}
        \item Claim 1: $U(g)$ is defined for all super-typical $g$. 
        \item Claim 2: There exists an absolute constant $N$ such that for each untypical gate $u$ the pre-image $U^{-1}(\set{u})$ contains at most $N$ elements. An easy argument shows that $N=21$ works. A more involved one shows that $N=4$ works too.
    \end{itemize}

    Before we prove those two claims let us finish the main part of the proof: We are now in the position to prove the following chain of inequalities:

    \begin{align*}
        (2\delta + o(1)) \abs{L^{(n)}} &\geq \abs{\{g\mid g \text{ is untypical} \}} \\
        &\geq \abs{\set{U(g)\mid g \text{ is super-typical}}}
        =\sum_{g\in \set{g\mid g \text{ is super-typical}}} \frac{1}{\abs{U^{-1}(\set{U(g)})}} \\
        &\geq \frac{1}{N} \abs{\{g\mid g \text{ is super-typical} \}} \\
        &\geq \frac{1}{N} (1 - \delta - o(1)) \abs{L^{(n)}} .
    \end{align*}

    The first line is just (\ref{proof:untypical-gates-estimate-1}). The second line follows from claim 1 and the third line from claim 2. The fourth line follows from (\ref{proof:estimate-super-typical-gates}). From this we deduce that

    \begin{equation*}
        \delta \geq \frac{1}{2N + 1} 
    \end{equation*}

    Using $N=21$ implies the theorem and the proof is complete (up to the two claims).

    \paragraph{Proof of claim 1:}

    \begin{figure}
       \centering
        \includegraphics[width=.5 \columnwidth]{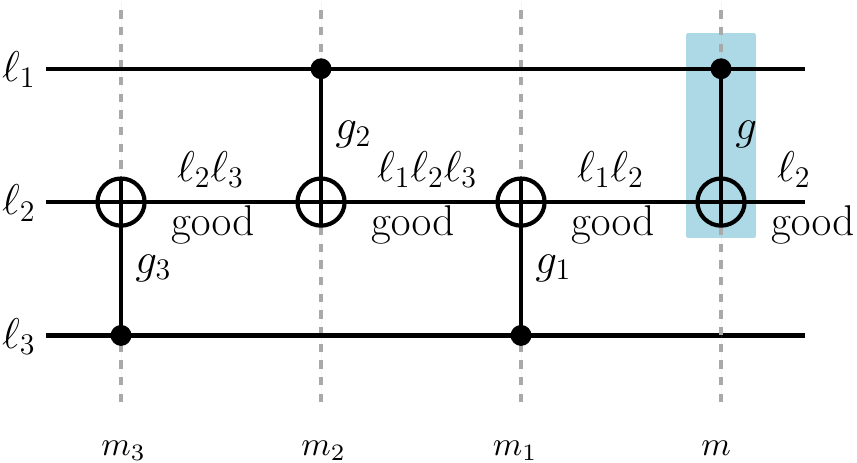}
        \caption{
        Under the assumption that all of the four gates $g$, $g_1$, $g_2$, and $g_3$ are typical the gates have to be arranged in some alternating fashion as shown. This on the other hand is also not possible because the middle label present at spacetime point $(k+1,m_3-1)$ would be regenerated by $g$, which is not possible if $g$ is typical. Let us note here that the definition of $\pred_3$ ensures that we can ignore all other gates for the things that matter in our analysis. Indeed, it can happen that e.g. between moments $m_3$ and $m$ other gates are applied. But by definition of the $g_i$ they do not target relevant qubits. E.g. it can happen that there is a gate with control at $k$ and target at $k-1$ somewhere between moments $m_3$ and $m$. But this is irrelevant for our analysis.}
        \label{fig:find-untypical-gate}
    \end{figure}

    The claim basically just says that for any typical gate $g$ for which $\pred_3(g)=(g_1,g_2,g_3)$ exists, at least one of the $g_i$ is untypical. We prove this by contradiction: Let us assume that there was a typical $g$, for which $\pred_3(g)=(g_1,g_2,g_3)$ exists so that all $g_i$ are typical too.

    Let $m$ be the moment in which $g$ acts and $m_i$ the moment in which $g_i$ acts. By definition we have $m>m_1>m_2>m_3$. Let us assume that the control qubit of $g$ is $k$ and the target qubit $k+1$ (without loss of generality we can assume that the target has a higher index, otherwise we can just mirror the situation).

    Let us make a short intermezzo to make a crucial observation: Consider a \CXNAME{} gate at moment $m$ and with control and target at $c$ and $t$. Assume it is typical. Then, the spacetime point $(c,m)$ is necessarily \emph{bad}. This follows from property \ref{item:symdiff_property} of the family of label sets $(L^{(n)})_n$ (two of the desired labels cannot be combined to another desired label).

    Using this crucial observation it is not hard to see that the gates of $g_i$ have to to placed as drawn in Fig.~\ref{fig:find-untypical-gate}. To see this, let us start by proving that $g_1$ has control at $k+2$ and target at $k+1$. By definition of $g_1$ it has to target either $k$ or $k+1$ (the qubits of $g$). There are four possible cases under a linear connectivity constraint:
    \begin{enumerate}
        \item control at $k-1$, target at $k$.
        \item control at $k$, target at $k+1$.
        \item control at $k+1$, target at $k$.
        \item control at $k+2$, target at $k+1$.
    \end{enumerate}
    The first case is not possible because typicality of $g_1$ would imply that the spacetime points $(k,m_1)$ to $(k,m)$ are good which in turn contradicts $g$ being typical due to our crucial observation. In the second case, the spacetime points $(k+1,m_1-1)$ and $(k+1,m)$ would contain the same label, again contradicting typicality of $g$. In the third case, the spacetime point $(k,m-1)$ would be good, which is not possible due to the same reasons as in the first case. Hence, only the last case remains - as desired. The arguments for the other two gates are basically the same.

    Hence, the gates are positioned as shown in Fig.~\ref{fig:find-untypical-gate}. On the other hand, this implies that the spacetime points $(k+1,m_3-1)$ and $(k+1,m)$ have the same label - contradicting typicality of $g$. This shows the first claim.
    
    \paragraph{Proof of claim 2:} Given an untypical gate $u$ we have to find an absolute upper bound $N$ on the number of typical gates $g$ such that $U(g)=u$. Assume that the support of $u$ (control and target) is $\{k,k+1\}$. Then any $g$ with $U(g)=u$ must be supported on $\{k-3,\ldots,k+4\}$. More precisely its support is one of seven possibilities $\{k-3,k-2\}$, $\{k-2,k-1\}$, ..., $\{k+3,k+4\}$. But for each of these seven supports there can only be at most three corresponding $g$ with $U(g)=u$ (the fourth one would have all of the other three gates in the way towards $u$). Hence in total there are at most three times seven (i.e. 21) $g$ with $U(g)=u$. In other words, $N=21$ is an upper bound as desired.

    \paragraph{Improving the bound $N$ in claim 2:} The preceding argument is rather crude and merely serves as a simple argument to see the existence of $N$. One may ask what the optimal value for $N$ is to ultimately find a better bound for $\gamma$. But note that the optimal $\gamma$ might still be larger than the optimal $N$ suggests (we believe that this is indeed the case).

    Again consider an untypical $u$. Using the fact that typical gates come at an alternating pattern (see Fig.~\ref{fig:find-untypical-gate}) we can see that there can be at most five gates $g$ such that $U(g)=u$. 

    In Fig.~\ref{fig:example-N-improved-bound} we show a \emph{hypothetical} scenario for which $u=U(g^{(i)})$ for $i\in\set{1,\ldots,5}$. Up to certain trivial changes this is actually the only scenario for which $U^{-1}(u)$ has (at most) five elements.

    There are still two cases to consider: The gate $g$ can either be present or not. Although we have $U(g)\neq u$ the presence of this gate can in principle influence the pre-image. Interestingly both cases are not possible! If the gate $g$ is not present it is easy to see that $g^{(1)}$ cannot be typical (which would imply $U(g^{(2)})=g^{(1)}\neq u$). In fact, label $\lab_3$ would be regenerated by $g^{(1)}$. If $g$ is present $g^{(2)}$ cannot be typical because it would regenerate $\lab_3$.

    Hence we showed that four is an upper bound for $N$ and hence the theorem holds for $\gamma=\frac{1}{9}$. \qedhere

    \begin{figure}[h]
       \centering
        \includegraphics[width=\columnwidth/2]{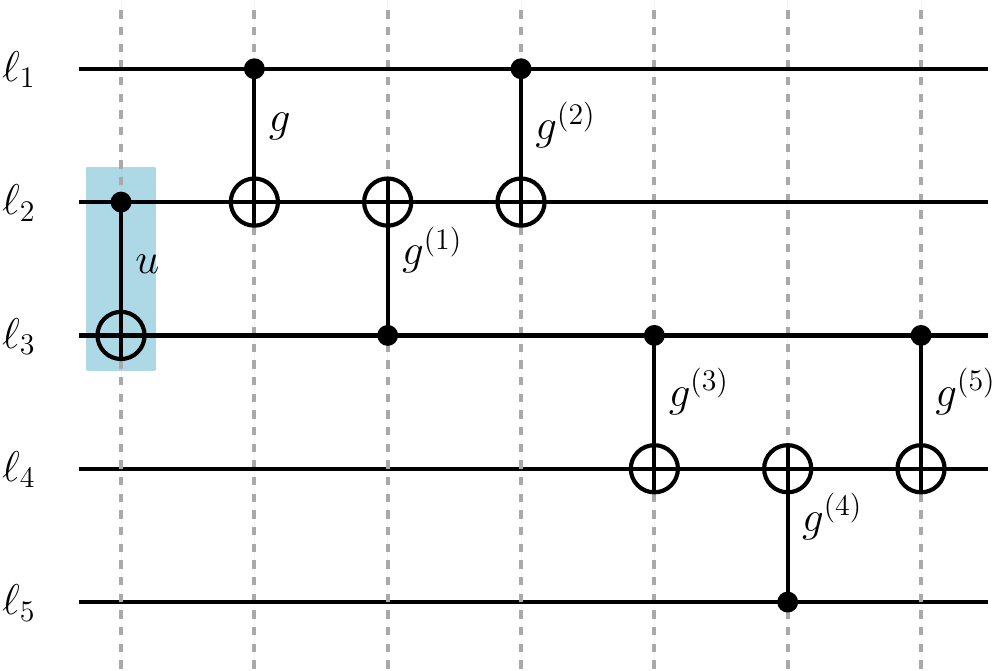}
        \caption{The only \emph{hypothetically} possible scenario for an untypical gate $u$ for which $U^{-1}(u)$ ($=\set{g^{(i)}\mid i=1,\ldots,5}$) has (at most) five elements. The gate $g$ can either be present or not. In both cases, using label tracking we see that this scenario is not possible concluding that $U^{-1}(u)$ has at most four elements for any untypical $u$.}
        \label{fig:example-N-improved-bound}
    \end{figure}
\end{proof}

\twocolumngrid


\begin{thebibliography}{71}%
\makeatletter
\providecommand \@ifxundefined [1]{%
 \@ifx{#1\undefined}
}%
\providecommand \@ifnum [1]{%
 \ifnum #1\expandafter \@firstoftwo
 \else \expandafter \@secondoftwo
 \fi
}%
\providecommand \@ifx [1]{%
 \ifx #1\expandafter \@firstoftwo
 \else \expandafter \@secondoftwo
 \fi
}%
\providecommand \natexlab [1]{#1}%
\providecommand \enquote  [1]{``#1''}%
\providecommand \bibnamefont  [1]{#1}%
\providecommand \bibfnamefont [1]{#1}%
\providecommand \citenamefont [1]{#1}%
\providecommand \href@noop [0]{\@secondoftwo}%
\providecommand \href [0]{\begingroup \@sanitize@url \@href}%
\providecommand \@href[1]{\@@startlink{#1}\@@href}%
\providecommand \@@href[1]{\endgroup#1\@@endlink}%
\providecommand \@sanitize@url [0]{\catcode `\\12\catcode `\$12\catcode
  `\&12\catcode `\#12\catcode `\^12\catcode `\_12\catcode `\%12\relax}%
\providecommand \@@startlink[1]{}%
\providecommand \@@endlink[0]{}%
\providecommand \url  [0]{\begingroup\@sanitize@url \@url }%
\providecommand \@url [1]{\endgroup\@href {#1}{\urlprefix }}%
\providecommand \urlprefix  [0]{URL }%
\providecommand \Eprint [0]{\href }%
\providecommand \doibase [0]{https://doi.org/}%
\providecommand \selectlanguage [0]{\@gobble}%
\providecommand \bibinfo  [0]{\@secondoftwo}%
\providecommand \bibfield  [0]{\@secondoftwo}%
\providecommand \translation [1]{[#1]}%
\providecommand \BibitemOpen [0]{}%
\providecommand \bibitemStop [0]{}%
\providecommand \bibitemNoStop [0]{.\EOS\space}%
\providecommand \EOS [0]{\spacefactor3000\relax}%
\providecommand \BibitemShut  [1]{\csname bibitem#1\endcsname}%
\let\auto@bib@innerbib\@empty
%</preamble>
\bibitem [{\citenamefont {Arute}\ \emph {et~al.}(2019)\citenamefont {Arute},
  \citenamefont {Arya}, \citenamefont {Babbush}, \citenamefont {Bacon},
  \citenamefont {Bardin}, \citenamefont {Barends}, \citenamefont {Biswas},
  \citenamefont {Boixo}, \citenamefont {Brandao}, \citenamefont {Buell},
  \citenamefont {Burkett}, \citenamefont {Chen}, \citenamefont {Chen},
  \citenamefont {Chiaro}, \citenamefont {Collins}, \citenamefont {Courtney},
  \citenamefont {Dunsworth}, \citenamefont {Farhi}, \citenamefont {Foxen},
  \citenamefont {Fowler}, \citenamefont {Gidney}, \citenamefont {Giustina},
  \citenamefont {Graff}, \citenamefont {Guerin}, \citenamefont {Habegger},
  \citenamefont {Harrigan}, \citenamefont {Hartmann}, \citenamefont {Ho},
  \citenamefont {Hoffmann}, \citenamefont {Huang}, \citenamefont {Humble},
  \citenamefont {Isakov}, \citenamefont {Jeffrey}, \citenamefont {Jiang},
  \citenamefont {Kafri}, \citenamefont {Kechedzhi}, \citenamefont {Kelly},
  \citenamefont {Klimov}, \citenamefont {Knysh}, \citenamefont {Korotkov},
  \citenamefont {Kostritsa}, \citenamefont {Landhuis}, \citenamefont
  {Lindmark}, \citenamefont {Lucero}, \citenamefont {Lyakh}, \citenamefont
  {Mandr{\`a}}, \citenamefont {McClean}, \citenamefont {McEwen}, \citenamefont
  {Megrant}, \citenamefont {Mi}, \citenamefont {Michielsen}, \citenamefont
  {Mohseni}, \citenamefont {Mutus}, \citenamefont {Naaman}, \citenamefont
  {Neeley}, \citenamefont {Neill}, \citenamefont {Niu}, \citenamefont {Ostby},
  \citenamefont {Petukhov}, \citenamefont {Platt}, \citenamefont {Quintana},
  \citenamefont {Rieffel}, \citenamefont {Roushan}, \citenamefont {Rubin},
  \citenamefont {Sank}, \citenamefont {Satzinger}, \citenamefont {Smelyanskiy},
  \citenamefont {Sung}, \citenamefont {Trevithick}, \citenamefont
  {Vainsencher}, \citenamefont {Villalonga}, \citenamefont {White},
  \citenamefont {Yao}, \citenamefont {Yeh}, \citenamefont {Zalcman},
  \citenamefont {Neven},\ and\ \citenamefont {Martinis}}]{Arute2019}%
  \BibitemOpen
  \bibfield  {author} {\bibinfo {author} {\bibfnamefont {F.}~\bibnamefont
  {Arute}}, \bibinfo {author} {\bibfnamefont {K.}~\bibnamefont {Arya}},
  \bibinfo {author} {\bibfnamefont {R.}~\bibnamefont {Babbush}}, \bibinfo
  {author} {\bibfnamefont {D.}~\bibnamefont {Bacon}}, \bibinfo {author}
  {\bibfnamefont {J.~C.}\ \bibnamefont {Bardin}}, \bibinfo {author}
  {\bibfnamefont {R.}~\bibnamefont {Barends}}, \bibinfo {author} {\bibfnamefont
  {R.}~\bibnamefont {Biswas}}, \bibinfo {author} {\bibfnamefont
  {S.}~\bibnamefont {Boixo}}, \bibinfo {author} {\bibfnamefont {F.~G. S.~L.}\
  \bibnamefont {Brandao}}, \bibinfo {author} {\bibfnamefont {D.~A.}\
  \bibnamefont {Buell}}, \bibinfo {author} {\bibfnamefont {B.}~\bibnamefont
  {Burkett}}, \bibinfo {author} {\bibfnamefont {Y.}~\bibnamefont {Chen}},
  \bibinfo {author} {\bibfnamefont {Z.}~\bibnamefont {Chen}}, \bibinfo {author}
  {\bibfnamefont {B.}~\bibnamefont {Chiaro}}, \bibinfo {author} {\bibfnamefont
  {R.}~\bibnamefont {Collins}}, \bibinfo {author} {\bibfnamefont
  {W.}~\bibnamefont {Courtney}}, \bibinfo {author} {\bibfnamefont
  {A.}~\bibnamefont {Dunsworth}}, \bibinfo {author} {\bibfnamefont
  {E.}~\bibnamefont {Farhi}}, \bibinfo {author} {\bibfnamefont
  {B.}~\bibnamefont {Foxen}}, \bibinfo {author} {\bibfnamefont
  {A.}~\bibnamefont {Fowler}}, \bibinfo {author} {\bibfnamefont
  {C.}~\bibnamefont {Gidney}}, \bibinfo {author} {\bibfnamefont
  {M.}~\bibnamefont {Giustina}}, \bibinfo {author} {\bibfnamefont
  {R.}~\bibnamefont {Graff}}, \bibinfo {author} {\bibfnamefont
  {K.}~\bibnamefont {Guerin}}, \bibinfo {author} {\bibfnamefont
  {S.}~\bibnamefont {Habegger}}, \bibinfo {author} {\bibfnamefont {M.~P.}\
  \bibnamefont {Harrigan}}, \bibinfo {author} {\bibfnamefont {M.~J.}\
  \bibnamefont {Hartmann}}, \bibinfo {author} {\bibfnamefont {A.}~\bibnamefont
  {Ho}}, \bibinfo {author} {\bibfnamefont {M.}~\bibnamefont {Hoffmann}},
  \bibinfo {author} {\bibfnamefont {T.}~\bibnamefont {Huang}}, \bibinfo
  {author} {\bibfnamefont {T.~S.}\ \bibnamefont {Humble}}, \bibinfo {author}
  {\bibfnamefont {S.~V.}\ \bibnamefont {Isakov}}, \bibinfo {author}
  {\bibfnamefont {E.}~\bibnamefont {Jeffrey}}, \bibinfo {author} {\bibfnamefont
  {Z.}~\bibnamefont {Jiang}}, \bibinfo {author} {\bibfnamefont
  {D.}~\bibnamefont {Kafri}}, \bibinfo {author} {\bibfnamefont
  {K.}~\bibnamefont {Kechedzhi}}, \bibinfo {author} {\bibfnamefont
  {J.}~\bibnamefont {Kelly}}, \bibinfo {author} {\bibfnamefont {P.~V.}\
  \bibnamefont {Klimov}}, \bibinfo {author} {\bibfnamefont {S.}~\bibnamefont
  {Knysh}}, \bibinfo {author} {\bibfnamefont {A.}~\bibnamefont {Korotkov}},
  \bibinfo {author} {\bibfnamefont {F.}~\bibnamefont {Kostritsa}}, \bibinfo
  {author} {\bibfnamefont {D.}~\bibnamefont {Landhuis}}, \bibinfo {author}
  {\bibfnamefont {M.}~\bibnamefont {Lindmark}}, \bibinfo {author}
  {\bibfnamefont {E.}~\bibnamefont {Lucero}}, \bibinfo {author} {\bibfnamefont
  {D.}~\bibnamefont {Lyakh}}, \bibinfo {author} {\bibfnamefont
  {S.}~\bibnamefont {Mandr{\`a}}}, \bibinfo {author} {\bibfnamefont {J.~R.}\
  \bibnamefont {McClean}}, \bibinfo {author} {\bibfnamefont {M.}~\bibnamefont
  {McEwen}}, \bibinfo {author} {\bibfnamefont {A.}~\bibnamefont {Megrant}},
  \bibinfo {author} {\bibfnamefont {X.}~\bibnamefont {Mi}}, \bibinfo {author}
  {\bibfnamefont {K.}~\bibnamefont {Michielsen}}, \bibinfo {author}
  {\bibfnamefont {M.}~\bibnamefont {Mohseni}}, \bibinfo {author} {\bibfnamefont
  {J.}~\bibnamefont {Mutus}}, \bibinfo {author} {\bibfnamefont
  {O.}~\bibnamefont {Naaman}}, \bibinfo {author} {\bibfnamefont
  {M.}~\bibnamefont {Neeley}}, \bibinfo {author} {\bibfnamefont
  {C.}~\bibnamefont {Neill}}, \bibinfo {author} {\bibfnamefont {M.~Y.}\
  \bibnamefont {Niu}}, \bibinfo {author} {\bibfnamefont {E.}~\bibnamefont
  {Ostby}}, \bibinfo {author} {\bibfnamefont {A.}~\bibnamefont {Petukhov}},
  \bibinfo {author} {\bibfnamefont {J.~C.}\ \bibnamefont {Platt}}, \bibinfo
  {author} {\bibfnamefont {C.}~\bibnamefont {Quintana}}, \bibinfo {author}
  {\bibfnamefont {E.~G.}\ \bibnamefont {Rieffel}}, \bibinfo {author}
  {\bibfnamefont {P.}~\bibnamefont {Roushan}}, \bibinfo {author} {\bibfnamefont
  {N.~C.}\ \bibnamefont {Rubin}}, \bibinfo {author} {\bibfnamefont
  {D.}~\bibnamefont {Sank}}, \bibinfo {author} {\bibfnamefont {K.~J.}\
  \bibnamefont {Satzinger}}, \bibinfo {author} {\bibfnamefont {V.}~\bibnamefont
  {Smelyanskiy}}, \bibinfo {author} {\bibfnamefont {K.~J.}\ \bibnamefont
  {Sung}}, \bibinfo {author} {\bibfnamefont {M.~D.}\ \bibnamefont
  {Trevithick}}, \bibinfo {author} {\bibfnamefont {A.}~\bibnamefont
  {Vainsencher}}, \bibinfo {author} {\bibfnamefont {B.}~\bibnamefont
  {Villalonga}}, \bibinfo {author} {\bibfnamefont {T.}~\bibnamefont {White}},
  \bibinfo {author} {\bibfnamefont {Z.~J.}\ \bibnamefont {Yao}}, \bibinfo
  {author} {\bibfnamefont {P.}~\bibnamefont {Yeh}}, \bibinfo {author}
  {\bibfnamefont {A.}~\bibnamefont {Zalcman}}, \bibinfo {author} {\bibfnamefont
  {H.}~\bibnamefont {Neven}},\ and\ \bibinfo {author} {\bibfnamefont {J.~M.}\
  \bibnamefont {Martinis}},\ }\bibfield  {title} {\bibinfo {title} {Quantum
  supremacy using a programmable superconducting processor},\ }\href
  {https://doi.org/10.1038/s41586-019-1666-5} {\bibfield  {journal} {\bibinfo
  {journal} {Nature}\ }\textbf {\bibinfo {volume} {574}},\ \bibinfo {pages}
  {505} (\bibinfo {year} {2019})}\BibitemShut {NoStop}%
\bibitem [{\citenamefont {Moses}\ \emph {et~al.}(2023)\citenamefont {Moses},
  \citenamefont {Baldwin}, \citenamefont {Allman}, \citenamefont {Ancona},
  \citenamefont {Ascarrunz}, \citenamefont {Barnes}, \citenamefont
  {Bartolotta}, \citenamefont {Bjork}, \citenamefont {Blanchard}, \citenamefont
  {Bohn}, \citenamefont {Bohnet}, \citenamefont {Brown}, \citenamefont
  {Burdick}, \citenamefont {Burton}, \citenamefont {Campbell}, \citenamefont
  {Campora}, \citenamefont {Carron}, \citenamefont {Chambers}, \citenamefont
  {Chan}, \citenamefont {Chen}, \citenamefont {Chernoguzov}, \citenamefont
  {Chertkov}, \citenamefont {Colina}, \citenamefont {Curtis}, \citenamefont
  {Daniel}, \citenamefont {DeCross}, \citenamefont {Deen}, \citenamefont
  {Delaney}, \citenamefont {Dreiling}, \citenamefont {Ertsgaard}, \citenamefont
  {Esposito}, \citenamefont {Estey}, \citenamefont {Fabrikant}, \citenamefont
  {Figgatt}, \citenamefont {Foltz}, \citenamefont {Foss-Feig}, \citenamefont
  {Francois}, \citenamefont {Gaebler}, \citenamefont {Gatterman}, \citenamefont
  {Gilbreth}, \citenamefont {Giles}, \citenamefont {Glynn}, \citenamefont
  {Hall}, \citenamefont {Hankin}, \citenamefont {Hansen}, \citenamefont
  {Hayes}, \citenamefont {Higashi}, \citenamefont {Hoffman}, \citenamefont
  {Horning}, \citenamefont {Hout}, \citenamefont {Jacobs}, \citenamefont
  {Johansen}, \citenamefont {Jones}, \citenamefont {Karcz}, \citenamefont
  {Klein}, \citenamefont {Lauria}, \citenamefont {Lee}, \citenamefont {Liefer},
  \citenamefont {Lu}, \citenamefont {Lucchetti}, \citenamefont {Lytle},
  \citenamefont {Malm}, \citenamefont {Matheny}, \citenamefont {Mathewson},
  \citenamefont {Mayer}, \citenamefont {Miller}, \citenamefont {Mills},
  \citenamefont {Neyenhuis}, \citenamefont {Nugent}, \citenamefont {Olson},
  \citenamefont {Parks}, \citenamefont {Price}, \citenamefont {Price},
  \citenamefont {Pugh}, \citenamefont {Ransford}, \citenamefont {Reed},
  \citenamefont {Roman}, \citenamefont {Rowe}, \citenamefont {Ryan-Anderson},
  \citenamefont {Sanders}, \citenamefont {Sedlacek}, \citenamefont {Shevchuk},
  \citenamefont {Siegfried}, \citenamefont {Skripka}, \citenamefont {Spaun},
  \citenamefont {Sprenkle}, \citenamefont {Stutz}, \citenamefont {Swallows},
  \citenamefont {Tobey}, \citenamefont {Tran}, \citenamefont {Tran},
  \citenamefont {Vogt}, \citenamefont {Volin}, \citenamefont {Walker},
  \citenamefont {Zolot},\ and\ \citenamefont {Pino}}]{Moses2023}%
  \BibitemOpen
  \bibfield  {author} {\bibinfo {author} {\bibfnamefont {S.~A.}\ \bibnamefont
  {Moses}}, \bibinfo {author} {\bibfnamefont {C.~H.}\ \bibnamefont {Baldwin}},
  \bibinfo {author} {\bibfnamefont {M.~S.}\ \bibnamefont {Allman}}, \bibinfo
  {author} {\bibfnamefont {R.}~\bibnamefont {Ancona}}, \bibinfo {author}
  {\bibfnamefont {L.}~\bibnamefont {Ascarrunz}}, \bibinfo {author}
  {\bibfnamefont {C.}~\bibnamefont {Barnes}}, \bibinfo {author} {\bibfnamefont
  {J.}~\bibnamefont {Bartolotta}}, \bibinfo {author} {\bibfnamefont
  {B.}~\bibnamefont {Bjork}}, \bibinfo {author} {\bibfnamefont
  {P.}~\bibnamefont {Blanchard}}, \bibinfo {author} {\bibfnamefont
  {M.}~\bibnamefont {Bohn}}, \bibinfo {author} {\bibfnamefont {J.~G.}\
  \bibnamefont {Bohnet}}, \bibinfo {author} {\bibfnamefont {N.~C.}\
  \bibnamefont {Brown}}, \bibinfo {author} {\bibfnamefont {N.~Q.}\ \bibnamefont
  {Burdick}}, \bibinfo {author} {\bibfnamefont {W.~C.}\ \bibnamefont {Burton}},
  \bibinfo {author} {\bibfnamefont {S.~L.}\ \bibnamefont {Campbell}}, \bibinfo
  {author} {\bibfnamefont {J.~P.}\ \bibnamefont {Campora}}, \bibinfo {author}
  {\bibfnamefont {C.}~\bibnamefont {Carron}}, \bibinfo {author} {\bibfnamefont
  {J.}~\bibnamefont {Chambers}}, \bibinfo {author} {\bibfnamefont {J.~W.}\
  \bibnamefont {Chan}}, \bibinfo {author} {\bibfnamefont {Y.~H.}\ \bibnamefont
  {Chen}}, \bibinfo {author} {\bibfnamefont {A.}~\bibnamefont {Chernoguzov}},
  \bibinfo {author} {\bibfnamefont {E.}~\bibnamefont {Chertkov}}, \bibinfo
  {author} {\bibfnamefont {J.}~\bibnamefont {Colina}}, \bibinfo {author}
  {\bibfnamefont {J.~P.}\ \bibnamefont {Curtis}}, \bibinfo {author}
  {\bibfnamefont {R.}~\bibnamefont {Daniel}}, \bibinfo {author} {\bibfnamefont
  {M.}~\bibnamefont {DeCross}}, \bibinfo {author} {\bibfnamefont
  {D.}~\bibnamefont {Deen}}, \bibinfo {author} {\bibfnamefont {C.}~\bibnamefont
  {Delaney}}, \bibinfo {author} {\bibfnamefont {J.~M.}\ \bibnamefont
  {Dreiling}}, \bibinfo {author} {\bibfnamefont {C.~T.}\ \bibnamefont
  {Ertsgaard}}, \bibinfo {author} {\bibfnamefont {J.}~\bibnamefont {Esposito}},
  \bibinfo {author} {\bibfnamefont {B.}~\bibnamefont {Estey}}, \bibinfo
  {author} {\bibfnamefont {M.}~\bibnamefont {Fabrikant}}, \bibinfo {author}
  {\bibfnamefont {C.}~\bibnamefont {Figgatt}}, \bibinfo {author} {\bibfnamefont
  {C.}~\bibnamefont {Foltz}}, \bibinfo {author} {\bibfnamefont
  {M.}~\bibnamefont {Foss-Feig}}, \bibinfo {author} {\bibfnamefont
  {D.}~\bibnamefont {Francois}}, \bibinfo {author} {\bibfnamefont {J.~P.}\
  \bibnamefont {Gaebler}}, \bibinfo {author} {\bibfnamefont {T.~M.}\
  \bibnamefont {Gatterman}}, \bibinfo {author} {\bibfnamefont {C.~N.}\
  \bibnamefont {Gilbreth}}, \bibinfo {author} {\bibfnamefont {J.}~\bibnamefont
  {Giles}}, \bibinfo {author} {\bibfnamefont {E.}~\bibnamefont {Glynn}},
  \bibinfo {author} {\bibfnamefont {A.}~\bibnamefont {Hall}}, \bibinfo {author}
  {\bibfnamefont {A.~M.}\ \bibnamefont {Hankin}}, \bibinfo {author}
  {\bibfnamefont {A.}~\bibnamefont {Hansen}}, \bibinfo {author} {\bibfnamefont
  {D.}~\bibnamefont {Hayes}}, \bibinfo {author} {\bibfnamefont
  {B.}~\bibnamefont {Higashi}}, \bibinfo {author} {\bibfnamefont {I.~M.}\
  \bibnamefont {Hoffman}}, \bibinfo {author} {\bibfnamefont {B.}~\bibnamefont
  {Horning}}, \bibinfo {author} {\bibfnamefont {J.~J.}\ \bibnamefont {Hout}},
  \bibinfo {author} {\bibfnamefont {R.}~\bibnamefont {Jacobs}}, \bibinfo
  {author} {\bibfnamefont {J.}~\bibnamefont {Johansen}}, \bibinfo {author}
  {\bibfnamefont {L.}~\bibnamefont {Jones}}, \bibinfo {author} {\bibfnamefont
  {J.}~\bibnamefont {Karcz}}, \bibinfo {author} {\bibfnamefont
  {T.}~\bibnamefont {Klein}}, \bibinfo {author} {\bibfnamefont
  {P.}~\bibnamefont {Lauria}}, \bibinfo {author} {\bibfnamefont
  {P.}~\bibnamefont {Lee}}, \bibinfo {author} {\bibfnamefont {D.}~\bibnamefont
  {Liefer}}, \bibinfo {author} {\bibfnamefont {S.~T.}\ \bibnamefont {Lu}},
  \bibinfo {author} {\bibfnamefont {D.}~\bibnamefont {Lucchetti}}, \bibinfo
  {author} {\bibfnamefont {C.}~\bibnamefont {Lytle}}, \bibinfo {author}
  {\bibfnamefont {A.}~\bibnamefont {Malm}}, \bibinfo {author} {\bibfnamefont
  {M.}~\bibnamefont {Matheny}}, \bibinfo {author} {\bibfnamefont
  {B.}~\bibnamefont {Mathewson}}, \bibinfo {author} {\bibfnamefont
  {K.}~\bibnamefont {Mayer}}, \bibinfo {author} {\bibfnamefont {D.~B.}\
  \bibnamefont {Miller}}, \bibinfo {author} {\bibfnamefont {M.}~\bibnamefont
  {Mills}}, \bibinfo {author} {\bibfnamefont {B.}~\bibnamefont {Neyenhuis}},
  \bibinfo {author} {\bibfnamefont {L.}~\bibnamefont {Nugent}}, \bibinfo
  {author} {\bibfnamefont {S.}~\bibnamefont {Olson}}, \bibinfo {author}
  {\bibfnamefont {J.}~\bibnamefont {Parks}}, \bibinfo {author} {\bibfnamefont
  {G.~N.}\ \bibnamefont {Price}}, \bibinfo {author} {\bibfnamefont
  {Z.}~\bibnamefont {Price}}, \bibinfo {author} {\bibfnamefont
  {M.}~\bibnamefont {Pugh}}, \bibinfo {author} {\bibfnamefont {A.}~\bibnamefont
  {Ransford}}, \bibinfo {author} {\bibfnamefont {A.~P.}\ \bibnamefont {Reed}},
  \bibinfo {author} {\bibfnamefont {C.}~\bibnamefont {Roman}}, \bibinfo
  {author} {\bibfnamefont {M.}~\bibnamefont {Rowe}}, \bibinfo {author}
  {\bibfnamefont {C.}~\bibnamefont {Ryan-Anderson}}, \bibinfo {author}
  {\bibfnamefont {S.}~\bibnamefont {Sanders}}, \bibinfo {author} {\bibfnamefont
  {J.}~\bibnamefont {Sedlacek}}, \bibinfo {author} {\bibfnamefont
  {P.}~\bibnamefont {Shevchuk}}, \bibinfo {author} {\bibfnamefont
  {P.}~\bibnamefont {Siegfried}}, \bibinfo {author} {\bibfnamefont
  {T.}~\bibnamefont {Skripka}}, \bibinfo {author} {\bibfnamefont
  {B.}~\bibnamefont {Spaun}}, \bibinfo {author} {\bibfnamefont {R.~T.}\
  \bibnamefont {Sprenkle}}, \bibinfo {author} {\bibfnamefont {R.~P.}\
  \bibnamefont {Stutz}}, \bibinfo {author} {\bibfnamefont {M.}~\bibnamefont
  {Swallows}}, \bibinfo {author} {\bibfnamefont {R.~I.}\ \bibnamefont {Tobey}},
  \bibinfo {author} {\bibfnamefont {A.}~\bibnamefont {Tran}}, \bibinfo {author}
  {\bibfnamefont {T.}~\bibnamefont {Tran}}, \bibinfo {author} {\bibfnamefont
  {E.}~\bibnamefont {Vogt}}, \bibinfo {author} {\bibfnamefont {C.}~\bibnamefont
  {Volin}}, \bibinfo {author} {\bibfnamefont {J.}~\bibnamefont {Walker}},
  \bibinfo {author} {\bibfnamefont {A.~M.}\ \bibnamefont {Zolot}},\ and\
  \bibinfo {author} {\bibfnamefont {J.~M.}\ \bibnamefont {Pino}},\ }\bibfield
  {title} {\bibinfo {title} {A race-track trapped-ion quantum processor},\
  }\href {https://doi.org/10.1103/PhysRevX.13.041052} {\bibfield  {journal}
  {\bibinfo  {journal} {Phys. Rev. X}\ }\textbf {\bibinfo {volume} {13}},\
  \bibinfo {pages} {041052} (\bibinfo {year} {2023})}\BibitemShut {NoStop}%
\bibitem [{\citenamefont {Kim}\ \emph {et~al.}(2023)\citenamefont {Kim},
  \citenamefont {Eddins}, \citenamefont {Anand}, \citenamefont {Wei},
  \citenamefont {van~den Berg}, \citenamefont {Rosenblatt}, \citenamefont
  {Nayfeh}, \citenamefont {Wu}, \citenamefont {Zaletel}, \citenamefont
  {Temme},\ and\ \citenamefont {Kandala}}]{Kim2023}%
  \BibitemOpen
  \bibfield  {author} {\bibinfo {author} {\bibfnamefont {Y.}~\bibnamefont
  {Kim}}, \bibinfo {author} {\bibfnamefont {A.}~\bibnamefont {Eddins}},
  \bibinfo {author} {\bibfnamefont {S.}~\bibnamefont {Anand}}, \bibinfo
  {author} {\bibfnamefont {K.~X.}\ \bibnamefont {Wei}}, \bibinfo {author}
  {\bibfnamefont {E.}~\bibnamefont {van~den Berg}}, \bibinfo {author}
  {\bibfnamefont {S.}~\bibnamefont {Rosenblatt}}, \bibinfo {author}
  {\bibfnamefont {H.}~\bibnamefont {Nayfeh}}, \bibinfo {author} {\bibfnamefont
  {Y.}~\bibnamefont {Wu}}, \bibinfo {author} {\bibfnamefont {M.}~\bibnamefont
  {Zaletel}}, \bibinfo {author} {\bibfnamefont {K.}~\bibnamefont {Temme}},\
  and\ \bibinfo {author} {\bibfnamefont {A.}~\bibnamefont {Kandala}},\
  }\bibfield  {title} {\bibinfo {title} {Evidence for the utility of quantum
  computing before fault tolerance},\ }\href
  {https://doi.org/10.1038/s41586-023-06096-3} {\bibfield  {journal} {\bibinfo
  {journal} {Nature}\ }\textbf {\bibinfo {volume} {618}},\ \bibinfo {pages}
  {500} (\bibinfo {year} {2023})}\BibitemShut {NoStop}%
\bibitem [{\citenamefont {Bluvstein}\ \emph {et~al.}(2024)\citenamefont
  {Bluvstein}, \citenamefont {Evered}, \citenamefont {Geim}, \citenamefont
  {Li}, \citenamefont {Zhou}, \citenamefont {Manovitz}, \citenamefont {Ebadi},
  \citenamefont {Cain}, \citenamefont {Kalinowski}, \citenamefont {Hangleiter},
  \citenamefont {Bonilla~Ataides}, \citenamefont {Maskara}, \citenamefont
  {Cong}, \citenamefont {Gao}, \citenamefont {Sales~Rodriguez}, \citenamefont
  {Karolyshyn}, \citenamefont {Semeghini}, \citenamefont {Gullans},
  \citenamefont {Greiner}, \citenamefont {Vuletić},\ and\ \citenamefont
  {Lukin}}]{bluvstein_logical_2024}%
  \BibitemOpen
  \bibfield  {author} {\bibinfo {author} {\bibfnamefont {D.}~\bibnamefont
  {Bluvstein}}, \bibinfo {author} {\bibfnamefont {S.~J.}\ \bibnamefont
  {Evered}}, \bibinfo {author} {\bibfnamefont {A.~A.}\ \bibnamefont {Geim}},
  \bibinfo {author} {\bibfnamefont {S.~H.}\ \bibnamefont {Li}}, \bibinfo
  {author} {\bibfnamefont {H.}~\bibnamefont {Zhou}}, \bibinfo {author}
  {\bibfnamefont {T.}~\bibnamefont {Manovitz}}, \bibinfo {author}
  {\bibfnamefont {S.}~\bibnamefont {Ebadi}}, \bibinfo {author} {\bibfnamefont
  {M.}~\bibnamefont {Cain}}, \bibinfo {author} {\bibfnamefont {M.}~\bibnamefont
  {Kalinowski}}, \bibinfo {author} {\bibfnamefont {D.}~\bibnamefont
  {Hangleiter}}, \bibinfo {author} {\bibfnamefont {J.~P.}\ \bibnamefont
  {Bonilla~Ataides}}, \bibinfo {author} {\bibfnamefont {N.}~\bibnamefont
  {Maskara}}, \bibinfo {author} {\bibfnamefont {I.}~\bibnamefont {Cong}},
  \bibinfo {author} {\bibfnamefont {X.}~\bibnamefont {Gao}}, \bibinfo {author}
  {\bibfnamefont {P.}~\bibnamefont {Sales~Rodriguez}}, \bibinfo {author}
  {\bibfnamefont {T.}~\bibnamefont {Karolyshyn}}, \bibinfo {author}
  {\bibfnamefont {G.}~\bibnamefont {Semeghini}}, \bibinfo {author}
  {\bibfnamefont {M.~J.}\ \bibnamefont {Gullans}}, \bibinfo {author}
  {\bibfnamefont {M.}~\bibnamefont {Greiner}}, \bibinfo {author} {\bibfnamefont
  {V.}~\bibnamefont {Vuletić}},\ and\ \bibinfo {author} {\bibfnamefont
  {M.~D.}\ \bibnamefont {Lukin}},\ }\bibfield  {title} {\bibinfo {title}
  {Logical quantum processor based on reconfigurable atom arrays},\ }\href
  {https://doi.org/10.1038/s41586-023-06927-3} {\bibfield  {journal} {\bibinfo
  {journal} {Nature}\ }\textbf {\bibinfo {volume} {626}},\ \bibinfo {pages}
  {58} (\bibinfo {year} {2024})}\BibitemShut {NoStop}%
\bibitem [{\citenamefont {Acharya}\ \emph {et~al.}(2024)\citenamefont
  {Acharya}, \citenamefont {Abanin}, \citenamefont {Aghababaie-Beni},
  \citenamefont {Aleiner}, \citenamefont {Andersen}, \citenamefont {Ansmann},
  \citenamefont {Arute}, \citenamefont {Arya}, \citenamefont {Asfaw},
  \citenamefont {Astrakhantsev}, \citenamefont {Atalaya}, \citenamefont
  {Babbush}, \citenamefont {Bacon}, \citenamefont {Ballard}, \citenamefont
  {Bardin}, \citenamefont {Bausch}, \citenamefont {Bengtsson}, \citenamefont
  {Bilmes}, \citenamefont {Blackwell}, \citenamefont {Boixo}, \citenamefont
  {Bortoli}, \citenamefont {Bourassa}, \citenamefont {Bovaird}, \citenamefont
  {Brill}, \citenamefont {Broughton}, \citenamefont {Browne}, \citenamefont
  {Buchea}, \citenamefont {Buckley}, \citenamefont {Buell}, \citenamefont
  {Burger}, \citenamefont {Burkett}, \citenamefont {Bushnell}, \citenamefont
  {Cabrera}, \citenamefont {Campero}, \citenamefont {Chang}, \citenamefont
  {Chen}, \citenamefont {Chen}, \citenamefont {Chiaro}, \citenamefont {Chik},
  \citenamefont {Chou}, \citenamefont {Claes}, \citenamefont {Cleland},
  \citenamefont {Cogan}, \citenamefont {Collins}, \citenamefont {Conner},
  \citenamefont {Courtney}, \citenamefont {Crook}, \citenamefont {Curtin},
  \citenamefont {Das}, \citenamefont {Davies}, \citenamefont {De~Lorenzo},
  \citenamefont {Debroy}, \citenamefont {Demura}, \citenamefont {Devoret},
  \citenamefont {Di~Paolo}, \citenamefont {Donohoe}, \citenamefont {Drozdov},
  \citenamefont {Dunsworth}, \citenamefont {Earle}, \citenamefont {Edlich},
  \citenamefont {Eickbusch}, \citenamefont {Elbag}, \citenamefont {Elzouka},
  \citenamefont {Erickson}, \citenamefont {Faoro}, \citenamefont {Farhi},
  \citenamefont {Ferreira}, \citenamefont {Burgos}, \citenamefont {Forati},
  \citenamefont {Fowler}, \citenamefont {Foxen}, \citenamefont {Ganjam},
  \citenamefont {Garcia}, \citenamefont {Gasca}, \citenamefont {Genois},
  \citenamefont {Giang}, \citenamefont {Gidney}, \citenamefont {Gilboa},
  \citenamefont {Gosula}, \citenamefont {Dau}, \citenamefont {Graumann},
  \citenamefont {Greene}, \citenamefont {Gross}, \citenamefont {Habegger},
  \citenamefont {Hall}, \citenamefont {Hamilton}, \citenamefont {Hansen},
  \citenamefont {Harrigan}, \citenamefont {Harrington}, \citenamefont {Heras},
  \citenamefont {Heslin}, \citenamefont {Heu}, \citenamefont {Higgott},
  \citenamefont {Hill}, \citenamefont {Hilton}, \citenamefont {Holland},
  \citenamefont {Hong}, \citenamefont {Huang}, \citenamefont {Huff},
  \citenamefont {Huggins}, \citenamefont {Ioffe}, \citenamefont {Isakov},
  \citenamefont {Iveland}, \citenamefont {Jeffrey}, \citenamefont {Jiang},
  \citenamefont {Jones}, \citenamefont {Jordan}, \citenamefont {Joshi},
  \citenamefont {Juhas}, \citenamefont {Kafri}, \citenamefont {Kang},
  \citenamefont {Karamlou}, \citenamefont {Kechedzhi}, \citenamefont {Kelly},
  \citenamefont {Khaire}, \citenamefont {Khattar}, \citenamefont {Khezri},
  \citenamefont {Kim}, \citenamefont {Klimov}, \citenamefont {Klots},
  \citenamefont {Kobrin}, \citenamefont {Kohli}, \citenamefont {Korotkov},
  \citenamefont {Kostritsa}, \citenamefont {Kothari}, \citenamefont
  {Kozlovskii}, \citenamefont {Kreikebaum}, \citenamefont {Kurilovich},
  \citenamefont {Lacroix}, \citenamefont {Landhuis}, \citenamefont {Lange-Dei},
  \citenamefont {Langley}, \citenamefont {Laptev}, \citenamefont {Lau},
  \citenamefont {Le~Guevel}, \citenamefont {Ledford}, \citenamefont {Lee},
  \citenamefont {Lee}, \citenamefont {Lensky}, \citenamefont {Leon},
  \citenamefont {Lester}, \citenamefont {Li}, \citenamefont {Li}, \citenamefont
  {Lill}, \citenamefont {Liu}, \citenamefont {Livingston}, \citenamefont
  {Locharla}, \citenamefont {Lucero}, \citenamefont {Lundahl}, \citenamefont
  {Lunt}, \citenamefont {Madhuk}, \citenamefont {Malone}, \citenamefont
  {Maloney}, \citenamefont {Mandr{\`a}}, \citenamefont {Manyika}, \citenamefont
  {Martin}, \citenamefont {Martin}, \citenamefont {Martin}, \citenamefont
  {Maxfield}, \citenamefont {McClean}, \citenamefont {McEwen}, \citenamefont
  {Meeks}, \citenamefont {Megrant}, \citenamefont {Mi}, \citenamefont {Miao},
  \citenamefont {Mieszala}, \citenamefont {Molavi}, \citenamefont {Molina},
  \citenamefont {Montazeri}, \citenamefont {Morvan}, \citenamefont {Movassagh},
  \citenamefont {Mruczkiewicz}, \citenamefont {Naaman}, \citenamefont {Neeley},
  \citenamefont {Neill}, \citenamefont {Nersisyan}, \citenamefont {Neven},
  \citenamefont {Newman}, \citenamefont {Ng}, \citenamefont {Nguyen},
  \citenamefont {Nguyen}, \citenamefont {Ni}, \citenamefont {Niu},
  \citenamefont {O'Brien}, \citenamefont {Oliver}, \citenamefont {Opremcak},
  \citenamefont {Ottosson}, \citenamefont {Petukhov}, \citenamefont {Pizzuto},
  \citenamefont {Platt}, \citenamefont {Potter}, \citenamefont {Pritchard},
  \citenamefont {Pryadko}, \citenamefont {Quintana}, \citenamefont
  {Ramachandran}, \citenamefont {Reagor}, \citenamefont {Redding},
  \citenamefont {Rhodes}, \citenamefont {Roberts}, \citenamefont {Rosenberg},
  \citenamefont {Rosenfeld}, \citenamefont {Roushan}, \citenamefont {Rubin},
  \citenamefont {Saei}, \citenamefont {Sank}, \citenamefont {Sankaragomathi},
  \citenamefont {Satzinger}, \citenamefont {Schurkus}, \citenamefont
  {Schuster}, \citenamefont {Senior}, \citenamefont {Shearn}, \citenamefont
  {Shorter}, \citenamefont {Shutty}, \citenamefont {Shvarts}, \citenamefont
  {Singh}, \citenamefont {Sivak}, \citenamefont {Skruzny}, \citenamefont
  {Small}, \citenamefont {Smelyanskiy}, \citenamefont {Smith}, \citenamefont
  {Somma}, \citenamefont {Springer}, \citenamefont {Sterling}, \citenamefont
  {Strain}, \citenamefont {Suchard}, \citenamefont {Szasz}, \citenamefont
  {Sztein}, \citenamefont {Thor}, \citenamefont {Torres}, \citenamefont
  {Torunbalci}, \citenamefont {Vaishnav}, \citenamefont {Vargas}, \citenamefont
  {Vdovichev}, \citenamefont {Vidal}, \citenamefont {Villalonga}, \citenamefont
  {Heidweiller}, \citenamefont {Waltman}, \citenamefont {Wang}, \citenamefont
  {Ware}, \citenamefont {Weber}, \citenamefont {Weidel}, \citenamefont {White},
  \citenamefont {Wong}, \citenamefont {Woo}, \citenamefont {Xing},
  \citenamefont {Yao}, \citenamefont {Yeh}, \citenamefont {Ying}, \citenamefont
  {Yoo}, \citenamefont {Yosri}, \citenamefont {Young}, \citenamefont {Zalcman},
  \citenamefont {Zhang}, \citenamefont {Zhu}, \citenamefont {Zobrist},
  \citenamefont {AI},\ and\ \citenamefont {{Collaborators}}}]{Acharya2024}%
  \BibitemOpen
  \bibfield  {author} {\bibinfo {author} {\bibfnamefont {R.}~\bibnamefont
  {Acharya}}, \bibinfo {author} {\bibfnamefont {D.~A.}\ \bibnamefont {Abanin}},
  \bibinfo {author} {\bibfnamefont {L.}~\bibnamefont {Aghababaie-Beni}},
  \bibinfo {author} {\bibfnamefont {I.}~\bibnamefont {Aleiner}}, \bibinfo
  {author} {\bibfnamefont {T.~I.}\ \bibnamefont {Andersen}}, \bibinfo {author}
  {\bibfnamefont {M.}~\bibnamefont {Ansmann}}, \bibinfo {author} {\bibfnamefont
  {F.}~\bibnamefont {Arute}}, \bibinfo {author} {\bibfnamefont
  {K.}~\bibnamefont {Arya}}, \bibinfo {author} {\bibfnamefont {A.}~\bibnamefont
  {Asfaw}}, \bibinfo {author} {\bibfnamefont {N.}~\bibnamefont
  {Astrakhantsev}}, \bibinfo {author} {\bibfnamefont {J.}~\bibnamefont
  {Atalaya}}, \bibinfo {author} {\bibfnamefont {R.}~\bibnamefont {Babbush}},
  \bibinfo {author} {\bibfnamefont {D.}~\bibnamefont {Bacon}}, \bibinfo
  {author} {\bibfnamefont {B.}~\bibnamefont {Ballard}}, \bibinfo {author}
  {\bibfnamefont {J.~C.}\ \bibnamefont {Bardin}}, \bibinfo {author}
  {\bibfnamefont {J.}~\bibnamefont {Bausch}}, \bibinfo {author} {\bibfnamefont
  {A.}~\bibnamefont {Bengtsson}}, \bibinfo {author} {\bibfnamefont
  {A.}~\bibnamefont {Bilmes}}, \bibinfo {author} {\bibfnamefont
  {S.}~\bibnamefont {Blackwell}}, \bibinfo {author} {\bibfnamefont
  {S.}~\bibnamefont {Boixo}}, \bibinfo {author} {\bibfnamefont
  {G.}~\bibnamefont {Bortoli}}, \bibinfo {author} {\bibfnamefont
  {A.}~\bibnamefont {Bourassa}}, \bibinfo {author} {\bibfnamefont
  {J.}~\bibnamefont {Bovaird}}, \bibinfo {author} {\bibfnamefont
  {L.}~\bibnamefont {Brill}}, \bibinfo {author} {\bibfnamefont
  {M.}~\bibnamefont {Broughton}}, \bibinfo {author} {\bibfnamefont {D.~A.}\
  \bibnamefont {Browne}}, \bibinfo {author} {\bibfnamefont {B.}~\bibnamefont
  {Buchea}}, \bibinfo {author} {\bibfnamefont {B.~B.}\ \bibnamefont {Buckley}},
  \bibinfo {author} {\bibfnamefont {D.~A.}\ \bibnamefont {Buell}}, \bibinfo
  {author} {\bibfnamefont {T.}~\bibnamefont {Burger}}, \bibinfo {author}
  {\bibfnamefont {B.}~\bibnamefont {Burkett}}, \bibinfo {author} {\bibfnamefont
  {N.}~\bibnamefont {Bushnell}}, \bibinfo {author} {\bibfnamefont
  {A.}~\bibnamefont {Cabrera}}, \bibinfo {author} {\bibfnamefont
  {J.}~\bibnamefont {Campero}}, \bibinfo {author} {\bibfnamefont {H.-S.}\
  \bibnamefont {Chang}}, \bibinfo {author} {\bibfnamefont {Y.}~\bibnamefont
  {Chen}}, \bibinfo {author} {\bibfnamefont {Z.}~\bibnamefont {Chen}}, \bibinfo
  {author} {\bibfnamefont {B.}~\bibnamefont {Chiaro}}, \bibinfo {author}
  {\bibfnamefont {D.}~\bibnamefont {Chik}}, \bibinfo {author} {\bibfnamefont
  {C.}~\bibnamefont {Chou}}, \bibinfo {author} {\bibfnamefont {J.}~\bibnamefont
  {Claes}}, \bibinfo {author} {\bibfnamefont {A.~Y.}\ \bibnamefont {Cleland}},
  \bibinfo {author} {\bibfnamefont {J.}~\bibnamefont {Cogan}}, \bibinfo
  {author} {\bibfnamefont {R.}~\bibnamefont {Collins}}, \bibinfo {author}
  {\bibfnamefont {P.}~\bibnamefont {Conner}}, \bibinfo {author} {\bibfnamefont
  {W.}~\bibnamefont {Courtney}}, \bibinfo {author} {\bibfnamefont {A.~L.}\
  \bibnamefont {Crook}}, \bibinfo {author} {\bibfnamefont {B.}~\bibnamefont
  {Curtin}}, \bibinfo {author} {\bibfnamefont {S.}~\bibnamefont {Das}},
  \bibinfo {author} {\bibfnamefont {A.}~\bibnamefont {Davies}}, \bibinfo
  {author} {\bibfnamefont {L.}~\bibnamefont {De~Lorenzo}}, \bibinfo {author}
  {\bibfnamefont {D.~M.}\ \bibnamefont {Debroy}}, \bibinfo {author}
  {\bibfnamefont {S.}~\bibnamefont {Demura}}, \bibinfo {author} {\bibfnamefont
  {M.}~\bibnamefont {Devoret}}, \bibinfo {author} {\bibfnamefont
  {A.}~\bibnamefont {Di~Paolo}}, \bibinfo {author} {\bibfnamefont
  {P.}~\bibnamefont {Donohoe}}, \bibinfo {author} {\bibfnamefont
  {I.}~\bibnamefont {Drozdov}}, \bibinfo {author} {\bibfnamefont
  {A.}~\bibnamefont {Dunsworth}}, \bibinfo {author} {\bibfnamefont
  {C.}~\bibnamefont {Earle}}, \bibinfo {author} {\bibfnamefont
  {T.}~\bibnamefont {Edlich}}, \bibinfo {author} {\bibfnamefont
  {A.}~\bibnamefont {Eickbusch}}, \bibinfo {author} {\bibfnamefont {A.~M.}\
  \bibnamefont {Elbag}}, \bibinfo {author} {\bibfnamefont {M.}~\bibnamefont
  {Elzouka}}, \bibinfo {author} {\bibfnamefont {C.}~\bibnamefont {Erickson}},
  \bibinfo {author} {\bibfnamefont {L.}~\bibnamefont {Faoro}}, \bibinfo
  {author} {\bibfnamefont {E.}~\bibnamefont {Farhi}}, \bibinfo {author}
  {\bibfnamefont {V.~S.}\ \bibnamefont {Ferreira}}, \bibinfo {author}
  {\bibfnamefont {L.~F.}\ \bibnamefont {Burgos}}, \bibinfo {author}
  {\bibfnamefont {E.}~\bibnamefont {Forati}}, \bibinfo {author} {\bibfnamefont
  {A.~G.}\ \bibnamefont {Fowler}}, \bibinfo {author} {\bibfnamefont
  {B.}~\bibnamefont {Foxen}}, \bibinfo {author} {\bibfnamefont
  {S.}~\bibnamefont {Ganjam}}, \bibinfo {author} {\bibfnamefont
  {G.}~\bibnamefont {Garcia}}, \bibinfo {author} {\bibfnamefont
  {R.}~\bibnamefont {Gasca}}, \bibinfo {author} {\bibfnamefont
  {{\'E}.}~\bibnamefont {Genois}}, \bibinfo {author} {\bibfnamefont
  {W.}~\bibnamefont {Giang}}, \bibinfo {author} {\bibfnamefont
  {C.}~\bibnamefont {Gidney}}, \bibinfo {author} {\bibfnamefont
  {D.}~\bibnamefont {Gilboa}}, \bibinfo {author} {\bibfnamefont
  {R.}~\bibnamefont {Gosula}}, \bibinfo {author} {\bibfnamefont {A.~G.}\
  \bibnamefont {Dau}}, \bibinfo {author} {\bibfnamefont {D.}~\bibnamefont
  {Graumann}}, \bibinfo {author} {\bibfnamefont {A.}~\bibnamefont {Greene}},
  \bibinfo {author} {\bibfnamefont {J.~A.}\ \bibnamefont {Gross}}, \bibinfo
  {author} {\bibfnamefont {S.}~\bibnamefont {Habegger}}, \bibinfo {author}
  {\bibfnamefont {J.}~\bibnamefont {Hall}}, \bibinfo {author} {\bibfnamefont
  {M.~C.}\ \bibnamefont {Hamilton}}, \bibinfo {author} {\bibfnamefont
  {M.}~\bibnamefont {Hansen}}, \bibinfo {author} {\bibfnamefont {M.~P.}\
  \bibnamefont {Harrigan}}, \bibinfo {author} {\bibfnamefont {S.~D.}\
  \bibnamefont {Harrington}}, \bibinfo {author} {\bibfnamefont {F.~J.~H.}\
  \bibnamefont {Heras}}, \bibinfo {author} {\bibfnamefont {S.}~\bibnamefont
  {Heslin}}, \bibinfo {author} {\bibfnamefont {P.}~\bibnamefont {Heu}},
  \bibinfo {author} {\bibfnamefont {O.}~\bibnamefont {Higgott}}, \bibinfo
  {author} {\bibfnamefont {G.}~\bibnamefont {Hill}}, \bibinfo {author}
  {\bibfnamefont {J.}~\bibnamefont {Hilton}}, \bibinfo {author} {\bibfnamefont
  {G.}~\bibnamefont {Holland}}, \bibinfo {author} {\bibfnamefont
  {S.}~\bibnamefont {Hong}}, \bibinfo {author} {\bibfnamefont {H.-Y.}\
  \bibnamefont {Huang}}, \bibinfo {author} {\bibfnamefont {A.}~\bibnamefont
  {Huff}}, \bibinfo {author} {\bibfnamefont {W.~J.}\ \bibnamefont {Huggins}},
  \bibinfo {author} {\bibfnamefont {L.~B.}\ \bibnamefont {Ioffe}}, \bibinfo
  {author} {\bibfnamefont {S.~V.}\ \bibnamefont {Isakov}}, \bibinfo {author}
  {\bibfnamefont {J.}~\bibnamefont {Iveland}}, \bibinfo {author} {\bibfnamefont
  {E.}~\bibnamefont {Jeffrey}}, \bibinfo {author} {\bibfnamefont
  {Z.}~\bibnamefont {Jiang}}, \bibinfo {author} {\bibfnamefont
  {C.}~\bibnamefont {Jones}}, \bibinfo {author} {\bibfnamefont
  {S.}~\bibnamefont {Jordan}}, \bibinfo {author} {\bibfnamefont
  {C.}~\bibnamefont {Joshi}}, \bibinfo {author} {\bibfnamefont
  {P.}~\bibnamefont {Juhas}}, \bibinfo {author} {\bibfnamefont
  {D.}~\bibnamefont {Kafri}}, \bibinfo {author} {\bibfnamefont
  {H.}~\bibnamefont {Kang}}, \bibinfo {author} {\bibfnamefont {A.~H.}\
  \bibnamefont {Karamlou}}, \bibinfo {author} {\bibfnamefont {K.}~\bibnamefont
  {Kechedzhi}}, \bibinfo {author} {\bibfnamefont {J.}~\bibnamefont {Kelly}},
  \bibinfo {author} {\bibfnamefont {T.}~\bibnamefont {Khaire}}, \bibinfo
  {author} {\bibfnamefont {T.}~\bibnamefont {Khattar}}, \bibinfo {author}
  {\bibfnamefont {M.}~\bibnamefont {Khezri}}, \bibinfo {author} {\bibfnamefont
  {S.}~\bibnamefont {Kim}}, \bibinfo {author} {\bibfnamefont {P.~V.}\
  \bibnamefont {Klimov}}, \bibinfo {author} {\bibfnamefont {A.~R.}\
  \bibnamefont {Klots}}, \bibinfo {author} {\bibfnamefont {B.}~\bibnamefont
  {Kobrin}}, \bibinfo {author} {\bibfnamefont {P.}~\bibnamefont {Kohli}},
  \bibinfo {author} {\bibfnamefont {A.~N.}\ \bibnamefont {Korotkov}}, \bibinfo
  {author} {\bibfnamefont {F.}~\bibnamefont {Kostritsa}}, \bibinfo {author}
  {\bibfnamefont {R.}~\bibnamefont {Kothari}}, \bibinfo {author} {\bibfnamefont
  {B.}~\bibnamefont {Kozlovskii}}, \bibinfo {author} {\bibfnamefont {J.~M.}\
  \bibnamefont {Kreikebaum}}, \bibinfo {author} {\bibfnamefont {V.~D.}\
  \bibnamefont {Kurilovich}}, \bibinfo {author} {\bibfnamefont
  {N.}~\bibnamefont {Lacroix}}, \bibinfo {author} {\bibfnamefont
  {D.}~\bibnamefont {Landhuis}}, \bibinfo {author} {\bibfnamefont
  {T.}~\bibnamefont {Lange-Dei}}, \bibinfo {author} {\bibfnamefont {B.~W.}\
  \bibnamefont {Langley}}, \bibinfo {author} {\bibfnamefont {P.}~\bibnamefont
  {Laptev}}, \bibinfo {author} {\bibfnamefont {K.-M.}\ \bibnamefont {Lau}},
  \bibinfo {author} {\bibfnamefont {L.}~\bibnamefont {Le~Guevel}}, \bibinfo
  {author} {\bibfnamefont {J.}~\bibnamefont {Ledford}}, \bibinfo {author}
  {\bibfnamefont {J.}~\bibnamefont {Lee}}, \bibinfo {author} {\bibfnamefont
  {K.}~\bibnamefont {Lee}}, \bibinfo {author} {\bibfnamefont {Y.~D.}\
  \bibnamefont {Lensky}}, \bibinfo {author} {\bibfnamefont {S.}~\bibnamefont
  {Leon}}, \bibinfo {author} {\bibfnamefont {B.~J.}\ \bibnamefont {Lester}},
  \bibinfo {author} {\bibfnamefont {W.~Y.}\ \bibnamefont {Li}}, \bibinfo
  {author} {\bibfnamefont {Y.}~\bibnamefont {Li}}, \bibinfo {author}
  {\bibfnamefont {A.~T.}\ \bibnamefont {Lill}}, \bibinfo {author}
  {\bibfnamefont {W.}~\bibnamefont {Liu}}, \bibinfo {author} {\bibfnamefont
  {W.~P.}\ \bibnamefont {Livingston}}, \bibinfo {author} {\bibfnamefont
  {A.}~\bibnamefont {Locharla}}, \bibinfo {author} {\bibfnamefont
  {E.}~\bibnamefont {Lucero}}, \bibinfo {author} {\bibfnamefont
  {D.}~\bibnamefont {Lundahl}}, \bibinfo {author} {\bibfnamefont
  {A.}~\bibnamefont {Lunt}}, \bibinfo {author} {\bibfnamefont {S.}~\bibnamefont
  {Madhuk}}, \bibinfo {author} {\bibfnamefont {F.~D.}\ \bibnamefont {Malone}},
  \bibinfo {author} {\bibfnamefont {A.}~\bibnamefont {Maloney}}, \bibinfo
  {author} {\bibfnamefont {S.}~\bibnamefont {Mandr{\`a}}}, \bibinfo {author}
  {\bibfnamefont {J.}~\bibnamefont {Manyika}}, \bibinfo {author} {\bibfnamefont
  {L.~S.}\ \bibnamefont {Martin}}, \bibinfo {author} {\bibfnamefont
  {O.}~\bibnamefont {Martin}}, \bibinfo {author} {\bibfnamefont
  {S.}~\bibnamefont {Martin}}, \bibinfo {author} {\bibfnamefont
  {C.}~\bibnamefont {Maxfield}}, \bibinfo {author} {\bibfnamefont {J.~R.}\
  \bibnamefont {McClean}}, \bibinfo {author} {\bibfnamefont {M.}~\bibnamefont
  {McEwen}}, \bibinfo {author} {\bibfnamefont {S.}~\bibnamefont {Meeks}},
  \bibinfo {author} {\bibfnamefont {A.}~\bibnamefont {Megrant}}, \bibinfo
  {author} {\bibfnamefont {X.}~\bibnamefont {Mi}}, \bibinfo {author}
  {\bibfnamefont {K.~C.}\ \bibnamefont {Miao}}, \bibinfo {author}
  {\bibfnamefont {A.}~\bibnamefont {Mieszala}}, \bibinfo {author}
  {\bibfnamefont {R.}~\bibnamefont {Molavi}}, \bibinfo {author} {\bibfnamefont
  {S.}~\bibnamefont {Molina}}, \bibinfo {author} {\bibfnamefont
  {S.}~\bibnamefont {Montazeri}}, \bibinfo {author} {\bibfnamefont
  {A.}~\bibnamefont {Morvan}}, \bibinfo {author} {\bibfnamefont
  {R.}~\bibnamefont {Movassagh}}, \bibinfo {author} {\bibfnamefont
  {W.}~\bibnamefont {Mruczkiewicz}}, \bibinfo {author} {\bibfnamefont
  {O.}~\bibnamefont {Naaman}}, \bibinfo {author} {\bibfnamefont
  {M.}~\bibnamefont {Neeley}}, \bibinfo {author} {\bibfnamefont
  {C.}~\bibnamefont {Neill}}, \bibinfo {author} {\bibfnamefont
  {A.}~\bibnamefont {Nersisyan}}, \bibinfo {author} {\bibfnamefont
  {H.}~\bibnamefont {Neven}}, \bibinfo {author} {\bibfnamefont
  {M.}~\bibnamefont {Newman}}, \bibinfo {author} {\bibfnamefont {J.~H.}\
  \bibnamefont {Ng}}, \bibinfo {author} {\bibfnamefont {A.}~\bibnamefont
  {Nguyen}}, \bibinfo {author} {\bibfnamefont {M.}~\bibnamefont {Nguyen}},
  \bibinfo {author} {\bibfnamefont {C.-H.}\ \bibnamefont {Ni}}, \bibinfo
  {author} {\bibfnamefont {M.~Y.}\ \bibnamefont {Niu}}, \bibinfo {author}
  {\bibfnamefont {T.~E.}\ \bibnamefont {O'Brien}}, \bibinfo {author}
  {\bibfnamefont {W.~D.}\ \bibnamefont {Oliver}}, \bibinfo {author}
  {\bibfnamefont {A.}~\bibnamefont {Opremcak}}, \bibinfo {author}
  {\bibfnamefont {K.}~\bibnamefont {Ottosson}}, \bibinfo {author}
  {\bibfnamefont {A.}~\bibnamefont {Petukhov}}, \bibinfo {author}
  {\bibfnamefont {A.}~\bibnamefont {Pizzuto}}, \bibinfo {author} {\bibfnamefont
  {J.}~\bibnamefont {Platt}}, \bibinfo {author} {\bibfnamefont
  {R.}~\bibnamefont {Potter}}, \bibinfo {author} {\bibfnamefont
  {O.}~\bibnamefont {Pritchard}}, \bibinfo {author} {\bibfnamefont {L.~P.}\
  \bibnamefont {Pryadko}}, \bibinfo {author} {\bibfnamefont {C.}~\bibnamefont
  {Quintana}}, \bibinfo {author} {\bibfnamefont {G.}~\bibnamefont
  {Ramachandran}}, \bibinfo {author} {\bibfnamefont {M.~J.}\ \bibnamefont
  {Reagor}}, \bibinfo {author} {\bibfnamefont {J.}~\bibnamefont {Redding}},
  \bibinfo {author} {\bibfnamefont {D.~M.}\ \bibnamefont {Rhodes}}, \bibinfo
  {author} {\bibfnamefont {G.}~\bibnamefont {Roberts}}, \bibinfo {author}
  {\bibfnamefont {E.}~\bibnamefont {Rosenberg}}, \bibinfo {author}
  {\bibfnamefont {E.}~\bibnamefont {Rosenfeld}}, \bibinfo {author}
  {\bibfnamefont {P.}~\bibnamefont {Roushan}}, \bibinfo {author} {\bibfnamefont
  {N.~C.}\ \bibnamefont {Rubin}}, \bibinfo {author} {\bibfnamefont
  {N.}~\bibnamefont {Saei}}, \bibinfo {author} {\bibfnamefont {D.}~\bibnamefont
  {Sank}}, \bibinfo {author} {\bibfnamefont {K.}~\bibnamefont
  {Sankaragomathi}}, \bibinfo {author} {\bibfnamefont {K.~J.}\ \bibnamefont
  {Satzinger}}, \bibinfo {author} {\bibfnamefont {H.~F.}\ \bibnamefont
  {Schurkus}}, \bibinfo {author} {\bibfnamefont {C.}~\bibnamefont {Schuster}},
  \bibinfo {author} {\bibfnamefont {A.~W.}\ \bibnamefont {Senior}}, \bibinfo
  {author} {\bibfnamefont {M.~J.}\ \bibnamefont {Shearn}}, \bibinfo {author}
  {\bibfnamefont {A.}~\bibnamefont {Shorter}}, \bibinfo {author} {\bibfnamefont
  {N.}~\bibnamefont {Shutty}}, \bibinfo {author} {\bibfnamefont
  {V.}~\bibnamefont {Shvarts}}, \bibinfo {author} {\bibfnamefont
  {S.}~\bibnamefont {Singh}}, \bibinfo {author} {\bibfnamefont
  {V.}~\bibnamefont {Sivak}}, \bibinfo {author} {\bibfnamefont
  {J.}~\bibnamefont {Skruzny}}, \bibinfo {author} {\bibfnamefont
  {S.}~\bibnamefont {Small}}, \bibinfo {author} {\bibfnamefont
  {V.}~\bibnamefont {Smelyanskiy}}, \bibinfo {author} {\bibfnamefont {W.~C.}\
  \bibnamefont {Smith}}, \bibinfo {author} {\bibfnamefont {R.~D.}\ \bibnamefont
  {Somma}}, \bibinfo {author} {\bibfnamefont {S.}~\bibnamefont {Springer}},
  \bibinfo {author} {\bibfnamefont {G.}~\bibnamefont {Sterling}}, \bibinfo
  {author} {\bibfnamefont {D.}~\bibnamefont {Strain}}, \bibinfo {author}
  {\bibfnamefont {J.}~\bibnamefont {Suchard}}, \bibinfo {author} {\bibfnamefont
  {A.}~\bibnamefont {Szasz}}, \bibinfo {author} {\bibfnamefont
  {A.}~\bibnamefont {Sztein}}, \bibinfo {author} {\bibfnamefont
  {D.}~\bibnamefont {Thor}}, \bibinfo {author} {\bibfnamefont {A.}~\bibnamefont
  {Torres}}, \bibinfo {author} {\bibfnamefont {M.~M.}\ \bibnamefont
  {Torunbalci}}, \bibinfo {author} {\bibfnamefont {A.}~\bibnamefont
  {Vaishnav}}, \bibinfo {author} {\bibfnamefont {J.}~\bibnamefont {Vargas}},
  \bibinfo {author} {\bibfnamefont {S.}~\bibnamefont {Vdovichev}}, \bibinfo
  {author} {\bibfnamefont {G.}~\bibnamefont {Vidal}}, \bibinfo {author}
  {\bibfnamefont {B.}~\bibnamefont {Villalonga}}, \bibinfo {author}
  {\bibfnamefont {C.~V.}\ \bibnamefont {Heidweiller}}, \bibinfo {author}
  {\bibfnamefont {S.}~\bibnamefont {Waltman}}, \bibinfo {author} {\bibfnamefont
  {S.~X.}\ \bibnamefont {Wang}}, \bibinfo {author} {\bibfnamefont
  {B.}~\bibnamefont {Ware}}, \bibinfo {author} {\bibfnamefont {K.}~\bibnamefont
  {Weber}}, \bibinfo {author} {\bibfnamefont {T.}~\bibnamefont {Weidel}},
  \bibinfo {author} {\bibfnamefont {T.}~\bibnamefont {White}}, \bibinfo
  {author} {\bibfnamefont {K.}~\bibnamefont {Wong}}, \bibinfo {author}
  {\bibfnamefont {B.~W.~K.}\ \bibnamefont {Woo}}, \bibinfo {author}
  {\bibfnamefont {C.}~\bibnamefont {Xing}}, \bibinfo {author} {\bibfnamefont
  {Z.~J.}\ \bibnamefont {Yao}}, \bibinfo {author} {\bibfnamefont
  {P.}~\bibnamefont {Yeh}}, \bibinfo {author} {\bibfnamefont {B.}~\bibnamefont
  {Ying}}, \bibinfo {author} {\bibfnamefont {J.}~\bibnamefont {Yoo}}, \bibinfo
  {author} {\bibfnamefont {N.}~\bibnamefont {Yosri}}, \bibinfo {author}
  {\bibfnamefont {G.}~\bibnamefont {Young}}, \bibinfo {author} {\bibfnamefont
  {A.}~\bibnamefont {Zalcman}}, \bibinfo {author} {\bibfnamefont
  {Y.}~\bibnamefont {Zhang}}, \bibinfo {author} {\bibfnamefont
  {N.}~\bibnamefont {Zhu}}, \bibinfo {author} {\bibfnamefont {N.}~\bibnamefont
  {Zobrist}}, \bibinfo {author} {\bibfnamefont {G.~Q.}\ \bibnamefont {AI}},\
  and\ \bibinfo {author} {\bibnamefont {{Collaborators}}},\ }\bibfield  {title}
  {\bibinfo {title} {Quantum error correction below the surface code
  threshold},\ }\href {https://doi.org/10.1038/s41586-024-08449-y} {\bibfield
  {journal} {\bibinfo  {journal} {Nature}\ } (\bibinfo {year}
  {2024})}\BibitemShut {NoStop}%
\bibitem [{\citenamefont {Brown}\ \emph {et~al.}(2010)\citenamefont {Brown},
  \citenamefont {Munro},\ and\ \citenamefont {Kendon}}]{Brown2010}%
  \BibitemOpen
  \bibfield  {author} {\bibinfo {author} {\bibfnamefont {K.~L.}\ \bibnamefont
  {Brown}}, \bibinfo {author} {\bibfnamefont {W.~J.}\ \bibnamefont {Munro}},\
  and\ \bibinfo {author} {\bibfnamefont {V.~M.}\ \bibnamefont {Kendon}},\
  }\bibfield  {title} {\bibinfo {title} {Using quantum computers for quantum
  simulation},\ }\href {https://doi.org/10.3390/e12112268} {\bibfield
  {journal} {\bibinfo  {journal} {Entropy}\ }\textbf {\bibinfo {volume} {12}},\
  \bibinfo {pages} {2268} (\bibinfo {year} {2010})}\BibitemShut {NoStop}%
\bibitem [{\citenamefont {Bernien}\ \emph {et~al.}(2017)\citenamefont
  {Bernien}, \citenamefont {Schwartz}, \citenamefont {Keesling}, \citenamefont
  {Levine}, \citenamefont {Omran}, \citenamefont {Pichler}, \citenamefont
  {Choi}, \citenamefont {Zibrov}, \citenamefont {Endres}, \citenamefont
  {Greiner}, \citenamefont {Vuletić},\ and\ \citenamefont
  {Lukin}}]{Bernien2017}%
  \BibitemOpen
  \bibfield  {author} {\bibinfo {author} {\bibfnamefont {H.}~\bibnamefont
  {Bernien}}, \bibinfo {author} {\bibfnamefont {S.}~\bibnamefont {Schwartz}},
  \bibinfo {author} {\bibfnamefont {A.}~\bibnamefont {Keesling}}, \bibinfo
  {author} {\bibfnamefont {H.}~\bibnamefont {Levine}}, \bibinfo {author}
  {\bibfnamefont {A.}~\bibnamefont {Omran}}, \bibinfo {author} {\bibfnamefont
  {H.}~\bibnamefont {Pichler}}, \bibinfo {author} {\bibfnamefont
  {S.}~\bibnamefont {Choi}}, \bibinfo {author} {\bibfnamefont {A.~S.}\
  \bibnamefont {Zibrov}}, \bibinfo {author} {\bibfnamefont {M.}~\bibnamefont
  {Endres}}, \bibinfo {author} {\bibfnamefont {M.}~\bibnamefont {Greiner}},
  \bibinfo {author} {\bibfnamefont {V.}~\bibnamefont {Vuletić}},\ and\
  \bibinfo {author} {\bibfnamefont {M.~D.}\ \bibnamefont {Lukin}},\ }\bibfield
  {title} {\bibinfo {title} {Probing many-body dynamics on a 51-atom quantum
  simulator},\ }\href {https://doi.org/10.1038/nature24622} {\bibfield
  {journal} {\bibinfo  {journal} {Nature}\ }\textbf {\bibinfo {volume} {551}},\
  \bibinfo {pages} {579} (\bibinfo {year} {2017})}\BibitemShut {NoStop}%
\bibitem [{\citenamefont {Cao}\ \emph {et~al.}(2019)\citenamefont {Cao},
  \citenamefont {Romero}, \citenamefont {Olson}, \citenamefont {Degroote},
  \citenamefont {Johnson}, \citenamefont {Kieferová}, \citenamefont
  {Kivlichan}, \citenamefont {Menke}, \citenamefont {Peropadre}, \citenamefont
  {Sawaya}, \citenamefont {Sim}, \citenamefont {Veis},\ and\ \citenamefont
  {Aspuru-Guzik}}]{Cao2019}%
  \BibitemOpen
  \bibfield  {author} {\bibinfo {author} {\bibfnamefont {Y.}~\bibnamefont
  {Cao}}, \bibinfo {author} {\bibfnamefont {J.}~\bibnamefont {Romero}},
  \bibinfo {author} {\bibfnamefont {J.~P.}\ \bibnamefont {Olson}}, \bibinfo
  {author} {\bibfnamefont {M.}~\bibnamefont {Degroote}}, \bibinfo {author}
  {\bibfnamefont {P.~D.}\ \bibnamefont {Johnson}}, \bibinfo {author}
  {\bibfnamefont {M.}~\bibnamefont {Kieferová}}, \bibinfo {author}
  {\bibfnamefont {I.~D.}\ \bibnamefont {Kivlichan}}, \bibinfo {author}
  {\bibfnamefont {T.}~\bibnamefont {Menke}}, \bibinfo {author} {\bibfnamefont
  {B.}~\bibnamefont {Peropadre}}, \bibinfo {author} {\bibfnamefont {N.~P.~D.}\
  \bibnamefont {Sawaya}}, \bibinfo {author} {\bibfnamefont {S.}~\bibnamefont
  {Sim}}, \bibinfo {author} {\bibfnamefont {L.}~\bibnamefont {Veis}},\ and\
  \bibinfo {author} {\bibfnamefont {A.}~\bibnamefont {Aspuru-Guzik}},\
  }\bibfield  {title} {\bibinfo {title} {Quantum chemistry in the age of
  quantum computing},\ }\href {https://doi.org/10.1021/acs.chemrev.8b00803}
  {\bibfield  {journal} {\bibinfo  {journal} {Chemical Reviews}\ }\textbf
  {\bibinfo {volume} {119}},\ \bibinfo {pages} {10856} (\bibinfo {year}
  {2019})}\BibitemShut {NoStop}%
\bibitem [{\citenamefont {Orús}\ \emph {et~al.}(2019)\citenamefont {Orús},
  \citenamefont {Mugel},\ and\ \citenamefont {Lizaso}}]{Orus2019}%
  \BibitemOpen
  \bibfield  {author} {\bibinfo {author} {\bibfnamefont {R.}~\bibnamefont
  {Orús}}, \bibinfo {author} {\bibfnamefont {S.}~\bibnamefont {Mugel}},\ and\
  \bibinfo {author} {\bibfnamefont {E.}~\bibnamefont {Lizaso}},\ }\bibfield
  {title} {\bibinfo {title} {Quantum computing for finance: Overview and
  prospects},\ }\href
  {https://doi.org/https://doi.org/10.1016/j.revip.2019.100028} {\bibfield
  {journal} {\bibinfo  {journal} {Reviews in Physics}\ }\textbf {\bibinfo
  {volume} {4}},\ \bibinfo {pages} {100028} (\bibinfo {year}
  {2019})}\BibitemShut {NoStop}%
\bibitem [{\citenamefont {Farhi}\ \emph {et~al.}(2000)\citenamefont {Farhi},
  \citenamefont {Goldstone}, \citenamefont {Gutmann},\ and\ \citenamefont
  {Sipser}}]{Farhi2000}%
  \BibitemOpen
  \bibfield  {author} {\bibinfo {author} {\bibfnamefont {E.}~\bibnamefont
  {Farhi}}, \bibinfo {author} {\bibfnamefont {j.}~\bibnamefont {Goldstone}},
  \bibinfo {author} {\bibfnamefont {S.}~\bibnamefont {Gutmann}},\ and\ \bibinfo
  {author} {\bibfnamefont {M.}~\bibnamefont {Sipser}},\ }\bibfield  {title}
  {\bibinfo {title} {Quantum computation by adiabatic evolution},\ }\href
  {https://arxiv.org/pdf/quant-ph/0001106} {\bibfield  {journal} {\bibinfo
  {journal} {arXiv:quant-ph/0001106}\ } (\bibinfo {year} {2000})}\BibitemShut
  {NoStop}%
\bibitem [{\citenamefont {E.}\ \emph {et~al.}(2014)\citenamefont {E.},
  \citenamefont {Goldstone},\ and\ \citenamefont {Gutmann}}]{Farhi2014}%
  \BibitemOpen
  \bibfield  {author} {\bibinfo {author} {\bibfnamefont {F.}~\bibnamefont
  {E.}}, \bibinfo {author} {\bibfnamefont {J.}~\bibnamefont {Goldstone}},\ and\
  \bibinfo {author} {\bibfnamefont {S.}~\bibnamefont {Gutmann}},\ }\bibfield
  {title} {\bibinfo {title} {A quantum approximate optimization algorithm},\
  }\href {https://arxiv.org/abs/1411.4028} {\bibfield  {journal} {\bibinfo
  {journal} {arXiv:1411.4028}\ } (\bibinfo {year} {2014})}\BibitemShut
  {NoStop}%
\bibitem [{\citenamefont {Cerezo}\ \emph {et~al.}(2021)\citenamefont {Cerezo},
  \citenamefont {Arrasmith}, \citenamefont {Babbush}, \citenamefont {Benjamin},
  \citenamefont {Endo}, \citenamefont {Fujii}, \citenamefont {{McClean}},
  \citenamefont {Mitarai}, \citenamefont {Yuan}, \citenamefont {Cincio},\ and\
  \citenamefont {Coles}}]{Cerezo2021}%
  \BibitemOpen
  \bibfield  {author} {\bibinfo {author} {\bibfnamefont {M.}~\bibnamefont
  {Cerezo}}, \bibinfo {author} {\bibfnamefont {A.}~\bibnamefont {Arrasmith}},
  \bibinfo {author} {\bibfnamefont {R.}~\bibnamefont {Babbush}}, \bibinfo
  {author} {\bibfnamefont {S.~C.}\ \bibnamefont {Benjamin}}, \bibinfo {author}
  {\bibfnamefont {S.}~\bibnamefont {Endo}}, \bibinfo {author} {\bibfnamefont
  {K.}~\bibnamefont {Fujii}}, \bibinfo {author} {\bibfnamefont {J.~R.}\
  \bibnamefont {{McClean}}}, \bibinfo {author} {\bibfnamefont {K.}~\bibnamefont
  {Mitarai}}, \bibinfo {author} {\bibfnamefont {X.}~\bibnamefont {Yuan}},
  \bibinfo {author} {\bibfnamefont {L.}~\bibnamefont {Cincio}},\ and\ \bibinfo
  {author} {\bibfnamefont {P.~J.}\ \bibnamefont {Coles}},\ }\bibfield  {title}
  {\bibinfo {title} {Variational quantum algorithms},\ }\href
  {https://doi.org/10.1038/s42254-021-00348-9} {\bibfield  {journal} {\bibinfo
  {journal} {Nature Reviews Physics}\ }\textbf {\bibinfo {volume} {3}},\
  \bibinfo {pages} {625} (\bibinfo {year} {2021})}\BibitemShut {NoStop}%
\bibitem [{\citenamefont {Preskill}(2018)}]{Preskill2018}%
  \BibitemOpen
  \bibfield  {author} {\bibinfo {author} {\bibfnamefont {J.}~\bibnamefont
  {Preskill}},\ }\bibfield  {title} {\bibinfo {title} {Quantum {C}omputing in
  the {NISQ} era and beyond},\ }\href
  {https://doi.org/10.22331/q-2018-08-06-79} {\bibfield  {journal} {\bibinfo
  {journal} {{Quantum}}\ }\textbf {\bibinfo {volume} {2}},\ \bibinfo {pages}
  {79} (\bibinfo {year} {2018})}\BibitemShut {NoStop}%
\bibitem [{\citenamefont {Gyongyosi}\ and\ \citenamefont
  {Imre}(2020)}]{Gyongyosi2020}%
  \BibitemOpen
  \bibfield  {author} {\bibinfo {author} {\bibfnamefont {L.}~\bibnamefont
  {Gyongyosi}}\ and\ \bibinfo {author} {\bibfnamefont {S.}~\bibnamefont
  {Imre}},\ }\bibfield  {title} {\bibinfo {title} {Circuit depth reduction for
  gate-model quantum computers},\ }\href
  {https://doi.org/10.1038/s41598-020-67014-5} {\bibfield  {journal} {\bibinfo
  {journal} {Scientific Reports}\ }\textbf {\bibinfo {volume} {10}},\ \bibinfo
  {pages} {11229} (\bibinfo {year} {2020})}\BibitemShut {NoStop}%
\bibitem [{\citenamefont {Zhang}\ \emph {et~al.}(2021)\citenamefont {Zhang},
  \citenamefont {Hayes}, \citenamefont {Qiu}, \citenamefont {Jin},
  \citenamefont {Chen},\ and\ \citenamefont {Zhang}}]{Zhang2021}%
  \BibitemOpen
  \bibfield  {author} {\bibinfo {author} {\bibfnamefont {C.}~\bibnamefont
  {Zhang}}, \bibinfo {author} {\bibfnamefont {A.~B.}\ \bibnamefont {Hayes}},
  \bibinfo {author} {\bibfnamefont {L.}~\bibnamefont {Qiu}}, \bibinfo {author}
  {\bibfnamefont {Y.}~\bibnamefont {Jin}}, \bibinfo {author} {\bibfnamefont
  {Y.}~\bibnamefont {Chen}},\ and\ \bibinfo {author} {\bibfnamefont {E.~Z.}\
  \bibnamefont {Zhang}},\ }\bibfield  {title} {\bibinfo {title} {Time-optimal
  qubit mapping},\ }in\ \href {https://doi.org/10.1145/3445814.3446706} {\emph
  {\bibinfo {booktitle} {Proceedings of the 26th ACM International Conference
  on Architectural Support for Programming Languages and Operating Systems}}},\
  \bibinfo {series and number} {ASPLOS '21}\ (\bibinfo  {publisher}
  {Association for Computing Machinery},\ \bibinfo {address} {New York, NY,
  USA},\ \bibinfo {year} {2021})\ p.\ \bibinfo {pages} {360–374}\BibitemShut
  {NoStop}%
\bibitem [{\citenamefont {Holmes}\ \emph {et~al.}(2020)\citenamefont {Holmes},
  \citenamefont {Johri}, \citenamefont {Guerreschi}, \citenamefont {Clarke},\
  and\ \citenamefont {Matsuura}}]{Holmes2020}%
  \BibitemOpen
  \bibfield  {author} {\bibinfo {author} {\bibfnamefont {A.}~\bibnamefont
  {Holmes}}, \bibinfo {author} {\bibfnamefont {S.}~\bibnamefont {Johri}},
  \bibinfo {author} {\bibfnamefont {G.~G.}\ \bibnamefont {Guerreschi}},
  \bibinfo {author} {\bibfnamefont {J.~S.}\ \bibnamefont {Clarke}},\ and\
  \bibinfo {author} {\bibfnamefont {A.~Y.}\ \bibnamefont {Matsuura}},\
  }\bibfield  {title} {\bibinfo {title} {Impact of qubit connectivity on
  quantum algorithm performance},\ }\href
  {https://doi.org/10.1088/2058-9565/ab73e0} {\bibfield  {journal} {\bibinfo
  {journal} {Quantum Science and Technology}\ }\textbf {\bibinfo {volume}
  {5}},\ \bibinfo {pages} {025009} (\bibinfo {year} {2020})}\BibitemShut
  {NoStop}%
\bibitem [{\citenamefont {Gidney}(2018)}]{Gidney2018}%
  \BibitemOpen
  \bibfield  {author} {\bibinfo {author} {\bibfnamefont {C.}~\bibnamefont
  {Gidney}},\ }\bibfield  {title} {\bibinfo {title} {Halving the cost of
  quantum addition},\ }\href
  {https://quantum-journal.org/papers/q-2018-06-18-74/} {\bibfield  {journal}
  {\bibinfo  {journal} {Quantum}\ }\textbf {\bibinfo {volume} {2}},\ \bibinfo
  {pages} {74} (\bibinfo {year} {2018})}\BibitemShut {NoStop}%
\bibitem [{\citenamefont {Schindler}\ \emph {et~al.}(2013)\citenamefont
  {Schindler}, \citenamefont {Nigg}, \citenamefont {Monz}, \citenamefont
  {Barreiro}, \citenamefont {Martinez}, \citenamefont {Wang}, \citenamefont
  {Quint}, \citenamefont {Brandl}, \citenamefont {Nebendahl}, \citenamefont
  {Roos}, \citenamefont {Chwalla}, \citenamefont {Hennrich},\ and\
  \citenamefont {Blatt}}]{Schindler2013}%
  \BibitemOpen
  \bibfield  {author} {\bibinfo {author} {\bibfnamefont {P.}~\bibnamefont
  {Schindler}}, \bibinfo {author} {\bibfnamefont {D.}~\bibnamefont {Nigg}},
  \bibinfo {author} {\bibfnamefont {T.}~\bibnamefont {Monz}}, \bibinfo {author}
  {\bibfnamefont {J.~T.}\ \bibnamefont {Barreiro}}, \bibinfo {author}
  {\bibfnamefont {E.}~\bibnamefont {Martinez}}, \bibinfo {author}
  {\bibfnamefont {S.~X.}\ \bibnamefont {Wang}}, \bibinfo {author}
  {\bibfnamefont {S.}~\bibnamefont {Quint}}, \bibinfo {author} {\bibfnamefont
  {M.~F.}\ \bibnamefont {Brandl}}, \bibinfo {author} {\bibfnamefont
  {V.}~\bibnamefont {Nebendahl}}, \bibinfo {author} {\bibfnamefont {C.~F.}\
  \bibnamefont {Roos}}, \bibinfo {author} {\bibfnamefont {M.}~\bibnamefont
  {Chwalla}}, \bibinfo {author} {\bibfnamefont {M.}~\bibnamefont {Hennrich}},\
  and\ \bibinfo {author} {\bibfnamefont {R.}~\bibnamefont {Blatt}},\ }\bibfield
   {title} {\bibinfo {title} {A quantum information processor with trapped
  ions},\ }\href {https://doi.org/10.1088/1367-2630/15/12/123012} {\bibfield
  {journal} {\bibinfo  {journal} {New Journal of Physics}\ }\textbf {\bibinfo
  {volume} {15}},\ \bibinfo {pages} {123012} (\bibinfo {year}
  {2013})}\BibitemShut {NoStop}%
\bibitem [{\citenamefont {Piltz}\ \emph {et~al.}(2016)\citenamefont {Piltz},
  \citenamefont {Sriarunothai}, \citenamefont {Ivanov}, \citenamefont {Wölk},\
  and\ \citenamefont {Wunderlich}}]{Piltz2016}%
  \BibitemOpen
  \bibfield  {author} {\bibinfo {author} {\bibfnamefont {C.}~\bibnamefont
  {Piltz}}, \bibinfo {author} {\bibfnamefont {T.}~\bibnamefont {Sriarunothai}},
  \bibinfo {author} {\bibfnamefont {S.~S.}\ \bibnamefont {Ivanov}}, \bibinfo
  {author} {\bibfnamefont {S.}~\bibnamefont {Wölk}},\ and\ \bibinfo {author}
  {\bibfnamefont {C.}~\bibnamefont {Wunderlich}},\ }\bibfield  {title}
  {\bibinfo {title} {Versatile microwave-driven trapped ion spin system for
  quantum information processing},\ }\href
  {https://doi.org/10.1126/sciadv.1600093} {\bibfield  {journal} {\bibinfo
  {journal} {Science Advances}\ }\textbf {\bibinfo {volume} {2}},\ \bibinfo
  {pages} {e1600093} (\bibinfo {year} {2016})}\BibitemShut {NoStop}%
\bibitem [{\citenamefont {Rempfer}\ and\ \citenamefont
  {Obenland}(2024)}]{Rempfer2024}%
  \BibitemOpen
  \bibfield  {author} {\bibinfo {author} {\bibfnamefont {B.}~\bibnamefont
  {Rempfer}}\ and\ \bibinfo {author} {\bibfnamefont {K.}~\bibnamefont
  {Obenland}},\ }\href@noop {} {\bibinfo {title} {Comparison of superconducting
  {NISQ} architectures}} (\bibinfo {year} {2024}),\ \Eprint
  {https://arxiv.org/abs/2409.02063} {arXiv:2409.02063 [quant-ph]} \BibitemShut
  {NoStop}%
\bibitem [{\citenamefont {Chamberland}\ \emph {et~al.}(2020)\citenamefont
  {Chamberland}, \citenamefont {Zhu}, \citenamefont {Yoder}, \citenamefont
  {Hertzberg},\ and\ \citenamefont {Cross}}]{Chamberland2020}%
  \BibitemOpen
  \bibfield  {author} {\bibinfo {author} {\bibfnamefont {C.}~\bibnamefont
  {Chamberland}}, \bibinfo {author} {\bibfnamefont {G.}~\bibnamefont {Zhu}},
  \bibinfo {author} {\bibfnamefont {T.~J.}\ \bibnamefont {Yoder}}, \bibinfo
  {author} {\bibfnamefont {J.~B.}\ \bibnamefont {Hertzberg}},\ and\ \bibinfo
  {author} {\bibfnamefont {A.~W.}\ \bibnamefont {Cross}},\ }\bibfield  {title}
  {\bibinfo {title} {Topological and subsystem codes on low-degree graphs with
  flag qubits},\ }\href {https://doi.org/10.1103/PhysRevX.10.011022} {\bibfield
   {journal} {\bibinfo  {journal} {Phys. Rev. X}\ }\textbf {\bibinfo {volume}
  {10}},\ \bibinfo {pages} {011022} (\bibinfo {year} {2020})}\BibitemShut
  {NoStop}%
\bibitem [{\citenamefont {Dupont}\ \emph {et~al.}(2023)\citenamefont {Dupont},
  \citenamefont {Evert}, \citenamefont {Hodson}, \citenamefont {Sundar},
  \citenamefont {Jeffrey}, \citenamefont {Yamaguchi}, \citenamefont {Feng},
  \citenamefont {Maciejewski}, \citenamefont {Hadfield}, \citenamefont {Alam},
  \citenamefont {Wang}, \citenamefont {Grabbe}, \citenamefont {Lott},
  \citenamefont {Rieffel}, \citenamefont {Venturelli},\ and\ \citenamefont
  {Reagor}}]{Dupont2023}%
  \BibitemOpen
  \bibfield  {author} {\bibinfo {author} {\bibfnamefont {M.}~\bibnamefont
  {Dupont}}, \bibinfo {author} {\bibfnamefont {B.}~\bibnamefont {Evert}},
  \bibinfo {author} {\bibfnamefont {M.~J.}\ \bibnamefont {Hodson}}, \bibinfo
  {author} {\bibfnamefont {B.}~\bibnamefont {Sundar}}, \bibinfo {author}
  {\bibfnamefont {S.}~\bibnamefont {Jeffrey}}, \bibinfo {author} {\bibfnamefont
  {Y.}~\bibnamefont {Yamaguchi}}, \bibinfo {author} {\bibfnamefont
  {D.}~\bibnamefont {Feng}}, \bibinfo {author} {\bibfnamefont {F.~B.}\
  \bibnamefont {Maciejewski}}, \bibinfo {author} {\bibfnamefont
  {S.}~\bibnamefont {Hadfield}}, \bibinfo {author} {\bibfnamefont {M.~S.}\
  \bibnamefont {Alam}}, \bibinfo {author} {\bibfnamefont {Z.}~\bibnamefont
  {Wang}}, \bibinfo {author} {\bibfnamefont {S.}~\bibnamefont {Grabbe}},
  \bibinfo {author} {\bibfnamefont {P.~A.}\ \bibnamefont {Lott}}, \bibinfo
  {author} {\bibfnamefont {E.~G.}\ \bibnamefont {Rieffel}}, \bibinfo {author}
  {\bibfnamefont {D.}~\bibnamefont {Venturelli}},\ and\ \bibinfo {author}
  {\bibfnamefont {M.~J.}\ \bibnamefont {Reagor}},\ }\bibfield  {title}
  {\bibinfo {title} {Quantum-enhanced greedy combinatorial optimization
  solver},\ }\href {https://doi.org/10.1126/sciadv.adi0487} {\bibfield
  {journal} {\bibinfo  {journal} {Science Advances}\ }\textbf {\bibinfo
  {volume} {9}},\ \bibinfo {pages} {eadi0487} (\bibinfo {year}
  {2023})}\BibitemShut {NoStop}%
\bibitem [{\citenamefont {Crooks}(2018)}]{Crooks2018}%
  \BibitemOpen
  \bibfield  {author} {\bibinfo {author} {\bibfnamefont {G.~E.}\ \bibnamefont
  {Crooks}},\ }\href@noop {} {\bibinfo {title} {Performance of the quantum
  approximate optimization algorithm on the maximum cut problem}} (\bibinfo
  {year} {2018}),\ \Eprint {https://arxiv.org/abs/1811.08419} {arXiv:1811.08419
  [quant-ph]} \BibitemShut {NoStop}%
\bibitem [{\citenamefont {Kivlichan}\ \emph {et~al.}(2018)\citenamefont
  {Kivlichan}, \citenamefont {McClean}, \citenamefont {Wiebe}, \citenamefont
  {Gidney}, \citenamefont {Aspuru-Guzik}, \citenamefont {Chan},\ and\
  \citenamefont {Babbush}}]{Kivlichan2018}%
  \BibitemOpen
  \bibfield  {author} {\bibinfo {author} {\bibfnamefont {I.~D.}\ \bibnamefont
  {Kivlichan}}, \bibinfo {author} {\bibfnamefont {J.}~\bibnamefont {McClean}},
  \bibinfo {author} {\bibfnamefont {N.}~\bibnamefont {Wiebe}}, \bibinfo
  {author} {\bibfnamefont {C.}~\bibnamefont {Gidney}}, \bibinfo {author}
  {\bibfnamefont {A.}~\bibnamefont {Aspuru-Guzik}}, \bibinfo {author}
  {\bibfnamefont {G.~K.-L.}\ \bibnamefont {Chan}},\ and\ \bibinfo {author}
  {\bibfnamefont {R.}~\bibnamefont {Babbush}},\ }\bibfield  {title} {\bibinfo
  {title} {Quantum simulation of electronic structure with linear depth and
  connectivity},\ }\href {https://doi.org/10.1103/PhysRevLett.120.110501}
  {\bibfield  {journal} {\bibinfo  {journal} {Phys. Rev. Lett.}\ }\textbf
  {\bibinfo {volume} {120}},\ \bibinfo {pages} {110501} (\bibinfo {year}
  {2018})}\BibitemShut {NoStop}%
\bibitem [{\citenamefont {O'Gorman}\ \emph {et~al.}(2019)\citenamefont
  {O'Gorman}, \citenamefont {Huggins}, \citenamefont {Rieffel},\ and\
  \citenamefont {Whaley}}]{Ogorman2019}%
  \BibitemOpen
  \bibfield  {author} {\bibinfo {author} {\bibfnamefont {B.}~\bibnamefont
  {O'Gorman}}, \bibinfo {author} {\bibfnamefont {W.~J.}\ \bibnamefont
  {Huggins}}, \bibinfo {author} {\bibfnamefont {E.~G.}\ \bibnamefont
  {Rieffel}},\ and\ \bibinfo {author} {\bibfnamefont {K.~B.}\ \bibnamefont
  {Whaley}},\ }\href@noop {} {\bibinfo {title} {Generalized swap networks for
  near-term quantum computing}} (\bibinfo {year} {2019}),\ \Eprint
  {https://arxiv.org/abs/1905.05118} {arXiv:1905.05118 [quant-ph]} \BibitemShut
  {NoStop}%
\bibitem [{\citenamefont {Hashim}\ \emph {et~al.}(2022)\citenamefont {Hashim},
  \citenamefont {Rines}, \citenamefont {Omole}, \citenamefont {Naik},
  \citenamefont {Kreikebaum}, \citenamefont {Santiago}, \citenamefont {Chong},
  \citenamefont {Siddiqi},\ and\ \citenamefont {Gokhale}}]{Hashim2022}%
  \BibitemOpen
  \bibfield  {author} {\bibinfo {author} {\bibfnamefont {A.}~\bibnamefont
  {Hashim}}, \bibinfo {author} {\bibfnamefont {R.}~\bibnamefont {Rines}},
  \bibinfo {author} {\bibfnamefont {V.}~\bibnamefont {Omole}}, \bibinfo
  {author} {\bibfnamefont {R.~K.}\ \bibnamefont {Naik}}, \bibinfo {author}
  {\bibfnamefont {J.~M.}\ \bibnamefont {Kreikebaum}}, \bibinfo {author}
  {\bibfnamefont {D.~I.}\ \bibnamefont {Santiago}}, \bibinfo {author}
  {\bibfnamefont {F.~T.}\ \bibnamefont {Chong}}, \bibinfo {author}
  {\bibfnamefont {I.}~\bibnamefont {Siddiqi}},\ and\ \bibinfo {author}
  {\bibfnamefont {P.}~\bibnamefont {Gokhale}},\ }\bibfield  {title} {\bibinfo
  {title} {Optimized swap networks with equivalent circuit averaging for
  qaoa},\ }\href {https://doi.org/10.1103/PhysRevResearch.4.033028} {\bibfield
  {journal} {\bibinfo  {journal} {Phys. Rev. Res.}\ }\textbf {\bibinfo {volume}
  {4}},\ \bibinfo {pages} {033028} (\bibinfo {year} {2022})}\BibitemShut
  {NoStop}%
\bibitem [{\citenamefont {Weidenfeller}\ \emph {et~al.}(2022)\citenamefont
  {Weidenfeller}, \citenamefont {Valor}, \citenamefont {Gacon}, \citenamefont
  {Tornow}, \citenamefont {Bello}, \citenamefont {Woerner},\ and\ \citenamefont
  {Egger}}]{Weidenfeller2022}%
  \BibitemOpen
  \bibfield  {author} {\bibinfo {author} {\bibfnamefont {J.}~\bibnamefont
  {Weidenfeller}}, \bibinfo {author} {\bibfnamefont {L.~C.}\ \bibnamefont
  {Valor}}, \bibinfo {author} {\bibfnamefont {J.}~\bibnamefont {Gacon}},
  \bibinfo {author} {\bibfnamefont {C.}~\bibnamefont {Tornow}}, \bibinfo
  {author} {\bibfnamefont {L.}~\bibnamefont {Bello}}, \bibinfo {author}
  {\bibfnamefont {S.}~\bibnamefont {Woerner}},\ and\ \bibinfo {author}
  {\bibfnamefont {D.~J.}\ \bibnamefont {Egger}},\ }\bibfield  {title} {\bibinfo
  {title} {Scaling of the quantum approximate optimization algorithm on
  superconducting qubit based hardware},\ }\href
  {https://doi.org/10.22331/q-2022-12-07-870} {\bibfield  {journal} {\bibinfo
  {journal} {{Quantum}}\ }\textbf {\bibinfo {volume} {6}},\ \bibinfo {pages}
  {870} (\bibinfo {year} {2022})}\BibitemShut {NoStop}%
\bibitem [{\citenamefont {Pei~Yuan}(2024)}]{Yuan2024}%
  \BibitemOpen
  \bibfield  {author} {\bibinfo {author} {\bibfnamefont {S.~Z.}\ \bibnamefont
  {Pei~Yuan}},\ }\bibfield  {title} {\bibinfo {title} {Full characterization of
  the depth overhead for quantum circuit compilation with arbitrary qubit
  connectivity constraint},\ }\href {https://arxiv.org/abs/2402.02403}
  {\bibfield  {journal} {\bibinfo  {journal} {arXiv:2402.02403}\ } (\bibinfo
  {year} {2024})}\BibitemShut {NoStop}%
\bibitem [{\citenamefont {Kaushal}\ \emph {et~al.}(2020)\citenamefont
  {Kaushal}, \citenamefont {Lekitsch}, \citenamefont {Stahl}, \citenamefont
  {Hilder}, \citenamefont {Pijn}, \citenamefont {Schmiegelow}, \citenamefont
  {Bermudez}, \citenamefont {Müller}, \citenamefont {Schmidt-Kaler},\ and\
  \citenamefont {Poschinger}}]{Kaushal2020}%
  \BibitemOpen
  \bibfield  {author} {\bibinfo {author} {\bibfnamefont {V.}~\bibnamefont
  {Kaushal}}, \bibinfo {author} {\bibfnamefont {B.}~\bibnamefont {Lekitsch}},
  \bibinfo {author} {\bibfnamefont {A.}~\bibnamefont {Stahl}}, \bibinfo
  {author} {\bibfnamefont {J.}~\bibnamefont {Hilder}}, \bibinfo {author}
  {\bibfnamefont {D.}~\bibnamefont {Pijn}}, \bibinfo {author} {\bibfnamefont
  {C.}~\bibnamefont {Schmiegelow}}, \bibinfo {author} {\bibfnamefont
  {A.}~\bibnamefont {Bermudez}}, \bibinfo {author} {\bibfnamefont
  {M.}~\bibnamefont {Müller}}, \bibinfo {author} {\bibfnamefont
  {F.}~\bibnamefont {Schmidt-Kaler}},\ and\ \bibinfo {author} {\bibfnamefont
  {U.}~\bibnamefont {Poschinger}},\ }\bibfield  {title} {\bibinfo {title}
  {Shuttling-based trapped-ion quantum information processing},\ }\href
  {https://doi.org/10.1116/1.5126186} {\bibfield  {journal} {\bibinfo
  {journal} {AVS Quantum Science}\ }\textbf {\bibinfo {volume} {2}},\ \bibinfo
  {pages} {014101} (\bibinfo {year} {2020})}\BibitemShut {NoStop}%
\bibitem [{\citenamefont {Bluvstein}\ \emph
  {et~al.}(2022{\natexlab{a}})\citenamefont {Bluvstein}, \citenamefont
  {Levine}, \citenamefont {Semeghini}, \citenamefont {Wang}, \citenamefont
  {Ebadi}, \citenamefont {Kalinowski}, \citenamefont {Keesling}, \citenamefont
  {Maskara}, \citenamefont {Pichler}, \citenamefont {Greiner}, \citenamefont
  {Vuleti{\'{c}}},\ and\ \citenamefont {Lukin}}]{Bluvstein2022}%
  \BibitemOpen
  \bibfield  {author} {\bibinfo {author} {\bibfnamefont {D.}~\bibnamefont
  {Bluvstein}}, \bibinfo {author} {\bibfnamefont {H.}~\bibnamefont {Levine}},
  \bibinfo {author} {\bibfnamefont {G.}~\bibnamefont {Semeghini}}, \bibinfo
  {author} {\bibfnamefont {T.~T.}\ \bibnamefont {Wang}}, \bibinfo {author}
  {\bibfnamefont {S.}~\bibnamefont {Ebadi}}, \bibinfo {author} {\bibfnamefont
  {M.}~\bibnamefont {Kalinowski}}, \bibinfo {author} {\bibfnamefont
  {A.}~\bibnamefont {Keesling}}, \bibinfo {author} {\bibfnamefont
  {N.}~\bibnamefont {Maskara}}, \bibinfo {author} {\bibfnamefont
  {H.}~\bibnamefont {Pichler}}, \bibinfo {author} {\bibfnamefont
  {M.}~\bibnamefont {Greiner}}, \bibinfo {author} {\bibfnamefont
  {V.}~\bibnamefont {Vuleti{\'{c}}}},\ and\ \bibinfo {author} {\bibfnamefont
  {M.~D.}\ \bibnamefont {Lukin}},\ }\bibfield  {title} {\bibinfo {title} {A
  quantum processor based on coherent transport of entangled atom arrays},\
  }\href {https://doi.org/10.1038/s41586-022-04592-6} {\bibfield  {journal}
  {\bibinfo  {journal} {Nature}\ }\textbf {\bibinfo {volume} {604}},\ \bibinfo
  {pages} {451} (\bibinfo {year} {2022}{\natexlab{a}})}\BibitemShut {NoStop}%
\bibitem [{\citenamefont {Zwerver}\ \emph {et~al.}(2023)\citenamefont
  {Zwerver}, \citenamefont {Amitonov}, \citenamefont {de~Snoo}, \citenamefont
  {Madzik}, \citenamefont {Rimbach-Russ}, \citenamefont {Sammak}, \citenamefont
  {Scappucci},\ and\ \citenamefont {Vandersypen}}]{Zwerver2023}%
  \BibitemOpen
  \bibfield  {author} {\bibinfo {author} {\bibfnamefont {A.}~\bibnamefont
  {Zwerver}}, \bibinfo {author} {\bibfnamefont {S.}~\bibnamefont {Amitonov}},
  \bibinfo {author} {\bibfnamefont {S.}~\bibnamefont {de~Snoo}}, \bibinfo
  {author} {\bibfnamefont {M.}~\bibnamefont {Madzik}}, \bibinfo {author}
  {\bibfnamefont {M.}~\bibnamefont {Rimbach-Russ}}, \bibinfo {author}
  {\bibfnamefont {A.}~\bibnamefont {Sammak}}, \bibinfo {author} {\bibfnamefont
  {G.}~\bibnamefont {Scappucci}},\ and\ \bibinfo {author} {\bibfnamefont
  {L.}~\bibnamefont {Vandersypen}},\ }\bibfield  {title} {\bibinfo {title}
  {Shuttling an electron spin through a silicon quantum dot array},\ }\href
  {https://doi.org/10.1103/PRXQuantum.4.030303} {\bibfield  {journal} {\bibinfo
   {journal} {PRX Quantum}\ }\textbf {\bibinfo {volume} {4}},\ \bibinfo {pages}
  {030303} (\bibinfo {year} {2023})}\BibitemShut {NoStop}%
\bibitem [{\citenamefont {K{\"u}nne}\ \emph {et~al.}(2024)\citenamefont
  {K{\"u}nne}, \citenamefont {Willmes}, \citenamefont {Oberl{\"a}nder},
  \citenamefont {Gorjaew}, \citenamefont {Teske}, \citenamefont {Bhardwaj},
  \citenamefont {Beer}, \citenamefont {Kammerloher}, \citenamefont {Otten},
  \citenamefont {Seidler}, \citenamefont {Xue}, \citenamefont {Schreiber},\
  and\ \citenamefont {Bluhm}}]{Künne2024}%
  \BibitemOpen
  \bibfield  {author} {\bibinfo {author} {\bibfnamefont {M.}~\bibnamefont
  {K{\"u}nne}}, \bibinfo {author} {\bibfnamefont {A.}~\bibnamefont {Willmes}},
  \bibinfo {author} {\bibfnamefont {M.}~\bibnamefont {Oberl{\"a}nder}},
  \bibinfo {author} {\bibfnamefont {C.}~\bibnamefont {Gorjaew}}, \bibinfo
  {author} {\bibfnamefont {J.~D.}\ \bibnamefont {Teske}}, \bibinfo {author}
  {\bibfnamefont {H.}~\bibnamefont {Bhardwaj}}, \bibinfo {author}
  {\bibfnamefont {M.}~\bibnamefont {Beer}}, \bibinfo {author} {\bibfnamefont
  {E.}~\bibnamefont {Kammerloher}}, \bibinfo {author} {\bibfnamefont
  {R.}~\bibnamefont {Otten}}, \bibinfo {author} {\bibfnamefont
  {I.}~\bibnamefont {Seidler}}, \bibinfo {author} {\bibfnamefont
  {R.}~\bibnamefont {Xue}}, \bibinfo {author} {\bibfnamefont {L.~R.}\
  \bibnamefont {Schreiber}},\ and\ \bibinfo {author} {\bibfnamefont
  {H.}~\bibnamefont {Bluhm}},\ }\bibfield  {title} {\bibinfo {title} {The
  spinbus architecture for scaling spin qubits with electron shuttling},\
  }\href {https://doi.org/10.1038/s41467-024-49182-4} {\bibfield  {journal}
  {\bibinfo  {journal} {Nature Communications}\ }\textbf {\bibinfo {volume}
  {15}},\ \bibinfo {pages} {4977} (\bibinfo {year} {2024})}\BibitemShut
  {NoStop}%
\bibitem [{\citenamefont {Schmitz}\ \emph {et~al.}(2021)\citenamefont
  {Schmitz}, \citenamefont {Sawaya}, \citenamefont {Johri},\ and\ \citenamefont
  {Matsuura}}]{Schmitz2021}%
  \BibitemOpen
  \bibfield  {author} {\bibinfo {author} {\bibfnamefont {A.~T.}\ \bibnamefont
  {Schmitz}}, \bibinfo {author} {\bibfnamefont {N.~P.}\ \bibnamefont {Sawaya}},
  \bibinfo {author} {\bibfnamefont {S.}~\bibnamefont {Johri}},\ and\ \bibinfo
  {author} {\bibfnamefont {A.~Y.}\ \bibnamefont {Matsuura}},\ }\bibfield
  {title} {\bibinfo {title} {Graph optimization perspective for low-depth
  trotter-suzuki decomposition},\ }\href
  {https://doi.org/10.48550/arXiv.2103.08602} {\bibfield  {journal} {\bibinfo
  {journal} {arXiv:2103.08602}\ } (\bibinfo {year} {2021})}\BibitemShut
  {NoStop}%
\bibitem [{\citenamefont {Meijer-van~de Griend}\ and\ \citenamefont
  {Meng~Li}(2023)}]{van-de-Griend2022}%
  \BibitemOpen
  \bibfield  {author} {\bibinfo {author} {\bibfnamefont {A.}~\bibnamefont
  {Meijer-van~de Griend}}\ and\ \bibinfo {author} {\bibfnamefont
  {S.}~\bibnamefont {Meng~Li}},\ }\bibfield  {title} {\bibinfo {title} {Dynamic
  qubit routing with cnot circuit synthesis for quantum compilation},\ }\href
  {https://doi.org/10.4204/EPTCS.394.18} {\bibfield  {journal} {\bibinfo
  {journal} {JournalElectronic Proceedings in Theoretical Computer Science,
  EPTCS}\ }\textbf {\bibinfo {volume} {394}},\ \bibinfo {pages} {363} (\bibinfo
  {year} {2023})}\BibitemShut {NoStop}%
\bibitem [{\citenamefont {Klaver}\ \emph {et~al.}(2024)\citenamefont {Klaver},
  \citenamefont {Rombouts}, \citenamefont {Fellner}, \citenamefont {Messinger},
  \citenamefont {Ender}, \citenamefont {Ludwig},\ and\ \citenamefont
  {Lechner}}]{Klaver2024}%
  \BibitemOpen
  \bibfield  {author} {\bibinfo {author} {\bibfnamefont {B.}~\bibnamefont
  {Klaver}}, \bibinfo {author} {\bibfnamefont {S.}~\bibnamefont {Rombouts}},
  \bibinfo {author} {\bibfnamefont {M.}~\bibnamefont {Fellner}}, \bibinfo
  {author} {\bibfnamefont {A.}~\bibnamefont {Messinger}}, \bibinfo {author}
  {\bibfnamefont {K.}~\bibnamefont {Ender}}, \bibinfo {author} {\bibfnamefont
  {K.}~\bibnamefont {Ludwig}},\ and\ \bibinfo {author} {\bibfnamefont
  {W.}~\bibnamefont {Lechner}},\ }\href@noop {} {\bibinfo {title} {Swap-less
  implementation of quantum algorithms}} (\bibinfo {year} {2024}),\ \Eprint
  {https://arxiv.org/abs/2408.10907} {arXiv:2408.10907 [quant-ph]} \BibitemShut
  {NoStop}%
\bibitem [{\citenamefont {Cowtan}\ \emph {et~al.}(2020)\citenamefont {Cowtan},
  \citenamefont {Dilkes}, \citenamefont {Duncan}, \citenamefont {Simmons},\
  and\ \citenamefont {Sivarajah}}]{Cowtan2019}%
  \BibitemOpen
  \bibfield  {author} {\bibinfo {author} {\bibfnamefont {A.}~\bibnamefont
  {Cowtan}}, \bibinfo {author} {\bibfnamefont {S.}~\bibnamefont {Dilkes}},
  \bibinfo {author} {\bibfnamefont {R.}~\bibnamefont {Duncan}}, \bibinfo
  {author} {\bibfnamefont {W.}~\bibnamefont {Simmons}},\ and\ \bibinfo {author}
  {\bibfnamefont {S.}~\bibnamefont {Sivarajah}},\ }\bibfield  {title} {\bibinfo
  {title} {{Phase Gadget Synthesis for Shallow Circuits}},\ }\href
  {https://doi.org/10.4204/EPTCS.318.13} {\bibfield  {journal} {\bibinfo
  {journal} {EPTCS}\ }\textbf {\bibinfo {volume} {318}},\ \bibinfo {pages}
  {213} (\bibinfo {year} {2020})}\BibitemShut {NoStop}%
\bibitem [{\citenamefont {Lechner}\ \emph {et~al.}(2015)\citenamefont
  {Lechner}, \citenamefont {Hauke},\ and\ \citenamefont
  {Zoller}}]{Lechner2015}%
  \BibitemOpen
  \bibfield  {author} {\bibinfo {author} {\bibfnamefont {W.}~\bibnamefont
  {Lechner}}, \bibinfo {author} {\bibfnamefont {P.}~\bibnamefont {Hauke}},\
  and\ \bibinfo {author} {\bibfnamefont {P.}~\bibnamefont {Zoller}},\
  }\bibfield  {title} {\bibinfo {title} {A quantum annealing architecture with
  all-to-all connectivity from local interactions},\ }\href
  {https://advances.sciencemag.org/content/1/9/e1500838} {\bibfield  {journal}
  {\bibinfo  {journal} {Science Advances}\ }\textbf {\bibinfo {volume} {1}},\
  \bibinfo {pages} {e1500838} (\bibinfo {year} {2015})}\BibitemShut {NoStop}%
\bibitem [{\citenamefont {Fellner}\ \emph
  {et~al.}(2022{\natexlab{a}})\citenamefont {Fellner}, \citenamefont
  {Messinger}, \citenamefont {Ender},\ and\ \citenamefont
  {Lechner}}]{Fellner2022}%
  \BibitemOpen
  \bibfield  {author} {\bibinfo {author} {\bibfnamefont {M.}~\bibnamefont
  {Fellner}}, \bibinfo {author} {\bibfnamefont {A.}~\bibnamefont {Messinger}},
  \bibinfo {author} {\bibfnamefont {K.}~\bibnamefont {Ender}},\ and\ \bibinfo
  {author} {\bibfnamefont {W.}~\bibnamefont {Lechner}},\ }\bibfield  {title}
  {\bibinfo {title} {Universal parity quantum computing},\ }\href
  {https://doi.org/10.1103/PhysRevLett.129.180503} {\bibfield  {journal}
  {\bibinfo  {journal} {Phys. Rev. Lett.}\ }\textbf {\bibinfo {volume} {129}},\
  \bibinfo {pages} {180503} (\bibinfo {year} {2022}{\natexlab{a}})}\BibitemShut
  {NoStop}%
\bibitem [{\citenamefont {Fellner}\ \emph
  {et~al.}(2022{\natexlab{b}})\citenamefont {Fellner}, \citenamefont
  {Messinger}, \citenamefont {Ender},\ and\ \citenamefont
  {Lechner}}]{Fellner2022app}%
  \BibitemOpen
  \bibfield  {author} {\bibinfo {author} {\bibfnamefont {M.}~\bibnamefont
  {Fellner}}, \bibinfo {author} {\bibfnamefont {A.}~\bibnamefont {Messinger}},
  \bibinfo {author} {\bibfnamefont {K.}~\bibnamefont {Ender}},\ and\ \bibinfo
  {author} {\bibfnamefont {W.}~\bibnamefont {Lechner}},\ }\bibfield  {title}
  {\bibinfo {title} {Applications of universal parity quantum computation},\
  }\href {https://doi.org/10.1103/PhysRevA.106.042442} {\bibfield  {journal}
  {\bibinfo  {journal} {Phys. Rev. A}\ }\textbf {\bibinfo {volume} {106}},\
  \bibinfo {pages} {042442} (\bibinfo {year} {2022}{\natexlab{b}})}\BibitemShut
  {NoStop}%
\bibitem [{\citenamefont {Park}\ and\ \citenamefont {Ahn}(2023)}]{Park2023}%
  \BibitemOpen
  \bibfield  {author} {\bibinfo {author} {\bibfnamefont {B.}~\bibnamefont
  {Park}}\ and\ \bibinfo {author} {\bibfnamefont {D.}~\bibnamefont {Ahn}},\
  }\bibfield  {title} {\bibinfo {title} {Reducing {CNOT} count in quantum
  {F}ourier transform for the linear nearest-neighbor architecture},\ }\href
  {https://doi.org/10.1038/s41598-023-35625-3} {\bibfield  {journal} {\bibinfo
  {journal} {Scientific Reports}\ }\textbf {\bibinfo {volume} {13}},\ \bibinfo
  {pages} {8638} (\bibinfo {year} {2023})}\BibitemShut {NoStop}%
\bibitem [{\citenamefont {Klaver}\ \emph {et~al.}(2025)\citenamefont {Klaver},
  \citenamefont {Ludwig}, \citenamefont {Messinger}, \citenamefont {Rombouts},
  \citenamefont {Fellner}, \citenamefont {Ender},\ and\ \citenamefont
  {Lechner}}]{Klaver2025}%
  \BibitemOpen
  \bibfield  {author} {\bibinfo {author} {\bibfnamefont {B.}~\bibnamefont
  {Klaver}}, \bibinfo {author} {\bibfnamefont {K.}~\bibnamefont {Ludwig}},
  \bibinfo {author} {\bibfnamefont {A.}~\bibnamefont {Messinger}}, \bibinfo
  {author} {\bibfnamefont {S.~M.~A.}\ \bibnamefont {Rombouts}}, \bibinfo
  {author} {\bibfnamefont {M.}~\bibnamefont {Fellner}}, \bibinfo {author}
  {\bibfnamefont {K.}~\bibnamefont {Ender}},\ and\ \bibinfo {author}
  {\bibfnamefont {W.}~\bibnamefont {Lechner}},\ }\bibfield  {title} {\bibinfo
  {title} {The parity flow formalism: Tracking quantum information throughout
  computation},\ }\href@noop {} {\bibfield  {journal} {\bibinfo  {journal}
  {arXiv:2505.09468}\ } (\bibinfo {year} {2025})}\BibitemShut {NoStop}%
\bibitem [{\citenamefont {Cheng}\ \emph {et~al.}(2023)\citenamefont {Cheng},
  \citenamefont {Deng}, \citenamefont {Gu}, \citenamefont {He}, \citenamefont
  {Hu}, \citenamefont {Huang}, \citenamefont {Li}, \citenamefont {Lin},
  \citenamefont {Lu}, \citenamefont {Lu}, \citenamefont {Qiu}, \citenamefont
  {Wang}, \citenamefont {Xin}, \citenamefont {Yu}, \citenamefont {Yung},
  \citenamefont {Zeng}, \citenamefont {Zhang}, \citenamefont {Zhong},
  \citenamefont {Peng}, \citenamefont {Nori},\ and\ \citenamefont
  {Yu}}]{Cheng2023}%
  \BibitemOpen
  \bibfield  {author} {\bibinfo {author} {\bibfnamefont {B.}~\bibnamefont
  {Cheng}}, \bibinfo {author} {\bibfnamefont {X.-H.}\ \bibnamefont {Deng}},
  \bibinfo {author} {\bibfnamefont {X.}~\bibnamefont {Gu}}, \bibinfo {author}
  {\bibfnamefont {Y.}~\bibnamefont {He}}, \bibinfo {author} {\bibfnamefont
  {G.}~\bibnamefont {Hu}}, \bibinfo {author} {\bibfnamefont {P.}~\bibnamefont
  {Huang}}, \bibinfo {author} {\bibfnamefont {J.}~\bibnamefont {Li}}, \bibinfo
  {author} {\bibfnamefont {B.-C.}\ \bibnamefont {Lin}}, \bibinfo {author}
  {\bibfnamefont {D.}~\bibnamefont {Lu}}, \bibinfo {author} {\bibfnamefont
  {Y.}~\bibnamefont {Lu}}, \bibinfo {author} {\bibfnamefont {C.}~\bibnamefont
  {Qiu}}, \bibinfo {author} {\bibfnamefont {H.}~\bibnamefont {Wang}}, \bibinfo
  {author} {\bibfnamefont {T.}~\bibnamefont {Xin}}, \bibinfo {author}
  {\bibfnamefont {S.}~\bibnamefont {Yu}}, \bibinfo {author} {\bibfnamefont
  {M.-H.}\ \bibnamefont {Yung}}, \bibinfo {author} {\bibfnamefont
  {J.}~\bibnamefont {Zeng}}, \bibinfo {author} {\bibfnamefont {S.}~\bibnamefont
  {Zhang}}, \bibinfo {author} {\bibfnamefont {Y.}~\bibnamefont {Zhong}},
  \bibinfo {author} {\bibfnamefont {X.}~\bibnamefont {Peng}}, \bibinfo {author}
  {\bibfnamefont {F.}~\bibnamefont {Nori}},\ and\ \bibinfo {author}
  {\bibfnamefont {D.}~\bibnamefont {Yu}},\ }\bibfield  {title} {\bibinfo
  {title} {Noisy intermediate-scale quantum computers},\ }\href
  {https://doi.org/10.1007/s11467-022-1249-z} {\bibfield  {journal} {\bibinfo
  {journal} {Frontiers of Physics}\ }\textbf {\bibinfo {volume} {18}},\
  \bibinfo {pages} {21308} (\bibinfo {year} {2023})}\BibitemShut {NoStop}%
\bibitem [{\citenamefont {Collins}\ \emph {et~al.}(2001)\citenamefont
  {Collins}, \citenamefont {Linden},\ and\ \citenamefont
  {Popescu}}]{ColLinNoa01}%
  \BibitemOpen
  \bibfield  {author} {\bibinfo {author} {\bibfnamefont {D.}~\bibnamefont
  {Collins}}, \bibinfo {author} {\bibfnamefont {N.}~\bibnamefont {Linden}},\
  and\ \bibinfo {author} {\bibfnamefont {S.}~\bibnamefont {Popescu}},\
  }\bibfield  {title} {\bibinfo {title} {Nonlocal content of quantum
  operations},\ }\href {https://doi.org/10.1103/PhysRevA.64.032302} {\bibfield
  {journal} {\bibinfo  {journal} {Phys. Rev. A}\ }\textbf {\bibinfo {volume}
  {64}},\ \bibinfo {pages} {032302} (\bibinfo {year} {2001})}\BibitemShut
  {NoStop}%
\bibitem [{\citenamefont {Sung}\ \emph {et~al.}(2021)\citenamefont {Sung},
  \citenamefont {Ding}, \citenamefont {Braum\"uller}, \citenamefont
  {Veps\"al\"ainen}, \citenamefont {Kannan}, \citenamefont {Kjaergaard},
  \citenamefont {Greene}, \citenamefont {Samach}, \citenamefont {McNally},
  \citenamefont {Kim}, \citenamefont {Melville}, \citenamefont {Niedzielski},
  \citenamefont {Schwartz}, \citenamefont {Yoder}, \citenamefont {Orlando},
  \citenamefont {Gustavsson},\ and\ \citenamefont {Oliver}}]{Sung2021}%
  \BibitemOpen
  \bibfield  {author} {\bibinfo {author} {\bibfnamefont {Y.}~\bibnamefont
  {Sung}}, \bibinfo {author} {\bibfnamefont {L.}~\bibnamefont {Ding}}, \bibinfo
  {author} {\bibfnamefont {J.}~\bibnamefont {Braum\"uller}}, \bibinfo {author}
  {\bibfnamefont {A.}~\bibnamefont {Veps\"al\"ainen}}, \bibinfo {author}
  {\bibfnamefont {B.}~\bibnamefont {Kannan}}, \bibinfo {author} {\bibfnamefont
  {M.}~\bibnamefont {Kjaergaard}}, \bibinfo {author} {\bibfnamefont
  {A.}~\bibnamefont {Greene}}, \bibinfo {author} {\bibfnamefont {G.~O.}\
  \bibnamefont {Samach}}, \bibinfo {author} {\bibfnamefont {C.}~\bibnamefont
  {McNally}}, \bibinfo {author} {\bibfnamefont {D.}~\bibnamefont {Kim}},
  \bibinfo {author} {\bibfnamefont {A.}~\bibnamefont {Melville}}, \bibinfo
  {author} {\bibfnamefont {B.~M.}\ \bibnamefont {Niedzielski}}, \bibinfo
  {author} {\bibfnamefont {M.~E.}\ \bibnamefont {Schwartz}}, \bibinfo {author}
  {\bibfnamefont {J.~L.}\ \bibnamefont {Yoder}}, \bibinfo {author}
  {\bibfnamefont {T.~P.}\ \bibnamefont {Orlando}}, \bibinfo {author}
  {\bibfnamefont {S.}~\bibnamefont {Gustavsson}},\ and\ \bibinfo {author}
  {\bibfnamefont {W.~D.}\ \bibnamefont {Oliver}},\ }\bibfield  {title}
  {\bibinfo {title} {Realization of high-fidelity cz and $zz$-free iswap gates
  with a tunable coupler},\ }\href {https://doi.org/10.1103/PhysRevX.11.021058}
  {\bibfield  {journal} {\bibinfo  {journal} {Phys. Rev. X}\ }\textbf {\bibinfo
  {volume} {11}},\ \bibinfo {pages} {021058} (\bibinfo {year}
  {2021})}\BibitemShut {NoStop}%
\bibitem [{\citenamefont {Wei}\ \emph {et~al.}(2024)\citenamefont {Wei},
  \citenamefont {Lauer}, \citenamefont {Pritchett}, \citenamefont {Shanks},
  \citenamefont {McKay},\ and\ \citenamefont {Javadi-Abhari}}]{Wei2024}%
  \BibitemOpen
  \bibfield  {author} {\bibinfo {author} {\bibfnamefont {K.~X.}\ \bibnamefont
  {Wei}}, \bibinfo {author} {\bibfnamefont {I.}~\bibnamefont {Lauer}}, \bibinfo
  {author} {\bibfnamefont {E.}~\bibnamefont {Pritchett}}, \bibinfo {author}
  {\bibfnamefont {W.}~\bibnamefont {Shanks}}, \bibinfo {author} {\bibfnamefont
  {D.~C.}\ \bibnamefont {McKay}},\ and\ \bibinfo {author} {\bibfnamefont
  {A.}~\bibnamefont {Javadi-Abhari}},\ }\bibfield  {title} {\bibinfo {title}
  {Native two-qubit gates in fixed-coupling, fixed-frequency transmons beyond
  cross-resonance interaction},\ }\href
  {https://doi.org/10.1103/PRXQuantum.5.020338} {\bibfield  {journal} {\bibinfo
   {journal} {PRX Quantum}\ }\textbf {\bibinfo {volume} {5}},\ \bibinfo {pages}
  {020338} (\bibinfo {year} {2024})}\BibitemShut {NoStop}%
\bibitem [{\citenamefont {Harrigan}\ \emph {et~al.}(2021)\citenamefont
  {Harrigan}, \citenamefont {Sung}, \citenamefont {Neeley}, \citenamefont
  {Satzinger}, \citenamefont {Arute}, \citenamefont {Arya}, \citenamefont
  {Atalaya}, \citenamefont {Bardin}, \citenamefont {Barends}, \citenamefont
  {Boixo}, \citenamefont {Broughton}, \citenamefont {Buckley}, \citenamefont
  {Buell}, \citenamefont {Burkett}, \citenamefont {Bushnell}, \citenamefont
  {Chen}, \citenamefont {Chen}, \citenamefont {{Ben Chiaro}}, \citenamefont
  {Collins}, \citenamefont {Courtney}, \citenamefont {Demura}, \citenamefont
  {Dunsworth}, \citenamefont {Eppens}, \citenamefont {Fowler}, \citenamefont
  {Foxen}, \citenamefont {Gidney}, \citenamefont {Giustina}, \citenamefont
  {Graff}, \citenamefont {Habegger}, \citenamefont {Ho}, \citenamefont {Hong},
  \citenamefont {Huang}, \citenamefont {Ioffe}, \citenamefont {Isakov},
  \citenamefont {Jeffrey}, \citenamefont {Jiang}, \citenamefont {Jones},
  \citenamefont {Kafri}, \citenamefont {Kechedzhi}, \citenamefont {Kelly},
  \citenamefont {Kim}, \citenamefont {Klimov}, \citenamefont {Korotkov},
  \citenamefont {Kostritsa}, \citenamefont {Landhuis}, \citenamefont {Laptev},
  \citenamefont {Lindmark}, \citenamefont {Leib}, \citenamefont {Martin},
  \citenamefont {Martinis}, \citenamefont {McClean}, \citenamefont {McEwen},
  \citenamefont {Megrant}, \citenamefont {Mi}, \citenamefont {Mohseni},
  \citenamefont {Mruczkiewicz}, \citenamefont {Mutus}, \citenamefont {Naaman},
  \citenamefont {Neill}, \citenamefont {Neukart}, \citenamefont {Niu},
  \citenamefont {O’Brien}, \citenamefont {O’Gorman}, \citenamefont {Ostby},
  \citenamefont {Petukhov}, \citenamefont {Putterman}, \citenamefont
  {Quintana}, \citenamefont {Roushan}, \citenamefont {Rubin}, \citenamefont
  {Sank}, \citenamefont {Skolik}, \citenamefont {Smelyanskiy}, \citenamefont
  {Strain}, \citenamefont {Streif}, \citenamefont {Szalay}, \citenamefont
  {Vainsencher}, \citenamefont {White}, \citenamefont {Yao}, \citenamefont
  {Yeh}, \citenamefont {Zalcman}, \citenamefont {Zhou}, \citenamefont {Neven},
  \citenamefont {Bacon}, \citenamefont {Lucero}, \citenamefont {Farhi},\ and\
  \citenamefont {Babbush}}]{harrigan_quantum_2021}%
  \BibitemOpen
  \bibfield  {author} {\bibinfo {author} {\bibfnamefont {M.~P.}\ \bibnamefont
  {Harrigan}}, \bibinfo {author} {\bibfnamefont {K.~J.}\ \bibnamefont {Sung}},
  \bibinfo {author} {\bibfnamefont {M.}~\bibnamefont {Neeley}}, \bibinfo
  {author} {\bibfnamefont {K.~J.}\ \bibnamefont {Satzinger}}, \bibinfo {author}
  {\bibfnamefont {F.}~\bibnamefont {Arute}}, \bibinfo {author} {\bibfnamefont
  {K.}~\bibnamefont {Arya}}, \bibinfo {author} {\bibfnamefont {J.}~\bibnamefont
  {Atalaya}}, \bibinfo {author} {\bibfnamefont {J.~C.}\ \bibnamefont {Bardin}},
  \bibinfo {author} {\bibfnamefont {R.}~\bibnamefont {Barends}}, \bibinfo
  {author} {\bibfnamefont {S.}~\bibnamefont {Boixo}}, \bibinfo {author}
  {\bibfnamefont {M.}~\bibnamefont {Broughton}}, \bibinfo {author}
  {\bibfnamefont {B.~B.}\ \bibnamefont {Buckley}}, \bibinfo {author}
  {\bibfnamefont {D.~A.}\ \bibnamefont {Buell}}, \bibinfo {author}
  {\bibfnamefont {B.}~\bibnamefont {Burkett}}, \bibinfo {author} {\bibfnamefont
  {N.}~\bibnamefont {Bushnell}}, \bibinfo {author} {\bibfnamefont
  {Y.}~\bibnamefont {Chen}}, \bibinfo {author} {\bibfnamefont {Z.}~\bibnamefont
  {Chen}}, \bibinfo {author} {\bibnamefont {{Ben Chiaro}}}, \bibinfo {author}
  {\bibfnamefont {R.}~\bibnamefont {Collins}}, \bibinfo {author} {\bibfnamefont
  {W.}~\bibnamefont {Courtney}}, \bibinfo {author} {\bibfnamefont
  {S.}~\bibnamefont {Demura}}, \bibinfo {author} {\bibfnamefont
  {A.}~\bibnamefont {Dunsworth}}, \bibinfo {author} {\bibfnamefont
  {D.}~\bibnamefont {Eppens}}, \bibinfo {author} {\bibfnamefont
  {A.}~\bibnamefont {Fowler}}, \bibinfo {author} {\bibfnamefont
  {B.}~\bibnamefont {Foxen}}, \bibinfo {author} {\bibfnamefont
  {C.}~\bibnamefont {Gidney}}, \bibinfo {author} {\bibfnamefont
  {M.}~\bibnamefont {Giustina}}, \bibinfo {author} {\bibfnamefont
  {R.}~\bibnamefont {Graff}}, \bibinfo {author} {\bibfnamefont
  {S.}~\bibnamefont {Habegger}}, \bibinfo {author} {\bibfnamefont
  {A.}~\bibnamefont {Ho}}, \bibinfo {author} {\bibfnamefont {S.}~\bibnamefont
  {Hong}}, \bibinfo {author} {\bibfnamefont {T.}~\bibnamefont {Huang}},
  \bibinfo {author} {\bibfnamefont {L.~B.}\ \bibnamefont {Ioffe}}, \bibinfo
  {author} {\bibfnamefont {S.~V.}\ \bibnamefont {Isakov}}, \bibinfo {author}
  {\bibfnamefont {E.}~\bibnamefont {Jeffrey}}, \bibinfo {author} {\bibfnamefont
  {Z.}~\bibnamefont {Jiang}}, \bibinfo {author} {\bibfnamefont
  {C.}~\bibnamefont {Jones}}, \bibinfo {author} {\bibfnamefont
  {D.}~\bibnamefont {Kafri}}, \bibinfo {author} {\bibfnamefont
  {K.}~\bibnamefont {Kechedzhi}}, \bibinfo {author} {\bibfnamefont
  {J.}~\bibnamefont {Kelly}}, \bibinfo {author} {\bibfnamefont
  {S.}~\bibnamefont {Kim}}, \bibinfo {author} {\bibfnamefont {P.~V.}\
  \bibnamefont {Klimov}}, \bibinfo {author} {\bibfnamefont {A.~N.}\
  \bibnamefont {Korotkov}}, \bibinfo {author} {\bibfnamefont {F.}~\bibnamefont
  {Kostritsa}}, \bibinfo {author} {\bibfnamefont {D.}~\bibnamefont {Landhuis}},
  \bibinfo {author} {\bibfnamefont {P.}~\bibnamefont {Laptev}}, \bibinfo
  {author} {\bibfnamefont {M.}~\bibnamefont {Lindmark}}, \bibinfo {author}
  {\bibfnamefont {M.}~\bibnamefont {Leib}}, \bibinfo {author} {\bibfnamefont
  {O.}~\bibnamefont {Martin}}, \bibinfo {author} {\bibfnamefont {J.~M.}\
  \bibnamefont {Martinis}}, \bibinfo {author} {\bibfnamefont {J.~R.}\
  \bibnamefont {McClean}}, \bibinfo {author} {\bibfnamefont {M.}~\bibnamefont
  {McEwen}}, \bibinfo {author} {\bibfnamefont {A.}~\bibnamefont {Megrant}},
  \bibinfo {author} {\bibfnamefont {X.}~\bibnamefont {Mi}}, \bibinfo {author}
  {\bibfnamefont {M.}~\bibnamefont {Mohseni}}, \bibinfo {author} {\bibfnamefont
  {W.}~\bibnamefont {Mruczkiewicz}}, \bibinfo {author} {\bibfnamefont
  {J.}~\bibnamefont {Mutus}}, \bibinfo {author} {\bibfnamefont
  {O.}~\bibnamefont {Naaman}}, \bibinfo {author} {\bibfnamefont
  {C.}~\bibnamefont {Neill}}, \bibinfo {author} {\bibfnamefont
  {F.}~\bibnamefont {Neukart}}, \bibinfo {author} {\bibfnamefont {M.~Y.}\
  \bibnamefont {Niu}}, \bibinfo {author} {\bibfnamefont {T.~E.}\ \bibnamefont
  {O’Brien}}, \bibinfo {author} {\bibfnamefont {B.}~\bibnamefont
  {O’Gorman}}, \bibinfo {author} {\bibfnamefont {E.}~\bibnamefont {Ostby}},
  \bibinfo {author} {\bibfnamefont {A.}~\bibnamefont {Petukhov}}, \bibinfo
  {author} {\bibfnamefont {H.}~\bibnamefont {Putterman}}, \bibinfo {author}
  {\bibfnamefont {C.}~\bibnamefont {Quintana}}, \bibinfo {author}
  {\bibfnamefont {P.}~\bibnamefont {Roushan}}, \bibinfo {author} {\bibfnamefont
  {N.~C.}\ \bibnamefont {Rubin}}, \bibinfo {author} {\bibfnamefont
  {D.}~\bibnamefont {Sank}}, \bibinfo {author} {\bibfnamefont {A.}~\bibnamefont
  {Skolik}}, \bibinfo {author} {\bibfnamefont {V.}~\bibnamefont {Smelyanskiy}},
  \bibinfo {author} {\bibfnamefont {D.}~\bibnamefont {Strain}}, \bibinfo
  {author} {\bibfnamefont {M.}~\bibnamefont {Streif}}, \bibinfo {author}
  {\bibfnamefont {M.}~\bibnamefont {Szalay}}, \bibinfo {author} {\bibfnamefont
  {A.}~\bibnamefont {Vainsencher}}, \bibinfo {author} {\bibfnamefont
  {T.}~\bibnamefont {White}}, \bibinfo {author} {\bibfnamefont {Z.~J.}\
  \bibnamefont {Yao}}, \bibinfo {author} {\bibfnamefont {P.}~\bibnamefont
  {Yeh}}, \bibinfo {author} {\bibfnamefont {A.}~\bibnamefont {Zalcman}},
  \bibinfo {author} {\bibfnamefont {L.}~\bibnamefont {Zhou}}, \bibinfo {author}
  {\bibfnamefont {H.}~\bibnamefont {Neven}}, \bibinfo {author} {\bibfnamefont
  {D.}~\bibnamefont {Bacon}}, \bibinfo {author} {\bibfnamefont
  {E.}~\bibnamefont {Lucero}}, \bibinfo {author} {\bibfnamefont
  {E.}~\bibnamefont {Farhi}},\ and\ \bibinfo {author} {\bibfnamefont
  {R.}~\bibnamefont {Babbush}},\ }\bibfield  {title} {\bibinfo {title} {Quantum
  approximate optimization of non-planar graph problems on a planar
  superconducting processor},\ }\href
  {https://doi.org/10.1038/s41567-020-01105-y} {\bibfield  {journal} {\bibinfo
  {journal} {Nature Physics}\ }\textbf {\bibinfo {volume} {17}},\ \bibinfo
  {pages} {332} (\bibinfo {year} {2021})}\BibitemShut {NoStop}%
\bibitem [{\citenamefont {Bluvstein}\ \emph
  {et~al.}(2022{\natexlab{b}})\citenamefont {Bluvstein}, \citenamefont
  {Levine}, \citenamefont {Semeghini}, \citenamefont {Wang}, \citenamefont
  {Ebadi}, \citenamefont {Kalinowski}, \citenamefont {Keesling}, \citenamefont
  {Maskara}, \citenamefont {Pichler}, \citenamefont {Greiner}, \citenamefont
  {Vuletić},\ and\ \citenamefont {Lukin}}]{bluvstein_quantum_2022}%
  \BibitemOpen
  \bibfield  {author} {\bibinfo {author} {\bibfnamefont {D.}~\bibnamefont
  {Bluvstein}}, \bibinfo {author} {\bibfnamefont {H.}~\bibnamefont {Levine}},
  \bibinfo {author} {\bibfnamefont {G.}~\bibnamefont {Semeghini}}, \bibinfo
  {author} {\bibfnamefont {T.~T.}\ \bibnamefont {Wang}}, \bibinfo {author}
  {\bibfnamefont {S.}~\bibnamefont {Ebadi}}, \bibinfo {author} {\bibfnamefont
  {M.}~\bibnamefont {Kalinowski}}, \bibinfo {author} {\bibfnamefont
  {A.}~\bibnamefont {Keesling}}, \bibinfo {author} {\bibfnamefont
  {N.}~\bibnamefont {Maskara}}, \bibinfo {author} {\bibfnamefont
  {H.}~\bibnamefont {Pichler}}, \bibinfo {author} {\bibfnamefont
  {M.}~\bibnamefont {Greiner}}, \bibinfo {author} {\bibfnamefont
  {V.}~\bibnamefont {Vuletić}},\ and\ \bibinfo {author} {\bibfnamefont
  {M.~D.}\ \bibnamefont {Lukin}},\ }\bibfield  {title} {\bibinfo {title} {A
  quantum processor based on coherent transport of entangled atom arrays},\
  }\href {https://doi.org/10.1038/s41586-022-04592-6} {\bibfield  {journal}
  {\bibinfo  {journal} {Nature}\ }\textbf {\bibinfo {volume} {604}},\ \bibinfo
  {pages} {451} (\bibinfo {year} {2022}{\natexlab{b}})}\BibitemShut {NoStop}%
\bibitem [{\citenamefont {Blekos}\ \emph {et~al.}(2024)\citenamefont {Blekos},
  \citenamefont {Brand}, \citenamefont {Ceschini}, \citenamefont {Chou},
  \citenamefont {Li}, \citenamefont {Pandya},\ and\ \citenamefont
  {Summer}}]{Blekos2024}%
  \BibitemOpen
  \bibfield  {author} {\bibinfo {author} {\bibfnamefont {K.}~\bibnamefont
  {Blekos}}, \bibinfo {author} {\bibfnamefont {D.}~\bibnamefont {Brand}},
  \bibinfo {author} {\bibfnamefont {A.}~\bibnamefont {Ceschini}}, \bibinfo
  {author} {\bibfnamefont {C.-H.}\ \bibnamefont {Chou}}, \bibinfo {author}
  {\bibfnamefont {R.-H.}\ \bibnamefont {Li}}, \bibinfo {author} {\bibfnamefont
  {K.}~\bibnamefont {Pandya}},\ and\ \bibinfo {author} {\bibfnamefont
  {A.}~\bibnamefont {Summer}},\ }\bibfield  {title} {\bibinfo {title} {A review
  on quantum approximate optimization algorithm and its variants},\ }\href
  {https://doi.org/https://doi.org/10.1016/j.physrep.2024.03.002} {\bibfield
  {journal} {\bibinfo  {journal} {Physics Reports}\ }\textbf {\bibinfo {volume}
  {1068}},\ \bibinfo {pages} {1} (\bibinfo {year} {2024})},\ \bibinfo {note} {a
  review on Quantum Approximate Optimization Algorithm and its
  variants}\BibitemShut {NoStop}%
\bibitem [{\citenamefont {Sachdeva}\ \emph {et~al.}(2024)\citenamefont
  {Sachdeva}, \citenamefont {Hartnett}, \citenamefont {Maity}, \citenamefont
  {Marsh}, \citenamefont {Wang}, \citenamefont {Winick}, \citenamefont
  {Dougherty}, \citenamefont {D.}, \citenamefont {Chong}, \citenamefont {Hush},
  \citenamefont {Mundada}, \citenamefont {Bentley}, \citenamefont {Biercuk},\
  and\ \citenamefont {Baum}}]{Sachdeva2024}%
  \BibitemOpen
  \bibfield  {author} {\bibinfo {author} {\bibfnamefont {N.}~\bibnamefont
  {Sachdeva}}, \bibinfo {author} {\bibfnamefont {G.~S.}\ \bibnamefont
  {Hartnett}}, \bibinfo {author} {\bibfnamefont {S.}~\bibnamefont {Maity}},
  \bibinfo {author} {\bibfnamefont {S.}~\bibnamefont {Marsh}}, \bibinfo
  {author} {\bibfnamefont {Y.}~\bibnamefont {Wang}}, \bibinfo {author}
  {\bibfnamefont {A.}~\bibnamefont {Winick}}, \bibinfo {author} {\bibfnamefont
  {R.}~\bibnamefont {Dougherty}}, \bibinfo {author} {\bibfnamefont
  {C.}~\bibnamefont {D.}}, \bibinfo {author} {\bibfnamefont {Y.~Q.}\
  \bibnamefont {Chong}}, \bibinfo {author} {\bibfnamefont {M.}~\bibnamefont
  {Hush}}, \bibinfo {author} {\bibfnamefont {P.~S.}\ \bibnamefont {Mundada}},
  \bibinfo {author} {\bibfnamefont {C.~D.~B.}\ \bibnamefont {Bentley}},
  \bibinfo {author} {\bibfnamefont {M.~J.}\ \bibnamefont {Biercuk}},\ and\
  \bibinfo {author} {\bibfnamefont {Y.}~\bibnamefont {Baum}},\ }\bibfield
  {title} {\bibinfo {title} {Quantum optimization using a 127-qubit gate-model
  ibm quantum computer can outperform quantum annealers for nontrivial binary
  optimization problems},\ }\href {https://doi.org/10.48550/arXiv.2406.01743}
  {\bibfield  {journal} {\bibinfo  {journal} {arXiv:2406.01743}\ } (\bibinfo
  {year} {2024})}\BibitemShut {NoStop}%
\bibitem [{\citenamefont {Pelofske}\ \emph {et~al.}(2024)\citenamefont
  {Pelofske}, \citenamefont {Bärtschi}, \citenamefont {Cincio}, \citenamefont
  {Golden},\ and\ \citenamefont {Eidenbenz}}]{Pelofske2023}%
  \BibitemOpen
  \bibfield  {author} {\bibinfo {author} {\bibfnamefont {E.}~\bibnamefont
  {Pelofske}}, \bibinfo {author} {\bibfnamefont {A.}~\bibnamefont {Bärtschi}},
  \bibinfo {author} {\bibfnamefont {L.}~\bibnamefont {Cincio}}, \bibinfo
  {author} {\bibfnamefont {J.}~\bibnamefont {Golden}},\ and\ \bibinfo {author}
  {\bibfnamefont {S.}~\bibnamefont {Eidenbenz}},\ }\bibfield  {title} {\bibinfo
  {title} {Scaling whole-chip {QAOA} for higher-order ising spin glass models
  on heavy-hex graphs},\ }\href {https://doi.org/10.1038/s41534-024-00906-w}
  {\bibfield  {journal} {\bibinfo  {journal} {npj Quantum Information}\
  }\textbf {\bibinfo {volume} {10}},\ \bibinfo {pages} {109} (\bibinfo {year}
  {2024})}\BibitemShut {NoStop}%
\bibitem [{\citenamefont {Javadi-Abhari}\ \emph {et~al.}(2024)\citenamefont
  {Javadi-Abhari}, \citenamefont {Treinish}, \citenamefont {Krsulich},
  \citenamefont {Wood}, \citenamefont {Lishman}, \citenamefont {Gacon},
  \citenamefont {Martiel}, \citenamefont {Nation}, \citenamefont {Bishop},
  \citenamefont {Cross}, \citenamefont {Johnson},\ and\ \citenamefont
  {Gambetta}}]{javadiabhari2024}%
  \BibitemOpen
  \bibfield  {author} {\bibinfo {author} {\bibfnamefont {A.}~\bibnamefont
  {Javadi-Abhari}}, \bibinfo {author} {\bibfnamefont {M.}~\bibnamefont
  {Treinish}}, \bibinfo {author} {\bibfnamefont {K.}~\bibnamefont {Krsulich}},
  \bibinfo {author} {\bibfnamefont {C.~J.}\ \bibnamefont {Wood}}, \bibinfo
  {author} {\bibfnamefont {J.}~\bibnamefont {Lishman}}, \bibinfo {author}
  {\bibfnamefont {J.}~\bibnamefont {Gacon}}, \bibinfo {author} {\bibfnamefont
  {S.}~\bibnamefont {Martiel}}, \bibinfo {author} {\bibfnamefont {P.~D.}\
  \bibnamefont {Nation}}, \bibinfo {author} {\bibfnamefont {L.~S.}\
  \bibnamefont {Bishop}}, \bibinfo {author} {\bibfnamefont {A.~W.}\
  \bibnamefont {Cross}}, \bibinfo {author} {\bibfnamefont {B.~R.}\ \bibnamefont
  {Johnson}},\ and\ \bibinfo {author} {\bibfnamefont {J.~M.}\ \bibnamefont
  {Gambetta}},\ }\href {https://arxiv.org/abs/2405.08810} {\bibinfo {title}
  {Quantum computing with qiskit}} (\bibinfo {year} {2024}),\ \Eprint
  {https://arxiv.org/abs/2405.08810} {arXiv:2405.08810 [quant-ph]} \BibitemShut
  {NoStop}%
\bibitem [{\citenamefont {Sivarajah}\ \emph {et~al.}(2020)\citenamefont
  {Sivarajah}, \citenamefont {Dilkes}, \citenamefont {Cowtan}, \citenamefont
  {Simmons}, \citenamefont {Edgington},\ and\ \citenamefont
  {Duncan}}]{Sivarajah2020}%
  \BibitemOpen
  \bibfield  {author} {\bibinfo {author} {\bibfnamefont {S.}~\bibnamefont
  {Sivarajah}}, \bibinfo {author} {\bibfnamefont {S.}~\bibnamefont {Dilkes}},
  \bibinfo {author} {\bibfnamefont {A.}~\bibnamefont {Cowtan}}, \bibinfo
  {author} {\bibfnamefont {W.}~\bibnamefont {Simmons}}, \bibinfo {author}
  {\bibfnamefont {A.}~\bibnamefont {Edgington}},\ and\ \bibinfo {author}
  {\bibfnamefont {R.}~\bibnamefont {Duncan}},\ }\bibfield  {title} {\bibinfo
  {title} {t|ket⟩: a retargetable compiler for nisq devices},\ }\href
  {https://doi.org/10.1088/2058-9565/ab8e92} {\bibfield  {journal} {\bibinfo
  {journal} {Quantum Science and Technology}\ }\textbf {\bibinfo {volume}
  {6}},\ \bibinfo {pages} {014003} (\bibinfo {year} {2020})}\BibitemShut
  {NoStop}%
\bibitem [{\citenamefont {Montanez-Barrera}\ \emph {et~al.}(2025)\citenamefont
  {Montanez-Barrera}, \citenamefont {Ji}, \citenamefont {von Spakovsky},
  \citenamefont {Bernal~Neira},\ and\ \citenamefont
  {Michielsen}}]{MontanezBarrera2025}%
  \BibitemOpen
  \bibfield  {author} {\bibinfo {author} {\bibfnamefont {J.~A.}\ \bibnamefont
  {Montanez-Barrera}}, \bibinfo {author} {\bibfnamefont {Y.}~\bibnamefont
  {Ji}}, \bibinfo {author} {\bibfnamefont {M.~R.}\ \bibnamefont {von
  Spakovsky}}, \bibinfo {author} {\bibfnamefont {D.~E.}\ \bibnamefont
  {Bernal~Neira}},\ and\ \bibinfo {author} {\bibfnamefont {K.}~\bibnamefont
  {Michielsen}},\ }\bibfield  {title} {\bibinfo {title} {Optimizing {QAOA}
  circuit transpilation with parity twine and swap network encodings},\ }\href
  {https://doi.org/10.48550/arXiv.2505.17944} {\bibfield  {journal} {\bibinfo
  {journal} {arXiv:2505.17944}\ } (\bibinfo {year} {2025})}\BibitemShut
  {NoStop}%
\bibitem [{\citenamefont {Shor}(1994)}]{Shor1994}%
  \BibitemOpen
  \bibfield  {author} {\bibinfo {author} {\bibfnamefont {P.}~\bibnamefont
  {Shor}},\ }\bibfield  {title} {\bibinfo {title} {Algorithms for quantum
  computation: discrete logarithms and factoring},\ }in\ \href
  {https://doi.org/10.1109/SFCS.1994.365700} {\emph {\bibinfo {booktitle}
  {Proceedings 35th Annual Symposium on Foundations of Computer Science}}}\
  (\bibinfo {year} {1994})\ pp.\ \bibinfo {pages} {124--134}\BibitemShut
  {NoStop}%
\bibitem [{\citenamefont {Kitaev}(1999)}]{Kitaev1999}%
  \BibitemOpen
  \bibfield  {author} {\bibinfo {author} {\bibfnamefont {A.~Y.}\ \bibnamefont
  {Kitaev}},\ }\bibfield  {title} {\bibinfo {title} {Quantum measurements and
  the abelian stabilizer problem},\ }\href
  {https://arxiv.org/abs/quant-ph/9511026} {\bibfield  {journal} {\bibinfo
  {journal} {arXiv:quant-ph/9511026}\ } (\bibinfo {year} {1999})}\BibitemShut
  {NoStop}%
\bibitem [{\citenamefont {Draper}(2000)}]{Draper2000}%
  \BibitemOpen
  \bibfield  {author} {\bibinfo {author} {\bibfnamefont {T.~G.}\ \bibnamefont
  {Draper}},\ }\bibfield  {title} {\bibinfo {title} {Addition on a quantum
  computer},\ }\href {https://arxiv.org/abs/quant-ph/0008033} {\bibfield
  {journal} {\bibinfo  {journal} {arXiv:quant-ph/0008033}\ } (\bibinfo {year}
  {2000})}\BibitemShut {NoStop}%
\bibitem [{\citenamefont {Ruiz-Perez}\ and\ \citenamefont
  {Garcia-Escartin}(2017)}]{Ruiz-Perez2017}%
  \BibitemOpen
  \bibfield  {author} {\bibinfo {author} {\bibfnamefont {L.}~\bibnamefont
  {Ruiz-Perez}}\ and\ \bibinfo {author} {\bibfnamefont {J.~C.}\ \bibnamefont
  {Garcia-Escartin}},\ }\bibfield  {title} {\bibinfo {title} {Quantum
  arithmetic with the quantum fourier transform},\ }\href
  {https://doi.org/10.1007/s11128-017-1603-1} {\bibfield  {journal} {\bibinfo
  {journal} {Quantum Information Processing}\ }\textbf {\bibinfo {volume}
  {16}},\ \bibinfo {pages} {152} (\bibinfo {year} {2017})}\BibitemShut
  {NoStop}%
\bibitem [{\citenamefont {Harrow}\ \emph {et~al.}(2009)\citenamefont {Harrow},
  \citenamefont {Hassidim},\ and\ \citenamefont {Lloyd}}]{Harrow2009}%
  \BibitemOpen
  \bibfield  {author} {\bibinfo {author} {\bibfnamefont {A.~W.}\ \bibnamefont
  {Harrow}}, \bibinfo {author} {\bibfnamefont {A.}~\bibnamefont {Hassidim}},\
  and\ \bibinfo {author} {\bibfnamefont {S.}~\bibnamefont {Lloyd}},\ }\bibfield
   {title} {\bibinfo {title} {Quantum algorithm for linear systems of
  equations},\ }\href {https://doi.org/10.1103/PhysRevLett.103.150502}
  {\bibfield  {journal} {\bibinfo  {journal} {Phys. Rev. Lett.}\ }\textbf
  {\bibinfo {volume} {103}},\ \bibinfo {pages} {150502} (\bibinfo {year}
  {2009})}\BibitemShut {NoStop}%
\bibitem [{\citenamefont {Maslov}(2007)}]{Maslov2007}%
  \BibitemOpen
  \bibfield  {author} {\bibinfo {author} {\bibfnamefont {D.}~\bibnamefont
  {Maslov}},\ }\bibfield  {title} {\bibinfo {title} {Linear depth stabilizer
  and quantum fourier transformation circuits with no auxiliary qubits in
  finite-neighbor quantum architectures},\ }\href
  {https://doi.org/10.1103/PhysRevA.76.052310} {\bibfield  {journal} {\bibinfo
  {journal} {Phys. Rev. A}\ }\textbf {\bibinfo {volume} {76}},\ \bibinfo
  {pages} {052310} (\bibinfo {year} {2007})}\BibitemShut {NoStop}%
\bibitem [{\citenamefont {Jin}\ \emph {et~al.}(2023)\citenamefont {Jin},
  \citenamefont {Gao}, \citenamefont {Guo}, \citenamefont {Chen}, \citenamefont
  {Hua}, \citenamefont {Zhang},\ and\ \citenamefont {Zhang}}]{Jin2023}%
  \BibitemOpen
  \bibfield  {author} {\bibinfo {author} {\bibfnamefont {Y.}~\bibnamefont
  {Jin}}, \bibinfo {author} {\bibfnamefont {X.}~\bibnamefont {Gao}}, \bibinfo
  {author} {\bibfnamefont {M.}~\bibnamefont {Guo}}, \bibinfo {author}
  {\bibfnamefont {H.}~\bibnamefont {Chen}}, \bibinfo {author} {\bibfnamefont
  {F.}~\bibnamefont {Hua}}, \bibinfo {author} {\bibfnamefont {C.}~\bibnamefont
  {Zhang}},\ and\ \bibinfo {author} {\bibfnamefont {E.~Z.}\ \bibnamefont
  {Zhang}},\ }\bibfield  {title} {\bibinfo {title} {Quantum fourier
  transformation circuits compilation},\ }\href
  {https://arxiv.org/abs/2312.16114} {\bibfield  {journal} {\bibinfo  {journal}
  {arXiv:2312.16114}\ } (\bibinfo {year} {2023})}\BibitemShut {NoStop}%
\bibitem [{\citenamefont {Gao}\ \emph {et~al.}(2024)\citenamefont {Gao},
  \citenamefont {Jin}, \citenamefont {Guo}, \citenamefont {Chen},\ and\
  \citenamefont {Zhang}}]{Gao2024}%
  \BibitemOpen
  \bibfield  {author} {\bibinfo {author} {\bibfnamefont {X.}~\bibnamefont
  {Gao}}, \bibinfo {author} {\bibfnamefont {Y.}~\bibnamefont {Jin}}, \bibinfo
  {author} {\bibfnamefont {M.}~\bibnamefont {Guo}}, \bibinfo {author}
  {\bibfnamefont {H.}~\bibnamefont {Chen}},\ and\ \bibinfo {author}
  {\bibfnamefont {E.~Z.}\ \bibnamefont {Zhang}},\ }\href@noop {} {\bibinfo
  {title} {Linear depth qft over ibm heavy-hex architecture}} (\bibinfo {year}
  {2024}),\ \Eprint {https://arxiv.org/abs/2402.09705} {arXiv:2402.09705
  [quant-ph]} \BibitemShut {NoStop}%
\bibitem [{\citenamefont {B\"aumer}\ \emph {et~al.}(2024)\citenamefont
  {B\"aumer}, \citenamefont {Tripathi}, \citenamefont {Seif}, \citenamefont
  {Lidar},\ and\ \citenamefont {Wang}}]{Baumer2024}%
  \BibitemOpen
  \bibfield  {author} {\bibinfo {author} {\bibfnamefont {E.}~\bibnamefont
  {B\"aumer}}, \bibinfo {author} {\bibfnamefont {V.}~\bibnamefont {Tripathi}},
  \bibinfo {author} {\bibfnamefont {A.}~\bibnamefont {Seif}}, \bibinfo {author}
  {\bibfnamefont {D.}~\bibnamefont {Lidar}},\ and\ \bibinfo {author}
  {\bibfnamefont {D.~S.}\ \bibnamefont {Wang}},\ }\bibfield  {title} {\bibinfo
  {title} {Quantum fourier transform using dynamic circuits},\ }\href
  {https://doi.org/10.1103/PhysRevLett.133.150602} {\bibfield  {journal}
  {\bibinfo  {journal} {Phys. Rev. Lett.}\ }\textbf {\bibinfo {volume} {133}},\
  \bibinfo {pages} {150602} (\bibinfo {year} {2024})}\BibitemShut {NoStop}%
\bibitem [{\citenamefont {Takahashi}\ \emph {et~al.}(2007)\citenamefont
  {Takahashi}, \citenamefont {Kunihiro},\ and\ \citenamefont
  {Ohta}}]{Takahashi2007}%
  \BibitemOpen
  \bibfield  {author} {\bibinfo {author} {\bibfnamefont {Y.}~\bibnamefont
  {Takahashi}}, \bibinfo {author} {\bibfnamefont {N.}~\bibnamefont
  {Kunihiro}},\ and\ \bibinfo {author} {\bibfnamefont {K.}~\bibnamefont
  {Ohta}},\ }\bibfield  {title} {\bibinfo {title} {The quantum fourier
  transform on a linear nearest neighbor architecture},\ }\href
  {https://doi.org/10.26421/QIC7.4-7} {\bibfield  {journal} {\bibinfo
  {journal} {Quantum Information \& Computation}\ }\textbf {\bibinfo {volume}
  {7}},\ \bibinfo {pages} {383} (\bibinfo {year} {2007})}\BibitemShut {NoStop}%
\bibitem [{\citenamefont {Childs}\ and\ \citenamefont
  {Wiebe}(2012)}]{Childs2012}%
  \BibitemOpen
  \bibfield  {author} {\bibinfo {author} {\bibfnamefont {A.~M.}\ \bibnamefont
  {Childs}}\ and\ \bibinfo {author} {\bibfnamefont {N.}~\bibnamefont {Wiebe}},\
  }\bibfield  {title} {\bibinfo {title} {Hamiltonian simulation using linear
  combinations of unitary operations},\ }\href
  {https://doi.org/https://doi.org/10.26421/QIC12.11-12} {\bibfield  {journal}
  {\bibinfo  {journal} {Quantum Information and Computation}\ }\textbf
  {\bibinfo {volume} {12}},\ \bibinfo {pages} {0901} (\bibinfo {year}
  {2012})}\BibitemShut {NoStop}%
\bibitem [{\citenamefont {Low}\ and\ \citenamefont {Chuang}(2017)}]{Low2017}%
  \BibitemOpen
  \bibfield  {author} {\bibinfo {author} {\bibfnamefont {G.~H.}\ \bibnamefont
  {Low}}\ and\ \bibinfo {author} {\bibfnamefont {I.~L.}\ \bibnamefont
  {Chuang}},\ }\bibfield  {title} {\bibinfo {title} {Optimal hamiltonian
  simulation by quantum signal processing},\ }\href
  {https://doi.org/10.1103/PhysRevLett.118.010501} {\bibfield  {journal}
  {\bibinfo  {journal} {Phys. Rev. Lett.}\ }\textbf {\bibinfo {volume} {118}},\
  \bibinfo {pages} {010501} (\bibinfo {year} {2017})}\BibitemShut {NoStop}%
\bibitem [{\citenamefont {Suzuki}(1976)}]{Suzuki1976}%
  \BibitemOpen
  \bibfield  {author} {\bibinfo {author} {\bibfnamefont {M.}~\bibnamefont
  {Suzuki}},\ }\bibfield  {title} {\bibinfo {title} {General decomposition
  theory of ordered operator exponentials},\ }\href@noop {} {\bibfield
  {journal} {\bibinfo  {journal} {Communications in Mathematical Physics}\
  }\textbf {\bibinfo {volume} {51}},\ \bibinfo {pages} {183–190} (\bibinfo
  {year} {1976})}\BibitemShut {NoStop}%
\bibitem [{\citenamefont {Zeng}\ \emph {et~al.}(2025)\citenamefont {Zeng},
  \citenamefont {Sun}, \citenamefont {Jiang},\ and\ \citenamefont
  {Zhao}}]{Zeng2025}%
  \BibitemOpen
  \bibfield  {author} {\bibinfo {author} {\bibfnamefont {P.}~\bibnamefont
  {Zeng}}, \bibinfo {author} {\bibfnamefont {J.}~\bibnamefont {Sun}}, \bibinfo
  {author} {\bibfnamefont {L.}~\bibnamefont {Jiang}},\ and\ \bibinfo {author}
  {\bibfnamefont {Q.}~\bibnamefont {Zhao}},\ }\bibfield  {title} {\bibinfo
  {title} {Simple and high-precision hamiltonian simulation by compensating
  trotter error with linear combination of unitary operations},\ }\href
  {https://doi.org/10.1103/PRXQuantum.6.010359} {\bibfield  {journal} {\bibinfo
   {journal} {PRX Quantum}\ }\textbf {\bibinfo {volume} {6}},\ \bibinfo {pages}
  {010359} (\bibinfo {year} {2025})}\BibitemShut {NoStop}%
\bibitem [{\citenamefont {Suzuki}(1991)}]{suzukiGeneralTheoryFractal1991a}%
  \BibitemOpen
  \bibfield  {author} {\bibinfo {author} {\bibfnamefont {M.}~\bibnamefont
  {Suzuki}},\ }\bibfield  {title} {\bibinfo {title} {General theory of fractal
  path integrals with applications to many-body theories and statistical
  physics},\ }\href {https://doi.org/10.1063/1.529425} {\bibfield  {journal}
  {\bibinfo  {journal} {Journal of Mathematical Physics}\ }\textbf {\bibinfo
  {volume} {32}},\ \bibinfo {pages} {400} (\bibinfo {year} {1991})}\BibitemShut
  {NoStop}%
\bibitem [{\citenamefont {Childs}\ \emph {et~al.}(2021)\citenamefont {Childs},
  \citenamefont {Su}, \citenamefont {Tran}, \citenamefont {Wiebe},\ and\
  \citenamefont {Zhu}}]{childsTheoryTrotterError2021}%
  \BibitemOpen
  \bibfield  {author} {\bibinfo {author} {\bibfnamefont {A.~M.}\ \bibnamefont
  {Childs}}, \bibinfo {author} {\bibfnamefont {Y.}~\bibnamefont {Su}}, \bibinfo
  {author} {\bibfnamefont {M.~C.}\ \bibnamefont {Tran}}, \bibinfo {author}
  {\bibfnamefont {N.}~\bibnamefont {Wiebe}},\ and\ \bibinfo {author}
  {\bibfnamefont {S.}~\bibnamefont {Zhu}},\ }\bibfield  {title} {\bibinfo
  {title} {Theory of {{Trotter Error}} with {{Commutator Scaling}}},\ }\href
  {https://doi.org/10.1103/PhysRevX.11.011020} {\bibfield  {journal} {\bibinfo
  {journal} {Physical Review X}\ }\textbf {\bibinfo {volume} {11}},\ \bibinfo
  {pages} {011020} (\bibinfo {year} {2021})}\BibitemShut {NoStop}%
\bibitem [{\citenamefont {Liu}\ \emph {et~al.}(2022)\citenamefont {Liu},
  \citenamefont {Barnes},\ and\ \citenamefont {Economou}}]{Economou2022}%
  \BibitemOpen
  \bibfield  {author} {\bibinfo {author} {\bibfnamefont {C.}~\bibnamefont
  {Liu}}, \bibinfo {author} {\bibfnamefont {E.}~\bibnamefont {Barnes}},\ and\
  \bibinfo {author} {\bibfnamefont {S.}~\bibnamefont {Economou}},\ }\bibfield
  {title} {\bibinfo {title} {Proposal for generating complex microwave graph
  states using superconducting circuits},\ }\href@noop {} {\bibfield  {journal}
  {\bibinfo  {journal} {arXiv:2201.00836}\ } (\bibinfo {year}
  {2022})}\BibitemShut {NoStop}%
\bibitem [{\citenamefont {Fowler}\ \emph {et~al.}(2004)\citenamefont {Fowler},
  \citenamefont {Devitt},\ and\ \citenamefont {Hollenberg}}]{Fowler2004}%
  \BibitemOpen
  \bibfield  {author} {\bibinfo {author} {\bibfnamefont {A.~G.}\ \bibnamefont
  {Fowler}}, \bibinfo {author} {\bibfnamefont {S.~J.}\ \bibnamefont {Devitt}},\
  and\ \bibinfo {author} {\bibfnamefont {L.~C.~L.}\ \bibnamefont
  {Hollenberg}},\ }\bibfield  {title} {\bibinfo {title} {Implementation of
  shor's algorithm on a linear nearest neighbour qubit array},\ }\href@noop {}
  {\bibfield  {journal} {\bibinfo  {journal} {Quantum Info. Comput.}\ }\textbf
  {\bibinfo {volume} {4}},\ \bibinfo {pages} {237–251} (\bibinfo {year}
  {2004})}\BibitemShut {NoStop}%
\end{thebibliography}
\end{document}